%% file: main.tex
\documentclass[11pt]{article}
\usepackage[utf8]{inputenc}
\usepackage[english]{babel}
\usepackage{csquotes} 

\usepackage[
  backend=biber,
  style=numeric,   % common in math; gives labels like [Smi23]
  url=false,          % suppress URLs
  doi=false,          % suppress DOIs
  isbn=false,         % suppress ISBNs
  giveninits=true,    % turns "John" into "J."
  maxnames=99,        % never truncate author lists with "et al."
  sorting=nty,        % nyt -> sort by name → year → title
  sortcites=true
]{biblatex}
\AtEveryBibitem{\clearfield{issn}}   % suppress ISSNs
\DeclareFieldFormat{eprint:arxiv}{%
  \href{https://arxiv.org/abs/#1}{arXiv\addcolon\space#1}%
}
\DeclareFieldFormat{eprint:stackexchange}{%
  \href{#1}{#1}%
}
\renewbibmacro{in:}{}

\usepackage{authblk}
\usepackage{graphicx}
\usepackage{amsmath,amssymb,amsthm}
\usepackage[dvipsnames]{xcolor}
\usepackage{xfrac}
\usepackage{parskip} % careful with some document classes, but OK here
\usepackage{fancyhdr}
\usepackage{vmargin}
\usepackage{listings}
\usepackage{dsfont}
\usepackage{tabu}
\usepackage{enumitem}
\usepackage{pdflscape}
\usepackage{systeme}
\usepackage{stmaryrd}
\usepackage{physics}
\usepackage{pdfpages}
\usepackage{float}
\usepackage{tikz-cd}
\usepackage{mathrsfs}
\usepackage{tensor}
\usepackage{upgreek}
\usepackage{tikzit}
\setmarginsrb{3 cm}{0 cm}{3 cm}{2.5 cm}{1 cm}{1.5 cm}{1 cm}{1.5 cm}
\usepackage{caption}
\usepackage{subcaption}
\usepackage[colorlinks,citecolor={blue},urlcolor={blue}]{hyperref}
\usepackage{varioref}
\usepackage[capitalise]{cleveref}
\usepackage{array}
\usepackage{makecell}

\usetikzlibrary{patterns.meta}

\newcommand{\bb}[1]{\mathbb{#1}}
\newcommand{\Grad}[1]{\textup{grad}_{#1}\,}

\newtheoremstyle{break}
{\topsep}{\topsep}% 
{\itshape}{}% 
{\bfseries}{}% 
\theoremstyle{break} 
\newtheorem{theorem}{Theorem}[section] 
\newtheorem{definition}[theorem]{Definition} 
\newtheorem{corollary}[theorem]{Corollary} 
\newtheorem{lemma}[theorem]{Lemma} 
\newtheorem{proposition}[theorem]{Proposition} 
\newtheorem{remark}[theorem]{Remark} 
\newtheorem{notation}[theorem]{Notation} 
 
\newtheorem{assumption}[theorem]{Assumption} 
 
\newtheorem*{theorem-non}{Theorem} 
\newtheorem*{corollary-non}{Corollary} 
\newtheorem*{conjecture-non}{Conjecture}
\newtheorem*{lemma-non}{Lemma} 

\crefname{assumption}{Assumption}{Assumptions} 
\crefname{notation}{notation}{notations}

\allowdisplaybreaks
\numberwithin{equation}{section}
\makeatletter
\renewcommand\@maketitle{%
  \newpage
  \null
  \begin{center}%
    {\LARGE \@title \par}%
    \vskip 1em%
    {\large
      \lineskip .5em%
      \begin{tabular}[t]{c}%
        \@author
      \end{tabular}\par}%
  \end{center}%
  \par
  \vskip 1.5em}
\makeatother

\crefformat{equation}{Eq.~(#2#1#3)} %So that the eq numbers appear in italics whenever needed. 

\renewcommand{\thefootnote}{\fnsymbol{footnote}}

\title{Rapid mixing for Gibbs measures in Riemannian manifolds}

\author[1,2]{Ángela Capel}
\author[3]{Marco Castrillón López}
\author[4]{Sofyan Iblisdir}
\author[5]{Angelo Lucia}
\author[6,7]{Pablo~Páez Velasco\thanks{pablopaez@ucm.es}}
\author[6,7]{David Pérez García}

\affil[1]{\textit{Department of Applied Mathematics and Theoretical Physics, University of Cambridge, Wilberforce Road, Cambridge, CB3 0WA, United Kingdom}}
\affil[2]{\textit{Fachbereich Mathematik, Universität Tübingen, 72076 Tübingen, Germany}}
\affil[3]{\textit{Departamento de Álgebra, Geometría y Topología, Facultad de Matematicas, Universidad Complutense de Madrid, 28040 Madrid, Spain}}
\affil[4]{\textit{Departament de Física Quàntica i Astrofísica \& Institut de Ciències del Cosmos,
Universitat de Barcelona, 08028 Barcelona, Spain}}
\affil[5]{\textit{Dipartimento di Matematica, Politecnico di Milano, 20133 Milano, Italy}}
\affil[6]{\textit{Departamento de An\'{a}lisis Matemático y Matemática Aplicada, Facultad de Matemáticas, Universidad Complutense de Madrid, 28040 Madrid, Spain}}
\affil[7]{\textit{Instituto de Ciencias Matemáticas, 28049 Madrid, Spain}}

\begin{document}
\maketitle
\renewcommand{\thefootnote}{\arabic{footnote}}
\setcounter{footnote}{0}

\begin{abstract}
Langevin dynamics on Riemannian manifolds is analyzed. Conditions ensuring the existence of a suitable logarithmic Sobolev inequality (rapid mixing to the Gibbs distribution) are identified. These conditions involve the curvature of the manifold, the inverse temperature, escaping directions from saddle points, and exclude barren plateaus and spurious local minima. We show that when these conditions are met, mixing times polynomial in the dimension of the manifold are achievable. This result is obtained through a relation between Langevin processes in the domain and in the image of a Riemannian submersion. Such a relation can be of independent interest. 
\end{abstract}
\newpage

\hypersetup{linkcolor=black}
\tableofcontents
\hypersetup{linkcolor=red}

\newpage 

\input{TikzFigures/preamble}

\input{Chapters/Intro}
\newpage
\input{Chapters/PoincareIneq}
\newpage
\input{Chapters/Lifting_Poincare}
\newpage
\input{Chapters/Upgrading}
\newpage 
\input{Chapters/ProofThm}
\newpage
\input{Chapters/Suboptimality}
\newpage
\input{Chapters/Examples}
\newpage

\section*{Acknowledgements}
The first version of this manuscript (\href{https://arxiv.org/pdf/2606.13453v1}{arXiv:2606.13453v1}) was developed entirely without the use of an LLM. We acknowledge use of ChatGPT Astra 6.0 to proofread the first version of the manuscript, which allowed us to correct several typographical errors. This revision also identified a mistake in our proof strategy, which we corrected by adding \cref{assumptionsym}. The revision allowed us to simplify the construction of the second Lyapunov function via the use of an explicit expression, suggested by ChatGPT, using the same reasoning and proof strategy of the first version. P.\,P.\,V. would like to thank Pablo Hidalgo Palencia for many fruitful conversations throughout the development of this work, and Tomasz Kania for generously providing us with the argument for the escaping time estimates of the generalized CIR process used in the first version. A.\,C. acknowledges the support of the Deutsche Forschungsgemeinschaft (DFG, German Research Foundation) - Project-ID 470903074 - TRR 352. This project was funded within the QuantERA II Programme which has received funding from the EU’s H2020 research and innovation programme under the GA No 101017733. M.\,C.\,L. was supported by  Project PID2024-156578NB-I00 funded by MICIU /AEI /10.13039/501100011033 / FEDER, EU. S.\,I. acknowledges financial support from the “Center of Excellence Maria de Maeztu 2025-2029” award to the Institute of Cosmos Sciences, grant CEX2024-001451-M, funded by MICIU/AEI/10.13039/501100011033. A.\,L. acknowledges support from the Italian Ministry of University and Research (MUR), through ``Programma per Giovani Ricercatori Rita Levi Montalcini'', the grant ``Dipartimento di Eccellenza 2023-2027'' of Dipartimento di Matematica, Politecnico di Milano, as well as the National Group of Mathematical Physics (GNFM) of the Italian Institute for High Mathematics (INdAM). P.\,P.\,V. acknowledges support of the Spanish Ministry of Science and Innovation MCIN/AEI/10.13039/501100011033 (CEX2023-001347-S,CEX2019-000904-S, CEX2019-000904-S-21-2). This work has been funded by the Spanish Ministry of Science, Innovation and Universities MICIU/AEI/10.13039/501100011033 (CEX2023-001347-S, PID2023-146758NB- I00), Comunidad de Madrid (TEC-2024/COM-84-QUITEMAD-CM), Universidad Complutense de Madrid (FEI-EU-22-06),  and the Ministry for Digital Transformation and of Civil Service of the Spanish Government through the QUANTUM ENIA project call - Quantum Spain project, and by the European Union through the Recovery, Transformation and Resilience Plan - NextGenerationEU within the framework of the Digital Spain 2026 Agenda. 

%\bibliographystyle{unsrtnat} % We choose the "plain" reference style
%\bibliography{bibliografia} % Entries are in
%\nocite{*}
\printbibliography[heading=bibintoc]

\newpage 

\appendix 
\input{Chapters/Geometry}
\newpage
\input{Chapters/Sobolev_Spaces}
\newpage 
\input{Chapters/BakrEmery}
\newpage
%\input{Chapters/Misc}
%\newpage
%\input{Chapters/PDEs}
%\newpage
\input{Chapters/Laplacian}
\end{document}

%% file: TikzFigures/preamble.tex
\definecolor{myred}{HTML}{E53935}
\definecolor{myblue}{HTML}{1E88E5}
\definecolor{mygreen}{HTML}{43A047}
\definecolor{myyellow}{HTML}{FDD835}
\definecolor{myorange}{HTML}{FB8C00}
\definecolor{mygold}{HTML}{F9A825}
\definecolor{mypurple}{HTML}{8E24AA}
\definecolor{mygray}{HTML}{BDBDBD}
\definecolor{mybrown}{HTML}{6D4C41}
\definecolor{mynavy}{HTML}{1A237E}
\definecolor{mypink}{HTML}{ffbfca}
\definecolor{myseagreen}{HTML}{26A69A}
\definecolor{myviolet}{HTML}{f07ef0}
\definecolor{mydarkblue}{HTML}{0D47A1}
\definecolor{mydarkcyan}{HTML}{E0FFFF}
\definecolor{darkgray}{rgb}{0.66, 0.66, 0.66}
\definecolor{mydarkgreen}{HTML}{1B5E20}
\definecolor{mydarkmagenta}{HTML}{AD1457}
\definecolor{mydarkorange}{HTML}{EF6C00}
\definecolor{lightblue}{rgb}{0.68, 0.85, 0.9}
\definecolor{lightcyan}{rgb}{0.88, 1.0, 1.0}
\definecolor{lightgray}{rgb}{0.83, 0.83, 0.83}
\definecolor{mylightgreen}{HTML}{81C784}
\definecolor{lightyellow}{rgb}{1.0, 1.0, 0.88}
\definecolor{myshadow}{rgb}{0.5, 0.5, 0.5}
\definecolor{pink}{rgb}{1.0, 0.75, 0.8}
\definecolor{violet}{rgb}{0.93, 0.51, 0.93}
\definecolor{myauxcolor}{RGB}{245, 255, 255}
\definecolor{mylightgreen}{RGB}{193, 225, 159}
\tikzstyle{heavier} = [line width=0.8pt]

%% file: Chapters/Intro.tex
\section{Introduction}
\label{sec:intro}

Given some sufficiently smooth function $F : M \to \bb{R}$ on a Riemannian manifold $(M, g)$, and some fixed positive constant $\beta > 0$---usually referred to as the inverse temperature---the problem of sampling with respect to the associated Gibbs distribution 
\[
\nu(x) := \frac{1}{Z} e^{-\beta F(x)},
\]
where $Z$ is the normalizing constant of $\nu$, $Z := \int_M e^{-\beta F} \, \textup{dVol}_g$, has played an essential role in several areas. For example, in lattice gauge theory, where $M$ represents the set of states a system can be in, and $F$ its energy, solutions to this problem are of paramount importance; without the ability to sample distributions of the above form we would not be able to predict masses of elementary particles or the nature of phase transitions in high-energy physics \cite{Creutz_2023}. The problem of sampling $\nu$ is also closely related to that of minimizing functions with constraints, which is crucial in numerical analysis \cite{boumal2022intromanifolds,absil2008optimization}. Gibbs sampling also has applications in differential privacy, the gold standard of privacy in machine learning \cite{diffpriv1,diffpriv2,diffpriv3}.

More often than not, the distribution $\nu$ cannot be sampled directly. For this reason, it is common to aim for approximate sampling through the use of some Markovian process whose fixed point is the target distribution $\nu$, see \cite{brooks2011handbook,gelman1995bayesian,mackay2003information,robert2004monte}. Markov chain Monte Carlo (MCMC) methods for sampling from Gibbs distributions in $\bb{R}^n$ have been extensively studied (see \cite{lau2022non} and the references therein). More recently, analogous methods have also been developed for Riemannian manifolds. For overviews on the existing literature on MCMC algorithms on manifolds, see, for example \cite[Section 3]{liu2022geometry}, \cite[Section 2.3]{durmus2026geodesic}, and \cite[Section 1.2]{goyal2019sampling}. 

This work focuses on a continuous-time Markov process, known as the Langevin diffusion (or Langevin dynamics), a stochastic process $X_t$ on $M$ that combines gradient descent and Brownian motion\footnote{Although discrete-time stochastic processes are more relevant to computer simulations, to be able to mobilize the conceptual tools that will lead to a proof of our main results, we have chosen to work in continuous time. Showing that a discretization preserves the convergence properties of the continuous-time Langevin dynamics is a non-trivial problem. In a setting similar to ours, this has been established in \cite{cheng2022efficient} for the Geometric Euler-Maruyama discretization.}. This process is known to have $\nu$ as its limiting distribution---i.e. for any initial condition, its probability density at time $t$, $\rho_t$, converges to $\nu$ as $t$ tends to infinity \cite{bakry2013analysis}. The Langevin diffusion and its discretizations have been extensively studied both in $\bb{R}^n$ and on manifolds.

The fact that the Langevin diffusion has $\nu$ as its stationary distribution is only part of the picture. From both a theoretical and practical perspective, it is equally important to understand the rate at which the process converges to equilibrium. Establishing quantitative convergence rates is considerably more challenging than proving convergence itself. In the simple case where $M$ is the Euclidean space, this issue has been studied extensively \cite{menz2014,chewi2025analysis,vempala2019rapid,wibisono2019proximal,gentiloni2026exponential}. In this setting, the relation between convergence rates, the dimension of $M$, the properties of $F$ and the value of $\beta$ is considerably well understood. The picture is much less clear for general Riemannian manifolds.

Some quantitative convergence results do exist for Gibbs sampling on Riemannian manifolds. For instance, \cite{moitra2020fast} establishes fast convergence guarantees for the Langevin dynamics, although only in a neighborhood of local minima, with applications to noisy matrix factorization, matrix sensing, and matrix completion. \cite{cheng2022efficient} bounds the discretization error of the continuous Langevin dynamics and establishes convergence guarantees for the resulting algorithm, although with rates that can be arbitrarily slow for non-convex functions. Under convexity assumptions, \cite{li2022mirror} obtains fast convergence guarantees for a discretization of the Langevin dynamics. Beyond Langevin-based methods, \cite{lee2018convergence} derives convergence rates for the discrete Riemannian Hamiltonian Monte Carlo algorithm, while \cite{goyal2019sampling} establishes fast convergence guarantees for the geodesic walk to sample from the uniform measure, and provides convergence guarantees for Gibbs sampling associated with convex functions. In a different direction, \cite{cai2026onlineriemanniangradientdescent} studies optimization rather than sampling, establishing local convergence guarantees for a Riemannian gradient descent algorithm in the context of tensor-network methods for quantum state tomography. 

In order to prove \textit{rapid} convergence for the density of the Langevin dynamics towards its stationary density, a classic approach is to derive a suitable \textit{logarithmic Sobolev inequality} (log-Sobolev for short), which in turn provides exponentially-decaying upper bounds for the total variation distance between the densities $\rho_t$ and $\nu$ \cite{bakry2013analysis}. Moreover, a log-Sobolev inequality also implies fast convergence of suitable discretizations of the Langevin dynamics \cite{wang2020fast,gatmiry2022convergenceriemannianlangevinalgorithm}.

In a recent work by M. B. Li and M. A. Erdogdu \cite{LiErd2022}, the Langevin diffusion process $X_t$ on the product of $n$-dimensional spheres was analyzed.  Under suitable assumptions on the function $F$, they establish a log-Sobolev inequality. In this work, we show that qualitatively similar results also hold for a much broader class of Riemannian manifolds, including those studied in \cite{LiErd2022}. Our construction of a log-Sobolev inequality follows their general strategy, but requires new techniques and ideas\footnote{Furthermore, in our approach, we avoid some non-verified assumptions used in some of the proofs of \cite[Appendix D]{LiErd2022Supp}.}.

Note that one may distinguish two regimes of the Langevin diffusion, in which $X_t$ exhibits different behaviors; namely when $\frac{1}{\beta} \ll 1$ and when $\frac{1}{\beta} \gg 1$---which can be understood as the \textit{low-temperature} and \textit{high-temperature} regimes. In the former, the evolution of $X_t$ is mainly dictated by the gradient of $F$, whereas in the latter, it is mainly noise. In this work, we will aim at proving log-Sobolev inequalities valid in the low-temperature regime. As in \cite{LiErd2022}, the major obstacle to overcome is ensuring that $X_t$ is able to \textit{escape} the saddle points of $F$ \textit{rapidly}. 

The high-temperature regime is interesting too \cite{cao2025dynamicalapproacharealaw}, but obtaining a log-Sobolev inequality is then simple whenever the Ricci curvature of $(M, g)$ is positive, i.e. $\textup{Ric}_g \geq \gamma g$, for some $\gamma > 0$. In this case, for sufficiently large values of $\frac{1}{\beta}$, there exists some constant $\kappa > 0$ such that
\[
\nabla^2F +\frac{1}{\beta} \textup{Ric}_g \geq \kappa g.
\]
This inequality, which is known as a \textit{curvature-dimension condition} (cf. \cref{BakryEmeryandLyapunov}) implies that $M$ satisfies a log-Sobolev inequality with the same constant by classic Bakry-Émery theory \cite[Section 5.7]{bakry2013analysis}. To be more precise, the condition 
\[
\frac{1}{\beta} \in \Omega\Big(\frac{|\lambda_{\min}(\nabla^2 F)|}{\gamma} \Big),
\]
is sufficient to obtain the desired log-Sobolev inequality. 

\subsection{Main results}

As mentioned earlier, this work is a study of rapid mixing for the Langevin diffusion on a Riemannian manifold in the low-temperature regime. It turns out that several technical difficulties arise when the global minimum of $F$ is not unique. We will see, however, that when the multiplicity of global minima is only caused by a symmetry of $F$, the situation improves. Such symmetries are very common in physics, as can be appreciated in the literature on lattice gauge theory or tensor networks \cite{RevModPhys,montvay1994quantum}. The study of lattice gauge theories and sampling and optimizing tensor networks are two potential fields of application of the results of this paper; first steps in this direction are taken in \cite{ourwork}. 

Informally, the possible symmetries of $F$ will be factored out of the original manifold $M$. On the resulting quotient manifold, the minimum will be unique. More precisely, let $(M, g)$ be a compact Riemannian manifold, and let $F: M \to \bb{R}$ be some smooth function on $M$. We assume that the symmetries of $F$ can be described via a \textit{suitable} action of some Lie group $G$ on $M$, i.e. there exists some compact and connected Lie group $G$ acting on $M$ and such that 
\[
F(gx) = F(x),
\]
for every $x \in M$ and every $g \in G$. If we consider the orbit space $M/G$, and the projection $\pi: M \to M/G$, we can define a function $\tilde F$ on $M/G$ such that $F = \tilde F \circ \pi$. Assuming that $\tilde F$ has a unique global minimum, our strategy will consist in studying Langevin dynamics on $M/G$, obtain a Poincaré inequality, and use it to lift it and upgrade it to a log-Sobolev inequality for $\nu$ in $M$. A key aspect of our derivation will be to endow $M/G$ with a Riemannian metric $h$ so as to make $\pi$ a \textit{Riemannian submersion} (cf. \cref{chapterRiemannianSubmersions}). We will study in parallel the Langevin dynamics on $(M, g)$ and $(M/G, h)$ associated with $F$ and $\tilde F$ respectively, i.e. those satisfying the formal\footnote{Although the expressions of the SDEs are non-rigorous, they provide a clear insight to the dynamics of the processes $X_t$ and $\tilde X_t$. Furthermore, they show the connection to the usual gradient descent. Nevertheless, these SDEs will not be used in this work. We will only use the definition provided in \cref{martingaleproblem}. See \cite[Theorem 1.3.6]{hsustochastic}.} SDEs 
\begin{equation}
\label{eq:defprocesos}
dX_t = -\textup{grad}_g\, F(X_t) dt + \sqrt{\frac{2}{\beta}} dW_t, \quad \textup{and} \quad d\tilde X_t = -\textup{grad}_h\, \tilde F(\tilde X_t) dt + \sqrt{\frac{2}{\beta}} d\tilde W_t,
\end{equation}
where $W_t$ and $\tilde{W}_t$ are Brownian motions in $M$ and $M/G$. Such pairs of processes have attracted considerable attention in recent years \cite{xuquotient,menon2024geometrydeeplinearnetwork,menon2026implicitregularizationlangevindynamics,HIM23}.

To illustrate the foregoing discussion with an example, let us consider the function
\[
F(X) = \frac{\Tr(X^\dagger A X)}{\Tr(X^\dagger BX)},
\]
where $X$ is an $n \times m$ complex matrix, with $n \geq m$ fixed, such that $X^\dagger X = \mathds{1}_m$, $A, B$ square Hermitian matrices, and $B$ positive-definite. Finding the minimum of this function is a well-known problem, known as trace-ratio minimization, relevant in principal component analysis \cite{jolliffe2011principal} or in the issue of graph embedding \cite{yan2005graph}. $F$ has a symmetry: $F(XU) = F(X)$ for every unitary matrix $U$ of dimension $m$. Nevertheless, considering the action of the unitary group $\textup{U}(m)$ on the space of $n \times m$ complex matrices satisfying $X^\dagger X = \mathds{1}_m$ given by
\[
(U, X) \mapsto XU,
\]
we can define a \textit{projected version} of $F$ which often has a unique global minimum \cite{shen2010tracequotient}. The trace ratio function will be discussed in further detail in \cref{sec:traceratio}.

The main results of this work are obtained when the following conditions are met. 

\begin{assumption}
\label{assumption3.7} 
There exists a compact and connected Lie group $G$ acting freely, isometrically and smoothly on $(M, g)$. 
\end{assumption}

Under \cref{assumption3.7}, it is known (cf. \cref{existenceRiemannianSubmersionMetrics}) that the quotient space $M/G$ can be endowed with a manifold structure and a Riemannian metric $h$ for which the projection
\[
\pi: (M, g) \to (M/G, h),
\]
is a \textit{Riemannian submersion}. 

Next, we need two assumptions regarding the curvature of $(M, g)$ and $(M/G, h)$. The first is needed for the lifting process of the Poincaré inequality (cf. \cref{SectionLift}), the next to obtain explicit curvature-dimension condition constants. 
\begin{assumption}[Curvature]
\label{assumption3.8} 
The fibers of $\pi$ have non-negative Ricci curvature. 
\end{assumption}

\begin{assumption}[Curvature]
\label{assumption3.9} 
There exists a constant $\mathbf{K} \geq 1$ such that for every $p\in M/G$, in normal coordinates centered at $p$, $|R_{ijk}^l(p)| \leq \mathbf{K}$ for every $i, j, k, l \in \{1, \dotsc, \dim( M/G)\}$, where $R$ denotes the Riemann curvature tensor of $(M/G, h)$. There also exist constants $R_M, R_{M/G} \geq 0$ that lower bound the Ricci curvatures of $(M, g)$ and $(M/G, h)$, respectively,
\[
\textup{Ric}_h \geq -R_{M/G}\,h,\quad \textup{\textit{and}}\quad \textup{Ric}_g \geq -R_M\, g.
\]
\end{assumption}

Note that whenever the sectional curvature of $(M, g)$ is lower bounded by some constant $\kappa$, the Ricci curvature---which is the sum of the sectional curvature associated with all the planes containing $v$---is lower bounded by $(\dim(M)-1)\kappa$. Furthermore, by O'Neill's curvature formulas (cf. \cref{CurvaturaSubmersion}), we know that the sectional curvature of $(M/G, h)$ is also lower bounded by $\kappa$, and so the Ricci curvature of $(M/G, h)$ is lower bounded by $(\dim(M/G)-1)\kappa$. 

The final geometric assumption regards the quotient manifold $(M/G, h)$. 
\begin{assumption}
\label{assumptionsym}
$(M/G, h)$ is a symmetric space. 
\end{assumption}
This assumption is imposed for technical reasons, as it allows us to obtain quantitative bounds the coefficients of $h$ and $h^{-1}$ in normal coordinates, as it is shown in \cref{sec:LaplacianAndGradientBound}. We expect that this assumption could either be replaced by the same assumption on $(M, g)$, or simply removed. The former approach would require a more detailed analysis on the curvature of the submersion $\pi$, while the latter would require a finer control of the curvature of $(M/G, h)$.

Next, we need to assume that the group action of $G$ on $M$ represents the symmetries of $F$. 
\begin{assumption}[Symmetry of $F$]
\label{assumption3.1}
$F$ is constant in the fibers of $\pi$---i.e. $F(x) = F(y)$ for every $x, y \in M$ such that $\pi(x) = \pi(y)$. 
\end{assumption}
As we mentioned earlier, whenever \cref{assumption3.1} holds, we can define a unique smooth function $\tilde F : M/G \to \bb{R}$ such that $\tilde F \circ \pi = F$. 

\begin{remark}
\label{rmk:LipschitzConstants}
Since $M$ is compact and $F$ is smooth, there exist constants $A_1, A_2$ such that $F$ is $A_1$-Lipschitz and $\Grad{g}F$ is $A_2$-Lipschitz. The same Lipschitz constants are valid for $\tilde F$ and its gradient, as it will be shown in \cref{prop:preservationofLipschitz}. We further assume that $\nabla^2 \tilde F$ is $A_3$-Lipschitz.
\end{remark}

The following assumption is made for the sake of simplicity, and without loss of generality.  
\begin{assumption}
\label{assumption3.2} 
$\min_{x \in M} F(x) = 0.$
\end{assumption} 
 
Our derivation of a Poincaré inequality on $M/G$ requires that $\tilde F$ has a unique (global) minimum. We believe this limitation is not fundamental, but an artifact of our construction; results qualitatively similar to ours should hold when $\tilde F$ has multiple minima. For $M = \bb{R}^n$, it has been shown that this difficulty can be overcome. In this context, the Langevin dynamics converge fast to a mixture of Gibbs distributions concentrated around the various minima of $\tilde F$ \cite{menz2014}. One avenue to deal with local minima is to turn them into saddle points by overparametrization \cite{safran2021effects,fukumizu2019semi}. Studying functions $\tilde F: M/G \to \bb{R}$ with multiple minima is the subject of ongoing work and will be discussed elsewhere. 
\begin{assumption}
\label{assumption3.3} 
$\tilde F$ has a unique minimum.
\end{assumption}
This assumption implies that $F$ has no non-global local minima, and all of the global minima of $F$ lie in a single fiber of $\pi$. For simplicity, we will use the term \textit{saddle point} for any critical point that is not the global minimum of $\tilde F$. With a slight abuse of terminology, local maxima will also be called saddle points in this work, as their treatment in the proofs is the same.

Let $\tilde{\mathcal{C}}$ and $\tilde{\mathcal{S}}$ denote the set of critical and saddle points of $\tilde F$, respectively. We need to assume that they are isolated and sufficiently spaced.
\begin{assumption}[Critical point spacing]
\label{assumption3.4} 
The critical points of $\tilde F$ are isolated points separated by some distance $D > 0$, i.e. 
\[
d_h(x_1, x_2) \geq D, \quad \forall x_1, x_2 \in \tilde{\mathcal{C}},\ x_1 \neq x_2,
\]
where $d_h$ denotes the geodesic distance on $(M/G, h)$. 
\end{assumption}

Again, this assumption is made for technical reasons but we believe it should eventually be possible to remove it. It is meant to control the behavior of $X_t$ when it escapes from saddle points. 

One could argue that \cref{assumption3.4} is generically satisfied. In fact, Morse functions---those for which the Hessian's kernel is zero at every critical point---defined on a compact manifold have a finite number of critical points \cite[Corollary 4.6]{cheeger1975comparison}. Moreover, Morse functions form a dense open subset of $C^\infty(M)$ when endowed with the strong topology (cf. \cite[Theorem 1.2, Chapter 6]{hirsch2012differential}). 

The next assumption on $\tilde F$ is standard in the analysis of non-convex optimization and sampling algorithms, both in Euclidean space and on Riemannian manifolds, and we believe that it is necessary to guarantee rapid mixing. 

\begin{assumption}
\label{assumption3.5.1}
For every saddle point $y \in \tilde{\mathcal{S}}$, there exists an escape direction. That is, there exists some constant $\lambda_* \in (0, 1]$ such that for every saddle point $y \in \tilde{\mathcal{S}}$ it holds that 
\[
\lambda_{\min}\big(\nabla^2 \tilde F(y)\big) \leq -\lambda_*.
\]
\end{assumption}

\begin{assumption}
\label{assumption3.5.2}
The global minimum is an ``attractor'' for the process $\tilde X_t$. In other words, the Hessian of $\tilde F$ at the global minimum is non-degenerate. More precisely, at the unique minimum $x^*$ of $\tilde F$, we have
\[
\lambda_{\min}\big(\nabla^2 \tilde{F}(x^*)\big) \geq \lambda_*.
\]
where $\lambda_*$ is the same constant as in \cref{assumption3.5.1}. 
\end{assumption}

\begin{remark}
Assumptions \ref{assumption3.5.1} and \ref{assumption3.5.2} hold if
\begin{enumerate}
    \item For every saddle point $y$ of $F$, $\lambda_{\min}\big(\nabla^2 F(y)\big) \leq -\lambda_*$.
    \item For every $x \in \pi^{-1}(x^*)$, and every $v \in \ker\big(\nabla^2 F(x)\big)$,  $v \in \ker(d\pi|_x)$. Furthermore, for every eigenvector $v \in \ker\big(\nabla^2 F(x)\big)^\bot$, $\nabla^2 F(x)[v] = \lambda_v v$, where $\lambda_v \geq \lambda_*$. 
\end{enumerate}
\end{remark}

As we will show in \cref{prop:projectedHessian}, the non-zero eigenvalues of $\nabla^2 F$ at any critical point are equal to those of $\nabla^2 \tilde F$. 

While the above assumption on the Hessian of $\tilde F$ allows us to escape saddle points, we also need to ensure that the process $\tilde X_t$ \textit{moves} sufficiently \textit{fast} whenever it is sufficiently far from the critical points. For that, we would like the evolution of $\tilde X_t$ (\cref{eq:defprocesos}) to be dominated by the gradient term. This requires the norm of the latter to be large enough away from $\tilde{\mathcal{C}}$. 

\begin{assumption}[No barren plateaus]
\label{assumption3.6} 
There exists a constant $0 < C_{\tilde F} \leq 1$ such that
\begin{equation}
\label{eq:boundgradientdistance}
|\textup{grad}_h\, \tilde F(x)|_{h} \geq C_{\tilde F} d_{h}(x, \tilde{\mathcal{C}}),    
\end{equation}
for every $x \in M/G$. 
\end{assumption}

\begin{remark}
This assumption holds whenever
 \[
|\textup{grad}_g\, F(x)|_g \geq C_{\tilde F} d_g(x, \mathcal{C}),
\]
for every $x \in M$, where $\mathcal{C}$ denotes the set of critical points of $F$. Indeed, the norms of $\Grad{g} F$ and $\Grad{h} \tilde F$ are equal (cf.  \cref{prop:gradientprojectedfunction}), and distances between points do not increase under Riemannian submersions. 
\end{remark}

Most of the assumptions required for our convergence analysis are standard in the analysis of non-convex optimization and sampling algorithms, both in Euclidean space and on Riemannian manifolds. In particular, assumptions guaranteeing the existence of escape directions at saddle points, together with lower bounds on the gradient norm away from critical points, are commonly used in the analysis of gradient-based methods; see for example \cite{jin2017escapesaddlepointsefficiently,chi2019,zhang2022symmetrygeometrytractablenonconvex} and the references therein. 

A natural question is whether these assumptions are satisfied for problems of practical interest. A substantial body of work has shown that many optimization problems possess a favorable geometric landscape, characterized by the absence of spurious local minima and the presence of saddle points with escape directions. Such results have been established for numerous non-convex optimization problems in Euclidean space, including matrix sensing, matrix completion, phase retrieval, low-rank matrix factorization, and neural network training; see for example \cite{luo2022nonconvexmatrixfactorizationgeodesically,chi2019,sun2018geometric,sun2015complete,ge2016matrix,bhojanapalli2016global,li2019non} and the references therein. Finally, arguments for escaping saddle points in optimization algorithms on Riemannian manifolds can be found in \cite{sun2019escaping,criscitiello2019efficiently,hou2021fastglobalconvergencelowrank}.

We are now in a position to state the two main results of this work. For $\beta$ large enough, we have
\begin{itemize}
    \item Rapid convergence of the distributions of $X_t$ and $\tilde X_t$ to their respective Gibbs distributions. 
    \item Concentration of the Gibbs distributions around the minima of $F$ and $\tilde F$. 
\end{itemize}

\begin{theorem}[Main result 1 - informal version]
\label{thm:MainInformal1}
Let $(M, g)$ be a compact Riemannian manifold of dimension $\dim(M) \geq 2$. Assume that there exists a group $G$ acting on $M$ and satisfying \cref{assumption3.7}. Let $\pi : (M, g) \to (M/G, h)$ be the associated Riemannian submersion, where $M$ and $M/G$ satisfy \cref{assumption3.8,assumption3.9} with constants $\mathbf{K}$, $R_{M}$ and $R_{M/G}$, and $M/G$ satisfies \cref{assumptionsym}. Let $F: M \to \bb{R}$ be a smooth function satisfying \crefrange{assumption3.1}{assumption3.2} and let $\tilde F: M/G \to \bb{R}$ be the unique function such that $F = \tilde F \circ \pi$. Further assume that $\tilde F$ satisfies \crefrange{assumption3.3}{assumption3.6}. Let $\beta > 0$ be such that
\[
\beta \in \Omega\Big(\textup{poly}\Big(R_{M/G}, \dim(M), A_2, A_3, \mathbf{K}, \frac{1}{i(M/G)}, \frac{1}{\mathit{conv}(M/G)}, \frac{1}{\lambda_*}, \frac{1}{D}, \frac{1}{C_{\tilde F}}\Big)\Big),
\]
where $C_{\tilde F}$ is the constant from \cref{eq:boundgradientdistance}, $\mathit{conv}(M/G)$ denotes the convexity radius of $(M/G, h)$, $i(M/G)$ denotes the injectivity radius of $(M/G, h)$, $D$ is the lower bound on the distance between any two critical points of $\tilde F$, $\lambda_*$ is the bound on the eigenvalues of $\nabla^2 \tilde F$ at the critical points, $A_2 \geq 1$ is the Lipschitz constant of the gradient of $F$ and $\tilde F$, and $A_3 \geq 1$ is the Lipschitz constant of the Hessian of $\tilde F$.

Then  the distributions $\rho_t$ and $\tilde \rho_t$ of the processes $X_t$ and $\tilde X_t$ solving the SDEs
\[
dX_t = -\textup{grad}_g\, F(X_t) dt + \sqrt{\frac{2}{\beta}} dW_t, \quad \textup{and} \quad d\tilde X_t = -\textup{grad}_h\, \tilde F(\tilde X_t) dt + \sqrt{\frac{2}{\beta}} d\tilde W_t,
\]
with initial uniform distribution on $M$ and $M/G$, respectively, converge exponentially fast to their respective Gibbs distributions $\nu$, and $\tilde \nu$, i.e. 
\begin{align*}
\norm{\nu - \rho_t}^2_{\textup{TV}} &\leq  \beta\, e^{-2\alpha t} \max_{y \in M} F(y), \quad \forall t \geq 0,\\
\norm{\tilde \nu - \tilde \rho_t}^2_{\textup{TV}} &\leq  \beta\, e^{-2\tilde \alpha t} \max_{y \in M/G} \tilde F(y) = \beta\, e^{-2\tilde \alpha t} \max_{y \in M} F(y) , \quad \forall t \geq 0,
\end{align*}
where the constants $\alpha$, and $\tilde \alpha$ correspond to the log-Sobolev inequality constants associated with the manifolds $M$ and $M/G$, which in particular scale as 
\begin{align*}
\frac{1}{\alpha} &\in O\Big(\textup{poly}\Big(\beta, A_2, R_M, \textup{diam}(M), \textup{diam}(G), \frac{1}{\lambda_*}\Big)\Big),\\
\frac{1}{\tilde \alpha} &\in O\Big(\textup{poly}\Big(\beta, A_2, R_{M/G}, \textup{diam}(M/G), \frac{1}{\lambda_*}\Big)\Big).
\end{align*}
where $\textup{diam}(G)$ denotes the diameter of $G$ as a fiber of $\pi$ and $\textup{diam}(M)$, $\textup{diam}(M/G)$ denote the diameter of $(M, g)$ and $(M/G, h)$. 
\end{theorem}

\begin{theorem}[Main result 2 - informal version]
\label{thm:MainInformal2}
Let $(M, g)$ be a compact Riemannian manifold of dimension $\dim(M) \geq 2$. Let $F: M \to \bb{R}$ be twice differentiable with an $A_2$-Lipschitz gradient, $A_2 \geq 1$. Let 
\[
\varepsilon_{\max} \in O(i(M)^2 A_2),
\]
where $i(M)$ denotes the injectivity radius of $M$. For every $\varepsilon \in (0, \varepsilon_{\max}]$ and $\delta \in (0, 1)$, if 
\begin{align*}
\beta &\in \Omega\bigg(\frac{\dim(M)^2\log (A_2, \textup{Vol}(M),\dim(M), \varepsilon^{-1}, \delta^{-1})}{\varepsilon}\bigg),
\end{align*}
the Gibbs distribution $\nu = \frac{1}{Z}e^{-\beta F}$ satisfies 
\[
\nu\left(F - \min_{y \in M} F(y) \geq \varepsilon\right) \leq \delta.
\]
\end{theorem}

The above results can be of interest when considering a family of minimization problems defined on a sequence of manifolds with growing dimension. It is then natural to wonder how $\beta$ and the convergence rates $\alpha, \tilde\alpha$ scale with the dimension of the manifold on which $F$ is defined. 
\begin{remark}
As a consequence of \cref{thm:MainInformal1}, if 
\[
\beta \in \Omega(\textup{poly}(\dim(M))),
\]
i.e. if the constants that lower bound $\beta$ grow at most polynomially with respect to the dimension of $M$, assuming furthermore that the diameter of $M, M/G$ and $G$, as well as the lower bound on the Ricci curvature of $M$ grow at most polynomially with respect to the dimension of $M$, then the log-Sobolev constants satisfy
\[
\frac{1}{\alpha} \in O(\textup{poly}(\dim(M))),\quad \textup{and} \quad \frac{1}{\tilde \alpha} \in O(\textup{poly}(\dim(M))).
\]

In particular, whenever $\max_{y \in M} F$ scales as $\textup{poly}(\dim(M))$, the distance between the distribution of $X_t$ and $\nu$ decays exponentially at a rate $\Omega\Big(\frac{1}{\textup{poly}(\dim(M))}\Big)$. Similarly, the distance between the distribution of $\tilde X_t$ and $\tilde \nu$ decays exponentially at a rate $\Omega\Big(\frac{1}{\textup{poly}(\dim(M/G))}\Big)$. 
\end{remark}

\subsection{Structure of the text}

The structure of the next five sections reflects our proof strategy. 
\begin{itemize}
    \item \cref{SectionPI}. We will prove that, under slightly milder assumptions than those considered for \cref{thm:MainInformal1}, the base space $M/G$ satisfies a Poincaré inequality with a constant that only depends on the bound on the eigenvalues of $\nabla^2 \tilde F$ at its critical points, $\lambda_*$ (cf. \cref{assumption3.5.1,assumption3.5.2}).
    \item \cref{SectionLift}. We will show how to \textit{lift} a Poincaré inequality from the base space to the total space of a Riemannian submersion $\pi$, when the Gibbs measure on $M$ is defined with respect to a function $F$ that is constant in the fibers of $\pi$, and the Gibbs measure on $M/G$ is defined with respect to the \textit{projected version} of $F$. This result will allow us to obtain a Poincaré inequality for $M$. 
    \item  \cref{SectionPItoLSI}. After a short review on classic results which relate the existence of a log-Sobolev inequality to mixing time estimates, we will show how to \textit{upgrade} a Poincaré inequality to a log-Sobolev inequality under the existence of a \textit{curvature-dimension condition}. This will ultimately allow us to prove the existence of a logarithmic Sobolev inequality for both $M/G$ and $M$, which will in turn translate to upper bounds for the distance between the distribution of the processes $X_t$ and $\tilde X_t$ (cf. \cref{eq:defprocesos}) and their associated Gibbs distributions. 
    \item \cref{sec:proof}. We will formalize and prove \cref{thm:MainInformal1}, regarding convergence rates of the processes $X_t$ and $\tilde X_t$.
    \item \cref{SectionSuboptimality}. We will formalize and prove \cref{thm:MainInformal2}, regarding the concentration of the Gibbs distributions $\nu$ and $\tilde \nu$. 
\end{itemize}

\cref{sec:traceratio} is an analysis of two scenarios in which most of the assumptions of \cref{thm:MainInformal1,thm:MainInformal2} can be easily verified, namely the trace quotient minimization problem, and the mean-field energy minimization problem associated with the two-dimensional ferromagnetic Ising model. 

In an effort to make our work as self-contained (and accessible) as possible, background is provided on Riemannian geometry in \cref{SectionExamples}, Sobolev spaces in \cref{sec:obolevspaces}, and Bakry-Émery theory in \cref{BakryEmeryandLyapunov}. \cref{sec:LaplacianAndGradientBound} is a discussion of a technical aspect of \cref{SectionPI}.

%% file: Chapters/PoincareIneq.tex
\section{Poincaré inequality under unique minimum}
\label{SectionPI}

This section is devoted to the main element in the proof of the rapid mixing result stated in \cref{thm:MainInformal1}: obtaining a suitable Poincaré inequality in the base space of the Riemannian submersion. 

\begin{definition}
\label{def:Langevingenerator}
Let $(M, g)$ be a Riemannian manifold, $F \in C^\infty(M)$, and $\beta > 0$. The \textup{generator of the Langevin dynamics} $\operatorname{L}$ is the operator acting on\footnote{In fact, $\operatorname{L}$ can be extended to act on a broader subset of functions of $L^2(M, \nu)$ via its \textit{closure} \cite{wang2006functional}.} $C^2(M)$ as 
\[
\operatorname{L}f := \langle -\Grad{g} F, \Grad{g}f\rangle_g + \frac{1}{\beta} \Delta_g f.
\]
\end{definition}

We will follow the notation and definitions of \cite{revuz2013continuous}. For an in-depth study of stochastic processes on Riemannian manifolds, see \cite{hsustochastic}.

\begin{definition}
Let $(\Omega, \mathcal{F}, \bb{P})$ be a probability space and $M$ be a smooth manifold. A \textup{continuous-time stochastic process} in $M$ is a family of random variables 
\[
\{Y_t: t \in \bb{R}^+\},
\]
where $Y_t : \Omega \to M$ for every $t \in \bb{R}^+$.
\end{definition}

\begin{definition}
\label{def:LangevinDiffusion}
Given some filtered probability space $(\Omega, \mathcal{F})$ and compact Riemannian manifold $(M,g)$, the \textup{Langevin diffusion process} $X_t$ is the unique $\mathcal{F}$-adapted semimartingale on $M$ such that 
\begin{equation}
\label{martingaleproblem}
M^f_t := f(X_t) - f(X_0) - \int_0^t \operatorname{L}f(X_s)ds,
\end{equation}
is an $\mathcal{F}$-adapted local martingale for all $f \in C^2(M)$.
\end{definition}

The existence and uniqueness of such a process with a given initial distribution is proven in \cite[Theorem 1.3.4 and 1.3.6]{hsustochastic}. We often say that $X_t$ is generated by $\operatorname{L}$ or that $X_t$ \textit{solves the martingale problem} associated with $\operatorname{L}$. The formal SDE associated with this process is
\begin{equation}
\label{LangevinDiffusionEq}
dX_t = -\Grad{g}F(X_t)\,dt + \sqrt{\frac{2}{\beta}}\,dW_t,
\end{equation}
where $W_t$ is a Brownian motion on $M$ and $X_0$ is initialized with a distribution $\rho_0$ supported on $M$.

Poincaré inequalities involve the notions of carré du champ operator and Markov triple. 

\begin{definition}[Carré du champ operator]
\label{def:carreduchampLangevin}
Let $(M, g)$ be a compact manifold, and let $\operatorname{L}$ be as given in \cref{def:Langevingenerator}. We define the \textup{carré du champ operator} $\Gamma$ on $C^2(M) \times C^2(M)$ as
\[
\Gamma(f_1,f_2) := \frac{1}{2}(\operatorname{L}(f_1f_2) - f_1\operatorname{L}f_2 - f_2\operatorname{L}f_1).
\]
\end{definition}

\begin{notation}
We will often use $\Gamma(f)$ to denote $\Gamma(f,f)$. 
\end{notation}

\begin{remark}
\label{rem:explicitexprLangevin}
The carré du champ operator associated with $\operatorname{L}$ as given in \cref{def:Langevingenerator} can be written as
\[
\Gamma(f_1, f_2) = \frac{1}{\beta}\langle \Grad{g}f_1, \Grad{g}f_2\rangle_g,    
\]
for any $f_1, f_2 \in C^2(M)$, and so 
\[
\Gamma(f) = \frac{1}{\beta}|\Grad{g}f|^2_g,
\]
for any $f \in C^2(M)$.
\end{remark}
\begin{proof}
The result automatically follows, using the well-known identities
\[
\Delta_g(f_1f_2) = f_1\Delta_g f_2 + f_2\Delta_g f_1 + 2 \langle \Grad{g}f_1, \Grad{g}f_2\rangle_g,
\]
and
\[
\Grad{g}(f_1f_2) = f_1\,\Grad{g}f_2 + f_2\, \Grad{g}f_1. 
\]
\end{proof}

\begin{definition}[Markov triple]
\label{def:markovtripleLangevin}
Let $(M, g)$ be a compact manifold, $F \in C^\infty(M)$, ${\beta >0}$, and $\operatorname{L}$ be the associated Langevin diffusion generator. The \textup{Markov triple} $(M, \nu, \Gamma)$ associated with $\operatorname{L}$ consists of the Gibbs measure 
\[
d\nu := \frac{1}{Z}e^{-\beta F} \textup{dVol}_g,
\]
where $Z$ is the normalizing constant, $Z := \int_M e^{-\beta F} \textup{dVol}_g$, and $\Gamma$ is the carré du champ operator associated with $\operatorname{L}$. 
\end{definition}

\begin{definition}[Poincaré inequality]
\label{defPoincareIneq}
Let $(M, g)$ be a compact Riemannian manifold, $F \in C^\infty(M)$, and $\beta >0$. The Markov triple $(M, \nu, \Gamma)$ satisfies a \textup{Poincaré inequality} with constant $\kappa > 0$, denoted $\textup{PI}(\kappa)$, if for every $f \in C^2(M)$,
\[
\int_M f^2 \,d\nu - \Big(\int_M f \,d\nu\Big)^2 \leq \frac{1}{\kappa} \int_M \Gamma(f) \, d\nu.
\]
\end{definition}

The Poincaré inequality constant $\kappa$ controls the rate of convergence of the Langevin diffusion process \cite[Chapter 4]{bakry2013analysis}. Bounding this constant is also the first step toward obtaining a log-Sobolev inequality. 

As we mentioned earlier, in this work we consider two Riemannian manifolds $(M, g)$ and $(B, h)$, where we assume that $(B, h)$ is a symmetric space. We further assume that there exists a---surjective---Riemannian submersion with totally geodesic fibers, $\pi: (M, g) \to (B, h)$. Note that, as we are assuming that $M$ is compact, $B$ must be compact as well. Now let $F: M \to \bb{R}$ be a smooth function that is constant in the fibers of $\pi$. The existence of such a function $F$ is equivalent to the existence of a function $\tilde F : B \to \bb{R}$ satisfying $F = \tilde F \circ \pi$. It is known that $\tilde F$ is smooth if and only if $F$ is smooth \cite[Theorem 9.21]{boumal2022intromanifolds}. Therefore, for such fiber-invariant function $F$, the Markov triple $(M, \nu, \Gamma)$ naturally induces another $(B, \tilde \nu, \tilde \Gamma)$ for $\tilde F$ and vice versa, where of course
\begin{equation}
\label{eq:defMarkovB}
\tilde{\operatorname{L}} := -\textup{grad}_h \tilde F + \frac{1}{\beta} \Delta_h,\quad
\tilde\Gamma(f) = \frac{1}{\beta}|\Grad{h}f|^2_h,\quad \textup{and}\quad 
d\tilde \nu = \frac{1}{\tilde Z} e^{-\beta \tilde F}\textup{dVol}_h, 
\end{equation}
where $\tilde Z$ is the normalizing constant for $\tilde \nu$ on $(B, h)$. 

As alluded in the previous section, when attempting to prove \textit{rapid mixing} for a Langevin process, the multiplicity of minima turns out to be a serious technical obstruction. If a submersion allows us to remove this multiplicity, it turns out that deriving a Poincaré inequality for the Markov triple $(B, \tilde \nu, \tilde \Gamma)$ is much simpler than for the original Markov triple $(M, \nu, \Gamma)$, hence our interest in the former. The problem of lifting such a Poincaré inequality to the latter Markov triple, and of upgrading it to a log-Sobolev inequality will be dealt with in \cref{SectionLift,SectionPItoLSI}, respectively.

To prove a Poincaré inequality for the Markov triple $(B, \tilde\nu, \tilde\Gamma)$, we will follow the steps of \cite[Appendix D]{LiErd2022Supp}. We will first define and prove the existence of two suitable Lyapunov functions. Next, we will show how these Lyapunov functions can be used to extend a \textit{local} Poincaré inequality, which is easy to obtain, to the whole manifold $B$. 

\subsection{A first Poincaré inequality}
\label{sec:introductionPI}

Let us restate the assumptions considered in \cref{sec:intro} that will be used in this section for completeness.
\begin{assumption}
\label{modassumption3.9}
There exists a constant $\mathbf{K} \geq 1$ such that for every $p\in B$, in normal coordinates centered at $p$, $|R_{ijk}^l(p)| \leq \mathbf{K}$ for every $i, j, k, l \in \{1, \dotsc, \dim(B)\}$, where $R$ denotes the Riemann curvature tensor of $(B, h)$. There also exist some constant $R_B \geq 0$ that lower bounds the Ricci curvature of $(B, h)$
\[
\textup{Ric}_h \geq -R_{B}\,h
\]
\end{assumption}

\begin{assumption}
\label{modassumptionsym}
$(B, h)$ is a symmetric space. 
\end{assumption}

\begin{assumption}
\label{modassumption3.1}
$F$ is constant in the fibers of $\pi$---i.e. $F(x) = F(y)$ for every $x, y \in M$ such that $\pi(x) = \pi(y)$. 
\end{assumption}

\begin{remark}
Since $M$ is compact and $F$ is smooth, we know that there exist constants $A_1, A_2$ such that $F$ is $A_1$-Lipschitz and $\Grad{g}F$ is $A_2$-Lipschitz. The same Lipschitz constants are valid for $\tilde F$ and its gradient, as it will be shown in \cref{prop:preservationofLipschitz}. As $B$ is also compact and $\tilde F$ is smooth, there exists some constant $A_3$ such that $\nabla^2 \tilde F$ is $A_3$-Lipschitz.
\end{remark}

\begin{assumption}
\label{modassumption3.2}
$\min_{x \in M} F(x) = 0$. 
\end{assumption} 

\begin{assumption}
\label{modassumption3.3} 
$\tilde F$ has a unique minimum.
\end{assumption}

\begin{assumption}
\label{modassumption3.4} 
The critical points of $\tilde F$ are isolated points separated by some distance $D > 0$, i.e. 
\[
d_h(x_1, x_2) \geq D, \quad \forall x_1, x_2 \in \tilde{\mathcal{C}},\ x_1 \neq x_2.
\]
\end{assumption}

\begin{assumption}
\label{modassumption3.5.1}
For every saddle point $y \in \tilde{\mathcal{S}}$, there exists an escape direction. That is, there exists some constant $\lambda_* \in (0, 1]$ such that for every saddle point $y \in \tilde{\mathcal{S}}$ it holds that 
\[
\lambda_{\min}\big(\nabla^2 \tilde F(y)\big) \leq -\lambda_* < 0.
\]
\end{assumption}

\begin{assumption}
\label{modassumption3.5.2}
The global minimum is an ``attractor'' for the process $\tilde X_t$. In other words, the Hessian of $\tilde F$ at the global minimum is non-degenerate. More precisely, at the unique minimum $x^*$ of $\tilde F$, we have
\[
\lambda_{\min}\big(\nabla^2 \tilde{F}(x^*)\big) \geq \lambda_*.
\]
where $\lambda_*$ is the same constant as in \cref{modassumption3.5.1}. 
\end{assumption}

\begin{assumption}
\label{modassumption3.6} 
There exists a constant $0 < C_{\tilde F} \leq 1$ such that
\begin{equation}
\label{eq9.41bis}
|\Grad{h}\tilde F(x)|_{h} \geq C_{\tilde F} d_{h}(x, \tilde{\mathcal{C}}),    
\end{equation}
for every $x \in B$. 
\end{assumption}

\begin{theorem}[Poincaré Inequality on $B$]
\label{prop9.12}    
Let $(B, h)$ be a compact Riemannian manifold satisfying \cref{modassumptionsym,modassumption3.9}. Let $\tilde F : B \to \bb{R}$ be a smooth function satisfying \crefrange{modassumption3.3}{modassumption3.6}. Further assume for simplicity that $\tilde F \geq 0$. Let $a, \beta > 0$ be such that
\begin{align*}
a^2 &\geq \max\Big\{\frac{24 A_2 \dim(B)}{C^2_{\tilde F}},  \frac{288}{\lambda_*}\Big\},
\\\beta &\geq \max\Bigg\{\frac{192^2 \dim(B)^5 A^2_2 A^2_3 \mathbf{K}^2 a^6}{\lambda_*^2}, \frac{9a^2}{D^2}, \frac{4a^2}{i(B)^2}, \frac{4R_{B}}{\lambda_*}, \frac{a^2}{\mathit{conv}(B)^2}\Bigg\},
\end{align*}
where $C_{\tilde F}$ is the constant from \cref{eq9.41bis}, $\mathit{conv}(B)$ denotes the convexity radius of $B$, $i(B)$ denotes the injectivity radius of $(B, h)$, $D$ is the lower bound on the distance between any two critical points of $\tilde F$, $\lambda_*$ is the bound on the eigenvalues of $\nabla^2 \tilde F$ at the critical points, and $A_2 \geq 1$ is the Lipschitz constant of the gradient of $F$ and $\tilde F$, and $A_3 \geq 1$ is the Lipschitz constants of the Hessian of $\tilde F$. 

Then  the Markov triple $(B, \tilde \nu, \tilde \Gamma)$ satisfies a $\textup{PI}(\kappa)$, where
\[
\kappa = \frac{\lambda_*}{120}.
\]
\end{theorem}

Informally, we will see that, whenever the \textit{amount of noise} $\frac{1}{\beta}$ is \textit{sufficiently small}, we obtain a Poincaré inequality with a constant that only depends on the value of $\lambda_*$, which in turn controls \textit{how fast} $\tilde X_t$ escapes from the saddle points of $\tilde F$. The constant $a$ sets a lower bound on the values of $\beta$ for which we can claim that the Poincaré inequality holds. 

The rest of the section will be devoted to the proof of this result, which will be given in \cref{sec2.3}. 

\subsection{Lyapunov functions}
\label{sec2.2}

\begin{definition}[Lyapunov function]
Let $(M, g)$ be a Riemannian manifold, and let $\operatorname{L}$ denote the Langevin generator from \cref{def:Langevingenerator} for some $F \in C^\infty(M)$ and $\beta > 0$. A function $W : M \rightarrow \bb{R}$, $W \in C^2(M)$ is said to be a \textup{Lyapunov function} on a measurable subset $U \subset M$ with parameters $\theta > 0$ and $b \geq 0$ if $W \geq 1$ and 
\[
\frac{\operatorname{L}W(x)}{W(x)} \leq -\theta + b\mathds{1}_U(x), \quad \forall x \in M.
\]
\end{definition}

\begin{definition}[Quasi-Lyapunov function]
Let $(M, g)$ be a compact Riemannian manifold, and $\operatorname{L}$ be the Langevin generator. Let $U$ be some open subset of $M$. A function $W: \overline{U} \rightarrow \bb{R}$ is said to be a \textup{quasi-Lyapunov function} on an open subset $U \subset M$ with parameter $\theta > 0$ if $W \in C^2(\overline{U})$, $W \geq 1$ and 
\[
\operatorname{L}W(x) \leq -\theta W(x), \quad \forall x \in U.
\]
\end{definition}

As we mentioned earlier, we are interested in obtaining Lyapunov functions\footnote{Our definition of Lyapunov functions is not standard, as they are often defined with slight modifications---see, for example \cite{Bakry2008}.} defined on (submanifolds of) $(B, h)$ and associated with the operator
\[
\tilde{\operatorname{L}} = - \Grad{h}\tilde F + \frac{1}{\beta} \Delta_h,
\]
where $\tilde F$ is the function from \cref{prop9.12} satisfying \crefrange{modassumption3.3}{modassumption3.6}.

\subsubsection{First Lyapunov function}

To construct our first Lyapunov function, we will use the following result, proven in \cite{LiErd2022Supp} (Lemma D.8). We include the proof for completeness, as it provides insight into the role of \cref{modassumption3.6}.
\begin{lemma}
\label{lemma9.8}
Let $(B, h)$ be a compact Riemannian manifold. Let $\tilde F: B \to \bb{R}$ be smooth, and suppose that it satisfies \crefrange{modassumption3.3}{modassumption3.6}. Let $\tilde{\mathcal{C}}$ be the set of critical points of $\tilde F$. Then, for every $\beta > 0$ and every $a > 0$ satisfying 
\[
a^2 \geq \frac{24A_2 \dim(B)}{C^2_{\tilde F}},
\]
where $A_2$ is the Lipschitz constant of the gradient of $\tilde F$, it holds that
\[
\frac{\Delta_h \tilde F(x)}{2} - \frac{\beta}{4}|\Grad{h}\tilde F(x)|^2_h \leq -A_2 \dim(B),\quad \forall x \in B : d_h(x, \tilde{\mathcal{C}})^2 \geq \frac{a^2}{4\beta}. 
\]
\end{lemma}
\begin{proof}
Using \cref{modassumption3.6} it holds that 
\[
|\Grad{h}\tilde F(x)|_h \geq C_{\tilde F} d_h(x, \tilde{\mathcal{C}}). 
\]
Thus, for every $x \in B$ such that $d_h(x, \tilde{\mathcal{C}})^2 \geq \frac{a^2}{4\beta}$ we obtain that 
\begin{align*}
\frac{\Delta_h \tilde F(x)}{2} - \frac{\beta}{4}|\Grad{h}\tilde F(x)|^2_h  &\leq \frac{\Delta_h \tilde F(x)}{2} - \frac{\beta}{4}C_{\tilde F}^2 d_h(x, \tilde{\mathcal{C}})^2
\\&\leq \frac{\Delta_h \tilde F(x)}{2} - \frac{a^2}{16}C_{\tilde F}^2
\\&\leq \frac{A_2\dim(B)}{2} - \frac{1}{16}C_{\tilde F}^2 \frac{24A_2 \dim(B)}{C^2_{\tilde F}}
\\&\leq -A_2 \dim(B),
\end{align*}
where the second-last inequality follows by the assumption $a^2 \geq \frac{24A_2 \dim(B)}{C^2_{\tilde F}}$. 
\end{proof}

In fact, \cref{modassumption3.6} is only used in the proof of \cref{lemma9.8}. Therefore, this assumption can be relaxed; it need only hold for every $x \in B$ satisfying $d_h(x, \tilde{\mathcal{C}})^2 \geq \frac{a^2}{4\beta}$, for the appropriate choice of $a$ and $\beta$.

\begin{notation}
For all $r > 0$ and $X \subset B$, we will denote the set of points within distance $r$ to $X$ as $U(r,X)$, i.e.
\[
U(r,X) := \{x \in B : d_h(x, X) < r\}.
\]
\end{notation}

Let $\tilde{\mathcal{S}}$ denote the saddle points of $\tilde F$. We are going to prove that 
\begin{equation}
\label{eq:defW1}
W_1: \overline{U}\Big(\frac{a}{2\sqrt{\beta}}, \tilde{\mathcal{S}}\Big)^c \to \bb{R},\quad x \mapsto  \exp\left(\frac{\beta}{2}\tilde F(x)\right)
\end{equation}
is a Lyapunov function on the geodesic ball $\mathcal{B}(\frac{a}{\sqrt{\beta}}, x^*)$, where $x^*$ is the unique minimum of $\tilde F$, $\overline{U}$ denotes the closure of $U$ and $\cdot^c$ denotes the complement. 

\begin{proposition}
\label{lem:deflyapunov1}
Under the assumptions of \cref{lemma9.8}, and the condition $\tilde F \geq 0$, the function $W_1$ from \cref{eq:defW1} is a Lyapunov function on $\mathcal{B}(\frac{a}{\sqrt{\beta}}, x^*)$ with respect to the Langevin generator
\[
\tilde{\operatorname{L}} = - \Grad{h}\tilde F + \frac{1}{\beta} \Delta_h,
\]
whenever 
\[
\frac{a}{\sqrt{\beta}} \leq \frac{D}{3}, 
\]
where $D> 0$ is any lower bound on the least distance between any two critical points of $\tilde F$.
\end{proposition}
\begin{proof}
From the choice of $\frac{a}{\sqrt{\beta}}$, we can decompose
\[
\overline{U}\Big(\frac{a}{2\sqrt{\beta}}, \tilde{\mathcal{S}}\Big)^c = \mathcal{B}\Big(\frac{a}{\sqrt{\beta}}, x^*\Big) \cup \overline{U}\Big(\frac{a}{2\sqrt{\beta}}, \tilde{\mathcal{C}}\Big)^c.
\]
Now, let us write
\begin{align*}
\tilde{\operatorname{L}} W_1(x) &= \Big\langle -\Grad{h}\tilde F(x), \Grad{h}\exp\Big(\frac{\beta}{2}\tilde F(x)\Big)\Big\rangle_h + \frac{1}{\beta} \Delta_h \exp\left(\frac{\beta}{2}\tilde F(x)\right) 
\\&= -\frac{\beta}{2} W_1(x) |\Grad{h}\tilde F(x)|^2_h + \frac{1}{\beta} \left(\frac{\beta}{2}W_1(x) \Delta_h \tilde F(x) + \frac{\beta^2}{4} W_1(x) |\Grad{h}\tilde F(x)|^2_h\right)
\\&= W_1(x) \Big(\frac{1}{2} \Delta_h \tilde F(x) - \frac{\beta}{4} |\Grad{h}\tilde F(x)|^2_h\Big).
\end{align*}

We can now apply \cref{lemma9.8} and the trivial bound involving the Lipschitz constant $A_2$, 
\[
\frac{1}{2} \Delta_h \tilde F(x) - \frac{\beta}{4} |\Grad{h}\tilde F(x)|^2_h \leq \frac{1}{2}\Delta_h \tilde F(x) \leq \frac{\dim(B) A_2}{2},
\]
which holds for every $x \in B$, to conclude that
\begin{equation}
\label{lyapunovW1}
\frac{\tilde{\operatorname{L}} W_1(x)}{W_1(x)} \leq - A_2 \dim(B) + \frac{3}{2} A_2 \dim(B) \mathds{1}_{\mathcal{B}\big(\frac{a}{\sqrt{\beta}}, x^*\big)}(x),
\end{equation}
for every $x \in \overline{U}\Big(\frac{a}{2\sqrt{\beta}}, \tilde{\mathcal{S}}\Big)^c $. \cref{lyapunovW1} proves that $W_1$ is a Lyapunov function on $\mathcal{B}\Big(\frac{a}{\sqrt{\beta}}, x^*\Big)$, as desired.
\end{proof}

\subsubsection{Second Lyapunov function}
\label{sec:secondlyapunovfunction}

Now, for every saddle point $y$ of $\tilde F$, we can construct a quasi-Lyapunov function in a sufficiently small ball centered around it. In order to define these quasi-Lyapunov functions, we will make use of the following auxiliary function, which can be understood as a \textit{projected distance function}.  
\begin{definition}
\label{def:tilder}
Let $y \in B$ be a fixed saddle point of $\tilde F$, and let $v \in T_y B$ be a fixed unit vector. For every $x$ outside the cut locus of $y$, we define, in normal coordinates 
\[
\tilde{r}_{y, v}(x) := \langle v, \log_y x\rangle.
\]
\end{definition}

The main reason why $\tilde r_{y, v}$ will be helpful is because the gradient of $\tilde r_{y, v}$ is always \textit{in the direction} of $v$, which is key in the proof of \cref{prop9.6}. This is clearly not the case for the gradient of the distance function $x \mapsto d_h(x, y)$. We also observe that, by Cauchy-Schwarz inequality
\begin{equation}
\label{eq:ineqtilderd}
|\tilde r_{y, v}(x)| \leq |\log_y x|_h = d_h(x, y),
\end{equation}
for all $y \in B$ and any $x \in B$ outside the cut locus of $y$. 

The quasi-Lyapunov functions are defined as follows: let $y \in B$ be a fixed saddle point of $\tilde F$, let $v \in T_y B$ be a unit eigenvector of $\nabla^2 \tilde F(y)$ that corresponds to the minimum eigenvalue of $\nabla^2 \tilde F(y)$. By \cref{modassumption3.5.1} we know that $\nabla^2 \tilde F(y)[v, v] = -\lambda_v$, for some $\lambda_v \geq \lambda_*$. Given some $a, \beta \geq 1$ we define $W_2^y : \overline{\mathcal{B}(\frac{a}{\sqrt{\beta}}, y)} \to \bb{R}$ as 
\begin{equation}
\label{eq:definitionW2}
W_2^y(x) := \exp\Big\{\frac{\lambda_*}{4}(a^2 - \beta \tilde r_{y, v}(x)^2)\Big\}.    
\end{equation}

Let us now prove that $W_2^y$ is a quasi-Lyapunov function on $\mathcal{B}(\frac{a}{\sqrt{\beta}}, y)$. A crucial ingredient in the proof is the fact that $v$ is the eigenvector of $\nabla^2 \tilde F(y)$ associated with its minimum eigenvalue. Actually, the fact that $W_2^y$ is a quasi-Lyapunov function on $\mathcal{B}(\frac{a}{\sqrt{\beta}}, y)$ encodes that the Langevin process $\tilde X_t$ is able to escape \textit{fast} from the saddle points.

\begin{proposition}[Second Lyapunov function]
\label{prop9.6}
Let $(B, h)$ be a compact manifold which is furthermore a symmetric space. Assume that there exists a constant $\mathbf{K} \geq 1$ such that for every $p\in B$, in normal coordinates centered at $p$, $|R_{ijk}^l(p)| \leq \mathbf{K}$ for every $i, j, k, l \in \{1, \dotsc, \dim(B)\}$, where $R$ denotes the Riemann curvature tensor of $(B, h)$. Let $\tilde F$ be a smooth function on $B$ satisfying \cref{modassumption3.5.1}.

Let $a \geq 1$,  
\begin{align*}
\beta &\geq \max\Bigg\{192^2 \dim(B)^5 A^2_2 A^2_3 \mathbf{K}^2 a^6, \frac{4a^2}{i(B)^2}\Bigg\},
\end{align*}
where $A_2, A_3 \geq 1$ are the Lipschitz constants of the gradient and the Hessian of $\tilde F$, and $i(B)$ denotes the injectivity radius of $(B, h)$. Let $W_2^y$ be defined as in \cref{eq:definitionW2}. Then, $W_2^y$ is a quasi-Lyapunov function on $\mathcal{B}( \frac{a}{\sqrt{\beta}},y)$ with constant $\frac{\lambda_*}{8}$, i.e. $W_2^y \in C^{\infty}(\overline{\mathcal{B}(\frac{a}{\sqrt{\beta}},y)})$, $W^y_2 \geq 1$ and 
\[
(\tilde{\operatorname{L}} W_2^y)(x) \leq -\frac{\lambda_*}{8} W_2^y(x),
\]
for every $x \in \mathcal{B}(\frac{a}{\sqrt{\beta}},y)$.
\end{proposition}
\begin{proof}
For the sake of clarity, we will use $\tilde r$ as a shorthand for $\tilde{r}_{y, v}$. First, note that $W_2^y \in C^{\infty}(\overline{\mathcal{B}(\frac{a}{\sqrt{\beta}},y)})$ follows by definition, as we are assuming that $\frac{a}{\sqrt{\beta}} \leq i(B)/2$. Also, since $|\tilde r(x)| \leq \frac{a}{\sqrt{\beta}}$ in $\mathcal{B}(\frac{a}{\sqrt{\beta}},y)$, it holds that 
\[
a^2 - \beta \tilde r(x)^2 \geq a^2 - \beta \frac{a^2}{\beta} \geq 0,
\]
which implies that $W_2^y \geq 1$ in $\mathcal{B}(\frac{a}{\sqrt{\beta}},y)$. 

In order to show that $W_2^y$ is a quasi-Lyapunov function, we need an explicit description of $\tilde{\operatorname{L}} W_2^y$. Let us first compute some intermediate expressions. First, by the chain rule,
\[
\Delta_h\Big(\exp\Big\{-\frac{\lambda_* \beta}{4} \tilde r^2\Big\}\Big)(x) = -\frac{\lambda_* \beta}{4}\exp\Big\{-\frac{\lambda_* \beta}{4} \tilde r^2(x)\Big\}  \Big(\Delta_h \tilde r^2(x) -\frac{\lambda_* \beta}{4}|\Grad{h} \tilde r^2(x)|^2_h \Big),
\]
which allows us to write the Laplacian of $W_2^y$ as 
\[
\Delta_h W_2^y(x) = -\frac{\lambda_* \beta}{4}W_2^y(x) \Big(\Delta_h \tilde r^2(x) -\frac{\lambda_* \beta}{4}|\Grad{h} \tilde r^2(x)|^2_h \Big).
\]
Applying again the chain rule for the gradient of $W_2^y$ we obtain
\[
\Grad{h} W_2^y(x) = \Grad{h}\exp\Big\{\frac{\lambda_*}{4}(a^2 - \beta \tilde r(x)^2)\Big\} = -\frac{\lambda_* \beta}{4}W_2^y(x) \Grad{h} \tilde r^2(x).
\]
Thus, $\tilde{\operatorname{L}} W_2^y$ can be written as
\begin{align}
\label{eq:expressionLW}
\begin{split}
(\tilde{\operatorname{L}} W_2^y)(x) &= \langle -\Grad{h}\tilde F(x), \Grad{h} W_2^y(x)\rangle_h + \frac{1}{\beta} \Delta_h W_2^y(x)
\\&= W_2^y(x) \Big\{ -\frac{\lambda_* \beta}{4} \langle -\Grad{h}\tilde F(x), \Grad{h} \tilde r^2(x)\rangle_h - \frac{\lambda_*}{4}  \Big(\Delta_h \tilde r^2(x) -\frac{\lambda_* \beta}{4}|\Grad{h} \tilde r^2(x)|^2_h \Big)\Big\}
\\&=W_2^y(x) \Big\{-\frac{\lambda_* \beta}{2} \tilde{\operatorname{L}}\big(\frac{1}{2} \tilde r^2\big)(x) + \frac{\lambda^2_* \beta}{4^2} |\Grad{h} \, \tilde r^2(x)|^2_h\Big\}.
\end{split}
\end{align}

In order to show that $W_2^y$ is a quasi-Lyapunov function, we need to find an upper bound for the expression in brackets on the right-hand side. Let us focus on $\tilde{\operatorname{L}}\big(\frac{1}{2} \tilde r^2\big)$, which can be written as
\begin{align}
\label{eq:expressionLrtilde}
\begin{split}
\big(\tilde{\operatorname{L}} \frac{1}{2}\tilde r^2\big)(x) &=  \langle -\Grad{h}\tilde F(x), \frac{1}{2} \Grad{h} \tilde{r}^2(x)\rangle_h + \frac{1}{2\beta} \Delta_h \tilde{r}^2(x)\\&=\langle -\Grad{h}\tilde F(x), \tilde{r}(x)\Grad{h}\tilde{r}(x) \rangle_h + \frac{1}{\beta}\left(\tilde{r}(x)\Delta_h \tilde{r}(x) + |\Grad{h}\tilde{r}(x)|^2_h\right). 
\end{split}
\end{align}
The first summand on the right-hand side of the last line in \cref{eq:expressionLrtilde} can be lower bounded as follows: since
\[
\beta \geq 192^2 \dim(B)^5 A^2_2 A^2_3 \mathbf{K}^2 a^6 \geq 72^2 \dim(B)^5 \mathbf{K}^2 a^2,
\]
from \cref{AuxLemma3} it follows that 
\[
|\langle -\Grad{h}\tilde F(x), \Grad{h} \tilde r(x) \rangle_h - \lambda_v \tilde r(x)| \leq (24\dim(B)^{5/2}A_2 + A_3)\frac{a^2}{\beta}, 
\]
for every $x \in \mathcal{B}(\frac{a}{\sqrt{\beta}}, y)$ and so 
\begin{align}
\label{eq:auxbound1}
\begin{split}
\tilde r(x)\langle -\Grad{h}\tilde F(x), \Grad{h} \tilde r(x) \rangle_h  &\geq \lambda_v \tilde r^2(x) - |\tilde r(x)| (24\dim(B)^{5/2}A_2 + A_3)\frac{a^2}{\beta}
\\&\geq \lambda_v \tilde r^2(x) - (24\dim(B)^{5/2}A_2 + A_3)\frac{a^3}{\beta\sqrt{\beta}},
\end{split}
\end{align}
where we used the bound $|\tilde r(x)| \leq \frac{a}{\sqrt{\beta}}$ for the last inequality.

Next, using again that $\frac{a}{\sqrt{\beta}} \leq \min\Big\{\frac{i(B)}{2}, \frac{1}{72 \dim(B)^{5/2} \mathbf{K}}\Big\}$ by assumption, we can apply \cref{lem:boundtermgradlaplacian} to conclude that\footnote{Despite $\Delta_h \tilde r$ being zero in the case of flat space and its gradient having constant norm one, bounding the Laplacian or the norm of the gradient of $\tilde r$ becomes non-trivial even in very simple settings such as the two-dimensional round sphere (cf. \cref{sec:notrivialboundlaplacian}). Nevertheless, when restricting $x$ to a sufficiently small neighborhood around $y$, we are able to obtain bounds for such quantities.}
\begin{align}
\label{eq:auxbound2}
\begin{split}
|\Grad{h}\tilde r(x)|^2_h &\leq 2,\\
|\Grad{h}\tilde r(x)|^2_h + \Tilde{r}(x)\Delta_h\Tilde{r}(x) &\geq \frac{1}{2},            
\end{split}
\end{align}
for every $x \in \mathcal{B}(\frac{a}{\sqrt{\beta}}, y)$. Inserting \cref{eq:auxbound1,eq:auxbound2} into \cref{eq:expressionLrtilde} shows that
\begin{align*}
\big(\tilde{\operatorname{L}} \frac{1}{2}\tilde r^2\big)(x) &\geq  \lambda_v \tilde r^2(x) - (24\dim(B)^{5/2}A_2 + A_3)\frac{a^3}{\beta\sqrt{\beta}} + \frac{1}{2\beta}
\\&= \lambda_v \tilde r^2(x) + \frac{1}{\beta}\Big(\frac{1}{2} - (24\dim(B)^{5/2}A_2 + A_3)\frac{a^3}{\sqrt{\beta}}\Big)
\end{align*}
for every $x \in \mathcal{B}(\frac{a}{\sqrt{\beta}}, y)$. As it also holds that 
\[
\beta \geq 192^2 \dim(B)^5 A^2_2 A^2_3 \mathbf{K}^2 a^6 \geq 4^2 (24\dim(B)^{5/2}A_2 + A_3)^2 a^6
\]
we find that 
\[
\frac{1}{2} - (24\dim(B)^{5/2}A_2 + A_3)\frac{a^3}{\sqrt{\beta}} \geq  \frac{1}{4}.
\]
This bound implies that, for every $x \in \mathcal{B}(\frac{a}{\sqrt{\beta}}, y)$
\[
\big(\tilde{\operatorname{L}} \frac{1}{2}\tilde r^2\big)(x) \geq \lambda_v \tilde{r}(x)^2 + \frac{1}{4\beta} \geq \lambda_* \tilde{r}(x)^2 + \frac{1}{4\beta}.
\]
Inserting this bound into \cref{eq:expressionLW} we see that
\begin{align*}
(\tilde{\operatorname{L}} W_2^y)(x) &= W_2^y(x) \Big\{-\frac{\lambda_* \beta}{2} \tilde{\operatorname{L}}(\frac{1}{2} \tilde r^2)(x) + \frac{\lambda^2_* \beta}{4^2} |\Grad{h} \, \tilde r^2(x)|^2_h\Big\} 
\\&\leq W_2^y(x) \Big\{-\frac{\lambda_* \beta}{2} \big(\lambda_* \tilde r^2(x) + \frac{1}{4\beta}\big) + \frac{\lambda^2_* \beta}{4^2} |\Grad{h}\, \tilde r^2(x)|^2_h\Big\} 
\end{align*}
Next, simply using the bound obtained in \cref{eq:auxbound2} for the norm of the gradient of $\tilde r$ in $\mathcal{B}(\frac{a}{\sqrt{\beta}}, y)$, we see that 
\begin{align*}
(\tilde{\operatorname{L}} W_2^y)(x) &\leq W_2^y(x) \Big\{-\frac{\lambda_* \beta}{2} \big(\lambda_* \tilde r^2(x) + \frac{1}{4\beta}\big) + \frac{\lambda^2_* \beta}{4} \tilde r^2(x)|\Grad{h} \tilde r(x)|_h^2\Big\} 
\\&\leq  W_2^y(x) \Big\{-\frac{\lambda_* \beta}{2} \big(\lambda_* \tilde r^2(x) + \frac{1}{4\beta}\big) + \frac{\lambda^2_* \beta}{2} \tilde r^2(x)\Big\}  
\\&= -\frac{\lambda_*}{8}W_2^y(x), 
\end{align*}
finishing the proof. 
\end{proof}
\subsection{Proof of Theorem \ref{prop9.12}}
\label{sec2.3}

Lastly, in this subsection we will prove that the Markov triple $(B, \tilde \nu, \tilde \Gamma)$ satisfies a Poincaré inequality with a---positive---constant that only depends on the value of $\lambda_*$, which is the positive constant controlling the escape directions from the saddle points (cf. \cref{modassumption3.5.1}), i.e. 
\[
\lambda_{\min}\big(\nabla^2 \tilde F(y)\big) \leq -\lambda_* < 0,\quad \forall y \in \tilde{\mathcal{S}}.
\]

Let $a, \beta > 0$ be fixed. We define the following sets in order to lighten the notation. 
\begin{align}
\label{eq:subsetsB}
\begin{split}
G_1 &:= \overline{U}\Big(\frac{a}{2\sqrt{\beta}}, \tilde{\mathcal{S}}\Big) = \Big\{x \in B\, |\, d_h(x, \tilde{\mathcal{S}})^2 \leq \frac{a^2}{4\beta}\Big\},\\
G_2 &:= U\Big(\frac{a}{\sqrt{\beta}}, \tilde{\mathcal{S}}\Big)  = \Big\{x \in B \, |\, d_h(x, \tilde{\mathcal{S}})^2 < \frac{a^2}{\beta}\Big\},\\
G_3 &:= \mathcal{B}(\frac{a}{\sqrt{\beta}},x^*) = \Big\{x \in B \, |\, d_h(x, x^*)^2 < \frac{a^2}{\beta}\Big\},\\
G_4 &:= B \setminus (G_1 \cup G_3),
\end{split}
\end{align}
where $x^*$ is the unique minimum of $\tilde F$. See \cref{fig:neighborhoods} for a visual representation.

\begin{figure}
    \centering
    \ctikzfig{TikzFigures/MapEnvironments}
    \caption{Sketch of the subsets defined in \cref{eq:subsetsB}. We have used striped filling to denote that some region is included in two subsets. The neighborhoods around the critical points are disjoint by \cref{modassumption3.4}, after choosing a sufficiently small value for $\frac{a}{\sqrt{\beta}}$.}
    \label{fig:neighborhoods}
\end{figure}

\subsubsection{Local Poincaré inequality near the global minimum}
\label{sec:localpoincare} 

Let us derive a Poincaré inequality for $(G_3, \left.\tilde \nu\right|_{G_3}, \tilde \Gamma|_{G_3})$, where $\left.\tilde \nu\right|_{G_3}$ denotes the Gibbs density restricted to $G_3$,
\[
\tilde \nu|_{G_3}(x) := \frac{\tilde \nu(x)}{\int_{G_3} \tilde \nu(y)dy} \mathds{1}_{G_3}(x),
\]
and $\tilde \Gamma|_{G_3}$ denotes the operator $\tilde \Gamma$ restricted to act on functions $f \in C^2(G_3)$. This inequality will follow from the geodesic convexity of $G_3$ along with the so-called \textit{curvature-dimension condition},
\[
\nabla^2\tilde F + \frac{1}{\beta}\textup{Ric}_h \geq \kappa h,
\]
with $\kappa > 0$. In fact, whenever this inequality holds on a \textit{geodesically convex} connected and open subset $A \subset B$, the triple $(A, \left.\tilde \nu\right|_A, \tilde \Gamma|_A)$ satisfies a $\textup{PI}(\kappa/2)$. We refer the reader to \cref{BakryEmeryandLyapunov} for more details. The geodesic convexity of the subset $A$ is crucial if one wants to directly obtain a Poincaré inequality from the curvature-dimension condition. For an in-depth study on how the convexity of a subset affects the interplay between the curvature-dimension condition and the Poincaré inequality, we refer to \cite{WangAnalysisforDiffusion,wang2006functional,Qian1997}.

In the following statement, which is analogous to \cite[Lemma D.12]{LiErd2022Supp}, we will obtain a Poincaré inequality for $(G_3, \left.\tilde \nu\right|_{G_3}, \tilde \Gamma|_{G_3})$ by using the curvature-dimension condition, as $G_3$ can always be assumed to be convex by taking $\frac{a}{\sqrt{\beta}}$ sufficiently small. Note that, unlike in \cite{LiErd2022Supp}, we are assuming that $\tilde F$ only has one global minimum in $B$ which simplifies the proof. 

\begin{lemma}[Poincaré Inequality on $G_3$]
\label{lemma9.9}
Suppose that $\tilde F$ satisfies \crefrange{modassumption3.3}{modassumption3.6} and assume that the Ricci curvature of $(B, h)$ is lower-bounded by some non-positive constant, i.e. $\textup{Ric}_{h} \geq -R_{B}\,h $ for some $R_{B} \geq 0$. Then for all choices of $\beta > 0$ such that
\[
\beta \geq \max \left\{a^2\frac{4A^2_3}{\lambda_*^2}, \frac{4R_{B}}{\lambda_*}, \frac{a^2}{\mathit{conv}(B)^2}\right\}, 
\]
where $\lambda_*$ is defined in \cref{modassumption3.5.1,modassumption3.5.2}, and $\mathit{conv}(B)$ denotes the convexity radius of $(B, h)$, it holds that the Markov triple $(G_3, \left.\tilde\nu\right|_{G_3}, \tilde \Gamma|_{G_3})$, satisfies a $\textup{PI}(\frac{\lambda_*}{8})$, i.e. 
\[
\int_{G_3} f^2\,d\tilde\nu|_{G_3} - \Big(\int_{G_3} f \, d\tilde\nu|_{G_3}\Big)^2 \leq \frac{8}{\lambda_* \beta} \int_{G_3} |\Grad{h}f|_h^2\, d\tilde\nu|_{G_3},
\]
for every $f \in C^2(G_3) \cap H^1(G_3)$, where $H^1$ denotes the Sobolev space of $L^2$ functions on $G_3$ with $L^2$ gradient (see \cref{sec:obolevspaces}). 
\end{lemma}
\begin{proof}
Let $x^*$ be the unique minimum of $\tilde F$. From the definition of $\tilde F$ having an $A_3$-Lipschitz Hessian (cf. \cref{def:lipschitzfunction}) we know that for every $x \in G_3$
\[
\max_{\substack{s \in T_{x^*}B\\|s|_h=1}} |\textup{P}^{-1}_{\log_{x^*}x} \circ \nabla^2 \tilde F(x) \circ \textup{P}_{\log_{x^*}x}[s] - \nabla^2 \tilde F(x^*)[s]|_h \leq A_3 d_h(x,x^*),
\]
where $\textup{P}_{\log_{x^*}x}$ denotes the parallel transport from $x^*$ to $x$ along the unique geodesic joining the two points. Since the parallel transport is an isometry, we can use Weyl's inequality to conclude that for every $x \in G_3$,
\begin{align*}
|\lambda_1(\nabla^2 \tilde F(x)) - \lambda_1(\nabla^2 \tilde F(x^*))| &\leq \max_{\substack{s \in T_{x^*}B\\|s|_h=1}} |\textup{P}^{-1}_{\log_{x^*}x} \circ \nabla^2 \tilde F(x) \circ \textup{P}_{\log_{x^*}x}[s] - \nabla^2 \tilde F(x^*)[s]|_h 
\\&\leq A_3 d_h(x,x^*),
\end{align*}
where $\lambda_1(\nabla^2\tilde F(x))$ denotes the smallest eigenvalue of the Hessian of $\tilde F$ at $x$. This way, we can apply \cref{modassumption3.5.2} to conclude that 
\[
\lambda_1(\nabla^2 \tilde F(x)) \geq \lambda_1(\nabla^2 \tilde F(x^*)) - A_3 d_h(x, x^*) \geq \frac{\lambda_*}{2},\quad \textup{whenever } d_h(x, x^*) \leq \frac{\lambda_*}{2A_3}.
\]
The above inequality is fulfilled on $G_3 = \{x \in B : d_h(x, x^*)^2 < \frac{a^2}{\beta}\}$, as by assumption $\frac{a^2}{\beta} \leq \frac{\lambda^2_*}{4 A^2_3}$. 

For this reason, and the fact that we assumed that $\textup{Ric}_h \geq -R_B$ where $R_B \geq 0$, then, given our choice of $\beta$,
\[
\nabla^2 \tilde F(x) + \frac{1}{\beta}\textup{Ric}_h \geq \left(\frac{\lambda_*}{2} - \frac{1}{\beta}R_B\right)h \geq \frac{\lambda_*}{4}h.
\]

Since we are ensuring that $\frac{a}{\sqrt{\beta}} \leq \mathit{conv}(B)$, we know that $(G_3, \left.\tilde\nu\right|_{G_3}, \tilde \Gamma|_{G_3})$ satisfies a $\textup{PI}(\lambda_*/8)$. See \cref{rem:CDimpliesPI} for more details. 
\end{proof}

\subsubsection{Extending the Poincaré inequality}
\label{sec:globalpoincare}

Having obtained a local Poincaré inequality on $G_3$, in this subsection we will extend it to obtain a Poincaré inequality for $(B, \tilde \nu, \tilde \Gamma)$. This will be done using the Lyapunov functions constructed in \cref{sec2.2},
\[
W_1(x) = \exp\Big(\frac{\beta}{2}\tilde F(x)\Big),\quad \textup{and}\quad W^y_2(x) = \exp\Big\{\frac{\lambda_*}{4}(a^2 - \beta \tilde r_{y, v}(x)^2)\Big\},\quad \forall y \in \tilde{\mathcal{S}}. 
\]

In order to extend our Poincaré inequality from $G_3$ to $B$, smooth bump functions will come in handy. The following two lemmas allow to construct one that vanishes on $G_1$, is equal to one on $G_2^c$, and such that its derivative is upper bounded by a suitable constant. 

\begin{lemma}[{\cite[Theorem 6.31]{lee2018introductionRiemannian}}]
\label{AuxLemma1}
Let $y \in B$ be some fixed point. The distance function $x \mapsto d_h(x, y)$ is a smooth function on $\mathcal{B}(i(B),y) \setminus \{y\}$ with unit gradient. 
\end{lemma}

\begin{lemma}
\label{lem:defbumpfunction}
Let 
\[
r := \frac{a}{\sqrt{\beta}} \leq \min\Big\{\frac{i(B)}{2}, \frac{D}{3}\Big\},
\]
where $D$ is the lower bound on the distance between any two saddle points of $\tilde F$ (cf. \cref{modassumption3.4}). Then  for every $\varepsilon > 0$ there exists a smooth partition function $\chi_{\varepsilon}: B \to [0, 1]$ such that $\chi_{\varepsilon} \equiv 1$ on $G_2^c$, $\chi_{\varepsilon} \equiv 0$ on $G_1$ and such that
\[
\norm{\tilde \Gamma(\chi_{\varepsilon})}_\infty \leq \frac{1}{\beta}\Big(\frac{2}{r} + \varepsilon\Big)^2.
\]
\end{lemma}
\begin{proof}
First, we can define a smooth bump function $\psi_{\varepsilon}: \bb{R} \to [0, 1]$ such that
\[
\psi_{\varepsilon}(x) = \begin{cases}
0, \quad & x \leq r/2,\\
\mathit{increasing}, \quad &x\in (r/2, r),\\
1,\quad & x\geq r,
\end{cases}
\]
with $\norm{\psi'_{\varepsilon}}_{\infty} \leq \frac{2}{r} + \varepsilon$ following standard methods (cf. \cite[Lemma G.13]{LiErd2022Supp}). Now, we define $\chi_{\varepsilon}: B \to [0, 1]$ as 
\[
\chi_{\varepsilon}(x) := \psi_{\varepsilon}(d_h(x,\tilde{\mathcal{S}})).
\]
Note that $\chi_{\varepsilon}(x)$ is constant in the interior of $G_1$ and in the complement of $G_2$. Furthermore, since we are assuming that $r \leq i(B)/2$, it holds by \cref{AuxLemma1} that the distance function is smooth in $G_2 \setminus \tilde{\mathcal{S}}$, and so $\chi_{\varepsilon}$ is smooth in $B$ and 
\[
\Grad{h}\chi_{\varepsilon} = \Grad{h}\psi_{\varepsilon}(d_h(x, \tilde{\mathcal{S}})) = \psi'_{\varepsilon}(d_h(x, \tilde{\mathcal{S}})) \Grad{h}d_h(x, \tilde{\mathcal{S}}),
\]
for every $x \in B$ such that $r/2 < d_h(x, \tilde{\mathcal{S}}) < r$ and $\Grad{h}\chi_{\varepsilon} \equiv 0$ elsewhere. Again, by \cref{AuxLemma1}, the distance function has a gradient with constant norm one, and so it holds that 
\[
|\Grad{h}\chi_{\varepsilon}(x)|_h = |\psi'(d_h(x, \tilde{\mathcal{S}}))|
\]
for every $x \in B$. For this reason, for every $x \in B$ such that $r/2 < d_h(x, \tilde{\mathcal{S}}) < r$,
\[
|\tilde \Gamma(\chi_{\varepsilon}(x))| = \frac{1}{\beta}|\Grad{h}\chi_{\varepsilon}(x)|^2 = \frac{1}{\beta}|\psi'_{\varepsilon}(d_h(x, \tilde{\mathcal{S}}))|^2 \leq \frac{1}{\beta}\Big(\frac{2}{r} + \varepsilon\Big)^2.
\]
Otherwise, $\tilde \Gamma(\chi_{\varepsilon}(x)) = 0$ which implies that
\[
\lVert\tilde \Gamma(\chi_{\varepsilon})\rVert_{\infty} = \frac{1}{\beta}  \lVert\Grad{h}\chi_{\varepsilon}\rVert^2_\infty = \frac{1}{\beta}\lVert\psi'_{\varepsilon}(d_h(\cdot, \tilde{\mathcal{S}}))\rVert^2_\infty \leq \frac{1}{\beta}\Big(\frac{2}{r} + \varepsilon\Big)^2.
\]
\end{proof}

Our final tool to derive a Poincaré inequality is an integration by parts formula on submanifolds of $B$. 
\begin{lemma}
\label{Neumannboundarycondition}
Let $\tilde{\operatorname{L}}$ and $\tilde\nu$ be as defined in \cref{eq:defMarkovB}, and let $N$ be a submanifold of $(B, h)$ of codimension $0$ with smooth boundary. For every pair of functions $W$, and $\phi$ in $C^2(\overline{N})$ such that $\phi$ vanishes at $\partial N$, it holds that 
\[
\int_N \phi \tilde{\operatorname{L}}W \,d\tilde\nu = - \frac{1}{\beta} \int_N \langle \Grad{h}\phi,\Grad{h}W\rangle_h \,d\tilde\nu .
\]
\end{lemma}
\begin{proof}
By the definition of $\tilde \nu$ and $\tilde{\operatorname{L}}$,
\begin{align}
\int_N \phi \tilde{\operatorname{L}}W \,d\tilde\nu &= \int_N \phi\langle-\Grad{h}\tilde F, \Grad{h}W\rangle_h \,d\tilde\nu + \frac{1}{\beta} \int_N \phi \Delta_h W \,d\tilde\nu 
\\&= \int_N \phi\langle-\Grad{h}\tilde F, \Grad{h}W\rangle_h \,d\tilde\nu + \frac{1}{\tilde Z\beta} \int_N \phi e^{-\beta \tilde F}\Delta_h W\; \textup{dVol}_h.
\end{align}
Using Green's identity, since we are assuming that $\phi$ vanishes at $\partial N$, it holds that 
\begin{align*}
\int_N \phi \tilde{\operatorname{L}}W \,d\tilde\nu &= \int_N \phi\langle-\Grad{h}\tilde F, \Grad{h}W\rangle_h \,d\tilde\nu - \frac{1}{\tilde Z\beta} \int_N \langle \Grad{h}(\phi e^{-\beta \tilde F}),\Grad{h}W\rangle_h\; \textup{dVol}_h
\\&= \int_N \phi\langle-\Grad{h}\tilde F, \Grad{h}W\rangle_h \,d\tilde\nu - \frac{1}{\beta} \int_N \langle \Grad{h}\phi,\Grad{h}W\rangle_h \,d\tilde\nu
\\&\quad + \int_N \phi \langle \Grad{h}\tilde F,\Grad{h}W\rangle_h \,d\tilde\nu
\\&= - \frac{1}{\beta} \int_N \langle \Grad{h}\phi,\Grad{h}W\rangle_h \,d\tilde\nu.
\end{align*}  
\end{proof}

We are now ready to prove \cref{prop9.12}. The proof strategy adopted in the setting considered in \cite{LiErd2022} remains valid in ours. We give here the full proof for the sake of completeness. 

\begin{proof}[Proof of \cref{prop9.12}]
Let $f \in C^2(B)$. Without loss of generality, we can assume that $\int_{G_3} f \, d\tilde\nu = 0$, as 
\[
\int_B f^2 \, d\tilde\nu - \Big(\int_B f \,d\tilde \nu \Big)^2 \leq \int_B (f + c)^2 \, d\tilde\nu,
\]
for every $c \in \bb{R}$. Our goal will be to upper bound the right-hand side of this inequality by some multiple of $\int_B |\Grad{h}f|^2_h \,d\tilde\nu$ in order to obtain a Poincaré inequality. 

Let $\varepsilon > 0$ and let $\chi_{\varepsilon}$ be the bump function from \cref{lem:defbumpfunction}. Since $\chi_{\varepsilon} \equiv 0$ in $G_1$ and $\chi_{\varepsilon} \equiv 1$ in $G_2^c$, we know that
\begin{align}
\label{firstinequality}
\begin{split}
\int_B f^2 \,d\tilde\nu &= \int_B (f(1 - \chi_{\varepsilon}) + f \chi_{\varepsilon})^2 \,d\tilde\nu 
\\&\leq 2\int_B f^2(1-\chi_{\varepsilon})^2 \, d\tilde\nu + 2\int_B f^2\chi_{\varepsilon}^2 \, d\tilde\nu
\\&= 2\int_{G_2} f^2(1-\chi_{\varepsilon})^2 \, d\tilde\nu + 2\int_{G_1^c} f^2\chi_{\varepsilon}^2 \, d\tilde\nu.
\end{split}
\end{align}

Note that we are assuming that $\frac{a}{\sqrt{\beta}} \leq \frac{D}{3}$. Thus, as the critical points of $\tilde F$ are isolated by \cref{modassumption3.4} we can write $G_2$ as a disjoint union
\[
G_2 = \bigsqcup_{y_i \in \tilde{\mathcal{S}}} \mathcal{B}\big(\frac{a}{\sqrt{\beta}}, y_i\big),
\]
where each ball $\mathcal{B}\big(\frac{a}{\sqrt{\beta}}, y_i\big)$ is a connected component of $G_2$.

In \cref{prop9.6} we showed that for every $y \in \tilde{\mathcal{S}}$, $W^y_2\in C^\infty(\overline{\mathcal{B}\big(\frac{a}{\sqrt{\beta}}, y\big)})$, $W^y_2 \geq 1$ and 
\begin{equation}
\label{eq:lyapunovW2}
\frac{(\tilde{\operatorname{L}}W^y_2)(x)}{W^y_2(x)} \leq -\frac{\lambda_*}{8},\quad \forall x \in \mathcal{B}\big(\frac{a}{\sqrt{\beta}}, y\big). 
\end{equation}
Using \cref{eq:lyapunovW2} and \cref{Neumannboundarycondition} we can bound the first summand of \cref{firstinequality} as follows: 
\begin{align*}
\int_{G_2} f^2(1-\chi_{\varepsilon})^2 \,d\tilde\nu &= \sum_{y \in \tilde{\mathcal{S}}}\int_{\mathcal{B}\big(\frac{a}{\sqrt{\beta}}, y\big)} f^2(1-\chi_{\varepsilon})^2 \,d\tilde\nu
\\&\leq
-\frac{8}{\lambda_*}\sum_{y \in \tilde{\mathcal{S}}}\int_{\mathcal{B}\big(\frac{a}{\sqrt{\beta}}, y\big)} \frac{\tilde{\operatorname{L}}W^y_2}{W^y_2}f^2(1-\chi_{\varepsilon})^2 \,d\tilde\nu
\\&= \frac{8}{\lambda_* \beta} \sum_{y \in \tilde{\mathcal{S}}}\int_{\mathcal{B}\big(\frac{a}{\sqrt{\beta}}, y\big)} \Big\langle \Grad{h}\frac{f^2(1-\chi_{\varepsilon})^2}{W^y_2}, \Grad{h}W^y_2\Big\rangle_h \,d\tilde\nu,
\end{align*}
since $(1-\chi_{\varepsilon})$ vanishes at $\partial G_2$. Using the chain rule we can further bound, for every $\mathcal{B}\big(\frac{a}{\sqrt{\beta}}, y\big)$
\begin{align}
\label{eq:ineqchainrule}
\begin{split}
\frac{8}{\lambda_* \beta} \int_{\mathcal{B}\big(\frac{a}{\sqrt{\beta}}, y\big)} \Big\langle \Grad{h}&\frac{f^2(1-\chi_{\varepsilon})^2}{W^y_2}, \Grad{h} W^y_2\Big\rangle_h \,d\tilde\nu
\\&= \frac{16}{\lambda_*\beta} \int_{\mathcal{B}\big(\frac{a}{\sqrt{\beta}}, y\big)} \frac{f(1-\chi_{\varepsilon})}{W^y_2}\langle \Grad{h}(f(1-\chi_{\varepsilon})), \Grad{h}W^y_2\rangle_h \,d\tilde\nu
\\&\quad - \frac{8}{\lambda_*\beta}\int_{\mathcal{B}\big(\frac{a}{\sqrt{\beta}}, y\big)} \frac{f^2(1-\chi_{\varepsilon})^2}{(W^y_2)^2} |\Grad{h}W^y_2|^2_h \,d\tilde\nu
\\&= \frac{8}{\lambda_* \beta} \int_{\mathcal{B}\big(\frac{a}{\sqrt{\beta}}, y\big)} |\Grad{h}(f(1-\chi_{\varepsilon}))|^2_h \,d\tilde\nu
\\&\quad- \frac{8}{\lambda_* \beta} \int_{\mathcal{B}\big(\frac{a}{\sqrt{\beta}}, y\big)} \Big|\Grad{h}(f(1-\chi_{\varepsilon}))- \frac{f(1-\chi_{\varepsilon})}{W^y_2}\Grad{h}W^y_2\Big|^2_h \,d\tilde\nu
\\&\leq \frac{8}{\lambda_* \beta} \int_{\mathcal{B}\big(\frac{a}{\sqrt{\beta}}, y\big)} |\Grad{h}(f(1-\chi_{\varepsilon}))|^2_h \,d\tilde\nu
\\&= \frac{8}{\lambda_*}\int_{\mathcal{B}\big(\frac{a}{\sqrt{\beta}}, y\big)} \tilde \Gamma(f(1-\chi_{\varepsilon})) \,d\tilde\nu.
\end{split}
\end{align}
Lastly, using the fact that 
\begin{equation}
\label{productGamma}
\tilde \Gamma(fg) \leq 2(f^2 \tilde \Gamma(g) + g^2 \tilde \Gamma(f)),\quad \forall f,g \in C^2(B),
\end{equation}
it holds that 
\begin{align*}
\sum_{y \in \tilde{\mathcal{S}}}\frac{8}{\lambda_*}\int_{\mathcal{B}\big(\frac{a}{\sqrt{\beta}}, y\big)} \tilde \Gamma(f(1-\chi_{\varepsilon})) \,d\tilde\nu &= \frac{8}{\lambda_*}\int_{G_2} \tilde \Gamma(f(1-\chi_{\varepsilon})) \,d\tilde\nu 
\\&\leq \frac{16}{\lambda_*} \left(\int_{G_2} f^2 \tilde \Gamma(\chi_{\varepsilon}) \,d\tilde\nu + \int_{G_2} (1-\chi_{\varepsilon})^2 \tilde \Gamma(f) \,d\tilde\nu\right)
\\&\leq \frac{16\lVert\tilde \Gamma(\chi_{\varepsilon})\rVert_\infty}{\lambda_*} \int_{G_2} f^2 \,d\tilde\nu + \frac{16}{\lambda_*}\int_{G_2} \tilde \Gamma(f) \,d\tilde\nu
\\&\leq \frac{16\lVert\tilde \Gamma(\chi_{\varepsilon})\rVert_\infty}{\lambda_*} \int_B f^2 \,d\tilde\nu + \frac{16}{\lambda_*}\int_B \tilde \Gamma(f) \,d\tilde\nu.
\end{align*}

Thus, we have concluded that 
\begin{equation}
\label{conclusion1}
2\int_{G_2} f^2(1-\chi_{\varepsilon})^2 d\tilde \nu \leq \frac{32\lVert\tilde \Gamma(\chi_{\varepsilon})\rVert_\infty}{\lambda_*} \int_B f^2 d\tilde \nu + \frac{32}{\lambda_*}\int_B \tilde \Gamma(f)d\tilde \nu.
\end{equation}

Moreover, in \cref{lem:deflyapunov1} we proved that $W_1$ is smooth, $W_1 \geq 1$ and 
\begin{equation}
\label{eq:lyapunovW1}
\frac{\tilde{\operatorname{L}}W_1(x)}{W_1(x)} \leq -A_2 \dim(B) + \frac{3}{2}A_2 \dim(B) \mathds{1}_{G_3}(x), \quad \forall x \in G_1^c \cup G_3.
\end{equation}
Thus, the second summand of \cref{firstinequality} can be bounded similarly, using \cref{eq:lyapunovW1} and \cref{Neumannboundarycondition}, since $\chi_{\varepsilon}$ vanishes on $\partial G_1$. Recall that $G_1^c = G_3 \cup G_4$. Thus, 
\begin{align*}
\int_{G_1^c} f^2 \chi_{\varepsilon}^2 \,d\tilde\nu &\leq \frac{1}{A_2 \dim(B)} \int_{G_1^c} \frac{-\tilde{\operatorname{L}}W_1}{W_1} f^2 \chi_{\varepsilon}^2 \,d\tilde\nu + \frac{3}{2} \int_{G_3} f^2 \,d\tilde\nu
\\&= \frac{1}{A_2 \beta \dim(B)} \int_{G_1^c} \Big\langle \Grad{h}\frac{f^2 \chi_{\varepsilon}^2}{W_1}, \Grad{h}W_1\Big\rangle_h \,d\tilde\nu +  \frac{3}{2} \int_{G_3} f^2 \,d\tilde\nu
\\&\leq \frac{1}{A_2 \dim(B)} \int_{G_1^c} \tilde \Gamma(f\chi_{\varepsilon}) \,d\tilde\nu + \frac{3}{2}\int_{G_3} f^2 \,d\tilde\nu, 
\end{align*}
where the last inequality is obtained following an analogous procedure as in \cref{eq:ineqchainrule}. Now, since we already proved that $(G_3, \left.\tilde\nu\right|_{G_3}, \tilde \Gamma|_{G_3})$ satisfies a $\textup{PI}(\frac{\lambda_*}{8})$, it holds that 
\[
\int_{G_3} f^2 \,d\tilde\nu = \frac{Z_{G_3}}{Z_B} \int_{G_3} f^2 \,\left.d\tilde\nu\right|_{G_3} \leq \frac{Z_{G_3}}{Z_B} \frac{8}{\lambda_*}\int_{G_3} \tilde \Gamma(f) \, \left.d\tilde \nu\right|_{G_3} = \frac{8}{\lambda_*}\int_{G_3} \tilde \Gamma(f) \,d\tilde\nu,
\]
where $Z_{G_3} = \int_{G_3} d\tilde\nu$ and $Z_B = \int_B d\tilde\nu$. Thus, 
\begin{align*}
\int_{G_1^c} f^2 \chi_{\varepsilon}^2 \,d\tilde\nu  &\leq \frac{1}{A_2 \dim(B)} \int_{G_1^c} \tilde \Gamma(f\chi_{\varepsilon}) \,d\tilde\nu + \frac{3}{2}\int_{G_3} f^2 \,d\tilde\nu 
\\&\leq \frac{1}{A_2 \dim(B)} \int_{G_1^c} \tilde \Gamma(f\chi_{\varepsilon}) \,d\tilde\nu + \frac{12}{\lambda_*}\int_{G_3} \tilde \Gamma(f) \,d\tilde\nu
\\&\leq \frac{1}{A_2 \dim(B)} \int_{B} \tilde \Gamma(f\chi_{\varepsilon}) \,d\tilde\nu + \frac{12}{\lambda_*}\int_B \tilde \Gamma(f) \,d\tilde\nu.
\end{align*}
Using \cref{productGamma} we can upper bound this expression,
\begin{align*}
\frac{1}{A_2 \dim(B)} \int_{B} \tilde \Gamma(f\chi_{\varepsilon}) \,d\tilde\nu + \frac{12}{\lambda_*}\int_B \tilde \Gamma(f) \,d\tilde\nu &\leq \frac{2\lVert\tilde \Gamma(\chi_{\varepsilon})\rVert_\infty}{A_2 \dim(B)} \int_{B} f^2 \,d\tilde\nu
\\&\quad + \left(\frac{2}{A_2 \dim(B)}+ \frac{12}{\lambda_*}\right) \int_{B} \tilde \Gamma(f) \,d\tilde\nu,
\end{align*}
which ultimately leads to
\begin{equation}
\label{conclusion2}
2\int_{G_1^c} f^2 \chi_{\varepsilon}^2 \,d\tilde\nu \leq  \frac{4\lVert\tilde \Gamma(\chi_{\varepsilon})\rVert_\infty}{A_2 \dim(B)} \int_{B} f^2 \,d\tilde\nu + \left(\frac{4}{A_2 \dim(B)} + \frac{24}{\lambda_*}\right) \int_{B} \tilde \Gamma(f) \,d\tilde\nu.
\end{equation}

Inserting \cref{conclusion1,conclusion2} into \cref{firstinequality}, we can conclude that
\begin{align*}
\int_B f^2 \,d\tilde\nu \leq 4\lVert\tilde \Gamma(\chi_{\varepsilon})\rVert_\infty\left(\frac{8}{\lambda_*} + \frac{1}{A_2 \dim(B)} \right)\int_{B} f^2 \,d\tilde\nu + \left(\frac{4}{A_2 \dim(B)} + \frac{56}{\lambda_*}\right) \int_{B} \tilde \Gamma(f) \,d\tilde\nu.
\end{align*}

Recall that $\chi_{\varepsilon}$ is such that
\[
\lVert\tilde \Gamma(\chi_{\varepsilon})\rVert_{\infty} \leq \frac{1}{\beta}\Big(\frac{2}{r} + \varepsilon\Big)^2,
\]
where $r = \frac{a}{\sqrt{\beta}}$. This implies that 
\[
\int_B f^2 \,d\tilde\nu \leq \frac{4}{\beta}\Big(\frac{2\sqrt{\beta}}{a} + \varepsilon\Big)^2\left(\frac{8}{\lambda_*} + \frac{1}{A_2 \dim(B)} \right)\int_{B} f^2 \,d\tilde\nu + \left(\frac{4}{A_2 \dim(B)} + \frac{56}{\lambda_*}\right) \int_{B} \tilde \Gamma(f) \,d\tilde\nu.
\]
Since the choice of $\varepsilon> 0$ was arbitrary, we can now take $\varepsilon \to 0$ and conclude that 
\[
\int_B f^2 \,d\tilde\nu \leq \frac{16}{a^2}\left(\frac{8}{\lambda_*} + \frac{1}{A_2 \dim(B)} \right)\int_{B} f^2 \,d\tilde\nu + \left(\frac{4}{A_2 \dim(B)} + \frac{56}{\lambda_*}\right) \int_{B} \tilde \Gamma(f) \,d\tilde\nu.
\]

Therefore, since $A_2 \dim(B) \geq 1 \geq \lambda_*$, it holds that
\begin{align*}
1 - \frac{16}{a^2} \left(\frac{8}{\lambda_*} + \frac{1}{A_2 \dim(B)} \right) &\geq 1 - \frac{16}{a^2} \frac{9}{\lambda_*}
\\&\geq 1 - \frac{144}{a^2 \lambda_*}
\\&\geq \frac{1}{2},
\end{align*}
where the last inequality follows from the assumption that $a^2 \geq \frac{288}{\lambda_*}$. Therefore, we know that 
\[
\int_B f^2 \,d\tilde\nu \leq \Big(\frac{8}{A_2 \dim(B)} + \frac{112}{\lambda_*}\Big) \int_{B} \tilde \Gamma(f) \,d\tilde\nu.
\]
Lastly, using again that $A_2 \dim(B) \geq 1 \geq \lambda_*$ we conclude that 
\[
\int_B f^2 \,d\tilde\nu \leq \frac{120}{\lambda_*} \int_{B} \tilde \Gamma(f) \,d\tilde\nu,
\]
and so the claim follows. 
\end{proof}

%% file: Chapters/Lifting_Poincare.tex
\section{Poincaré inequalities through Riemannian submersions}
\label{SectionLift}

In this section, we will study the Poincaré inequality in the context of Riemannian submersions\footnote{For a short introduction to Riemannian submersions, see \cref{SectionExamples}.}. In particular, we will show how, given a Riemannian submersion $\pi: (M, g) \to (B, h)$, the existence of a Poincaré inequality on the base space $B$ can be \textit{lifted} to obtain a Poincaré inequality on the total space $M$ and vice versa.

Let $(M, g)$ be a complete Riemannian manifold---not necessarily compact\footnote{Although in the rest of the sections of this work $M$ is always assumed to be compact, we have decided to write this section for any Riemannian manifold $(M, g)$, not necessarily compact, as we believe that the result shown here is relevant on its own.}---and let $G$ be some compact and connected Lie group acting smoothly, freely and isometrically on $M$. Let $h$ be the metric on $M/G$ for which the projection $\pi: (M, g) \to (M/G, h)$ is a Riemannian submersion. As in \cref{sec:intro,SectionPI}, we consider a smooth function $F : M \to \bb{R}$ which is constant in the fibers of $\pi$---i.e. $F(gx) = F(x)$ for every $x \in M$ and every $g \in G$---and let $\tilde F$ be the unique function on $M/G$ satisfying $F = \tilde F \circ \pi$. We again define the Gibbs distributions on $M$ and $M/G$ as
\[
\nu = \frac{1}{Z} e^{-\beta F},\quad \tilde \nu = \frac{1}{\tilde Z} e^{-\beta \tilde F},
\]
where $Z$ and $\tilde Z$ are assumed to be positive and finite. 

Let us assume that a Poincaré inequality holds for either $M$ or $M/G$. That is, there exists some constant $\kappa > 0$ such that for every $f \in C^2(M) \cap H^1(M, \nu)$,
\[
\int_M f^2 \, d\nu - \Bigg(\int_M f \, d\nu\Bigg)^2 \leq \frac{1}{\kappa\beta} \int_M |\textup{grad}_g \ f|^2_g \, d\nu,
\]
or there exists some constant $\tilde \kappa > 0$ such that for every $f \in C^2(M/G) \cap H^1(M/G, \tilde \nu)$,
\[
\int_{M/G} f^2 \, d\tilde \nu - \Bigg(\int_{M/G} f \, d\tilde \nu\Bigg)^2 \leq \frac{1}{\tilde \kappa\beta} \int_{M/G} |\textup{grad}_{h} \ f|^2_{h} \, d\tilde \nu,
\]
where $H^1(M, \nu), H^1(M/G, \tilde \nu)$ denote the Sobolev spaces\footnote{For a short introduction to Sobolev spaces on manifolds, we refer the reader to \cref{sec:obolevspaces}.} of square-integrable functions with square-integrable gradient on $M$ and $M/G$ with respect to $\nu$ and $\tilde \nu$, respectively. In either case, assuming that one of the inequalities holds, we will identify sufficient conditions under which the other holds as well. In order to do so, we will make use of two main results. The first one, which will be discussed in \cref{sec:fubini}, can be understood as a version of Fubini's theorem for Riemannian submersions. The second, which will be shown in \cref{sec:liftingandloweringpoincare}, guarantees that compact connected manifolds with non-negative Ricci curvature always verify a Poincaré inequality with respect to their normalized Riemannian measure. We will then prove the equivalence between a Poincaré inequality with respect to $\nu$ in $M$ and one with respect to $\tilde \nu$ in $M/G$.

Note that, when $M$---and thus $M/G$---is compact, $C^2(M) \cap H^1(M, \nu) = C^2(M)$ and $C^2(M/G) \cap H^1(M/G, \tilde \nu) = C^2(M/G)$, and so the above Poincaré inequalities correspond to the Poincaré inequalities introduced in \cref{defPoincareIneq} for the Markov triples $(M, \nu, \Gamma)$ and $(M/G, \tilde \nu, \tilde \Gamma)$, where 
\[
\Gamma(f) = \frac{1}{\beta} |\Grad{g}f|^2_g, \quad \tilde \Gamma(\tilde f) = \frac{1}{\beta} |\Grad{h}\tilde f|^2_{h},
\]
for every $f \in C^2(M)$ and every $\tilde f \in C^2(M/G)$.

\subsection{Integrals and Riemannian submersions}
\label{sec:fubini}

The following theorem will be used to derive a variation of Fubini's theorem in the context of Riemannian submersions.

\begin{theorem}[{\cite[Chapter III, Section 2, Theorem 1]{sulanke1972differentialgeometrie}}]
\label{thm:fubinimanifolds}
Let $M$ and $N$ be two smooth manifolds of dimensions $m$ and $n$, respectively ($m \geq n$) with compatible orientations. Let $\varphi : M \to N$ be a $C^1$ function. Assume that the set of critical values of $\varphi$ has measure zero. Let $f : M \to \bb{R}$ be a measurable function. Let $\omega \in \Omega^{m-n}(M)$ be an $(m-n)$-form on $M$, and let $\eta \in \Omega^n(N)$ be an $n$-form on $N$. If $f$ is integrable with respect to $\omega \wedge \varphi^*\eta$ on $M$---where $\varphi^* \eta$ denotes the pullback of $\eta$ by $\varphi$---then, for almost every $x \in N$, 
\[
F(x) := \begin{cases}\int_{\varphi^{-1}(x)} f\, \omega, \quad& \textup{if } \varphi^{-1}(x) \neq \emptyset,\\
0\quad& \text{otherwise,}\end{cases}
\]
is well-defined, measurable, integrable with respect to $\eta$ and 
\[
\int_M f\; \omega \wedge \varphi^* \eta = \int_N F\; \eta.
\]
Moreover, if 
\[
\hat{F}(x) := \begin{cases}
\int_{\varphi^{-1}(x)} |f|\, |\omega|,\quad& \textup{if } \varphi^{-1}(x) \neq \emptyset,\\
0\quad& \text{otherwise,}
\end{cases}
\]
is integrable with respect to $\eta$ on $N$, then $f$ is integrable with respect to $\omega \wedge \varphi^* \eta$ on $M$. 
\end{theorem}

In particular, if we take $\varphi$ in \cref{thm:fubinimanifolds} to be a Riemannian submersion between two Riemannian manifolds, and define $\omega$ as the volume form of the fibers of $\varphi$, and $\eta$ as the volume form of the base space with respect to the induced metric, we obtain the following result.

\begin{corollary}
\label{LiftingSobolev}
Let $\pi : (M, g) \rightarrow (N, h)$ be a Riemannian submersion. Let us denote by $\hat{g}(p) := g|_{\pi^{-1}(p)}$ the induced metric on the fiber $\pi^{-1}(p)$, for every $p \in N$. For any integrable function $f$ on $M$ it holds that $\int_{\pi^{-1}(x)} f(y)\; \textup{dVol}_{\hat{g}(x)}(y)$ is integrable on $N$ and
\[
\int_M f(x)\; \textup{dVol}_g(x) = \int_N \left(\int_{\pi^{-1}(p)} f(y)\; \textup{dVol}_{\hat{g}(p)}(y) \right)\; \textup{dVol}_h(p).
\]
Moreover, if 
\[
\hat{F}(p) := \int_{\pi^{-1}(p)} |f(y)|\,\textup{dVol}_{\hat g(p)}(y)
\]
is integrable on $N$, then $f$ is integrable on $M$. 
\end{corollary}
\begin{proof}
Let $\eta =  \textup{dVol}_{h(\pi(x))}$, and $\omega = \textup{dVol}_{\hat{g}(\pi(x))}$. Since $\pi$ is a Riemannian submersion, it has no critical points and 
\[
\textup{dVol}_{\hat{g}(\pi(x))} \wedge \pi^* \textup{dVol}_{h(\pi(x))} = \textup{dVol}_{g(x)},
\]
which follows from the decomposition of $g$ into its horizontal and vertical components (cf. \cref{chapterRiemannianSubmersions}). 
\end{proof}

Two straightforward consequences of \cref{LiftingSobolev} arise when considering $G$-bundles with compact isometric fibers. Firstly, we can strengthen the result obtained in \cref{thm:fubinimanifolds} to ensure regularity of
\[
x \mapsto \int_{\pi^{-1}(x)} f(y)\, \textup{dVol}_{\hat{g}(x)}(y),
\]
provided that $f$ is sufficiently regular. Secondly, given a function $\tilde F$ on $B$, and $F$ defined on $M$ verifying $F = \tilde F \circ \pi$, we can relate the normalizing constant of the Gibbs measure associated with $F$ with that of the Gibbs measure associated with $\tilde F$ and the volume of $G$. 
\begin{proposition}[Upgraded regularity]
\label{rem:upgradedregularity}
Consider a principal $G$-bundle which is also a Riemannian submersion, $\pi: (M, g) \to (M/G, h)$, where $G$ is a compact Lie group acting isometrically on $(M, g)$. Further assume that the fibers of $\pi$ are all isometric to some fixed compact Lie group $(G, \hat{g})$. For every function $f \in C^k(M)$, the function
\[
\hat{F}(x) := \int_{\pi^{-1}(x)} f(y)\, \textup{dVol}_{\hat{g}}(y)
\]
is well-defined and belongs to $C^k(M/G)$. 
\end{proposition}
\begin{proof}
For every $x_0 \in M/G$ we can consider an open neighborhood $U$ around $x_0$ and a local trivialization, so that in local coordinates 
\[
\hat F(x) = \int_{\pi^{-1}(x)} f(y) \,\textup{dVol}_{\hat{g}}(y) = \int_{G} f(x, y)\, \textup{dVol}_{\hat{g}}(y),
\]
for every $x \in U$. Now, since $G$ is compact and $f$ is continuous on $M$, we can apply the dominated convergence theorem to conclude that $\hat F$ is continuous at $x_0$. Thus, $\hat F$ is continuous in $M/G$. Furthermore, if $f$ is assumed to be in $C^k(M)$, then, using the compactness of $G$ and the regularity of $f$, we can differentiate under the integral sign to conclude that $\hat F$ is $C^k(M/G)$. 
\end{proof}

Lastly, let us state the variation of Fubini's theorem in $G$-bundles which are also Riemannian submersions, when integrating with respect to the Gibbs measure associated with a function which is constant in the fibers of the submersion $\pi$. As we mentioned previously, this will allow us to obtain a relation between the normalizing constants of the Gibbs measures and the volume of $G$. 

\begin{proposition}[Fubini's theorem for Gibbs measures]
\label{Fubini}
Consider a principal $G$-bundle which is also a Riemannian submersion, $\pi: (M, g) \to (M/G, h)$. Further assume that the fibers of $\pi$ are all isometric to some compact Lie group $(G, \hat{g})$. Let $F$ be some function on $M$ that is constant on the fibers of $\pi$, and $\tilde{F}$ be such that $\tilde{F} \circ \pi = F$. Consider the Gibbs measures $d\nu(x) = \frac{1}{Z}e^{-\beta F(x)}\textup{dVol}_g(x)$ on $M$ and $d\tilde{\nu}(x) = \frac{1}{\tilde{Z}}e^{-\beta \tilde{F}(x)}\textup{dVol}_{h}(x)$ on $M/G$. Provided both $\nu$ and $\tilde \nu$ are probability measures, given any integrable function $f$ on $M$ with respect to $\nu$, the function 
\[
\tilde f: x \mapsto \tilde f(x) = \int_{\pi^{-1}(x)} f(y)\, \textup{dVol}_{\hat g}(y)
\]
is integrable on $M/G$ with respect to $\tilde \nu$ and 
\[
\int_M f(x) d\nu(x) = \frac{\tilde{Z}}{Z}\int_{M/G} \tilde f(x)\, d\tilde{\nu}(x).
\]
\end{proposition}
\begin{proof}
The result simply follows by rewriting 
\begin{align*}
\int_M f(x) d\nu(x) &= \frac{1}{Z}\int_M f(x)e^{-\beta F(x)}\; \textup{dVol}_{g}(x) 
\\&= \frac{1}{Z}\int_{M/G} \left(\int_{\pi^{-1}(x)} f(y)e^{-\beta F(y)}\; \textup{dVol}_{\hat{g}}(y)  \right)\; \textup{dVol}_{h}(x) 
\\&= \frac{1}{Z}\int_{M/G} \left(\int_{\pi^{-1}(x)} f(y)\; \textup{dVol}_{\hat{g}}(y)  \right)\; e^{-\beta \tilde{F}(x)}\; \textup{dVol}_{h}(x) 
\\&= \frac{\tilde{Z}}{Z}\int_{M/G} \left(\int_{\pi^{-1}(x)} f(y)\; \textup{dVol}_{\hat{g}}(y)  \right)\; d\tilde{\nu}(x)
\\&= \frac{\tilde{Z}}{Z}\int_{M/G} \tilde f(x)\, d\tilde{\nu}(x).
\end{align*}
\end{proof}

From this expression, we obtain the following relation between the volume of $G$ and the normalizing constants $Z$ and $\tilde{Z}$ when considering $f \equiv 1$. 
\begin{equation}
\label{volumenzetas}
1 = \int_M d\nu(x) = \frac{\tilde{Z}}{Z}\int_{M/G} \textup{Vol}_{\hat{g}}(G)\; d\tilde{\nu}(x) = \frac{\tilde{Z}}{Z} \textup{Vol}_{\hat{g}}(G).
\end{equation}

\subsection{Lifting and lowering a Poincaré inequality}
\label{sec:liftingandloweringpoincare}

The following result allows us to guarantee that every compact and connected manifold of non-negative Ricci curvature satisfies a Poincaré inequality with respect to its normalized Riemannian measure, i.e. the probability measure associated with its volume form, which corresponds to the Gibbs measure at infinite temperature. Note that this result on its own is not useful in our setting; indeed, we are interested in proving the existence of a Poincaré inequality with respect to the Gibbs measure associated with a given function $F$ at finite temperature, which in general is different from the Riemannian measure. Nevertheless, this result is crucial for the lifting process of the Poincaré inequality, as will be shown in the proof of \cref{thmliftPI}. 

\begin{theorem}[{\cite{zhong1984estimate}}]
\label{teoremaPIGenerico}
Let $(N,h)$ be a compact, connected Riemannian manifold of finite volume $\textup{Vol}(N)$. Assume that $N$ has non-negative Ricci curvature. Let $d\mu = \frac{\textup{dVol}_h}{\textup{Vol}(N)}$ be the normalized Riemannian measure on $(N, h)$ and let $D$ be its diameter. Then  the Markov triple $(N, \mu, |\Grad{h}\cdot |_h^2)$ satisfies a $\textup{PI}(\frac{\pi^2}{D^2})$, i.e. for every $f \in C^2(N)$,
\[
\int_N f^2\; d\mu - \left(\int_N f\; d\mu \right)^2 \leq \frac{D^2}{\pi^2}\int_N |\Grad{h}f|_h^2 \; d\mu,
\]
which can be rewritten as
\[
\int_N f^2\; \textup{dVol}_h - \frac{1}{\textup{Vol}(N)}\left(\int_N f\; \textup{dVol}_h  \right)^2 \leq \frac{D^2}{\pi^2}\int_N |\Grad{h} f|_h^2 \; \textup{dVol}_h .
\]
\end{theorem}

We now proceed to state and prove the main two results of the section, namely \cref{thmliftPI,thmlowerPI}. The former will allow us to obtain a Poincaré inequality in the total space of a Riemannian submersion, provided the base space satisfies a Poincaré inequality, while the latter will allow us to do the opposite. Recall that, as we have mentioned previously, these Poincaré inequalities are always considered with respect to the Gibbs measure associated with a function which is constant in the fibers of the Riemannian submersion.

\begin{theorem}
\label{thmliftPI}
Let $(M, g)$ be a complete manifold and let $\pi: (M, g) \to (M/G, h)$ be a principal $G$-bundle that is also a Riemannian submersion, where $G$ is connected and acts isometrically on $(M, g)$. Assume that the fibers are all isometric to a compact Lie group $(G, \hat{g})$. Further assume that $(G, \hat{g})$ has non-negative Ricci curvature. Let $\textup{diam}(G)$ be the diameter of $G$. Consider a $C^2$ function $F: M \to \bb{R}$ which is constant on the fibers of $\pi$. Let $\tilde{F}: M/G \to \bb{R}$ be the (unique) function on $M/G$ such that $F = \tilde{F} \circ \pi$. Let 
\[
d\tilde{\nu} = \frac{1}{\tilde{Z}} e^{-\beta \tilde{F}(x)}\; \textup{dVol}_{h},\quad  d\nu = \frac{1}{Z}e^{-\beta F(x)}\textup{dVol}_g,
\]
be the Gibbs probability measures on $M/G$ and $M$, respectively. If $(M/G, \tilde{\nu}, \tilde \Gamma)$ satisfies a $\textup{PI}(\tilde{\kappa})$, then $(M, \nu, \Gamma)$ also satisfies a $\textup{PI}(\kappa)$ with
\[
\frac{1}{\kappa} = \max\Big\{\frac{1}{\tilde{\kappa}}, \frac{\textup{diam}(G)^2 \beta }{\pi^2}\Big\},
\]
where $\Gamma(\cdot) = \frac{1}{\beta}|\textup{grad}_g\, \cdot|^2_{g}$, and $\tilde{\Gamma}(\cdot) = \frac{1}{\beta} |\textup{grad}_{h} \,\cdot|^2_{h}$.
\end{theorem}

\begin{theorem}
\label{thmlowerPI}
Under the conditions of \cref{thmliftPI}, where $(G, \hat{g})$ may not have non-negative Ricci curvature, if $(M, \nu, \Gamma)$ satisfies a $\textup{PI}(\kappa)$ then $(M/G, \tilde{\nu}, \tilde{\Gamma})$ also satisfies a $\textup{PI}(\kappa)$. 
\end{theorem} 

For the proofs of \cref{thmliftPI,thmlowerPI}, we will use two auxiliary lemmas that relate the horizontal and the vertical components of the gradient of a function in the total space of a Riemannian submersion to the gradient with respect to the metrics of the base space and the fibers. These components are denoted by $\cdot^\mathcal{H}$ and $\cdot^\mathcal{V}$, respectively (see \cref{notverticalhorizontal}).

\begin{lemma}
\label{LemmaGradienteHorizontal}
Let $\pi: (M, g) \to (M/G, h)$ be a principal $G$-bundle which is also a Riemannian submersion, where $G$ is a compact Lie group acting isometrically on $(M, g)$. Assume that the fibers are all isometric to $(G, \hat{g})$. Then  for every $x \in M/G$ and for every $f \in C^2(M)$, the following holds:
\begin{equation*}
\left|\frac{1}{\textup{Vol}_{\hat{g}}(G)}\textup{grad}_{h}\int_{\pi^{-1}(x)} f(y)\; \textup{dVol}_{\hat{g}}(y)\right|_{h}^2
\leq \frac{1}{\textup{Vol}_{\hat{g}}(G)}  \int_{\pi^{-1}(x)} \left|\left(\textup{grad}_{g} \, f(y)\right)^{\mathcal{H}}\right|_{g}^2 \textup{dVol}_{\hat{g}}(y).
\end{equation*}
\end{lemma}
\begin{proof}
First, recall that as $(G, \hat g)$ has finite volume and $f \in C^2(M)$, it holds that $|f|$ and  $|(\textup{grad}_g\, f)^{\mathcal{H}}|^2_g$ are both integrable in $\pi^{-1}(x)$ for every $x \in M/G$, so the above expression is well-defined.

Let us now fix $x \in M/G$. We choose normal coordinates $(z^1, \dotsc, z^k)$ in $M/G$ centered at \mbox{$x = (0, \dotsc, 0)$}. Thus, $\{\frac{\partial}{\partial z^1}, \dotsc, \frac{\partial}{\partial z^k}\}$ is an orthonormal basis for $T_x M/G$, i.e. \mbox{$\left.h_{ij}\right|_x = \delta_{ij}$}. 

Furthermore, we can extend these coordinates to $\pi^{-1}(x) \setminus \mathcal{X} (\simeq G \setminus \mathcal{X})$, where $\mathcal{X}$ is a measure-zero set. By a slight abuse of notation, we denote such coordinates as 
\[
(z^1, \dotsc, z^k, y^1, \dotsc, y^{n-k}).
\]
This way, every point $y \in \pi^{-1}(x) \setminus \mathcal{X}$ can be written as $(0, \dotsc, 0, y^1, \dotsc, y^{n-k})$. Furthermore, by horizontally lifting the frame $\{\frac{\partial}{\partial z^1}, \dotsc, \frac{\partial}{\partial z^k}\}$ to the fiber $\pi^{-1}(x) \setminus \mathcal{X}$, we obtain a frame
\[
\left\{\frac{\partial}{\partial z^1}, \dotsc, \frac{\partial}{\partial z^k}, \frac{\partial}{\partial y^1}, \dotsc, \frac{\partial}{\partial y^{n-k}}\right\}
\]
is such that $\left\{\frac{\partial}{\partial z^1}, \dotsc, \frac{\partial}{\partial z^k}\right\}$ are orthonormal with respect to $g$ and horizontal with respect to $\pi$ (cf. \cref{verticalhorizontaltangent}), $\left\{\frac{\partial}{\partial y^1}, \dotsc, \frac{\partial}{\partial y^{n-k}}\right\}$ are vertical and $g(\frac{\partial}{\partial z^i}, \frac{\partial}{\partial y^j}) = 0$ for every $i \in \{1, \dotsc, k\}$, $j \in \{1, \dotsc, n-k\}$. 

From our choice of coordinates, we know that on $\pi^{-1}(x) \setminus \mathcal{X}$, $g$ is of the form 
\begin{equation}
\label{matrixmetric}
\left(\begin{array}{ c | c }
    \mathds{1} & 0 \\
    \hline
    0 & A
  \end{array}\right),
\end{equation}
where $\mathds{1}$ is the identity matrix of size $k \times k$ and $A$ is an $(n-k) \times (n-k)$ matrix. 

In such coordinates, since the metric $h$ at $x$ is simply the Euclidean metric,
\[
\left.\textup{grad}_{h}( \cdot)\right|_x = h^{ij}(x)\frac{\partial }{\partial z^i}(\cdot)\frac{\partial }{\partial z^j} = \delta^{ij}\frac{\partial }{\partial z^i}(\cdot)\frac{\partial }{\partial z^j}, 
\]
and so
\begin{multline*}
\left|\frac{1}{\textup{Vol}_{\hat{g}}(G)} \textup{grad}_{h} \int_{\pi^{-1}(x)} f(y)\; \textup{dVol}_{\hat{g}}(y)\right|_{h}^2 = 
\frac{1}{\textup{Vol}_{\hat{g}}(G)^2} \sum_{i = 1}^k \left(\frac{\partial}{\partial z^i} \int_{\pi^{-1}(x)} f(y)\; \textup{dVol}_{\hat{g}}(y)\right)^2
\\= \frac{1}{\textup{Vol}_{\hat{g}}(G)^2} \sum_{i = 1}^k \left(\frac{\partial}{\partial z^i} \int_{G \setminus \mathcal{X}} f(0, \dotsc, 0, y)\; \textup{dVol}_{\hat{g}}(y)\right)^2
\\= \frac{1}{\textup{Vol}_{\hat{g}}(G)^2} \sum_{i = 1}^k \left(\int_{G \setminus \mathcal{X}} \frac{\partial}{\partial z^i} f(0, \dotsc, 0, y)\; \textup{dVol}_{\hat{g}}(y)\right)^2
\\\leq \frac{1}{\textup{Vol}_{\hat{g}}(G)} \sum_{i = 1}^k \int_{G \setminus \mathcal{X}} \left(\frac{\partial}{\partial z^i} f(0, \dotsc, 0, y)\right)^2 \textup{dVol}_{\hat{g}}(y),
\end{multline*}
where the last inequality is Cauchy-Schwarz's inequality: 
\[
\left(\int_G g(z)\; \textup{dVol}_{\hat{g}}(z)\right)^2 \leq \int_G g(z)^2\; \textup{dVol}_{\hat{g}}(z) \int_{G} \textup{dVol}_{\hat{g}}(z).
\]

On the other hand, let us rewrite 
\begin{equation}
\label{intofgrad}
\frac{1}{\textup{Vol}_{\hat{g}}(G)} \int_{\pi^{-1}(x)} |(\textup{grad}_g\, f(y))^{\mathcal{H}}|^2_g\; \textup{dVol}_{\hat{g}}(y)
\end{equation}
using the coordinates from the beginning of this proof, for some point in $\pi^{-1}(x) \setminus \mathcal{X}$,
\[
(\textup{grad}_g\, f)^{\mathcal{H}} = g^{ij}\frac{\partial f}{\partial x^i}\frac{\partial }{\partial z^j},
\]
where $x^i = z^i$ for $i \in \{1, \dotsc, k\}$ and $x^i = y^{i-k}$ for $i \in \{k+1, \dotsc, n\}$. Moreover, since $g$ in $\pi^{-1}(x) \setminus \mathcal{X}$ is of the form shown in \cref{matrixmetric} we can conclude that
\[
(\textup{grad}_g\, f)^{\mathcal{H}} = \delta^{ij}\frac{\partial f}{\partial z^i}\frac{\partial }{\partial z^j}.
\]
Therefore, its norm with respect to $g$ is simply
\[
|(\textup{grad}_g\, f(y))^{\mathcal{H}}|^2_g = g\left(\delta^{ij}\frac{\partial f}{\partial z^i}\frac{\partial }{\partial z^j}, \delta^{kl}\frac{\partial f}{\partial z^k}\frac{\partial }{\partial z^l}\right) = \delta^{ij}\frac{\partial f}{\partial z^i}\frac{\partial f}{\partial z^j}.
\]
This allows us to write \cref{intofgrad} as 
\[
\frac{1}{\textup{Vol}_{\hat{g}}(G)} \int_{G \setminus \mathcal{X}} \sum_{i = 1}^k \left(\frac{\partial}{\partial z^i}f(0,\dotsc,0, y)\right)^2 \textup{dVol}_{\hat{g}}(y),
\]
finishing the proof.
\end{proof}

\begin{lemma}
\label{LemmaGradienteVertical}
Let $\pi: (M, g) \to (M/G, h)$ be a principal $G$-bundle, for some compact Lie group $G$, that is also a Riemannian submersion. Assume that the fibers are all isometric to $(G, \hat{g})$. Then for every $x \in M/G$ and for every $f \in C^1(M)$, the following holds:
\[
\left|\textup{grad}_{\hat{g}}\, f(y)\right|_{\hat{g}}^2 = \left|\left(\textup{grad}_{g} \,f(y)\right)^{\mathcal{V}}\right|_{g}^2
\]
for almost every $y \in \pi^{-1}(x)$. 

\end{lemma}
\begin{proof}
Fix $x \in M/G$. Working in the same coordinates as in the proof of \cref{LemmaGradienteHorizontal} we know that on $\pi^{-1}(x) \setminus \mathcal{X}$, $g$ is of the form 
\[
\left(\begin{array}{ c | c }
    \mathds{1} & 0 \\
    \hline
    0 & A
  \end{array}\right).
\]
Using the same argument as before, we obtain
\[
\left(\textup{grad}_{g}\, f\right)^{\mathcal{V}} = g^{ij} \frac{\partial f}{\partial y^i}\frac{\partial }{\partial y^j},
\]
finishing the proof. 
\end{proof}

Finally, we prove the following auxiliary technical lemma before turning to the main results of this section; it is stated separately for the sake of clarity.

\begin{lemma}
\label{lem:lemmaregularityintegral}
Under the conditions of \cref{thmliftPI}, let $f \in C^2(M) \cap H^1(M, \nu)$. Then 
\[
\hat F(x) := \int_{\pi^{-1}(x)} f(y)\,\textup{dVol}_{\hat{g}}(y)
\]
is a function in $C^2(M/G) \cap H^1(M/G, \tilde \nu)$.
\end{lemma}
\begin{proof}
The fact that $\hat F$ belongs to $C^2(M/G)$ follows from the regularity of $f$ by \cref{rem:upgradedregularity}. Now, using Cauchy-Schwarz's inequality
\[
\hat F(x)^2 = \left(\int_{\pi^{-1}(x)} f(y)\, \textup{dVol}_{\hat{g}}(y)\right)^2 \leq \textup{Vol}_{\hat{g}}(G)\int_{\pi^{-1}(x)} f^2(y)\, \textup{dVol}_{\hat{g}}(y),
\]
and so by \cref{Fubini}
\begin{align*}
\int_{M/G} \hat F(x)^2\,  d\tilde \nu(x) &\leq  \textup{Vol}_{\hat{g}}(G) \int_{M/G} \left(\int_{\pi^{-1}(x)} f^2(y)\, \textup{dVol}_{\hat{g}}(y)\right) d\tilde \nu (x)
\\&= \frac{Z}{\tilde Z}\textup{Vol}_{\hat{g}}(G) \int_M f^2(x) \,d\nu(x) < \infty,
\end{align*}
as $f \in H^1(M, \nu)$, which implies that $\hat F \in L^2(M/G, \tilde \nu)$. Moreover, from \cref{LemmaGradienteHorizontal}, we know that for every $x \in M/G$
\[
|\textup{grad}_{h}\,  \hat F(x)|^2_{h} \leq \textup{Vol}_{\hat{g}}(G) \int_{\pi^{-1}(x)} \left|\left(\textup{grad}_{g}\, f(y)\right)^{\mathcal{H}}\right|_{g}^2 \,\textup{dVol}_{\hat{g}}(y).
\]
Using again \cref{Fubini} we can conclude that 
\begin{align*}
\int_{M/G} |\textup{grad}_{h}\,  \hat F(x)|^2_{h}\, d\tilde \nu(x) &\leq \textup{Vol}_{\hat{g}}(G) \int_{M/G} \left( \int_{\pi^{-1}(x)} \left|\left(\textup{grad}_{g} \,f(y)\right)^{\mathcal{H}}\right|_{g}^2 \,\textup{dVol}_{\hat{g}}(y)\right) d\tilde\nu(x)
\\&= \textup{Vol}_{\hat{g}}(G) \frac{Z}{\tilde Z} \int_M \left|\left(\textup{grad}_{g} \,f(x)\right)^{\mathcal{H}}\right|_{g}^2\, d\nu(x)
\\&\leq \textup{Vol}_{\hat{g}}(G) \frac{Z}{\tilde Z} \int_M \left|\textup{grad}_{g} \,f(x)\right|_{g}^2 \,d\nu(x) < \infty,
\end{align*}
as $f \in H^1(M, \nu)$,  which allows us to conclude that $ \hat F \in H^1(M/G, \tilde\nu)$. 
\end{proof}

With the above lemmas in mind, we will now proceed to the proof of the lifting result for Poincaré inequalities. 

\begin{proof}[{Proof of \cref{thmliftPI}}]
First, recall that $(M/G, \tilde{\nu}, \tilde \Gamma)$ satisfies a $\textup{PI}(\tilde{\kappa})$ iff for every $f \in C^2(M/G) \cap H^1(M/G, \tilde{\nu})$
\[
\int_{M/G} f^2\; d\tilde{\nu} - \left(\int_{M/G} f \;d\tilde{\nu} \right)^2 \leq \frac{1}{\tilde{\kappa}\beta} \int_{M/G} |\textup{grad}_{h}\, f|_{h}^2\, d\tilde{\nu}.
\]
Let $f \in C^2(M) \cap H^1(M, \nu)$. Without loss of generality, we can assume that $\int_M f d\nu = 0$. Using \cref{Fubini} we can write
\begin{equation}
\label{eq:initialeqPoincare}
\int_M f^2\; d\nu = \frac{\tilde{Z}}{Z} \int_{M/G} \left(\int_{\pi^{-1}(x)} f^2(y) \,\textup{dVol}_{\hat{g}}(y)\right) d\tilde{\nu}(x).    
\end{equation}
As $f \in C^2(\pi^{-1}(x))$ for every $x \in M/G$, we can now use the fact that the fibers satisfy a Poincaré inequality; we apply \cref{teoremaPIGenerico} to \cref{eq:initialeqPoincare} and obtain
\begin{align}
\label{eq:firstauxineqpoincare}
\begin{split}
\int_M f^2\; d\nu &\leq \frac{\tilde{Z}}{Z}\left\{\frac{1}{\textup{Vol}_{\hat{g}}(G)}\int_{M/G} \left(\int_{\pi^{-1}(x)} f(y)\, \textup{dVol}_{\hat{g}}(y) \right)^2 d\tilde{\nu}(x)\right.
\\&\quad\left. + C \int_{M/G} \left(\int_{\pi^{-1}(x)} \left|\textup{grad}_{\hat{g}} \,f(y)\right|_{\hat{g}}^2 \,\textup{dVol}_{\hat{g}}(y)\right) d\tilde{\nu}(x) \right\},
\end{split}
\end{align}
where $C := \frac{\textup{diam}(G)^2}{\pi^2}$.

Using \cref{LemmaGradienteVertical}, we know that
\[
\left|\textup{grad}_{\hat{g}}\, f\right|_{\hat{g}}^2 = \left|\left(\textup{grad}_{g} \,f\right)^{\mathcal{V}}\right|_{g}^2,
\]
and so, using the fact that $f \in H^1(M, \nu)$, we can apply \cref{Fubini} to the second term on the right-hand side of \cref{eq:firstauxineqpoincare} to conclude that 
\[
\int_M f^2\; d\nu \leq \frac{\tilde{Z}}{Z\textup{Vol}_{\hat{g}}(G)}\int_{M/G} \left(\int_{\pi^{-1}(x)} f(y)\, \textup{dVol}_{\hat{g}}(y) \right)^2 d\tilde{\nu}(x) + C \int_{M} \left|\left(\textup{grad}_{g}\, f\right)^{\mathcal{V}}\right|_{g}^2 d\nu.
\]

Now, it follows from \cref{lem:lemmaregularityintegral} that
\[
x \mapsto \int_{\pi^{-1}(x)} f(y) \,\textup{dVol}_{\hat{g}}(y)
\]
is a function in $C^2(M/G) \cap H^1(M/G, \tilde \nu)$. Since we are assuming that $(M/G, \tilde{\nu}, \tilde{\Gamma})$ satisfies a Poincaré inequality and that $\int_M f d\nu = 0$ we can further bound
\begin{align*}
\int_M f^2\; d\nu &\leq \frac{\tilde{Z}}{Z\textup{Vol}_{\hat{g}}(G)}\Biggl\{\left( \int_{M/G} \left(\int_{\pi^{-1}(x)} f(y)\, \textup{dVol}_{\hat{g}}(y)\right)\; d\tilde{\nu}(x)\right)^2
\\&\quad + \frac{1}{\tilde{\kappa}\beta} \int_{M/G} \left|\textup{grad}_{h} \int_{\pi^{-1}(x)} f(y) \,\textup{dVol}_{\hat{g}}(y) \right|^2_{h} d\tilde{\nu}(x)\Biggr\}
\\&\quad + C \int_{M} \left|\left(\textup{grad}_{g} \,f\right)^{\mathcal{V}}\right|_{g}^2 d\nu 
\\&= \frac{\tilde{Z}}{Z\textup{Vol}_{\hat{g}}(G)}\Biggl\{\frac{Z^2}{\tilde{Z}^2}\left(\int_M f\, d\nu\right)^2 + \frac{1}{\tilde{\kappa}\beta}\int_{M/G} \left|\textup{grad}_{h} \int_{\pi^{-1}(x)} f(y) \,\textup{dVol}_{\hat{g}}(y) \right|^2_{h} d\tilde{\nu}(x)\Biggr\}
\\&\quad + C \int_{M} \left|\left(\textup{grad}_{g} \,f\right)^{\mathcal{V}}\right|_{g}^2 d\nu
\\&= \frac{\tilde{Z}}{Z}\frac{1}{\textup{Vol}_{\hat{g}}(G)\tilde{\kappa}\beta}\int_{M/G} \left|\textup{grad}_{h} \int_{\pi^{-1}(x)} f(y) \, \textup{dVol}_{\hat{g}}(y) \right|^2_{h} d\tilde{\nu}(x)
\\&\quad + C \int_{M} \left|\left(\textup{grad}_{g} \, f\right)^{\mathcal{V}}\right|_{g}^2 d\nu
\\&= \frac{1}{\textup{Vol}_{\hat{g}}(G)^2 \tilde{\kappa}\beta}\int_{M/G} \left|\textup{grad}_{h} \int_{\pi^{-1}(x)} f(y)\, \textup{dVol}_{\hat{g}}(y) \right|^2_{h} d\tilde{\nu}(x)
\\&\quad + C \int_{M} \left|\left(\textup{grad}_{g} \, f\right)^{\mathcal{V}}\right|_{g}^2 d\nu,
\end{align*}
where the last equality follows from \cref{volumenzetas}. Lastly, we apply \cref{LemmaGradienteHorizontal} and \cref{Fubini} to conclude that 
\begin{align*}
\int_M f^2\; d\nu &\leq \frac{1}{\textup{Vol}_{\hat{g}}(G)\tilde{\kappa}\beta }\int_{M/G} \left(\int_{\pi^{-1}(x)} |\left(\textup{grad}_{g} \,f(y)\right)^{\mathcal{H}}|_{g}^2 \,\textup{dVol}_{\hat{g}}(y)\right) d\tilde{\nu}(x)
\\&\quad + C \int_{M} \left|\left(\textup{grad}_{g}\, f\right)^{\mathcal{V}}\right|_{g}^2 d\nu 
\\&= \frac{\tilde{Z}}{Z}\frac{1}{\tilde{\kappa}\beta}\int_{M/G} \left(\int_{\pi^{-1}(x)} |\left(\textup{grad}_{g}\, f(y)\right)^{\mathcal{H}}|_{g}^2 \,\textup{dVol}_{\hat{g}}(y)\right) d\tilde{\nu}(x)
\\&\quad+ C \int_{M} \left|\left(\textup{grad}_{g} \,f\right)^{\mathcal{V}}\right|_{g}^2 d\nu 
\\&= \frac{1}{\tilde{\kappa}\beta}\int_{M} \left|(\textup{grad}_{g}\, f)^{\mathcal{H}} \right|_{g}^2 d\nu
+ C \int_{M} \left|\left(\textup{grad}_{g}\, f\right)^{\mathcal{V}}\right|_{g}^2 d\nu.
\end{align*}

Therefore, if $(M/G, \tilde{\nu}, \tilde{\Gamma})$ satisfies a $\textup{PI}(\tilde{\kappa})$, $(M, \nu, \Gamma)$ satisfies a $\textup{PI}(\kappa)$ verifying
\[
\frac{1}{\kappa} = \max \Big\{\frac{1}{\tilde{\kappa}}, \beta C\Big\}= \max\Big\{\frac{1}{\tilde{\kappa}}, \frac{\textup{diam}(G)^2 \beta }{\pi^2}\Big\},
\]
where $\textup{diam}(G)$ is the diameter of the fibers with the restricted metric $(G, \hat{g})$. 
\end{proof}

Lastly, let us prove that Poincaré inequalities can be lowered from the total space to the base space of a Riemannian submersion. Unlike for the previous result, we now do not need to use the fact that the fibers satisfy a Poincaré inequality, and the Poincaré constant obtained in the base space is the same as that of the total space. 

\begin{proof}[{Proof of \cref{thmlowerPI}}]
Let $\tilde f \in C^2(M/G)\cap H^1(M/G, \tilde{\nu})$, and let $f: M \to \bb{R}$ be the unique function such that $\tilde f \circ \pi = f$. In particular, $f$ is constant in the fibers of $\pi$ and belongs to $C^2(M)$. Using \cref{LiftingSobolev} and the fact that $\tilde f \in H^1(M/G, \tilde \nu)$ we know that 
\[
x \mapsto \int_{\pi^{-1}(x)} f^2(y) e^{-\beta F(y)}\,\textup{dVol}_{\hat g}(y) = \tilde f^2(x) e^{-\beta \tilde F(x)}\, \textup{Vol(G)}
\] 
is integrable on $M/G$. Thus $f^2 e^{-\beta F}$ is integrable on $M$---i.e. $f \in L^2(M, \nu)$. Substituting $f^2$ by $|\textup{grad}_g\, f|^2_g$ in the above expression and using \cref{prop:gradientprojectedfunction} we can similarly conclude that $|\textup{grad}_g\, f|^2_g\, e^{-\beta F}$ is integrable on $M$, which implies that $f \in H^1(M, \nu)$. 

Without loss of generality, assume that $\int_{M/G} \tilde f\; d\tilde{\nu} = 0$. Thus, by \cref{Fubini} we know that 
\begin{align*}
\int_M f(x)\, d\nu(x) &= \frac{\tilde Z}{Z} \int_{M/G} \left(\int_{\pi^{-1}(x)} f(y)\, \textup{dVol}_{\hat{g}}(y)\right) d\tilde{\nu}(x)
\\&= \frac{\tilde Z}{Z} \int_{M/G} \left(\tilde f(x) \int_{\pi^{-1}(x)}  \textup{dVol}_{\hat{g}}(y)\right) d\tilde{\nu}(x)
\\&= \frac{\tilde Z}{Z} \textup{Vol}(G) \int_{M/G} \tilde f(x)\, d\tilde{\nu}(x) = 0.
\end{align*}

Furthermore, using again \cref{Fubini}, we know that
\begin{align*}
\int_M f^2(x) \; d\nu(x) = \frac{\tilde Z}{Z} \textup{Vol}(G) \int_{M/G} \tilde{f}^2(x) \;d\tilde{\nu}(x).
\end{align*}

Thus, since $(M, \nu, \Gamma)$ satisfies a $\textup{PI}(\kappa)$ and $f \in C^2(M) \cap H^1(M, \nu)$, we know that
\begin{align*}
\int_{M/G} \tilde f^2(x)\; d\tilde{\nu}(x) &= \frac{Z}{\tilde{Z} \textup{Vol}(G)} \int_M f^2(x)\; d\nu(x) 
\\&\leq \frac{Z}{\tilde{Z} \textup{Vol}(G)} \frac{1}{\kappa \beta} \int_M |\textup{grad}_g \,f(x)|^2_g\; d\nu(x) 
\\&= \frac{1}{\textup{Vol}(G)}\frac{1}{\kappa \beta} \int_{M/G} \left(\int_{\pi^{-1}(x)} |\textup{grad}_g\, f(y)|^2_g \,\textup{dVol}_{\hat{g}}(y)\right)d\tilde{\nu}(x). 
\end{align*}

Moreover, by \cref{prop:gradientprojectedfunction} we can conclude that for every $x \in M$, it holds that 
\[
|\textup{grad}_g\, f(x)|^2_g = |\textup{grad}_{h}\, \tilde f(\pi(x))|^2_{h},
\]
and so
\begin{align*}
\int_{M/G} \tilde f^2(x)\; d\tilde{\nu}(x) &\leq \frac{1}{\textup{Vol}(G)}\frac{1}{\kappa \beta} \int_{M/G} \left(\int_{\pi^{-1}(x)} |\textup{grad}_{h}\, \tilde f(x)|^2_{h}\, \textup{dVol}_{\hat{g}}(y)\right)d\tilde{\nu}(x)
\\&= \frac{1}{\kappa \beta} \int_{M/G} |\textup{grad}_{h}\, \tilde f(x)|^2_{h}\; d\tilde{\nu}(x),
\end{align*}
finishing the proof. 
\end{proof}

Although the lowering result for the Poincaré inequality is not required for the proof of our main result or in the examples studied, we include it here as it may be of independent interest in this and other related settings. 

%% file: Chapters/Upgrading.tex
\section{From a Poincaré to a log-Sobolev inequality and rapid mixing}
\label{SectionPItoLSI}

In this section, we show how to derive a log-Sobolev inequality from a Poincaré inequality together with the curvature-dimension condition. The key tool is a weaker version of the log-Sobolev inequality, known as a \textit{defective} log-Sobolev inequality. The use of such weakened inequalities (including the related \textit{weak} log-Sobolev inequalities) is common practice in this type of argument (cf. \cite{cattiaux2007weak,carlen2004logarithmic}).

\begin{definition}[Logarithmic Sobolev inequalities]
We say that the Markov triple $(M, \nu, \Gamma)$ satisfies a \textup{defective logarithmic Sobolev inequality} with constants $\alpha > 0$ and $A \geq 0$, \textup{LSI}$(\alpha, A)$, if, for every probability measure $\mu$ such that\footnote{Recall that $\mu \ll \nu$ is used to denote that $\mu$ is \textit{absolutely continuous} with respect to $\nu$, i.e. for every measurable set $A$, $\nu(A) = 0$ implies that $\mu(A) = 0$.} $\mu \ll \nu$ with $f := \frac{d\mu}{d\nu} \in C^2(M)$,
\begin{equation}
\label{defectiveLSI}    
H(\mu | \nu) \leq \frac{1}{2\alpha} I(\mu | \nu) + A,
\end{equation}
where $H(\mu | \nu) := \int_M f \log f d\nu$ is known as the \textup{relative entropy} (or Kullback-Leibler divergence) of the probability measure $\mu$ with respect to $\nu$, and $I(\mu | \nu) := \int_M \frac{\Gamma(f)}{f} d\nu$ is known as the \textup{Fisher information} of $\mu$ with respect to $\nu$.\footnote{Here, we are using the convention $0 \log0 = 0$ and $\Gamma(0)/0 = 0$.} 

We say the Markov triple $(M, \nu, \Gamma)$ satisfies a \textup{tight logarithmic Sobolev inequality}---or simply \textup{log-Sobolev inequality}---with constant $\alpha > 0$, \textup{LSI}($\alpha$), if \cref{defectiveLSI} is satisfied with $A = 0$, i.e.
\[
H(\mu | \nu) \leq  \frac{1}{2\alpha}I(\mu | \nu).
\]
\end{definition}

The log-Sobolev inequality allows us to study the rate at which the distribution of the Langevin diffusion $X_t$, $\rho_t$, tends to the Gibbs distribution $\nu$. In order to have a notion of \textit{distance} between two probability distributions, we can make use of the so-called \textit{total variation distance}. 

\begin{definition}[Total variation distance]
Let $\mu$ and $\nu$ be two probability measures on a measure space $(E, \mathcal{F})$. Then  we define their total variation distance as 
\[
\norm{\mu - \nu}_{\text{TV}} := \sup\{|\mu(A) - \nu(A)| : A \in \mathcal{F}\}. 
\]
\end{definition}

We are interested in obtaining an upper bound for the total variation distance between $\rho_t$ and $\nu$, which decreases \textit{fast} with $t$. In fact, one of the implications of having a log-Sobolev inequality is that the total variation distance between the distribution of $X_t$, $\rho_t$ and the Gibbs distribution $\nu$ decreases exponentially fast with respect to $t$. 

\begin{proposition}
\label{thm:LSIImpliesHypercontract}
If the Markov triple $(M, \nu, \Gamma)$ satisfies an $\textup{LSI}(\alpha)$, then 
\[
\norm{\rho_t - \nu}^2_{\textup{TV}} \leq \frac{1}{2} e^{-2\alpha t} H(\rho_0 | \nu),\quad \forall t \geq 0.
\]
\end{proposition}
\begin{proof}
Indeed, by \cite[Theorem 5.2.1]{bakry2013analysis} we know that 
\begin{equation}
\label{eq:boundKL}
H(\rho_t | \nu) \leq e^{-2\alpha t} H(\rho_0 | \nu), \quad \forall t \geq 0. 
\end{equation}
Using the Pinsker-Csizsár-Kullback inequality, we get
\[
\norm{\rho_t - \nu}^2_{\textup{TV}} \leq \frac{1}{2} H(\rho_t | \nu) \leq \frac{1}{2} e^{-2\alpha t} H(\rho_0 | \nu).
\]
\end{proof}

A log-Sobolev inequality shows that the total variation distance between $\rho_t$ and $\nu$ decreases exponentially fast with time. In many situations, it is also crucial to obtain bounds on $H(\rho_0 |\nu)$. 

\begin{proposition}
\label{BoundInitialRelativeEntropy}
Let $\mu$ be the uniform distribution on a compact Riemannian manifold $(M, g)$, let $F : M \to \bb{R}$ be a smooth function satisfying $\min_{x \in M} F(x) \geq 0$, and let $\nu$ be the Gibbs distribution $\nu = \frac{1}{Z}e^{-\beta F}$. Then 
\[
H(\mu | \nu) \leq \beta \max_{x \in M} F(x). 
\]
\end{proposition}
\begin{proof}
The fact that $\mu \ll \nu$ follows immediately. Now, since $d\mu = \frac{1}{\textup{Vol}(M)} \textup{dVol}_g(x)$ and $d\nu = \frac{1}{Z}e^{-\beta F(x)} \textup{dVol}_g(x)$,
\[
\frac{d\mu}{d\nu} = \frac{Z}{\textup{Vol}(M)} e^{\beta F(x)},
\]
which implies that 
\[
\log \frac{d\mu}{d\nu} = \log (\frac{Z}{\textup{Vol}(M)}) + \beta F(x).
\]
This expression allows us to write the relative entropy of $\mu$ with respect to $\nu$ as
\[
H(\mu | \nu) = \frac{1}{\textup{Vol}(M)} \int_M \log \frac{d\mu}{d\nu} \textup{dVol}_g(x) = \log (\frac{Z}{\textup{Vol}(M)}) + \frac{\beta}{\textup{Vol}(M)} \int_M F(x) \textup{dVol}_g(x).
\]
Now, note that since $F \geq 0$,
\[
Z = \int_M e^{-\beta F(x)} \textup{dVol}_g(x) \leq \int_M \textup{dVol}_g(x) = \textup{Vol}_g(M),
\]
and so we can upper bound the relative entropy as
\[
H(\mu | \nu) \leq \log(1) + \frac{\beta}{\textup{Vol}(M)} \int_M \max_{x \in M} F(x) \textup{dVol}_g(x) \leq \beta \max_{x \in M} F(x).
\]
\end{proof}

In particular, whenever the initial distribution of the Langevin diffusion $X_t$ is uniform on $M$, $H(\rho_0| \nu)$ can be upper bounded by the product of $\beta$ and the maximum of $F$. 

The Poincaré inequality yields an analogue of \cref{eq:boundKL} with the relative entropy replaced by the $\chi^2$ divergence (cf. \cite[Theorem 4.2.5]{bakry2013analysis}), which in turn yields the total variation distance bound $\norm{\rho_t - \nu}^2_{\textup{TV}} \leq \frac{1}{2} e^{-2\kappa t} \chi^2(\rho_0 | \nu)$ (cf. \cite[Equation 2.27]{Tsybakov2008}) where $\kappa$ is the Poincaré inequality constant. However, the relative entropy $H$ is upper bounded by the logarithm of  $1+ \chi^2$ (\cite[Lemma 2.7]{Tsybakov2008}). For this reason, while both the Poincaré and log-Sobolev inequalities provide exponentially decaying bounds on the total variation distance between $\rho_t$ and $\nu$, the prefactor obtained via the latter is exponentially smaller. 

\subsection{From Poincaré to log-Sobolev}

\begin{theorem}[{Poincaré inequality and curvature-dimension condition imply log-Sobolev inequality}]
\label{prop9.15}
Consider the Markov triple $(M, \nu, \Gamma)$ and let $\beta \geq 1$. Suppose that the following two conditions hold:
\begin{enumerate}
    \item $(M, \nu, \Gamma)$ satisfies a curvature-dimension condition with constant $-\kappa_1 < 0$, where $\kappa_1 \geq1$, denoted $\textup{CD}(-\kappa_1)$, i.e.
    \begin{equation}
    \label{CD3}
    \nabla^2 F + \frac{1}{\beta} \textup{Ric}_g \geq -\kappa_1 g.
    \end{equation}
    \item $(M, \nu, \Gamma)$ satisfies a $\textup{PI}(\kappa_2)$ with $\kappa_2 > 0$, where $\kappa_2 \leq 1$.
\end{enumerate}
Then  $(M, \nu, \Gamma)$ satisfies an $\textup{LSI}(\alpha)$ with
\[
\frac{1}{\alpha} = \frac{4\beta \kappa_1 \textup{diam}^*(M)^2}{\kappa_2},
\]
where $\textup{diam}^*(M) = \max\{\textup{diam}(M), 1\}$, and $\textup{diam}(M)$ is the diameter of $M$.
\end{theorem}

\begin{remark}
Note that we can assume that $\kappa_1 \geq 1$ since $\textup{CD}(-\kappa_1)$ with $0 < \kappa_1 < 1$ implies $\textup{CD}(-\kappa')$ for any $\kappa' \geq 1$. Similarly, we can assume that $1 \geq \kappa_2 > 0$, since $(M, \nu, \Gamma)$ satisfying a $\textup{PI}(\kappa_2)$ with $\kappa_2 > 1$ implies a $\textup{PI}(1)$.
\end{remark}

\begin{remark}
\cref{prop9.15} is a generalization of \cite[Proposition D.17]{LiErd2022Supp}. Note that we are allowing for the curvature-dimension condition to hold with respect to the negative constant $-\kappa_1$, as opposed to when we obtained a Poincaré inequality from the curvature-dimension condition, in which case we needed the constant to be positive (cf. \cref{lemma9.9,rem:CDimpliesPI}).
\end{remark}

To prove \cref{prop9.15}, we will proceed as in \cite[Section D.6]{LiErd2022Supp}. We will first obtain an upper bound on the relative entropy of $\rho_t$---the distribution of the Langevin diffusion $X_t$---with respect to $\nu$, under the assumption that the curvature-dimension condition holds. This upper bound will be written in terms of the Fisher information and the Wasserstein distance between $\rho_t$ and $\nu$. We will next get a defective log-Sobolev inequality from it. Finally, this inequality will be tightened using standard techniques. 

\begin{definition}[Wasserstein distance]
Let $(\mathcal{X}, d)$ be a separable and complete metric space, and let $p \in [1, \infty)$. Let $\mu$ and $\eta$ be two probability measures of $\mathcal{X}$. We define the \textup{Wasserstein distance} of order $p$ between $\mu$ and $\eta$ as
\[
W_p(\mu, \eta) := \left(\inf_{\pi \in \Pi(\mu, \eta)} \int_{\mathcal{X} \times \mathcal{X}} d(x, y)^p \, d\pi(x, y)\right)^{1/p},
\]
where $\Pi(\mu, \eta)$ is the set of joint probability measures on $\mathcal{X} \times \mathcal{X}$ whose marginals are $\mu$ and $\eta$. 
\end{definition}

As we mentioned earlier, the Wasserstein distance allows us to upper bound the relative entropy between two probability distributions in the presence of a curvature-dimension condition. 

\begin{proposition}[{HWI Inequality, \cite[Corollary 20.13]{villani2008optimal}}]
\label{theorem9.13}
Let $(M, g)$ be a complete Riemannian manifold equipped with a measure $d\mu = e^{-V}\textup{dVol}_g$, $V \in C^2(M)$, satisfying the following curvature-dimension bound: 
\[
\nabla^2 V + \textup{Ric}_g \geq \kappa g,
\]
for some $\kappa \in \bb{R}$. Then  if $\mu \in P_2(M)$, i.e. $\mu$ is a probability measure and
\[
\int_M d_g(x, x_0)^2 d\mu < \infty,\quad \forall x_0 \in M,
\]
for any $\eta \in P_2(M)$, $\eta \ll \mu$, it holds that
\begin{equation}
\label{eq:prelimineqH}
H(\eta | \mu) \leq W_2(\eta, \mu)\sqrt{\mathbf{I}(\eta | \mu)} - \frac{\kappa}{2}W_2(\eta, \mu)^2,    
\end{equation}
where $\mathbf{I}(\eta | \mu)$ is the \textit{unadjusted} Fisher information, i.e. $\mathbf{I}(\eta | \mu) := \int_M \frac{|\Grad{g} f|_g^2}{f} d\mu$, with $f = \frac{d\eta}{d\mu}$. 
\end{proposition}

The following lemma provides an upper bound on the Wasserstein distance that will be useful to derive a defective log-Sobolev inequality from \cref{eq:prelimineqH}. 

\begin{lemma}[{\cite[Theorem 6.15]{villani2008optimal}}]
\label{9.14Li}
Let $\mu, \eta$ be two probability measures on a separable and complete metric space $(\mathcal{X}, d)$. Then for every $p \geq 1$ and every $x_0 \in \mathcal{X}$, it holds that
\[
W_p(\mu, \eta) \leq 2^{\frac{1}{p'}}\left\{\int_{\mathcal{X}} d(x_0, x)^p \, d|\mu - \eta|(x)\right\}^{\frac{1}{p}},
\]
where $\frac{1}{p} + \frac{1}{p'} = 1$.
\end{lemma}

Using this result we can easily upper bound the Wasserstein distance by the diameter of the manifold. This result is analogous to \cite[Particular case 6.16]{villani2008optimal}.
\begin{remark}
\label{cotaWasserstein}
Let $(M, g)$ be a compact Riemannian manifold with diameter $\textup{diam}(M)$. Then  for any two probability measures $\mu, \eta \in P_2(M)$ it holds that
\[
W_2(\mu, \eta)^2 \leq 4\, \textup{diam}(M)^2.
\]
\end{remark}

\begin{proof}
Setting $p = 2$ in \cref{9.14Li} we obtain 
\[
W_2(\mu, \eta)^2 \leq 2\, \textup{diam}(M)^2 \int_{M} d|\mu - \eta|(x) = 4\,\textup{diam}(M)^2 \norm{\mu - \eta}_{\text{TV}}.
\]
Now
\[
\norm{\mu - \eta}_{\text{TV}} = \sup\{|\mu(A) - \eta(A)| : A \in \mathcal{F}\} \leq 1,
\]
yielding
\[
W_2(\mu, \eta)^2 \leq 4\, \textup{diam}(M)^2.
\]
\end{proof}

Lastly, as we mentioned earlier, we will end the proof of \cref{prop9.15} with a \textit{tightening} argument, which will allow us to obtain a tight log-Sobolev inequality from the existence of a defective log-Sobolev inequality and a Poincaré inequality. To this end, we will use the following result, which we include without proof. 

\begin{proposition}[{Tightening of a log-Sobolev inequality with a Poincaré inequality - \cite[Proposition 5.1.3]{bakry2013analysis}}]
\label{proptightening}
Assume that $(M, \nu, \Gamma)$ satisfies an $\textup{LSI}(C_1, C_2)$ and a $\textup{PI}(C')$. Then  the Markov triple satisfies a log-Sobolev inequality with constant
\[
\frac{1}{\frac{1}{C_1} + \frac{1}{C'}(\frac{C_2}{2} + 1)}.
\]
\end{proposition}

With the above results in mind, we are now ready to prove \cref{prop9.15}. 

\begin{proof}[Proof of \cref{prop9.15}]
Let $\nu$ be the Gibbs distribution, and let $\mu$ be a probability measure on $M$ satisfying $\mu \ll \nu$ and such that $\frac{d\mu}{d\nu} \in C^2(M)$. Consider the relative entropy between $\mu$ and $\nu$, $H(\mu|\nu)$. In order to apply \cref{theorem9.13}, we have to rewrite $\nu$ as $\nu = e^{-V}$ for a suitable function $V$. It suffices to consider $V = \beta F + \log(Z)$. For this choice of $V$, the curvature-dimension condition of \cref{theorem9.13} is fulfilled. Indeed,
\begin{align*}
\nabla^2 V + \textup{Ric}_g = \nabla^2(\beta F + \log(Z)) + \textup{Ric}_g = \beta \nabla^2 F + \textup{Ric}_g \geq -\beta \kappa_1 g,
\end{align*}
where the inequality follows by multiplying the curvature-dimension condition written in \cref{CD3} by $\beta$. Therefore, applying \cref{theorem9.13} we conclude that
\begin{equation}
\label{ineqrelativeentropy}
H(\mu | \nu) \leq W_2(\mu, \nu) \sqrt{\mathbf{I}(\mu | \nu)} + \frac{\beta \kappa_1}{2} W_2(\mu, \nu)^2,
\end{equation}
where
\[
\mathbf{I}(\mu | \nu) = \beta I(\mu | \nu),
\]
and $I(\mu | \nu)$ is the Fisher information. Furthermore, we can apply Young's inequality 
\[
ab \leq \frac{\varepsilon}{2}a^2 + \frac{1}{2\varepsilon}b^2, \quad \forall \varepsilon > 0,\ \forall a, b \geq 0.
\]
to \cref{ineqrelativeentropy} with $a = \sqrt{\mathbf{I}(\mu | \nu)}$ and $b = W_2(\mu, \nu)$ to obtain
\begin{align*}
H(\mu | \nu) &\leq \frac{\varepsilon}{2}\mathbf{I}(\mu | \nu) + \left(\frac{1}{2\varepsilon} + \frac{\beta \kappa_1}{2}\right)W_2(\mu, \nu)^2\\
&\leq \frac{\varepsilon}{2}\mathbf{I}(\mu | \nu) + \left(\frac{1}{2\varepsilon}+ \frac{\beta \kappa_1}{2}\right)4\,\textup{diam}(M)^2\\
&= \frac{\varepsilon\beta}{2}I(\mu | \nu) + \left(\frac{1}{2\varepsilon}+ \frac{\beta \kappa_1}{2}\right)4\, \textup{diam}(M)^2,
\end{align*}
where the second inequality follows from \cref{cotaWasserstein}. 

The above inequality is a defective logarithmic Sobolev inequality (cf. \cref{defectiveLSI}) with constants $\alpha = \frac{1}{\varepsilon \beta}$ and $A = \left(\frac{1}{2\varepsilon}+ \frac{\beta \kappa_1}{2}\right)4\, \textup{diam}(M)^2$. Thus, we can apply the tightening result from \cref{proptightening} to conclude that $(M, \nu, \Gamma)$ satisfies a tight logarithmic Sobolev inequality with constant 
\[
\frac{1}{\varepsilon \beta + \frac{1}{\kappa_2}\left(\left(\frac{1}{2\varepsilon}+ \frac{\beta \kappa_1}{2}\right)2\,\textup{diam}(M)^2 + 1\right)},
\]
i.e.
\[
H(\mu | \nu) \leq \frac{1}{2}\left(\varepsilon\beta + \frac{1}{\kappa_2}\left\{\left(\frac{1}{2\varepsilon}+ \frac{\beta \kappa_1}{2}\right)2\,\textup{diam}(M)^2 + 1\right\}\right)I(\mu | \nu).
\]
Since this inequality works for all $\varepsilon > 0$ we can choose $\varepsilon = \frac{\textup{diam}(M)}{\sqrt{\beta \kappa_2}}$ to obtain the least upper bound, i.e.
\[
\frac{1}{\tilde\alpha} := 2\,\textup{diam}(M)\sqrt{\frac{\beta}{\kappa_2}} + \frac{\beta \kappa_1 \textup{diam}(M)^2 + 1}{\kappa_2}.
\]

Let us conclude by obtaining a more readable log-Sobolev inequality constant. Indeed as we are assuming without loss of generality that $\beta \geq 1$, $0 < \kappa_2 \leq 1$, $\kappa_1 \geq 1$ and defining $\textup{diam}^*(M) := \max\{1, \textup{diam}(M)\}$ and we can bound 
\begin{align*}
2\,\textup{diam}^*(M)\sqrt{\frac{\beta}{\kappa_2}} + \frac{\beta \kappa_1 \textup{diam}^*(M)^2 + 1}{\kappa_2} &\leq \frac{2\beta \,\textup{diam}^*(M) + \beta \kappa_1 \textup{diam}^*(M)^2 + 1}{\kappa_2}\\
&\leq \frac{4\beta \kappa_1 \textup{diam}^*(M)^2}{\kappa_2},
\end{align*}
finishing the proof. 
\end{proof}

%% file: Chapters/ProofThm.tex
\section{Rapid mixing for Gibbs measures in Riemannian manifolds}
\label{sec:proof}

We are now in a position to state and prove our rapid mixing result for Langevin dynamics on Riemannian manifolds. An informal presentation thereof was given in \cref{thm:MainInformal1}. To this end, we will put together the main results of \cref{SectionPI,SectionLift,SectionPItoLSI}. We will first prove that $(M, \nu, \Gamma)$ and $(M/G, \tilde\nu, \tilde \Gamma)$ both satisfy a log-Sobolev inequality. Next, we will see how the rapid mixing result follows. 

\begin{theorem}[{\cref{thm:MainInformal1}} - formal version]
\label{thm:MainFormal1}
Let $(M, g)$ be a compact Riemannian manifold. Let $G$ be a group action on $M$ and satisfying \cref{assumption3.7}. Let $\pi: (M, g) \to (M/G, h)$ be the induced Riemannian submersion and assume that it has totally geodesic fibers. Further assume that the fibers of $\pi$---which are isometric to $(G, \hat{g})$, where $\hat{g}$ corresponds to their induced metric---satisfy \cref{assumption3.8}. Assume that $M$ and $M/G$ satisfy \cref{assumption3.9} with constants $\mathbf{K}$, $R_M$ and $R_{M/G}$, and $M/G$ satisfies \cref{assumptionsym}.

Let $F: M \to \bb{R}$ be a smooth function satisfying \crefrange{assumption3.1}{assumption3.2} and let $\tilde F: M/G \to \bb{R}$ be the unique function such that $F = \tilde F \circ \pi$. Assume that $\tilde F$ satisfies \crefrange{assumption3.3}{assumption3.6}. Let $a, \beta > 0$ be such that
\begin{align*}
a^2 &\geq \max\Big\{\frac{24 A_2 \dim(M/G)}{C^2_{\tilde F}}, \frac{288}{\lambda_*}\Big\},
\\\beta &\geq \max\Bigg\{\frac{192^2 \dim(M/G)^5 A^2_2 A^2_3 \mathbf{K}^2 a^6}{\lambda_*^2},  \frac{9a^2}{D^2}, \frac{4a^2}{i(M/G)^2}, \frac{4R_{M/G}}{\lambda_*}, \frac{a^2}{\mathit{conv}(M/G)^2}\Bigg\},
\end{align*}
where $C_{\tilde F}$ is the constant from \cref{eq:boundgradientdistance}, $\mathit{conv}(M/G)$ denotes the convexity radius of $(M/G, h)$, $i(M/G)$ denotes the injectivity radius of $(M/G, h)$, $D$ is the lower bound on the distance between any two critical points of $\tilde F$, $\lambda_*$ is the bound on the eigenvalues of $\nabla^2 \tilde F$ at the critical points, $A_2 \geq 1$ is the Lipschitz constant of the gradient of $F$ and $\tilde F$, and $A_3 \geq 1$ is the Lipschitz constant of the Hessian of $\tilde F$. 

Then  the Markov triple $(M, \nu, \Gamma)$ satisfies an $\textup{LSI}(\alpha_M)$, with
\begin{equation}
\label{eq:constantalphaM}
\frac{1}{\alpha_M} = 4\beta (A_2 + R_M) \textup{diam}^*(M)^2\max \Big\{\frac{120}{\lambda_*}, \frac{\textup{diam}(G)^2 \beta}{\pi^2}\Big\},
\end{equation}
where $\textup{diam}(G)$ denotes the diameter of $G$ as a fiber of $\pi$ and $\textup{diam}^*(M)$ denotes the maximum between the diameter of $(M, g)$ and $1$. Furthermore, the Markov triple $(M/G, \tilde \nu, \tilde \Gamma)$ satisfies an $\textup{LSI}(\alpha_{M/G})$, with
\begin{equation}
\label{eq:constantalphaMG}
\frac{1}{\alpha_{M/G}} = \frac{480\beta (A_2 + R_{M/G}) \textup{diam}^*(M/G)^2}{\lambda_*}.    
\end{equation}
\end{theorem}

Note that there are two extra assumptions needed to prove the existence of a log-Sobolev inequality for the Markov triples $(M, \nu, \Gamma)$ and $(M/G, \tilde \nu, \tilde \Gamma)$, compared to those needed to prove the existence of a Poincaré inequality for $(M/G, \tilde \nu, \tilde \Gamma)$. The first is that the Ricci curvature of the fibers of $\pi$ is non-negative, in order to be able to lift the Poincaré inequality from $M/G$ to $M$. The second is the existence of a lower bound on the Ricci curvature of $M$, which allows us to tighten the Poincaré inequality lifted to $M$ to a log-Sobolev inequality. 

Also observe that in the above theorem we claim that the fibers of $\pi$ are isometric. This is always the case whenever the Riemannian submersion $\pi$ has totally geodesic fibers \cite[Proposition 1.8]{escobales1975riemannian}. 

As we mentioned before, the proof of \cref{thm:MainFormal1} only consists in putting together the results obtained in the previous sections. In particular, we will first use the Poincaré inequality theorem---which corresponds to \cref{prop9.12}---to conclude that $(B, \tilde \nu, \tilde \Gamma)$ satisfies a Poincaré inequality. Then  using the lifting result shown in \cref{thmliftPI} we will obtain a Poincaré inequality on $(M, \nu, \Gamma)$. Lastly, using the tightening arguments from \cref{prop9.15} we will \textit{upgrade} both Poincaré inequalities to obtain logarithmic Sobolev inequalities.

\begin{proof}[Proof of \cref{thm:MainFormal1}]
First, recall that under the assumptions of the theorem, we can apply \cref{prop9.12} and conclude that the Markov triple $(M/G, \tilde \nu, \tilde \Gamma)$ satisfies a $\textup{PI}(\kappa_{M/G})$ with
\[
\kappa_{M/G} = \frac{\lambda_*}{120}. 
\]
Now, we can lift this Poincaré inequality from $M/G$ to $M$ by simply applying \cref{thmliftPI}, and claim that $(M, \nu, \Gamma)$ satisfies a $\textup{PI}(\kappa_M)$ with
\begin{align*}
\frac{1}{\kappa_{M}} &= \max \Big\{\frac{1}{\kappa_{M/G}}, \frac{\textup{diam}(G)^2 \beta}{\pi^2}\Big\}\\
&= \max \Big\{\frac{120}{\lambda_*}, \frac{\textup{diam}(G)^2 \beta}{\pi^2}\Big\}. 
\end{align*}

After obtaining a Poincaré inequality for $(M, \nu, \Gamma)$ and $(M/G, \tilde \nu, \tilde \Gamma)$, it only remains to tighten them to logarithmic Sobolev inequalities. This can be done using \cref{prop9.15}. Indeed, note that by \cref{rmk:LipschitzConstants} it holds that 
\begin{align*}
\nabla^2_g F &\geq -A_2 g,\\
\nabla^2_h \tilde F &\geq -A_2 h,    
\end{align*}
and so, as we are assuming that $\textup{Ric}_g \geq -R_M\, g$, $\textup{Ric}_h \geq -R_{M/G}\, h$ for some constants $R_M, R_{M/G} \geq 0$, we can conclude that
\begin{align*}
\nabla_g^2 F + \frac{1}{\beta}\textup{Ric}_g &\geq -(A_2 + R_M)\,g,\\
\nabla_h^2 \tilde F + \frac{1}{\beta}\textup{Ric}_h &\geq -(A_2 + R_{M/G})\, h.
\end{align*}
Thus, by \cref{prop9.15} we can conclude that $(M, \nu, \Gamma)$ satisfies an $\textup{LSI}(\alpha_M)$ with 
\begin{align*}
\frac{1}{\alpha_M} &= \frac{4\beta (A_2 + R_M) \textup{diam}^*(M)^2}{\kappa_M}\\
&= 4\beta (A_2 + R_M) \textup{diam}^*(M)^2\max \Big\{\frac{120}{\lambda_*}, \frac{\textup{diam}(G)^2 \beta}{\pi^2}\Big\}.
\end{align*}
and $(M/G, \tilde \nu, \tilde \Gamma)$ satisfies an $\textup{LSI}(\alpha_{M/G})$ with
\begin{align*}
\frac{1}{\alpha_{M/G}} = \frac{4\beta (A_2 + R_{M/G}) \textup{diam}^*(M/G)^2}{\kappa_{M/G}} = \frac{480\beta (A_2 + R_{M/G}) \textup{diam}^*(M/G)^2}{\lambda_*}.
\end{align*}
\end{proof}

As we mentioned previously, the rapid mixing result shown in \cref{thm:MainInformal1} becomes a corollary of \cref{thm:MainFormal1}.  
\begin{corollary}
\label{cor:rapidmixing}
Under the assumptions of \cref{thm:MainFormal1}, the Langevin diffusion process $X_t$ defined in \cref{LangevinDiffusionEq} with initial uniform distribution on $M$ converges exponentially fast to the Gibbs distribution $\nu$, i.e. 
\[
\norm{\nu - \rho_t}^2_{\textup{TV}} \leq  \beta\, e^{-2\alpha_M t} \max_{y \in M} F(y), \quad \forall t \geq 0, 
\]
where $\rho_t$ denotes the distribution of $X_t$ and $\alpha_M$ denotes the log-Sobolev inequality constant of $(M, \nu, \Gamma)$. Similarly, the Langevin diffusion process $\tilde X_t$ defined analogously with initial uniform distribution on $M/G$ also converges exponentially fast to the Gibbs distribution $\tilde \nu$, i.e. 
\[
\norm{\tilde \nu - \tilde \rho_t}^2_{\textup{TV}} \leq  \beta\, e^{-2\alpha_{M/G} t} \max_{y \in M/G} \tilde F(y) = \beta\, e^{-2\alpha_{M/G} t} \max_{y \in M} F(y) , \quad \forall t \geq 0,
\]
where $\tilde \rho_t$ denotes the distribution of $\tilde X_t$ and $\alpha_{M/G}$ denotes the log-Sobolev inequality constant of $(M/G, \tilde \nu, \tilde \Gamma)$. 
\end{corollary}
\begin{proof}
In \cref{thm:MainFormal1} we proved that $(M, \nu, \Gamma)$ satisfies an $\textup{LSI}(\alpha_M)$ (see \cref{eq:constantalphaM}) and the Markov triple $(M/G, \tilde \nu, \tilde \Gamma)$ satisfies an $\textup{LSI}(\alpha_{M/G})$ (see \cref{eq:constantalphaMG}). Using \cref{thm:LSIImpliesHypercontract} we can obtain an upper bound on the total variation distance between the densities of $X_t$ and $\tilde X_t$ and the Gibbs distributions $\nu$ and $\tilde \nu$,
\[
||\rho_t - \nu||_{\textup{TV}}^2 \leq \frac{1}{2} e^{-2\alpha_M t} H(\rho_0|\nu),\quad \textup{and}\quad ||\tilde \rho_t - \tilde \nu||_{\textup{TV}}^2 \leq \frac{1}{2} e^{-2\alpha_{M/G}\, t} H(\tilde \rho_0|\tilde \nu),
\]
for every $t \geq 0$. Moreover, by \cref{BoundInitialRelativeEntropy}, we can bound the relative entropy between the uniform distribution and the Gibbs distribution on $M$ and $M/G$ as
\[
H(\rho_0 | \nu) \leq \beta \max_{y \in M} F(y),\quad \textup{and}\quad H(\tilde \rho_0 | \tilde \nu) \leq \beta \max_{y \in M/G} \tilde F(y) = \beta \max_{y \in M} F(y),
\]
concluding the proof.
\end{proof}

%% file: Chapters/Suboptimality.tex
\section{Suboptimality of the Gibbs distribution}
\label{SectionSuboptimality}

The aim of this section is to prove that, given a compact Riemannian manifold $(M, g)$ and a sufficiently smooth function $F: M \to \bb{R}$, the Gibbs distribution associated with $F$, $\nu = \frac{1}{Z}e^{-\beta F}$, \textit{finds} the minima of $F$. This result was stated informally in \cref{thm:MainInformal2}.

More formally, we will show that, by choosing appropriate values for $\beta$, which scale polynomially with the dimension of the manifold, and logarithmically with respect to the volume of the manifold and the Lipschitz constant of the gradient of $F$, the distribution $\nu$ gets concentrated around the minima of the function $F$.

\begin{theorem}[{\cref{thm:MainInformal2}} - formal version]
\label{thm:MainFormal2}
Let $(M, g)$ be a compact Riemannian manifold of dimension $d \geq 2$, and let $F \in C^2(M)$ with an $A_2$-Lipschitz gradient, where without loss of generality $A_2 \geq 1$. Let 
\[
\varepsilon_{\max} := \min \Big\{\frac{i(M)^2 A_2}{8}, 1\Big\},
\]
where $i(M)$ denotes the injectivity radius of $M$. For every $\varepsilon \in (0, \varepsilon_{\max}]$ and $\delta \in (0, 1)$, if
\[
\beta \geq \frac{2}{\varepsilon}\left( \frac{1}{2} +d\log\frac{d^3A_2 \textup{Vol}^*(M)\sqrt{2\pi}}{\delta\varepsilon} \right),
\]
where $\textup{Vol}^*(M) := \max\{\textup{Vol}(M), 1\}$, then the Gibbs distribution $\nu(x) = \frac{1}{Z} e^{-\beta F(x)}$ satisfies
\[
\nu\left(F - \min_{y \in M} F(y) \geq \varepsilon\right) \leq \delta.
\]
\end{theorem}

In the above theorem one can assume, without loss of generality, that 
\[
\min_{x\in M} F(x) = 0.
\]

This section is based on \cite[Appendix C]{LiErd2022Supp}. Note that the statement of the theorem---which corresponds to \cite[Theorem 2.5]{LiErd2022}---has remained almost identical to the original, as well as the auxiliary lemmas needed to prove it. Nevertheless, some of the proofs have been adapted to generalize the result to an arbitrary compact manifold $M$. 

\begin{lemma}[{\cite[Lemma C.1]{LiErd2022Supp}}] 
\label{lem:firstlemsubopt}
Let $(M, g)$, $\varepsilon_{\max}$ and $F$ be as in \cref{thm:MainFormal2}. Let $x^* \in M$ be any global minimum of $F$. Then  for every $\varepsilon \in (0, \varepsilon_{\max}]$ the following bound holds:
\[
\nu(F - \min_{y \in M} F(y) \geq \varepsilon ) \leq \frac{e^{-\beta \varepsilon} \textup{Vol}(M)}{\int_{\mathcal{B}(R, x^*)} e^{-\beta \frac{A_2}{2} d_g(x^*, x)^2}\; \textup{dVol}_g(x)},
\]
where $R := \sqrt{\frac{2\varepsilon}{A_2}}$, $\mathcal{B}(R, x^*)$ denotes the geodesic ball of radius $R$ centered at $x^*$, and $d_g(\cdot, \cdot)$ is the geodesic distance of $(M, g)$. 
\end{lemma}

Note that we require the geodesic ball of radius $R$ centered at $x^*$ to be contained in the geodesic ball with radius equal to the injectivity radius of $M$. Under this extra assumption, the original proof of \cref{lem:firstlemsubopt} remains valid. For this reason, we do not include a proof for it in this work.  

For the next auxiliary lemma, which is an adapted version of \cite[Lemma C.2]{LiErd2022Supp}, we will incorporate a proof, as the techniques used are different.

\begin{lemma}
Let $(M, g)$ be a compact Riemannian manifold of dimension $d \geq 2$. Let $A_2 > 0$ be a fixed constant,
\[
\varepsilon_{\max} = \min\Big\{\frac{i(M)^2 A_2}{8}, 1\Big\}
\]
and let $\varepsilon \in (0, \varepsilon_{\max}]$, $R = \frac{2\varepsilon}{A_2}$. Then  for every
\[
\beta \geq \frac{1+2\log2}{2\varepsilon},
\]
and every $y \in M$ it holds that
\[
\int_{\mathcal{B}(R, y)} e^{-\beta \frac{A_2}{2} d_g(y, x)^2}  \; \textup{dVol}_g(x) \geq 2^{d-1} \pi^{1/2} \frac{\Gamma(\frac{d+1}{2})^{d-1}}{\Gamma(\frac{d}{2})^d d^{d-1}} \left(\frac{1}{\beta A_2}\right)^{d/2}e^{-1/2}.
\]
\end{lemma}
\begin{proof}
By the Gauss lemma, since $R \leq i(M)/2$,
\begin{equation}
\label{eq:rewrittenintegral}
\int_{\mathcal{B}(R, y)} e^{-\beta \frac{A_2}{2} d_g(y, x)^2} \; \textup{dVol}_g(x) = \int_0^R e^{-\beta\frac{A_2}{2}\rho^2} \left(\int_{S_{y}(\rho)} d\sigma \right)\; d\rho,
\end{equation}
where
\[
S_{y}(\rho) := \{\text{exp}_{y}(v)\, : \, |v|_g = \rho\}
\]
is the geodesic sphere of radius $\rho$ centered at $y$, $\exp_{y}$ denotes the exponential map centered at $y$, and $d\sigma$ is the induced volume form on $S_{y}(\rho)$. Now, since $M$ is compact, by \cite[Proposition 14]{Croke1980}, it holds that 
\[
\textup{Vol}(S_{y}(\rho)) \geq  2^{d-1} \frac{\textup{Vol}_{g_{\mathit{round}}}(\bb{S}^{d-1})^d}{\textup{Vol}_{g_{\mathit{round}}}(\bb{S}^{d})^{d-1} d^{d-1}} \rho^{d-1},
\]
for all $\rho \leq i(M)/2$, where $\textup{Vol}_{g_{\mathit{round}}}(\bb{S}^{d})$ denotes the volume of the $d$-dimensional round sphere of radius $1$, which is given by (cf. \cite{Gray1979})  
\[
\textup{Vol}_{g_{\mathit{round}}}(\bb{S}^{d}) = \frac{2 \pi^{(d+1)/2}}{\Gamma(\frac{d+1}{2})}.
\]
Thus, 
\begin{align*}
\textup{Vol}(S_{y}(\rho)) &\geq  2^{d-1} \frac{\textup{Vol}_{g_{\mathit{round}}}(\bb{S}^{d-1})^d}{\textup{Vol}_{g_{\mathit{round}}}(\bb{S}^{d})^{d-1} d^{d-1}} \rho^{d-1}
\\&= 2^{d-1} \frac{\left(\frac{2 \pi^{d/2}}{\Gamma(\frac{d}{2})}\right)^d}{\left(\frac{2 \pi^{(d+1)/2}}{\Gamma(\frac{d+1}{2})}\right)^{d-1} d^{d-1}} \rho^{d-1}\\
&=2^{d} \pi^{1/2} \frac{\Gamma(\frac{d+1}{2})^{d-1}}{\Gamma(\frac{d}{2})^d d^{d-1}} \rho^{d-1}
\\&= \mathcal{A}_d\, \rho^{d-1},
\end{align*}
where
\[
\mathcal{A}_d := 2^{d} \pi^{1/2} \frac{\Gamma(\frac{d+1}{2})^{d-1}}{\Gamma(\frac{d}{2})^d d^{d-1}}.
\]

This way, we can lower-bound the right-hand side of \cref{eq:rewrittenintegral} as
\[
\int_0^R e^{-\beta\frac{A_2}{2}\rho^2} \left(\int_{S_{y}(\rho)}\; d\sigma \right)\; d\rho \geq \mathcal{A}_d \int_0^R e^{-\beta\frac{A_2}{2}\rho^2} \rho^{d-1}\; d\rho.
\]
Applying the change of variable $r = \sqrt{\beta A_2}\rho$ we get
\begin{align}
\label{eq1}
\nonumber
\int_0^R e^{-\beta\frac{A_2}{2}\rho^2} \rho^{d-1}\; d\rho &= \int_0^{\sqrt{\beta A_2}R} e^{-\frac{r^2}{2}} \left(\frac{1}{\beta A_2}\right)^{d/2} r^{d-1}\; dr 
\\&= \left(\frac{1}{\beta A_2}\right)^{d/2} \int_0^{\sqrt{\beta A_2}R} e^{-\frac{r^2}{2}} r^{d-1}\; dr.    
\end{align}
Now, since $d \geq 2$ and $R = \sqrt{\frac{2\varepsilon}{A_2}}$, if we choose $\beta$ such that 
\[
\beta \geq \frac{1}{A_2 R^2} = \frac{1}{2\varepsilon},
\]
it holds that 
\[
\sqrt{\beta A_2}R \geq 1,
\]
and so we can lower bound \cref{eq1} as follows:
\begin{align*}
\left(\frac{1}{\beta A_2}\right)^{d/2} \int_0^{\sqrt{\beta A_2}R} e^{-\frac{r^2}{2}} r^{d-1}\; dr &\geq \left(\frac{1}{\beta A_2}\right)^{d/2} \int_1^{\sqrt{\beta A_2}R} e^{-\frac{r^2}{2}} r^{d-1}\; dr\\
&\geq \left(\frac{1}{\beta A_2}\right)^{d/2} \int_1^{\sqrt{\beta A_2}R} e^{-\frac{r^2}{2}} r \; dr \\
&= \left(\frac{1}{\beta A_2}\right)^{d/2} \left.\left(-e^{-\frac{r^2}{2}}\right)\right|_{1}^{\sqrt{\beta A_2} R}\\
&= \left(\frac{1}{\beta A_2}\right)^{d/2} \left(e^{-1/2} - e^{-\frac{\beta A_2 R^2}{2}}\right)\\
&= \left(\frac{1}{\beta A_2}\right)^{d/2} \left(e^{-1/2} - e^{-\beta \varepsilon}\right).
\end{align*}

Furthermore, 
\[
e^{-1/2} - e^{-\beta \varepsilon} \geq \frac{1}{2}e^{-1/2} \iff \log\frac{1}{2} - \frac{1}{2} \geq -\beta \varepsilon \iff \beta \geq \frac{1}{2\varepsilon} + \frac{\log 2}{\varepsilon} = \frac{1+2\log 2}{2\varepsilon}.
\]
Therefore, for such values of $\beta$
\[
\left(\frac{1}{\beta A_2}\right)^{d/2} \int_0^{\sqrt{\beta A_2}R} e^{-\frac{r^2}{2}} r^{d-1}\; dr \geq \frac{1}{2} \left(\frac{1}{\beta A_2}\right)^{d/2}e^{-1/2},
\]
which allows us to conclude that
\[
\int_{\mathcal{B}(R, y)} e^{-\beta \frac{A_2}{2} d_g(y, x)^2} \; \textup{dVol}_g(x) \geq 2^{d-1} \pi^{1/2} \frac{\Gamma(\frac{d+1}{2})^{d-1}}{\Gamma(\frac{d}{2})^d d^{d-1}} \left(\frac{1}{\beta A_2}\right)^{d/2}e^{-1/2}  .
\]
\end{proof}

With the above two auxiliary lemmas in mind, we will now write the proof of \cref{thm:MainFormal2}. It has the same structure as in \cite{LiErd2022Supp}. Nevertheless, we have decided to incorporate it in our work for completeness, as the theorem is now stated for any general compact manifold $M$. 

\begin{proof}[Proof of \cref{thm:MainFormal2}]

From the previous lemmas we obtained the following bound:
\[
\nu(F \geq \varepsilon) \leq \frac{e^{-\beta \varepsilon} \textup{Vol}(M)}{2^{d-1} \pi^{1/2} \frac{\Gamma(\frac{d+1}{2})^{d-1}}{\Gamma(\frac{d}{2})^d d^{d-1}} \left(\frac{1}{\beta A_2}\right)^{d/2}e^{-1/2}},
\]
for $\beta \geq \frac{1+2\log 2}{2\varepsilon}$. If we now define 
\[
\mathscr{B}_d := \frac{e^{1/2}\Gamma(\frac{d}{2})^d d^{d-1}A_2^{d/2} \textup{Vol}(M)}{2^{d-1} \pi^{1/2} \Gamma(\frac{d+1}{2})^{d-1} },
\]
it only remains to find a value of $\beta$ such that 
\begin{equation}
\label{goalineq}
\beta^{d/2} e^{-\beta \varepsilon} \mathscr{B}_d \leq \delta.    
\end{equation}
Taking logarithms in this expression we obtain
\[
\frac{d}{2}\log \beta - \beta\varepsilon + \log \mathscr{B}_d \leq \log \delta,
\]
which can be rewritten as 
\begin{equation}
\label{eq:inequalitybeta}
\beta\varepsilon - \frac{d}{2} \log \beta \geq \log \frac{1}{\delta} + \log \mathscr{B}_d.
\end{equation}
Now, note that 
\[
\frac{\beta\varepsilon}{2} - \frac{d}{2} \log \beta
\]
is an increasing function of $\beta$ for $\beta \geq \frac{d}{\varepsilon}$. Therefore, for such values
\[
\frac{\beta\varepsilon}{2} - \frac{d}{2}\log \beta \geq \frac{d}{2}\Big(1 - \log \frac{d}{\varepsilon}\Big).
\]
Inserting this inequality into \cref{eq:inequalitybeta}, it is easy to see that we only need to find the values of $\beta$ for which
\[
\frac{\beta\varepsilon}{2} + \frac{d}{2}\Big(1 - \log \frac{d}{\varepsilon}\Big) \geq \log \frac{1}{\delta} + \log \mathscr{B}_d,
\]
that is, 
\[
\beta \geq \frac{2}{\varepsilon}\left(\log \frac{1}{\delta} + \log \mathscr{B}_d +  \frac{d}{2}\Big(\log \frac{d}{\varepsilon} - 1\Big)\right).
\]

Let us express the logarithm of $\mathscr{B}_d$ explicitly:
\begin{align*}
\log \mathscr{B}_d &= \log \left(\frac{e^{1/2}\Gamma(\frac{d}{2})^d d^{d-1}A_2^{d/2} \textup{Vol}(M)}{2^{d-1} \pi^{1/2} \Gamma(\frac{d+1}{2})^{d-1} }\right)
\\&= \frac{1}{2} + \log \Gamma\big(\frac{d+1}{2}\big) + d\log\frac{\Gamma\left(\frac{d}{2}\right)}{\Gamma\left(\frac{d+1}{2}\right)} + (d-1)\log d + \frac{d}{2}\log A_2 + \log \textup{Vol}(M)
\\&\quad - (d-1)\log 2 -\frac{1}{2}\log\pi 
\\&\leq \frac{1}{2} + \log \Gamma\big(\frac{d+1}{2}\big) + d\log\frac{\Gamma\left(\frac{d}{2}\right)}{\Gamma\left(\frac{d+1}{2}\right)} +  (d-1)\log d + \frac{d}{2}\log A_2 + \log \textup{Vol}(M)
\\&\leq \frac{1}{2} + \log \Gamma\big(\frac{d+1}{2}\big) + d\log\frac{\Gamma\left(\frac{d}{2}\right)}{\Gamma\left(\frac{d+1}{2}\right)} + d\log (dA_2) + \log \textup{Vol}(M).
\end{align*}

Furthermore, $\frac{\Gamma\left(\frac{d}{2}\right)}{\Gamma\left(\frac{d+1}{2}\right)}$ is a decreasing function of $d$ for every $d \geq 2$. Thus, 
\begin{align*}
\log \mathscr{B}_d &\leq \frac{1}{2} + \log \Gamma\big(\frac{d+1}{2}\big) + d\log\frac{\Gamma(1)}{\Gamma\left(\frac{3}{2}\right)} + d\log (dA_2) + \log \textup{Vol}(M)
\\&\leq \frac{1}{2} + \log \Gamma\big(\frac{d+1}{2}\big) + d\log (2dA_2) + \log \textup{Vol}(M).
\end{align*}

Therefore, we can choose 
\begin{align*}
\beta &\geq \frac{2}{\varepsilon}\left(\frac{1}{2} + \log \Gamma\big(\frac{d+1}{2}\big) + d\log (\frac{2d^2A_2}{\delta \varepsilon}) + \log \textup{Vol}(M)\right)\\    
&\geq \frac{2}{\varepsilon}\left(\log \frac{1}{\delta} + \log \mathscr{B}_d +  \frac{d}{2}(\log \frac{d}{\varepsilon} - 1)\right).
\end{align*}

Now, for every $d \geq 2$ we can use the following upper bound on the Gamma function (see \cite[Theorem 1.5]{batir2008inequalities})
\[
\Gamma\left(\frac{d+1}{2}\right) < \sqrt{2\pi} \left(\frac{d}{2e}\right)^{d/2},
\]
and so
\[
\log \Gamma\left(\frac{d+1}{2}\right) < \log \sqrt{2\pi} + \frac{d}{2}\log \frac{d}{2e}
\]
Thus, a suitable lower bound for $\beta$ is 
\[
\beta \geq \frac{2}{\varepsilon}\left(\frac{1}{2} + \log \sqrt{2\pi} + \frac{d}{2}\log \frac{d}{2e} + d\log (\frac{2d^2A_2}{\delta \varepsilon}) + \log \textup{Vol}(M)\right)
\]
Lastly, to simplify this, we note that whenever
\[
\beta \geq \frac{2}{\varepsilon}\left( \frac{1}{2} +d\log\frac{d^3A_2 \textup{Vol}^*(M)\sqrt{2\pi}}{\delta\varepsilon} \right),
\]
where $\textup{Vol}^*(M) := \max\{\textup{Vol}(M), 1\}$, \cref{goalineq} holds, finishing the proof. 
\end{proof}

\begin{remark}
The choice of $\beta$ in \cref{thm:MainFormal2} scales as 
\[
\Omega\left(\frac{d}{\varepsilon}\log\left(\frac{dA_2 \textup{Vol}^*(M)}{\delta \varepsilon}\right)\right).
\]
\end{remark}

In particular, in order for $\beta$ to grow polynomially with the dimension $d$, the value of $\textup{Vol}(M)$ and $A_2$ can grow exponentially with the dimension of $M$.  

%% file: Chapters/Examples.tex
\section{Examples}
\label{sec:traceratio}

In this section we will consider two scenarios in which most of the assumptions of \cref{thm:MainFormal1} can be easily verified---in particular, the existence of a unique minimum for the \textit{projected} function, and those regarding the eigenvalues of its Hessian at the critical points. We will study the functions considered in the trace ratio minimization problem, and in the problem of minimizing the energy of a two-dimensional classical ferromagnetic Ising spin model.

Note that similar results---in particular, results establishing absence of spurious local minima and the existence of saddle points with escape directions---have been established for numerous non-convex optimization problems in Euclidean space, including matrix sensing, matrix completion, phase retrieval, low-rank matrix factorization, and neural network training; see for example \cite{luo2022nonconvexmatrixfactorizationgeodesically,chi2019,sun2018geometric,sun2015complete,ge2016matrix,bhojanapalli2016global,li2019non} and the references therein.

\subsection{Trace ratio minimization}

The trace ratio minimization problem consists in finding the minimum of the function
\begin{align}
\label{eq:traceratiofunction}
\begin{split}
f: \textup{V}_k(\bb{C}^n) \to \bb{R},\quad X \mapsto \frac{\Tr(X^\dagger A X)}{\Tr(X^\dagger B X)},
\end{split}
\end{align}
where $\textup{V}_k(\bb{C}^n)$ denotes the Stiefel manifold of linear isometries from $\bb{C}^k$ to $\bb{C}^n$, with $n \geq 2$ and $1 \leq k < n$ (cf. \cref{curvaturestiefel}), and the matrices $A, B \in \bb{C}^{n \times n}$ are Hermitian with $B$ positive definite. 

Finding the minimum of this function---or the maximum, which corresponds to the minimization problem for $(-A, B)$---has attracted sustained attention in the literature. This problem is related to \textit{principal component analysis} (PCA) \cite{jolliffe2011principal} and more generally to dimensionality reduction techniques, such as graph embedding \cite{yan2005graph}. Furthermore, very similar problems are studied in the context of quantum information \cite{yanguez2025efficient,DMRGFrank}. 

In order to apply our results to the function studied in this problem, we want to find a \textit{projected version} of $f$ which has a unique minimum. To that end, let us first \textit{lift} it to $\textup{U}(n)$ by defining 
\begin{equation}
\label{eq:traceratioUn}
F: \textup{U}(n) \to \bb{R},\quad U := (X|Y) \mapsto \frac{\Tr(X^\dagger A X)}{\Tr(X^\dagger B X)},
\end{equation}
where $X$ corresponds to the first $k$ columns of $U$ and $Y$ corresponds to its last $n-k$ columns. If we consider the projection map
\[
\pi_1: \textup{U}(n) \to \textup{V}_k(\bb{C}^n),\quad (X | Y) \mapsto X
\]
studied in \cref{curvaturestiefel}, it is clear that
\[
F = f \circ \pi_1.
\] 

If we now consider the actions of the groups $\textup{U}(k) \times \textup{U}(n-k)$ and $\textup{U}(k)$ on $\textup{U}(n)$ and $\textup{V}_k(\bb{C}^n)$, respectively, defined as

\begin{minipage}{0.45\linewidth}
\begin{align*}
\textup{U}(k) \times \textup{U}(n-k) \times \textup{U}(n) &\to \textup{U}(n)\\
((A, B),U) &\mapsto U \begin{pmatrix}
A & 0\\
0 & B
\end{pmatrix},
\end{align*}
\end{minipage}
\begin{minipage}{0.45\linewidth}
\begin{align*}
\textup{U}(k) \times \textup{V}_k(\bb{C}^n) &\to \textup{V}_k(\bb{C}^n)\\
(U,X)  &\mapsto XU
\end{align*}
\end{minipage}

\noindent we showed that they induce the projection maps onto the Grassmann manifold
\[
\pi_2: \textup{V}_k(\bb{C}^n) \to \textup{Gr}_k(\bb{C}^n)\quad \textup{and}\quad \pi_3: \textup{U}(n) \to \textup{Gr}_k(\bb{C}^n),
\]
as seen in \cref{curvaturegrassmann}. Furthermore, $F$ is invariant under the above action on $\textup{U}(n)$: given any $U = (X | Y ) \in \textup{U}(n)$ and any $(A, B) \in \textup{U}(k) \times \textup{U}(n-k)$, it holds that 
\begin{align*}
F\bigg(U \begin{pmatrix}
A& 0\\
0 & B
\end{pmatrix}\bigg) = F((XA|YB)) = F((X|Y)) = F(U).
\end{align*}
Therefore, we can \textit{project} $F$ onto $\textup{Gr}_k(\bb{C}^n)$, obtaining the function 
\begin{align}
\label{eq:traceratiograssmann}
\begin{split}
\tilde f: \textup{Gr}_k(\bb{C}^n) \to \bb{R},\quad P \mapsto \frac{\Tr(P A)}{\Tr(P B)},
\end{split}
\end{align}
where we are using the description of the Grassmann manifold given in \cref{curvaturegrassmann}, i.e. 
\[
\textup{Gr}_k(\bb{C}^n) = \{P \in \bb{C}^{n \times n} : P^\dagger = P, P^2 = P, \Tr(P) = k\}.
\]
The function $\tilde f$ can also be seen as a \textit{projected version} of $f$ via $\pi_2$, by simply rewriting 
\[
\frac{\Tr(X^\dagger A X)}{\Tr(X^\dagger B X)} = \frac{\Tr(XX^\dagger A)}{\Tr(XX^\dagger B)},
\]
and denoting $P = XX^\dagger$. 

Recall that, if we endow $\textup{U}(n)$ with its bi-invariant metric $g$, there exists a metric $h_1$ on $\textup{V}_k(\bb{C}^n)$ and a metric $h_2$ on $\textup{Gr}_k(\bb{C}^n)$ for which the projections
\[
\pi_1: (\textup{U}(n), g) \to (\textup{V}_k(\bb{C}^n), h_1),\quad \textup{and}\quad 
\pi_3: (\textup{U}(n), g) \to (\textup{Gr}_k(\bb{C}^n), h_2),
\]
are Riemannian submersions with totally geodesic fibers. Therefore, by \cref{prop:inducedtotgeodfibers} we know that $\pi_2: (\textup{V}_k(\bb{C}^n), h_1) \to (\textup{Gr}_k(\bb{C}^n), h_2)$ is also a Riemannian submersion with totally geodesic fibers.

This way, we have defined three functions---namely $F$, $f$ and $\tilde f$---verifying 
\[
F = f \circ \pi_1,\quad \textup{and}\quad F = \tilde f \circ \pi_3,
\]
(see \cref{fig:figuresettingex1}). Having defined the trace minimization problem on $\textup{U}(n)$, $\textup{V}_k(\bb{C}^n)$ and $\textup{Gr}_k(\bb{C}^n)$, our aim will be to obtain log-Sobolev inequalities for the Markov triples 
\[
(\textup{U}(n), \nu_F, \Gamma_g),\quad (\textup{V}_k(\bb{C}^n), \nu_f, \Gamma_{h_1}), \quad \textup{and}\quad (\textup{Gr}_k(\bb{C}^n), \nu_{\tilde f}, \Gamma_{h_2}), 
\]
induced by the operators
\begin{equation}
\label{eq:diferentesoperators}
\operatorname{L}_F := -\textup{grad}_g\, F + \frac{1}{\beta}\Delta_g,\ \operatorname{L}_f := -\textup{grad}_{h_1}\, f + \frac{1}{\beta}\Delta_{h_1},\ \textup{and}\ \operatorname{L}_{\tilde f} := -\textup{grad}_{h_2} \, \tilde f + \frac{1}{\beta}\Delta_{h_2},
\end{equation}
where $\beta > 0$ is some fixed constant. Recall that, when considering these operators, the associated Gibbs measures are given by
\[
d\nu_F = \frac{1}{Z_F} e^{-\beta F}\textup{dVol}_g,\quad d\nu_f = \frac{1}{Z_f}e^{-\beta f}\textup{dVol}_{h_1},\quad \text{and}\quad d\nu_{\tilde f} = \frac{1}{Z_{\tilde f}} e^{-\beta \tilde f}\textup{dVol}_{h_2},
\]
where $Z_F, Z_f$ and $Z_{\tilde f}$ are their respective normalizing constants, and the carré du champ operators are given by
\[
\Gamma_g(\phi_1) = \frac{1}{\beta}|\textup{grad}_g\, \phi_1|^2_g,\quad \Gamma_{h_1}(\phi_2) = \frac{1}{\beta}|\textup{grad}_{h_1}\, \phi_2|^2_{h_1},\quad 
\Gamma_{h_2}(\phi_3) = \frac{1}{\beta}|\textup{grad}_{h_2}\, \phi_3|^2_{h_2},
\]
for any $C^2$ functions $\phi_1, \phi_2, \phi_3$ in $\textup{U}(n)$, $\textup{V}_k(\bb{C}^n)$ and $\textup{Gr}_k(\bb{C}^n)$, respectively.

\begin{figure}
    \centering
\begin{tikzpicture}[scale=1]
    \node[anchor=center] at (2.155, 2.465) {$\big((\textup{U}(n), g), F\big)$};
    \node[anchor=center] at (2.155, 1.337) {$\big((\textup{V}_k(\mathbb{C}^n), h_1), f\big)$};
    \node[anchor=center] at (2.155, 0.208) {$\big((\textup{Gr}_k(\mathbb{C}^n), h_2), \tilde f\,\big)$};
    \draw[->] (0.6, 1.265) .. controls (0.35, 0.889) and (0.35, 0.513) .. (0.6, 0.137);
    \draw[->] (0.6, 2.465) .. controls (0.35, 2.089) and (0.35, 1.713) .. (0.6, 1.337);
    \draw[->] (3.8, 2.465) .. controls (4.1, 1.713) and (4.1, 0.96) .. (3.8, 0.208);
    \node[anchor=center] at (0.1, 1.901) {$\pi_1$};
    \node[anchor=center] at (0.1, 0.701) {$\pi_2$};
    \node[anchor=center] at (4.4, 1.337) {$\pi_3$};
\end{tikzpicture}
    \caption{Sketch describing the Markov triples, functions and Riemannian submersions considered in this subsection, where $f$, $F$ and $\tilde f$ are as defined in \cref{eq:traceratioUn,eq:traceratiofunction,eq:traceratiograssmann}, respectively.}
    \label{fig:figuresettingex1}
\end{figure}

The function $\tilde f$ has been studied in depth\footnote{Although their work focuses on real Grassmann manifolds, their results automatically apply to the complex case after replacing symmetric matrices by Hermitian matrices, and the transpose with the conjugate transpose.} in \cite{shen2010tracequotient}. Two of the main results proven in the aforementioned reference concern the study of local minima and global minima of $\tilde f$. In particular, they prove that $\tilde f$ has no spurious local minima, and that under generic assumptions, its global minimum is unique and the Hessian of $\tilde f$ at the global minimum is non-degenerate.

\begin{proposition}[{\cite[Theorem 2]{shen2010tracequotient}}]
\label{thm:nolocalminimatraceratio}
Every local minimum of $\tilde f$ is a global minimum. 
\end{proposition}

\begin{proposition}[{\cite[Lemma 1]{shen2010tracequotient}}]
Let $P$ be a global minimum of $\tilde f$. Assume that there is a gap between the $k$-th and the $k+1$-th smallest eigenvalue of $\Phi(P) := A - \tilde{f}(P) B$. Then  the global minimum is unique and its Hessian is non-degenerate at $P$. 
\end{proposition}

Let us now characterize the critical points of $\tilde f$. We want to ensure that they are isolated, separated by a \textit{sufficiently large} distance, and that there always exists an escape direction for every saddle point. This will be done in the following subsection. 

\subsubsection{Critical points of \texorpdfstring{$\tilde f$}{f}}
\label{sec:criticalpointstraceratio}

Although it is difficult to characterize the critical points of $\tilde f$ in the general case, we will consider two settings in which the critical points of $\tilde f$ can be described explicitly, namely when $\tilde f$ is defined on $\textup{Gr}_1(\bb{C}^n)$ and when $B = \mathds{1}$. We believe that similar results can be derived in more general scenarios. 

First, note that in \cite[Corollary 1]{shen2010tracequotient} it is proven that the critical points of $\tilde f$ correspond to those $P \in \textup{Gr}_k(\bb{C}^n)$ for which there exists a common orthonormal eigenbasis for $\Phi(P) = A - \tilde{f}(P) B$ and $P$. This way, the question of finding the set of critical points of $\tilde f$ reduces to characterizing the set 
\[
\bigg\{P \in \textup{Gr}_k(\bb{C}^n) : \bigg[P, A - \frac{\Tr(PA)}{\Tr(PB)}B\bigg] = 0\bigg\}.
\]
To do so, we will make use of the following simultaneous diagonalization result.  
\begin{proposition}[{\cite[Theorem 7.6.4]{horn2012matrix}}]
\label{thm:simultdiag}
Let $A, B$ be two Hermitian matrices. If $B$ is positive definite, then there exists a non-singular matrix $S$ such that $B = SS^\dagger$ and $A = S \Lambda S^{\dagger}$, where $\Lambda$ is real-diagonal. Furthermore, the diagonal entries of $\Lambda$ correspond to the eigenvalues of $B^{-1}A$. 
\end{proposition}

Lastly, to study the escape directions of $\tilde f$ at its saddle points, we will analyze its Hessian. It is shown in \cite{shen2010tracequotient} that the Hessian of $\tilde f$ at a critical point $P$ in the direction of $X = \begin{pmatrix}
X_{11} & X_{12}\\
X_{12}^\dagger & X_{22}
\end{pmatrix} \in T_P\textup{Gr}_k(\bb{C}^n)$---where $X_{11}, X_{12}, X_{22}$ are matrices of size $k \times k$, $k \times (n-k)$ and $(n-k) \times (n-k)$, respectively---is given by
\[
\nabla^2 \tilde f(P)[X, X] = \frac{2}{\Tr(BP)} \sum_{i = 1}^k \sum_{j = 1}^{n-k} |z_{ij}|^2 (\sigma_j - \lambda_i),
\]
where $(z_{ij})$ are the coordinates of $X_{12}$, $\{\sigma_j\}_{j = 1}^{n-k}$ correspond to the eigenvalues of $\Phi(P)$ associated with the zero eigenvectors of $P$, and $\{\lambda_i\}_{i = 1}^{k}$ correspond to the eigenvalues of $\Phi(P)$ associated with the non-zero eigenvectors of $P$.

\paragraph{First case: rank-one \texorpdfstring{$P$}{P}.}\mbox{}\\
Let us first consider the case when $\tilde f$ is defined on $\textup{Gr}_1(\bb{C}^n)$. Here, we can rewrite \cref{eq:traceratiofunction} as
\begin{align}
\label{eq:traceratioGr1}
\begin{split}
\tilde f: \textup{Gr}_1(\bb{C}^n) \to \bb{R},\quad \ketbra{\varphi} \mapsto \frac{\Tr(\ketbra{\varphi} A)}{\Tr(\ketbra{\varphi} B)},
\end{split}
\end{align}
where $\ket{\varphi}$ is some unit vector in $\bb{C}^n$, and $\bra{\varphi} := \ket{\varphi}^\dagger$. Moreover, the set of critical points of $\tilde f$ can be easily described. 

\begin{proposition}
\label{prop:criticalpointsGr1}
Let $\tilde f$ be as defined in \cref{eq:traceratioGr1}, and assume that $B^{-1}A$ has $n$ distinct eigenvalues. Then  $\tilde f$ has exactly $n$ critical points.
\end{proposition}
\begin{proof}
As we saw earlier, the critical points of $\tilde f$ correspond to the set of rank-one projectors $\ketbra{\varphi}$ that satisfy
\[
[\ketbra{\varphi}, \Phi(\ketbra{\varphi})] = 0.
\]
This condition can be rewritten as
\[
\Phi(\ketbra{\varphi}) \ketbra{\varphi} = \ketbra{\varphi} \Phi(\ketbra{\varphi}).
\]
Now, since $\braket{\varphi} = 1$, it holds that 
\begin{align*}
\Phi(\ketbra{\varphi}) \ket{\varphi} &=\ketbra{\varphi} \Phi(\ketbra{\varphi}) \ket{\varphi}\\
&= \bra{\varphi}(A - \tilde f(\ketbra{\varphi})B)\ket{\varphi}\ket{\varphi}
\\&= (\bra{\varphi}A\ket{\varphi} - \tilde f(\ketbra{\varphi})\bra{\varphi}B\ket{\varphi})\ket{\varphi}
\\&= 0,
\end{align*}
and so 
\begin{equation}
\label{eq:psieig}
A\ket{\varphi} = \tilde f(\ketbra{\varphi}) B \ket{\varphi}.    
\end{equation}

Using \cref{thm:simultdiag} we know that there exists some non-singular matrix $S$ such that $A = S\Lambda S^\dagger$ and $B = SS^\dagger$, where $\Lambda$ is a diagonal matrix with diagonal entries equal to the eigenvalues of $B^{-1}A$. Therefore, \cref{eq:psieig} can be rewritten as
\[
S\Lambda S^\dagger \ket{\varphi} = \tilde f(\ketbra{\varphi}) S S^\dagger \ket{\varphi},
\]
and so 
\[
\Lambda S^\dagger \ket{\varphi} = \tilde f(\ketbra{\varphi}) S^\dagger \ket{\varphi},
\]
which implies that $S^\dagger \ket{\varphi}$ is an eigenvector of $\Lambda$. Since we are assuming that $\Lambda$ has $n$ distinct eigenvalues, then, up to a complex phase, there are only $n$ possible choices for $S^\dagger \ket{\varphi}$, which in turn implies that $\tilde f$ has $n$ critical points. 
\end{proof}

We have proven that the set of critical points of $\tilde f$ corresponds to
\[
\{\ketbra{\varphi} \in \textup{Gr}_1(\bb{C}^n): S^\dagger \ket{\varphi} \textup{ is an eigenvector of } \Lambda\}, 
\]
where $S$ and $\Lambda$ are as defined in the previous proof. Now, using the expression for the distance in the Grassmann manifold \cite[Eq. 5.3]{bendokat2024grassmann} we know that the minimum distance between any two critical points of $\tilde f$ is given by
\[
D := \min\Big\{\sqrt{2}\arccos(|\bra{\varphi_1}\ket{\varphi_2}|) \in [0, \frac{\pi}{\sqrt{2}}] : \ketbra{\varphi_1} \neq \ketbra{\varphi_2} \textup{ critical points of } \tilde f\Big\}.
\]

To conclude, we analyze the Hessian of $\tilde f$ at the critical points. As we mentioned earlier, at any critical point $\ketbra{\varphi}$ of $\tilde f$, its Hessian in the direction of $X \in T_P \textup{Gr}_1(\bb{C}^n)$ can be written as 
\[
\nabla^2 \tilde f(\ketbra{\varphi})[X, X] = \frac{2}{\Tr(B\ketbra{\varphi})} \sum_{j = 1}^{n-1} |z_{1j}|^2 (\sigma_j - \lambda_1).
\]
Moreover, recall that $\lambda_1$ is the eigenvalue of $\Phi(\ketbra{\varphi})$ associated with $\ket{\varphi}$, which we showed in the proof of \cref{prop:criticalpointsGr1} to be zero. Thus, we can rewrite this expression as 
\[
\nabla^2 \tilde f(\ketbra{\varphi})[X, X] = \frac{2}{\Tr(B\ketbra{\varphi})} \sum_{j = 1}^{n-1} |z_{1j}|^2 \sigma_j,
\]
where the values $\sigma_j$ correspond to the eigenvalues of 
\[
A - \tilde f(\ketbra{\varphi}) B
\]
not associated with $\ket{\varphi}$. Therefore, since $h_2(X,X) = 2\sum_{j}|z_{1j}|^2$, denoting by $\tilde{\mathcal{S}}$ the set of saddle points of $\tilde f$, we know that for every $P \in \tilde{\mathcal{S}}$,
\[
\lambda_{\min}(\nabla^2\tilde f(P)) \leq \max_{P \in \tilde{\mathcal{S}}} \min \Big\{\frac{\sigma_i(P)}{\Tr(BP)} : i \in \{1, \dotsc, n-1\}\Big\}, 
\]
where $\sigma_i(P)$ denotes the $i$-th eigenvalue $\sigma_i$ at the critical point $P$, and at the global minimum $P^*$
\[
\lambda_{\min}(\nabla^2\tilde f(P^*)) = \min \Big\{\frac{\sigma_i(P^*)}{\Tr(BP^*)} : i \in \{1, \dotsc, n-1\}\Big\}, 
\]
so we can choose 
\[
\lambda_* = \min\Big\{1, -\max_{P \in \tilde{\mathcal{S}}} \min \Big\{\frac{\sigma_i(P)}{\Tr(BP)} : i \in \{1, \dotsc, n-1\}\Big\}, \min \Big\{\frac{\sigma_i(P^*)}{\Tr(BP^*)} : i \in \{1, \dotsc, n-1\}\Big\}\Big\}.
\]
\paragraph{Second case: \texorpdfstring{$B = \mathds{1}$}{B = 1}.}\mbox{}\\
We can rewrite \cref{eq:traceratiofunction} as 
\begin{align}
\label{eq:traceratiograssmannBid}
\begin{split}
\tilde f: \textup{Gr}_k(\bb{C}^n) \to \bb{R},\quad P \mapsto \frac{\Tr(P A)}{k}.
\end{split}
\end{align}
In this case, assuming that $A$ has a simple spectrum, any two critical points of $\tilde f$ lie in each other's cut locus, which in particular implies that
\[
d(P, Q) \geq i(\textup{Gr}_k(\bb{C}^n)),
\]
for any two critical points $P, Q$ of $\tilde f$ (cf. \cref{secinjectivity}). 

Describing the cut locus of a point in a manifold is not an easy task in general. Nevertheless, for the case of the Grassmann manifold with the induced metric, the cut locus of a point $P = U U^\dagger$ can be described easily \cite[Section 5.1]{bendokat2024grassmann}. Indeed, for any $P = UU^\dagger \in \textup{Gr}_k(\bb{C}^n)$, its cut locus $C(P)$ is given by
\begin{equation}
\label{eq:cutlocusdescription}
C(P) = \{ Q = YY^\dagger \in \textup{Gr}_k(\bb{C}^n) : \rank(U^\dagger Y) < k\}.
\end{equation}
Therefore, the cut locus of a projector $P$ in $\textup{Gr}_k(\bb{C}^n)$ corresponds to those projectors whose associated subspace contains at least one orthogonal direction to the subspace on which $P$ projects. 

\begin{proposition}
Let $\tilde f$ be defined as in \cref{eq:traceratiograssmannBid}, and assume that $A$ has $n$ distinct eigenvalues. Then  $\tilde f$ has $\binom{n}{k}$ distinct critical points, which are separated by a distance of at least $\frac{\sqrt{2}\pi}{2}$. 
\end{proposition}
\begin{proof}
When $B = \mathds{1}$, any critical point $P$ of $\tilde f$ satisfies
\[
[P, A - \tilde f(P)] = 0,
\]
which is equivalent to the condition that
\[
[P, A] = 0.
\]
This way, for every critical point $P$ of $\tilde f$, we know that $P$ and $A$ can be diagonalized in the same basis. Assuming again that $A$ has $n$ distinct eigenvalues with associated eigenvectors $\{\ket{\varphi_i}\}_{i = 1}^n$, it must be the case that we can write $P$ as 
\[
P = \sum_{l = 1}^{k} \ketbra{\varphi_{i_l}},
\]
where $\{i_1 \dotsc, i_k\} \subset \{1, \dotsc, n\}$. Thus, there are $\binom{n}{k}$ choices of $P$. 

Using the above description of the critical points, we know that every two critical points will differ in at least one orthogonal direction. Thus, by \cref{eq:cutlocusdescription}, any two critical points $P$ and $Q$ will lie in each other's cut locus. Therefore, their distance is lower bounded by the injectivity radius of $\textup{Gr}_k(\bb{C}^n)$ which in turn is at least $\frac{\sqrt{2}\pi}{2}$ (cf. \cref{cor:injectivityradii}). 
\end{proof}

Let us now analyze the Hessian of $\tilde f$ (cf. \cref{eq:traceratiograssmannBid}) at any critical point $P$. Recall from the previous result that we can write $A$ and $P$ as 
\[
A = \sum_{i = 1}^k \lambda_i \ketbra{\varphi_i} + \sum_{i = k+1}^{n} \sigma_{i-k} \ketbra{\varphi_i},\quad \textup{and}\quad P = \sum_{i = 1}^k \ketbra{\varphi_i},
\]
where $\{\ket{\varphi_i}\}_{i = 1}^n$ is a common eigenbasis for both matrices, and $\{\lambda_i\}_{i = 1}^k$, $\{\sigma_i\}_{i=1}^{n-k}$ are the associated eigenvalues for $A$, which are assumed to be distinct. 

Furthermore, from the expression of the Hessian of $\tilde f$ at some critical point $P$, and using the fact that $\Phi(P) = A - \tilde f(P)$, we know that
\[
\nabla^2 \tilde f(P)[X, X] = \frac{2}{\Tr(BP)} \sum_{i = 1}^k \sum_{j = 1}^{n-k} |z_{ij}|^2 (\sigma_j - \lambda_i).
\]

This way, denoting by $\tilde{\mathcal{S}}$ the set of saddle points of $\tilde f$, and denoting $E_1 < E_2 < \dotsc < E_n$ the distinct eigenvalues of $A$, we can use again the fact that $h_2(X,X) = 2\sum_{ij}|z_{ij}|^2$ to conclude that for every $P \in \tilde{\mathcal{S}}$, 
\[
\lambda_{\min}(\nabla^2\tilde f(P)) \leq - \frac{E_{k+1} - E_k}{k},
\]
and at the global minimum $P^*$, 
\[
\lambda_{\min}(\nabla^2\tilde f(P^*)) = \frac{E_{k+1} - E_k}{k},
\]
so we can choose
\[
\lambda_* = \min\Big\{1, \frac{E_{k+1} - E_k}{k}\Big\}
\]

\subsubsection{Rapid mixing and suboptimality for the Gibbs distributions}

Lastly, let us apply \cref{thm:MainFormal1,thm:MainFormal2,cor:rapidmixing} to the trace ratio problem. Recall that we are interested in the functions $F, f$ and $\tilde f$ defined on $(\textup{U}(n), g)$, $(\textup{V}_k(\bb{C}^n), h_1)$, $(\textup{Gr}_k(\bb{C}^n), h_2)$ as presented in \cref{eq:traceratiofunction,eq:traceratioUn,eq:traceratiograssmann}. 

The following result allows us to guarantee that, for sufficiently large values of $\beta$, the associated Gibbs distribution to $F$, $f$, and $\tilde f$ \textit{find} the global minima. 

\begin{proposition}
\label{prop:mainthmtrace1}
Let $F$, $f$ and $\tilde f$ be defined as in \cref{eq:traceratiofunction,eq:traceratioUn,eq:traceratiograssmann} and let $A_2 \geq 1$ be the Lipschitz constant of the gradient of $F$, $f$ and $\tilde f$.  Let $\operatorname{L}_F$, $\operatorname{L}_f$ and $\operatorname{L}_{\tilde f}$ be the operators defined as in \cref{eq:diferentesoperators}. For every $\varepsilon \in (0, 1]$ and $\delta \in (0, 1)$, if $\beta$ is such that 
\[
\beta \geq \frac{2}{\varepsilon}\bigg(\frac{1}{2} + 6n^2\log n + \frac{n^3(n+1)}{2}\log(2\pi) + n^2 \log \frac{A_2 \sqrt{2\pi}}{\delta\varepsilon}\bigg),
\]
then the Gibbs distribution $\nu_F(x) = \frac{1}{Z_F} e^{-\beta F(x)}$ on $\textup{U}(n)$ is such that
\[
\nu_F\left(F - \min_{y \in \textup{U}(n)} F(y) \geq \varepsilon\right) \leq \delta.
\]
If we define 
\[
\varepsilon^M_{\max} := \min \Big\{\frac{\pi^2 A_2}{16}, 1\Big\},
\]
then, for every $\varepsilon \in (0, \varepsilon^M_{\max}]$ and $\delta \in (0, 1)$, if $\beta$ is such that 
\[
\beta \geq \frac{2}{\varepsilon}\bigg(\frac{1}{2} + k(2n-k)\bigg(3\log(k(2n-k))  +\log \frac{A_2 \sqrt{2\pi}}{\delta\varepsilon}\bigg) + k^2(2n-k)\frac{2n-k+1}{2}\log(2\pi) \bigg),
\]
the Gibbs distribution $\nu_f(x) = \frac{1}{Z_f} e^{-\beta f(x)}$ on $\textup{V}_k(\bb{C}^n)$ is such that
\[
\nu_f\left(f - \min_{y \in \textup{V}_k(\bb{C}^n)} f(y) \geq \varepsilon\right) \leq \delta.
\]
and if $\beta$ is such that
\[
\beta \geq \frac{2}{\varepsilon}\bigg(\frac{1}{2} + 2k(n-k)\bigg(3\log (2k(n-k)) + \log \frac{A_2 \sqrt{2\pi}}{\delta\varepsilon} + k(n-k)\log(2\pi)\bigg) \bigg),
\]
the Gibbs distribution $\nu_{\tilde f}(x) = \frac{1}{Z_{\tilde f}} e^{-\beta \tilde f(x)}$ on $\textup{Gr}_k(\bb{C}^n)$ is such that
\[
\nu_{\tilde f}\left(\tilde f - \min_{y \in \textup{Gr}_k(\bb{C}^n)} \tilde f(y) \geq \varepsilon\right) \leq \delta.
\]
\end{proposition}
\begin{proof}
In order to prove this result, it suffices to bound the injectivity radius and the volume of $\textup{U}(n)$, $\textup{V}_k(\bb{C}^n)$ and $\textup{Gr}_k(\bb{C}^n)$, and apply \cref{thm:MainFormal2}. The bounds obtained in \cref{cor:injectivityradii} for the injectivity radii of these manifolds are
\[
i(\textup{U}(n)) \geq \pi, \quad i(\textup{V}_k(\bb{C}^n)) \geq \frac{\sqrt{2}\pi}{2}, \quad \textup{and}\quad i(\textup{Gr}_k(\bb{C}^n)) \geq \frac{\sqrt{2}\pi}{2}, 
\]
and the bounds on their volume can be deduced from the expressions shown in \cref{dimensionandvolume}, i.e. 
\[
\textup{Vol}(\textup{U}(n)) \leq (2\pi)^{n(n+1)/2}, \quad \textup{Vol}(\textup{V}_k(\bb{C}^n)) \leq (2\pi)^{nk - \frac{k(k-1)}{2}},\quad \textup{Vol}(\textup{Gr}_k(\bb{C}^n)) \leq (2\pi)^{nk -k^2}.
\]
Lastly, the dimensions of $\textup{U}(n)$, $\textup{V}_k(\bb{C}^n)$ and $\textup{Gr}_k(\bb{C}^n)$ can be substituted by the explicit expressions found in \cref{dimensionandvolume}, i.e. 
\[
\dim(\textup{U}(n)) = n^2,\quad \dim(\textup{V}_k(\bb{C}^n)) =2nk-k^2, \quad \dim(\textup{Gr}_k(\bb{C}^n)) = 2k(n-k).
\]
\end{proof}

Next, let us first restate some of the assumptions needed in order for \cref{thm:MainFormal1} to hold.  
\begin{assumption}
\label{ass:sec8.1assumption1}
The critical points of $\tilde f$ are isolated points separated by some distance $D > 0$, i.e. 
\[
d_{h_2}(x_1, x_2) \geq D, \quad \forall x_1, x_2 \in \tilde{\mathcal{C}},\ x_1 \neq x_2,
\]
where $\tilde{\mathcal{C}}$ denotes the set of critical points of $\tilde f$. 
\end{assumption}

\begin{assumption}
\label{ass:sec8.1assumption2.1}
For every saddle point $y \in \tilde{\mathcal{S}}$, there exists an escape direction. That is, there exists some constant $\lambda_* \in (0, 1]$ such that for every saddle point $y \in \tilde{\mathcal{S}}$ it holds that 
\[
\lambda_{\min}\big(\nabla^2 \tilde f(y)\big) \leq -\lambda_* < 0.
\]
\end{assumption}

\begin{assumption}
\label{ass:sec8.1assumption2.2}
The Hessian of $\tilde f$ at the global minimum is non-degenerate. More precisely, at the unique minimum $x^*$ of $\tilde f$, we have
\[
\lambda_{\min}\big(\nabla^2 \tilde{f}(x^*)\big) \geq \lambda_*.
\]
where $\lambda_*$ is the same constant as in \cref{ass:sec8.1assumption2.1}. 
\end{assumption}

\begin{assumption}
\label{ass:sec8.1assumption3}
There exists a constant $0 < C_{\tilde f} \leq 1$ such that
\begin{equation*}
|\Grad{h_2}\tilde f(x)|_{h_2} \geq C_{\tilde f} d_{h_2}(x, \tilde{\mathcal{C}}),    
\end{equation*}
for every $x \in \textup{Gr}_k(\bb{C}^n)$. 
\end{assumption}

\begin{proposition}
\label{prop:mainthmtrace2}
Let $F$, $f$ and $\tilde f$ be defined as in \cref{eq:traceratiofunction,eq:traceratioUn,eq:traceratiograssmann}. Assume that there is a gap between the $k$-th and the $k+1$-th smallest eigenvalue of 
\[
\Phi(P^*)= A - \tilde f(P^*)B,
\]
where $P^*$ is the global minimum of $\tilde f$. Furthermore, assume that $\tilde f$ satisfies \crefrange{ass:sec8.1assumption1}{ass:sec8.1assumption3}---where in particular, $\lambda_*$ depends on the gap of $\Phi(P^*)$. Let $\operatorname{L}_F$, $\operatorname{L}_f$ and $\operatorname{L}_{\tilde f}$ be the operators defined as in \cref{eq:diferentesoperators}, and let $a, \beta > 0$ be such that
\begin{align*}
a^2 &\geq \max\Big\{\frac{48 A_2 k(n-k)}{C^2_{\tilde f}}, \frac{288}{\lambda_*}\Big\},
\\\beta &\geq \max\Bigg\{\frac{384^2 2^5 k^5 (n-k)^5 A^2_2 A^2_3 a^6}{\lambda_*^2},  \frac{9a^2}{D^2}, \frac{8a^2}{\pi^2}\Bigg\},
\end{align*}
where $A_2 \geq 1$ is the Lipschitz constant of the gradient of $F$, $f$ and $\tilde f$, and $A_3 \geq 1$ is the Lipschitz constant of the Hessian of $\tilde f$. Then  the Markov triple $(\textup{U}(n), \nu_F, \Gamma_g)$ associated with the operator $\operatorname{L}_F$ satisfies an $\textup{LSI}(\alpha)$ verifying 
\[
\frac{1}{\alpha} = 4\beta A_2 \pi^2 n (1 + \sqrt{2})^2\max \Big\{\frac{120}{\lambda_*}, n (1 + \sqrt{2})^2 \beta\Big\},
\]
the Markov triple $(\textup{V}_k(\bb{C}^n), \nu_{f}, \Gamma_{h_1})$ associated with the operator $\operatorname{L}_f$ satisfies an $\textup{LSI}(\hat \alpha)$ verifying 
\[
\frac{1}{\hat \alpha} = 4\beta A_2 \pi^2 n (1 + \sqrt{2})^2\max \Big\{\frac{120}{\lambda_*}, k (1 + \sqrt{2})^2 \beta\Big\},
\] 
and the Markov triple $(\textup{Gr}_k(\bb{C}^n), \nu_{\tilde f}, \Gamma_{h_2})$ associated with the operator $\operatorname{L}_{\tilde f}$ satisfies an $\textup{LSI}(\tilde \alpha)$ verifying
\[
\frac{1}{\tilde \alpha} = \frac{480\beta A_2  \pi^2 n (1 + \sqrt{2})^2}{\lambda_*}.
\]
\end{proposition}
\begin{proof}
We begin by applying \cref{thm:MainFormal1} to the Markov triples $(\textup{U}(n), \nu_F, \Gamma_g)$ and \linebreak $(\textup{Gr}_k(\bb{C}^n), \nu_{\tilde f}, \Gamma_{h_2})$. For every $X, Y, Z \in \mathfrak{u}(n)$ it is known (cf. \cite[Proposition 31, Section 11]{OneilSemi}) that $R_{h_2}({d\pi_3} X, {d\pi_3} Y){d\pi_3} Z = d\pi_3([[X,Y], Z])$. In particular, for every $x \in \textup{U}(n)$, choosing normal coordinates centered at $\pi_3(x)$ and unit vectors, it follows from \cref{ineqliebracket} that $|({R_{h_2}})_{ijkl}\pi_3(x)|\leq 2$ and so we can take $\mathbf{K} = 2$ in \cref{thm:MainFormal1}. Furthermore, we also showed in \cref{LemaCurvaturaUn,prop:curvGrassmann} that the Ricci curvatures of the unitary group and the Grassmann manifold are non-negative, and so the constants $R_{\textup{U}(n)}$ and $R_{\textup{Gr}_k(\bb{C}^n)}$ from \cref{thm:MainFormal1} can be taken to be $0$.

Also, the Riemannian submersion $\pi_3 : (\textup{U}(n), g) \to (\textup{Gr}_k(\bb{C}^n), h_2)$ has totally geodesic fibers, which implies by \cref{CurvaturaSubmersion,vanishingT} that the sectional curvature of its fibers coincides with that of the total space.

The result now follows by substituting the manifold-dependent constants in the statement of \cref{thm:MainFormal1} with the known bounds derived in \cref{SectionExamples}. We use the lower bounds for the injectivity radii of $(\textup{U}(n), g)$ and $(\textup{Gr}_k(\bb{C}^n), h_2)$, along with the lower bound on the convexity radius of $(\textup{Gr}_k(\bb{C}^n), h_2)$ obtained in \cref{cor:injectivityradii}, i.e.
\[
\mathit{conv}(\textup{Gr}_k(\bb{C}^n)) \geq \frac{\sqrt{2}\pi}{4}.
\]
We also use the upper bound on the diameter of $(\textup{U}(n), g)$, and $(\textup{Gr}_k(\bb{C}^n), h_2)$ found in \cref{diamUn,eq:diamaterboundstiefelgrassmann}, i.e.
\[
\textup{diam}(\textup{Gr}_k(\bb{C}^n)) \leq \textup{diam}(\textup{U}(n)) \leq \pi \sqrt{n}(1 + \sqrt{2}),
\]
and substitute the dimensions of $\textup{U}(n)$ and $\textup{Gr}_k(\bb{C}^n)$. 

Lastly, we obtain a bound on the diameter of the fibers $\textup{U}(k) \times \textup{U}(n-k)$ with the induced metric, which can be derived via \cref{diameterProduct} and \cref{diamUn}, i.e. 
\[
\textup{diam}(\textup{U}(k) \times \textup{U}(n-k)) \leq \sqrt{\pi^2 k (1+ \sqrt{2})^2 + \pi^2 (n-k) (1+\sqrt{2})^2} = \pi (1 + \sqrt{2}) \sqrt{n}.
\]

All of the above bounds allow us to obtain the log-Sobolev inequality constants for the Markov triples $(\textup{U}(n), \nu_F, \Gamma_g)$ and $(\textup{Gr}_k(\bb{C}^n), \nu_{\tilde f}, \Gamma_{h_2})$. 

To obtain the log-Sobolev inequality of the triple $(\textup{V}_k(\bb{C}^n), \nu_f, \Gamma_{h_1})$, we proceed as in the proof of \cref{thm:MainFormal1}: we first apply \cref{prop9.12} to $(\textup{Gr}_k(\bb{C}^n), \nu_{\tilde f}, \Gamma_{h_2})$ to find that it satisfies a $\textup{PI}(\kappa)$ with
\[
\kappa = \frac{\lambda_*}{120}. 
\]
Then  using the lifting theorem for Poincaré inequalities shown in \cref{thmliftPI}---as the fibers $U(k)$ have non-negative induced Ricci curvature---we lift the Poincaré inequality via $\pi_2$ to ensure that $(\textup{V}_k(\bb{C}^n), \nu_f, \Gamma_{h_1})$ satisfies a $\textup{PI}(\tilde \kappa)$ where
\[
\frac{1}{\tilde \kappa} = \max\Big\{\frac{120}{\lambda_*}, (1+\sqrt{2})^2  k \beta\Big\},
\]
which can be then tightened using \cref{prop9.15} to obtain an $\textup{LSI}(\hat{\alpha})$, where
\[
\frac{1}{\hat{\alpha}} = \frac{4\beta A_2 \textup{diam}(\textup{V}_k(\bb{C}^n))^2}{\tilde \kappa}.
\]
Lastly, using the bound on the diameter of the Stiefel manifold shown in \cref{eq:diamaterboundstiefelgrassmann} we see that 
\[
\textup{diam}(\textup{V}_k(\bb{C}^n)) \leq \pi \sqrt{n}(1 + \sqrt{2}),
\]
and so the result follows. 
\end{proof}

Recall that in \cref{sec:criticalpointstraceratio} we proved that the critical points of $\tilde f$ are isolated when it is defined on $\textup{Gr}_1(\bb{C}^n)$ or when $B = \mathds{1}$, assuming a simple spectrum of $B^{-1}A$. Furthermore, we also gave an expression for the minimum distance between any two critical points, and the constant $\lambda_*$ from \cref{ass:sec8.1assumption2.1}.

\begin{remark}
In the case when $\tilde f$ is defined on $\textup{Gr}_1(\bb{C}^n)$ or $B = \mathds{1}$, \cref{ass:sec8.1assumption1} holds, i.e. the critical points of $\tilde f$ are isolated, and the constant $D$---which lower bounds the distance between any two critical points---is given by 
\[
D = \frac{\sqrt{2}\pi}{2},
\]
when $B = \mathds{1}$---which in particular does not depend on $k$ or $n$---and by
\[
D = \min\Big\{\sqrt{2}\arccos(|\bra{\varphi_1}\ket{\varphi_2}|) \in [0, \frac{\pi}{\sqrt{2}}] : \ketbra{\varphi_1} \neq \ketbra{\varphi_2} \textup{ critical points of } \tilde f\Big\},
\]
when $\tilde f$ is defined on $\textup{Gr}_1(\bb{C}^n)$.

Moreover, whenever the spectrum of $B^{-1}A$ is simple, \cref{ass:sec8.1assumption2.1,ass:sec8.1assumption2.2} hold, and the constant $\lambda_*$ can be written when $B = \mathds{1}$ as 
\[
\lambda_* := \min\Big\{1, \frac{E_{k+1} - E_k}{k}\Big\},
\]
where $E_{k+1}$ and $E_k$ are the $k+1$-th and $k$-th smallest eigenvalues of $A$, respectively, and as 
\[
\lambda_* := \min\Big\{1, -\max_{P \in \tilde{\mathcal{S}}} \Big\{\frac{\min_{i \in \{1, \dotsc, n-1\}} \sigma_i(P)}{\Tr(BP)} \Big\}, \frac{\min_{i \in \{1, \dotsc, n-1\}} \sigma_i(P^*)}{\Tr(BP^*)}\Big\}.
\]
when $\tilde f$ is defined on $\textup{Gr}_1(\bb{C}^n)$, where $P^*$ denotes the minimum of $\tilde f$. 
\end{remark}

Lastly, let us restate \cref{cor:rapidmixing}.

\begin{corollary}
Let $X_t$, $\hat X_t$ and $\tilde X_t$ be the Langevin diffusion processes generated by $\operatorname{L}_F$, $\operatorname{L}_f$, and $\operatorname{L}_{\tilde f}$, with initial uniform distribution on $\textup{U}(n)$, $\textup{V}_k(\bb{C}^n)$ and $\textup{Gr}_k(\bb{C}^n)$, respectively. Under the assumptions of \cref{prop:mainthmtrace2}, they converge exponentially fast to the Gibbs distributions $\nu_F$, $\nu_f$, and $\nu_{\tilde f}$, respectively, i.e. 
\begin{align*}
\norm{\nu_F - \rho_t}^2_{\textup{TV}} &\leq  \beta\, e^{-2\alpha t} (\max_{y \in \textup{U}(n)} F(y) - \min_{y \in \textup{U}(n)} F(y)), \\
\norm{\nu_f - \hat \rho_t}^2_{\textup{TV}} &\leq  \beta\, e^{-2\hat \alpha t} (\max_{y \in \textup{V}_k(\bb{C}^n)} f(y) - \min_{y \in \textup{V}_k(\bb{C}^n)} f(y))= \beta\, e^{-2\hat\alpha t} (\max_{y \in \textup{U}(n)} F(y) - \min_{y \in M} F(y)),\\
||\nu_{\tilde f} - \tilde \rho_t||^2_{\textup{TV}} &\leq  \beta\, e^{-2\tilde \alpha t} (\max_{y \in \textup{Gr}_k(\bb{C}^n)} \tilde f(y) - \min_{y \in \textup{Gr}_k(\bb{C}^n)} \tilde f(y))= \beta\, e^{-2\tilde \alpha t} (\max_{y \in \textup{U}(n)} F(y) - \min_{y \in \textup{U}(n)} F(y)),
\end{align*}
for every $t \geq 0$, where $\rho_t, \hat \rho_t$ and $\tilde{\rho}_t$ denote the distributions of $X_t, \hat X_t$ and $\tilde X_t$, and $\alpha, \hat \alpha, \tilde \alpha$ denote the log-Sobolev inequality constants obtained in \cref{prop:mainthmtrace2}.
\end{corollary}

\subsection{Two-dimensional Ising model}

The second function that we consider is defined as
\begin{align}
\label{eq:isingfunction}
\begin{split}
F : \textup{SU}(2)^{\times n^2} &\rightarrow \mathbb{R}\\
(U_1, \dotsc, U_{n^2}) &\mapsto \bra{0}^{\otimes n^2} \big(U_1^\dagger \otimes \dotsm \otimes U_{n^2}^\dagger \big)H \big(U_1 \otimes \dotsm \otimes U_{n^2}\big) \ket{0}^{\otimes n^2},
\end{split}
\end{align}
where $n \geq 2$ and $H$ corresponds to a ferromagnetic Ising model Hamiltonian, i.e.
\[
H = -\sum_{\langle I, J\rangle \in \Lambda} \sigma_z^I \sigma_z^J - \sum_{I\in V} h_I \sigma_z^I,
\]
over a square lattice $(V, \Lambda)$ with nearest-neighbor interactionsa and open boundary conditions. $h_{(i, j)}\in \bb{R}$ are defined as
\begin{equation}
\label{eq:definitionmagneticfield}
h_{(i, j)} = \begin{cases}
4 + \varepsilon_{(i, j)} &\quad \text{if } 1 < i < n \text{ and } 1 < j < n,\\
3 + \varepsilon_{(i, j)} &\quad \text{if } 1 < i < n\text{ xor } 1 < j < n,\\
2 + \varepsilon_{(i, j)} &\quad \text{otherwise},
\end{cases}
\end{equation}
with $\varepsilon_{(i, j)} > 0$ for every $(i, j) \in V$, and $\sigma_z$ denotes the $Z$-Pauli matrix, i.e. $\sigma_z = \begin{pmatrix}
1 & 0 \\
0 & -1
\end{pmatrix}$. We denote by $\ket{0}$ and $\ket{1}$ the eigenvectors of $\sigma_z$ with eigenvalue $1$ and $-1$, respectively. $F$ is commonly regarded as the expectation value of $H$ in a given state. Previous results for performing efficient Gibbs sampling for this model (as a model with discrete spins) are known \cite{jerrum1993polynomial}. 

As in the previous example, the function $F$ defined in \cref{eq:isingfunction} has a clear symmetry: let $i \in \{1, \dotsc, n^2\}$ and let $U_i, \tilde U_i \in \textup{SU}(2)$ be such that 
\[
\tilde U_i\ket{0} = e^{i\theta} U_i\ket{0},
\]
for some $\theta \in [0, 2\pi)$. Then  
\[
F(U_1, \dotsc, U_i, \dotsc, U_{n^2}) = F(U_1, \dotsc, \tilde U_i, \dotsc, U_{n^2}).
\]

The symmetry of $F$ can be studied in terms of a group action on $\textup{SU}(2)$. Indeed, if we consider the action of $\textup{U}(1)$ on $\textup{SU}(2)$ defined as
\begin{align}
\label{eq:actiononSU2}
\begin{split}
\textup{SU}(2) \times \textup{U}(1) &\to \textup{SU}(2)\\
(X, e^{i\theta}) &\mapsto X\begin{pmatrix}
e^{i\theta} & 0\\
0 & e^{-i\theta}
\end{pmatrix},
\end{split}
\end{align}
it is clear that $F$ is invariant under this action in any of its coordinates $(U_1, \dotsc, U_{n^2})$. 

Note that the action shown in \cref{eq:actiononSU2} is free: for every $U \in \textup{SU}(2)$ and every $e^{i\theta} \in \textup{U}(1)$, it holds that
\[
U\begin{pmatrix}
e^{i\theta} & 0\\
0 & e^{-i\theta}
\end{pmatrix} = U
\iff
\begin{pmatrix}
e^{i\theta} & 0\\
0 & e^{-i\theta}
\end{pmatrix} = \mathds{1}
\]
which implies that $e^{i\theta} = 1$. 

Moreover, the action is smooth and isometric whenever $\textup{SU}(2)$ is endowed with its bi-invariant metric $g_{\mathit{bi}}$. Therefore, by \cref{existenceRiemannianSubmersionMetrics}, there exists a Riemannian metric $h$ on $\textup{SU}(2)/\textup{U}(1)$ such that the projection $\pi: (\textup{SU}(2), g_{\mathit{bi}}) \to (\textup{SU}(2)/\textup{U}(1), h)$ is a Riemannian submersion, which furthermore has totally geodesic fibers (cf. \cref{rmk:totallygeodesic}). 

In fact, $\pi$ corresponds to a well-studied Riemannian submersion, known as the Hopf fibration between $\bb{S}^3$ and $\bb{S}^2$. To illustrate this, let us introduce the following auxiliary result.
\begin{proposition}
\label{prop:isomSUS}
Let $\textup{SU}(2)$ be endowed with the bi-invariant metric $g_{\mathit{bi}}$. Then 
\[
(\textup{SU}(2), g_{\mathit{bi}}) \simeq (\bb{S}^3, 2g_{\textit{round}}),
\]
where $g_{\textit{round}}$ is the standard round metric on the sphere $\bb{S}^3$.
\end{proposition}
\begin{proof}
First, note that
\[
\textup{SU}(2) = \left\{\begin{pmatrix}
a & -\overline{b} \\
b & \overline{a}
\end{pmatrix}\ :\ a, b \in \bb{C},\  |a|^2+|b|^2=1\right\} \simeq \bb{S}^3.
\]
Recall from \cref{prop:liealgebracident} that its Lie algebra $\mathfrak{su}(2)$ consists of $2 \times 2$ anti-Hermitian matrices with zero trace, 
\[
\mathfrak{su}(2) = T_{\mathds{1}} \textup{SU}(2) = \left\{ \begin{pmatrix}
ix & z\\
-\overline{z} & -ix
\end{pmatrix}\ :\ x \in \bb{R},\ z \in \bb{C}\right\}.
\]

Moreover, the bi-invariant metric of $\textup{SU}(2)$ at the identity is given by
\[
\Tr \Bigg(\begin{pmatrix}
-ix_2 & -z_2\\
\overline{z_2} & ix_2
\end{pmatrix}
\begin{pmatrix}
ix_1 & z_1\\
-\overline{z_1} & -ix_1
\end{pmatrix}\Bigg) = x_1x_2 + z_2 \overline{z_1} + z_1 \overline{z_2} + x_1 x_2 
\]
for any two tangent vectors $\begin{pmatrix}
ix_2 & z_2\\
-\overline{z_2} & -ix_2
\end{pmatrix}$, $\begin{pmatrix}
ix_1 & z_1\\
-\overline{z_1} & -ix_1
\end{pmatrix} \in T_{\mathds{1}}\textup{SU}(2)$, which corresponds to $2g_{\mathit{Eucl}}$ in $\bb{R}^3$.

Now, the metric $g_{\mathit{bi}}$ is extended from $T_\mathds{1} \textup{SU}(2)$ to $T\textup{SU}(2)$ as 
\[
g_{\mathit{bi}}(AG, BG) := \Tr(G^\dagger B^\dagger A G) = \Tr (B^\dagger A).
\]

Note that the map
\begin{align*}
\bb{R}^4 \rightarrow \bb{R}^4,\quad 
\begin{pmatrix}
x & y\\
-\overline{y} & \overline{x}
\end{pmatrix} \mapsto \begin{pmatrix}
x & y\\
-\overline{y} & \overline{x}
\end{pmatrix}G
\end{align*}
where $x, y \in \bb{C}$, is an isometry in $\bb{R}^4$, and so the extension of the scalar product coincides with twice the usual scalar product in the round sphere $\bb{S}^3$, i.e. $2g_{\mathit{round}}$. 
\end{proof}

With this result in mind, it is clear that the action defined in \cref{eq:actiononSU2} is equivalent to the action of $\textup{U}(1)$ on $\bb{S}^3$ given by 
\begin{align*}
\textup{U}(1) \times \bb{S}^3 \to \bb{S}^3,\quad (e^{i\theta}, x) \mapsto e^{i\theta}x.
\end{align*} 
This action induces a Riemannian submersion, known as the Hopf fibration between $(\bb{S}^3, \lambda g_{\mathit{round}})$ and $(\bb{S}^2, \frac{\lambda}{4}g_{\mathit{round}})$, for any $\lambda > 0$ \cite[Section 4.3, Chapter 3]{petersen2006riemannian}. Using \cref{prop:isomSUS} it holds that the Riemannian submersion with totally geodesic fibers
\[
\pi:(\textup{SU}(2), g_{\mathit{bi}}) \to (\textup{SU}(2)/\textup{U}(1), h),
\]
corresponds to the Hopf fibration 
\[
\pi: (\bb{S}^3, 2 g_{\mathit{round}}) \to (\bb{S}^2, \frac{1}{2}g_{\mathit{round}}).
\]

Let us now study the function $F$ on $\textup{SU}(2)^{\times n^2}$ and its \textit{projected version} $\tilde F$ on $(\bb{S}^2)^{\times n^2}$, which can be understood as the function $F$ defined on the product of Bloch spheres. 

\subsubsection{Critical points}

Let us characterize the critical points of $F$. To do so, we will first show that every eigenvector of $H$ is a critical point of $F$, and then prove that only the eigenvectors of $H$ can be critical points. 

Let us introduce an auxiliary result that will allow us to conclude that the eigenvectors of $H$ are critical points of $F$. 
\begin{lemma}
\label{LemmaMatrixTheory}
Let $A$ be an $n$-square complex matrix. Then
\[
\textup{Tr}(AX) = 0
\]
for every Hermitian matrix $X$ if and only if $A = 0$. 
\end{lemma}

With this result in mind, let us study the critical points of an auxiliary function $f$. 
\begin{proposition}
\label{CriticalPointsOneGate}
Let $\hat{H}$ be some Hermitian matrix and let $f$ be defined as 
\begin{align*}
f: &\ \textup{SU}(n) \rightarrow \bb{R}\\
&U \mapsto \bra{0}U^\dagger \hat{H} U \ket{0},
\end{align*}
where $\ket{0}$ denotes some fixed unit vector in $\bb{C}^{n}$. Then  $U\ket{0}$ is an eigenvector of $\hat{H}$ if and only if $U$ is a critical point of $f$. 
\end{proposition}

\begin{proof}
Let us give an expression for the differential of $f$ at any point $U$ in the direction of a tangent vector $iXU$, where $X$ is some traceless Hermitian matrix. To do so, we consider the bi-invariant metric of $\textup{SU}(n)$ and parametrize $f$ along a geodesic starting at $U$ in the direction of $iXU$, 
\[
\gamma(t) = e^{itX} U,
\]
yielding
\[
f(t) := f(\gamma(t)) = \bra{0} U^\dagger e^{-itX} \hat{H} e^{itX} U \ket{0}.
\]
Computing its derivative, we obtain
\begin{align*}
f'(t) &= -i \bra{0} U^\dagger e^{-itX}X\hat{H}e^{itX}U\ket{0} + i\bra{0} U^\dagger e^{-itX}\hat{H}Xe^{itX}U\ket{0} 
\\&= i\bra{0} U^\dagger e^{-itX}[\hat{H}, X]e^{itX}U\ket{0}.
\end{align*}
Using the notation $\ket{U} := U\ket{0}$ we can rewrite this expression as
\[
f'(t) = i\bra{U} e^{-itX}[\hat{H}, X]e^{itX}\ket{U},
\]
which allows us to conclude that
\[
df|_U(iXU) = f'(0) = i\bra{U}[\hat{H}, X]\ket{U}.
\]

Assume now that $U\ket{0}$ is an eigenvector of $\hat{H}$, i.e. $\hat{H}\ket{U} = E \ket{U}$ for some $E \in \bb{R}$. Thus, for every $iX \in \mathfrak{su}(n)$
\[
df|_U (iXU) = i\bra{U}[\hat{H}, X]\ket{U} = i\bra{U}\hat{H}X - X\hat{H}\ket{U} = iE \bra{U}X - X\ket{U} = 0. 
\]
To prove the converse, let $U \in \textup{SU}(n)$ and suppose that
\begin{equation}
\label{eq:eq1}
\bra{U}[\hat{H}, X]\ket{U} = 0,\quad \forall iX \in \mathfrak{su}(n).    
\end{equation}
We can write $\ket{U}$ in the (orthonormal) eigenbasis of $\hat{H}$; $\ket{U} = \sum_{i = 1}^n \alpha_i \ket{\varphi_i}$, where $\hat{H} \ket{\varphi_i} = E_i \ket{\varphi_i}$. With this expression, \cref{eq:eq1} becomes
\[
\left(\sum_{i = 1}^n \overline{\alpha_i} \bra{\varphi_i}\right)\hat{H}X\left(\sum_{i = 1}^n \alpha_i \ket{\varphi_i}\right) = \left(\sum_{i = 1}^n \overline{\alpha_i} \bra{\varphi_i}\right)X\hat{H}\left(\sum_{i = 1}^n \alpha_i \ket{\varphi_i}\right),
\]
which implies that
\[
\left(\sum_{i = 1}^n \overline{\alpha_i} E_i \bra{\varphi_i}\right)X\left(\sum_{i = 1}^n \alpha_i \ket{\varphi_i}\right) = \left(\sum_{i = 1}^n \overline{\alpha_i} \bra{\varphi_i}\right)X\left(\sum_{i = 1}^n \alpha_i E_i\ket{\varphi_i}\right),
\]
and so
\[
\sum_{i, j = 1}^n \overline{\alpha_i} \alpha_j E_i \bra{\varphi_i}X\ket{\varphi_j} = \sum_{i, j = 1}^n \overline{\alpha_i} \alpha_j E_j \bra{\varphi_i}X\ket{\varphi_j}.
\]

Grouping all the terms on the left-hand side we obtain
\[
\sum_{i, j= 1}^n \overline{\alpha_i} \alpha_j (E_i-E_j) \bra{\varphi_i}X\ket{\varphi_j} = 0.
\]
If we define the matrix $A = \left(\overline{\alpha_j}\alpha_i (E_j - E_i)\right)_{ij}$ we can rewrite the expression as 
\[
\sum_{i, j = 1}^n \overline{\alpha_i} \alpha_j (E_i-E_j) \bra{\varphi_i}X\ket{\varphi_j} = \sum_{i, j = 1}^n A_{ji} X_{ij} = \textup{Tr}(AX) = 0 
\]
for every $iX \in \mathfrak{su}(n)$.

In order to apply \cref{LemmaMatrixTheory}, take $Y$ to be any Hermitian matrix. Therefore, we can write $Y = \Lambda + X$ where $\Lambda$ is diagonal and $X$ is such that $iX \in \mathfrak{su}(n)$, thus 
\[
\textup{Tr}(AY) = \textup{Tr}(A\Lambda) + \textup{Tr}(AX) = \textup{Tr}(A\Lambda) = 0,
\]
where the last equality follows from the fact that $\Lambda$ is diagonal. Indeed
\[
\textup{Tr}(A\Lambda) = \sum_{i = 1}^n A_{ii}\Lambda_{ii} = \sum_{i = 1}^n |\alpha_i|^2 (E_i-E_i)\Lambda_{ii} = 0.
\]
Using \cref{LemmaMatrixTheory} we can conclude that
\[
\overline{\alpha_j}\alpha_i (E_j - E_i) = 0,\quad \forall i, j \in \{1, \dotsc, n\}.
\]
This implies that, if all of the eigenvalues are different, at most one $\alpha_i$ is non-zero, implying that $\ket{U}$ is an eigenvector of $\hat{H}$. 
In the case where some eigenvalues are equal, we can conclude that $\ket{U}$ has to be a linear combination of eigenvectors associated with the same eigenvalue, implying that $\ket{U}$ is too an eigenvector of $\hat{H}$. 
\end{proof}

In the next auxiliary result, we will show how to compute the gradient of a function defined on the product of two manifolds, when endowed with the product metric.
\begin{lemma}
\label{productrulegradient}
Let $(M_1, g_1)$ and $(M_2, g_2)$ be two Riemannian manifolds and let $M_1 \times M_2$ be the product manifold endowed with the product metric $g$. Let us define
\begin{align*}
\hat{f}: &\ M_1 \times M_2 \rightarrow \bb{R}\\
&\quad(x, y) \mapsto f_1(x)f_2(y),
\end{align*}
where $f_i: M_i \rightarrow \bb{R}$ are smooth functions. Then 
\[
\textup{grad}_{g}\, \hat{f}(x,y) = (f_2(y) \textup{grad}_{g_1}\, f_1(x),\, f_1(x) \textup{grad}_{g_2}\, f_2(y)), \quad \forall (x, y) \in M_1 \times M_2.
\]
\end{lemma}
\begin{proof}
Let $(x, y) \in M_1 \times M_2$, let $W \in T_{(x,y)} (M_1 \times M_2)$ and let $\gamma(t)$ be the geodesic such that $\gamma(0) = (x, y)$ and $\gamma'(0) = W$. Writing both $W$ and $\gamma(t)$ in coordinates, we know that $W = (U, V)$ with $U \in T_x M_1$ and $V \in T_y M_2$, and $\gamma(t) = (\gamma_1(t), \gamma_2(t))$ with $\gamma_1(0) = x$, $\gamma_1'(0) = U$, $\gamma_2(0) = y$,  and $\gamma_2'(0) = V$. 

Using the definition of the differential of $\hat{f}$ and the product rule in $\bb{R}$ we conclude that 
\begin{align*}
d\hat{f}|_{(x,y)}(U, V) = \left.\frac{d}{dt}\right|_{t = 0} \hat{f}(\gamma(t)) &= \left.\frac{d}{dt}\right|_{t = 0} \hat{f}(\gamma_1(t), \gamma_2(t))
\\&= \left.\frac{d}{dt}\right|_{t = 0} f_1(\gamma_1(t))f_2(\gamma_2(t))\\
&= f_2(y)\left.\frac{d}{dt}\right|_{t = 0} f_1(\gamma_1(t)) + f_1(x)\left.\frac{d}{dt}\right|_{t = 0} f_2(\gamma_2(t))
\\&=f_2(y)df_1|_{x}(U) + f_1(x)df_2|_{y}(V).
\end{align*}
Therefore,
\begin{align*}
g(\textup{grad}_g\, \hat{f}, W) = d\hat{f}(W) &= f_2(y)df_1(U) + f_1(x)df_2(V) 
\\&= f_2(y)g_1(\textup{grad}_{g_1}\,f_1, U) + f_1(x)g_{2}(\textup{grad}_{g_2}\,f_2, V)
\\&= g\left((f_2(y)\textup{grad}_{g_1}\,f_1, f_1(x)\textup{grad}_{g_2}\,f_2), (U, V)\right),
\end{align*}
and the claim follows.     
\end{proof}

With the above results in mind, we will now study the critical points of the function $F$ defined in \cref{eq:isingfunction}. In particular, we will prove the following result.
\begin{proposition}
\label{thm:descriptioncriticalpointsIsing}
Let $F$ be defined as in \cref{eq:isingfunction}. Then  every critical point of $F$ corresponds to an eigenvector of $H$. 
\end{proposition}

Let us first introduce some notation; using $(x,y)$ coordinates on the square lattice $V$ we can express $F$ explicitly as
\begin{align*}
F(U_{(1,1)}, \dotsc, U_{(n,n)}) &= -\sum_{\langle (i_1, i_2), (j_1, j_2)\rangle \in \Lambda} \bra{0}U_{(i_1, i_2)}^\dagger \sigma_z U_{(i_1, i_2)}\ket{0}\bra{0}U_{(j_1, j_2)}^\dagger \sigma_z U_{(j_1, j_2)}\ket{0} 
\\&\quad- \sum_{(i_1, i_2) \in V}h_{(i_1, i_2)} \bra{0}U_{(i_1, i_2)}^\dagger \sigma_z U_{(i_1, i_2)}\ket{0},
\end{align*}
where $h_{(i_1, i_2)}$ is as defined in \cref{eq:definitionmagneticfield}. To simplify the notation, we will denote
\[
f_{(i, j)} := \bra{0}U_{(i, j)}^\dagger \sigma_z U_{(i, j)}\ket{0},
\]
and we will denote the gradient as $\nabla$ for simplicity. This way, we can rewrite $F$ as 
\[
F = -\sum_{\langle (i_1, i_2), (j_1, j_2)\rangle \in \Lambda} f_{(i_1, i_2)}f_{(j_1, j_2)}- \sum_{(i_1, i_2) \in V}h_{(i_1, i_2)} f_{(i_1, i_2)}.
\]

Note that, using \cref{productrulegradient}, we know that the gradient of $F$ with respect to $U_{(i, j)}$ is 
\begin{align}
\label{localgradient}
\begin{split}
\nabla_{(i, j)}F &:= -\left(f_{(i-1, j)}+f_{(i+1, j)}\right.
\\&\quad\quad\left. +f_{(i, j-1)}+f_{(i, j+1)}+h_{(i, j)}\right)\nabla f_{(i,j)},   
\end{split}
\end{align}
where some of the summands may be omitted (i.e. when $i-1=0, j-1=0, i+1>n$ or $j+1>n$). 

Also recall from \cref{CriticalPointsOneGate}, that $\nabla f_{(i, j)}$ only vanishes when $f_{(i, j)} = \pm 1$, that is, when $U_{i, j}\ket{0}$ is---up to a phase---$\ket{0}$ or $\ket{1}$, respectively.

We are interested in characterizing the critical points of $F$, i.e, when
\[
\nabla_{(i, j)} F = 0\quad \forall i,j \in \{1, \dotsc, n\}.
\]
First of all, note that when $\ket{U}$ is an eigenvector of $H$, then it holds that $\nabla f_{(i, j)} = 0$ for every $(i, j) \in V$ and so it is a critical point of $F$. 

\begin{proof}[Proof of \cref{thm:descriptioncriticalpointsIsing}]
Assume that we are working on an $n \times n$ lattice, with $n \geq 2$. We will prove that every critical point corresponds to an eigenvector of $H$. To do so, we will consider a critical point and assume that at least one of the $\nabla f_{(i, j)}$ is non-zero. This will be done in three steps:

\begin{enumerate}
    \item We will first assume that $\nabla f_{(1,1)}$ is non-zero, and reach a contradiction. Therefore, by the symmetry of the problem we can conclude that every critical point must have 
    \[
    \nabla f_{(1, 1)} = \nabla f_{(1, n)} = \nabla f_{(n, 1)} = \nabla f_{(n, n)} = 0.
    \]
    \item For the next step, assuming that $\nabla f_{(1,1)} = \nabla f_{(1, n)} = 0$ we will prove that if $\nabla f_{(1, k)} \neq 0$ then we cannot be at a critical point.
    
    Again, by the symmetry of the problem, for the third step we may assume that $\nabla f_{(1,i)} = \nabla f_{(i, 1)} = 0$ for every $i = 1, \dotsc, n$. 
\item For the last step of the proof, we will assume that $\nabla f_{(k, l)} \neq 0$ but $\nabla f_{(k-1, l)} = \nabla f_{(k, l-1)} = 0$ for some $(k, l)$ with $k \in \{2, \dots, n-1\}, l \in \{2,\dotsc,  n-1\}$ and prove that we cannot be at a critical point.  
\end{enumerate}

\paragraph{First step: corners.}

Assume that $\nabla f_{(1,1)} \neq 0$. Therefore, by \cref{localgradient}
\[
f_{(1,2)} + f_{(2,1)} + \varepsilon_{(1,1)} + 2 = 0,
\]
which cannot happen, as $-1 \leq f_{(i, j)} \leq 1$ for every $(i, j) \in V$. Therefore, we may now assume that the gradients $\nabla f_{(i, j)}$ of the corners vanish, and continue with the second step of our proof.

\paragraph{Second step: first row.}

Assume that, for some $1 < k < n$, $\nabla f_{(1,k)} \neq 0$, and $\nabla f_{(1, k-1)} = 0$. This implies that 
\begin{equation*}
f_{(1,k-1)} + f_{(1,k+1)} + f_{(2,k)} + 3 + \varepsilon_{(1, k)} = 0.
\end{equation*}
Since $\nabla f_{(1,k-1)} = 0$ by assumption, it has to be the case that $f_{(1, k-1)} = \pm1$, and so
\begin{equation*}
f_{(1,k+1)} + f_{(2,k)} + 2 + \varepsilon_{(1,k)} = 0,\quad \textup{or}\quad f_{(1,k+1)} + f_{(2,k)} + 4 + \varepsilon_{(1,k)} = 0,
\end{equation*}
which in either case is a contradiction.

Again, by the symmetry of the problem, we can assume that all of the $f_{(i,j)}$ laying on the \textit{outer layer} of our lattice have vanishing gradients. 

\paragraph{Third step: inner spins.}
It only remains to prove the case when $\nabla f_{(i, j)} \neq 0$ for some $1 < i, j < n$, and 
\[
\nabla f_{(i-1, j)} = \nabla f_{(i, j-1)} = 0.
\]
In this case 
\[
f_{(i-1, j)} + f_{(i, j-1)} + f_{(i+1, j)} + f_{(i, j+1)} + 4 + \varepsilon_{(i, j)} = 0,
\]
which cannot happen, as 
\[
f_{(i-1, j)} + f_{(i, j-1)} \in \{-2, 0, 2\}.
\] 
\end{proof}

\begin{remark}
Note that the proof of \cref{thm:descriptioncriticalpointsIsing} only relies on the magnetic field being greater than the interaction degree at each site. For this reason, by considering a suitable magnetic field for each spin, the proof can be easily adapted to hold for any dimension and any interaction graph. 
\end{remark}

\cref{thm:descriptioncriticalpointsIsing} allows us to conclude that the set of critical points of $F$ can be written as
\[
\{ (U_1, \dotsc, U_{n^2}) \in \textup{SU}(2)^{\times n^2} : (U_1 \otimes \cdots \otimes U_{n^2})\ket{0}^{\otimes n^2} = e^{i\theta}\ket{x} : x \in \{0, 1\}^{n^2}, \theta \in [0, 2\pi)\}.
\]
Therefore, if we consider the Riemannian submersion with totally geodesic fibers $\pi$ from $\textup{SU}(2)^{\times n^2}$ to $(\bb{S}^2)^{\times n^2}$, and consider $\tilde F$, the unique function on $(\bb{S}^2)^{\times n^2}$ such that $F = \tilde F \circ \pi$, we can conclude that $\tilde F$ has exactly $2^{n^2}$ critical points, which correspond---via a slight abuse of notation where $\ket{\psi} \equiv e^{i\theta}\ket{\psi}$ for every vector $\ket{\psi}$ and any $\theta \in [0, 2\pi)$---to the points in $(\bb{S}^2)^{\times n^2}$ given by
\[
\{\ket{x} : x \in \{0, 1\}^{n^2}\}.
\]
Thus, we can lower bound the distance between any two critical points of $\tilde F$. 
\begin{lemma}
\label{lem:distancecriticalpoints}
Let $\tilde F$ be the unique function on $((\bb{S}^2)^{\times n^2}, \frac{1}{2}g_{\mathit{round}})$ such that $F = \tilde F \circ \pi$. Then for any two distinct critical points $x, y$ of $\tilde F$, it holds that 
\[
d_{\frac{1}{2}g_{\mathit{round}}}(x, y) \geq \frac{\sqrt{2}\pi}{2}.
\]
\end{lemma}
\begin{proof}
For any two distinct eigenvectors $\ket{x} = \ket{x_1, \dotsc, x_{n^2}}$, $\ket{x'} = \ket{x'_1, \dotsc, x'_{n^2}}$ of $H$ there exists some $j \in \{1, \dotsc, n^2\}$ such that $\{x_j, x'_j \}= \{0,1\}$. Since any point in $\bb{S}^2$ can be written as
\[
\cos\frac{\theta}{2} \ket{0} + \sin \frac{\theta}{2} e^{i \varphi} \ket{1},
\]
it is clear that the points in the Bloch sphere $\bb{S}^2$ corresponding to $x_j$ and $x'_j$ are antipodal. Therefore, the distance between any two distinct critical points, $\ket{x}$ and $\ket{x'}$, when considering the round metric in $\bb{S}^2$ scaled by $\frac{1}{2}$, is lower bounded by $\frac{\sqrt{2}}{2}\pi$. Indeed, the distance between two antipodal points in $(\bb{S}^2,g_{\mathit{round}})$ is $\pi$, and when scaling the metric by $\frac{1}{2}$, the distances get scaled by a factor of $\frac{1}{\sqrt{2}}$. 
\end{proof}

Lastly, let us study the Hessian of $F$ at the critical points. 
\begin{lemma}
\label{lem:escapingeigenvalueIsing}
Let $F$ be as defined in \cref{eq:isingfunction}, and let $\tilde F$ be its \textit{projected version} on $((\bb{S}^2)^{\times n^2}, h)$. Then $\tilde F$ has a unique minimum, a unique maximum and the rest of the critical points are saddle points. Moreover, for every saddle point $y$ of $\tilde F$, $\lambda_{\min}\big(\nabla^2 \tilde F(y)\big) \leq -\lambda_* < 0$, where
\[
\lambda_* := \min\{1, \min_{(i, j) \in V} \varepsilon_{(i, j)}\}.
\]
\end{lemma}
\begin{proof}
Let us parametrize $F$ along a geodesic 
\[
\gamma(t) = (e^{it\tilde X_1}U_1, \dotsc, e^{it\tilde X_{n^2}}U_{n^2}).
\]
Taking derivatives, we may conclude that
\begin{align*}
\nabla^2 F|_{\overrightarrow{U}}(i\overrightarrow{XU},i\overrightarrow{XU}) &= - \sum_{i, j = 1}^{n^2} \bra{U} X_i X_j H \ket{U} - \sum_{i, j = 1}^{n^2} \bra{U} H X_i X_j \ket{U}
\\&\quad + 2\sum_{i, j = 1}^{n^2} \bra{U} X_i H X_j \ket{U},
\end{align*}
where we denote by $X_i$ the tangent vector $\tilde X_i$ tensorized with the identity, i.e. writing each $\tilde X_i$ in the eigenbasis $\{\ket{0}, \ket{1}\}$,
\[
X_i = \mathds{1} \otimes \dotsm \otimes \mathds{1} \otimes \left(\sum_{k,l = 0}^1 \alpha_{k, l}^{(j)} \ket{k}\bra{l}\right)\otimes \mathds{1} \otimes \dotsm \otimes \mathds{1}. 
\]
We also denote 
\[
\overrightarrow{U} := (U_1, \dotsc, U_{n^2}),\ \ket{U} := U_1 \otimes \dotsm \otimes  U_{n^2}\ket{0}^{\otimes n^2},
\]
and 
\[
\overrightarrow{XU} := (X_1X_2 \dotsm X_{n^2}) (U_1 \otimes U_2 \otimes \dotsm \otimes U_{n^2}),
\]
for simplicity. 

If $\ket{U}$ is an eigenvector of $H$ with eigenvalue $E$, then
\begin{equation}
\label{eqHessianMultiGate}
\left.\nabla^2 F\right|_{\overrightarrow{U}}(i\overrightarrow{XU},i\overrightarrow{XU}) = 2\left(\sum_{i, j = 1}^{n^2} \bra{U} X_i H X_j \ket{U} - E\sum_{i, j = 1}^{n^2} \bra{U} X_i X_j \ket{U}\right).
\end{equation}
Moreover, writing $\ket{U}$ as 
\[
\ket{U} = \ket{l_1} \otimes \ket{l_2} \otimes \dotsm \otimes \ket{l_{n^2}},
\]
where $l_i \in \{0, 1\}$ for every $i \in \{1, \dotsc, n^2\}$, we can conclude that
\begin{align*}
X_j \ket{U} &= \ket{l_1} \otimes \overset{j)}{\dotsm} \otimes \left(\sum_{k = 0}^1 \alpha_{k, l_j}^{(j)} \ket{k}\right)\otimes \dotsm \otimes \ket{l_{n^2}}
\\&= \sum_{k = 0}^1 \alpha_{k, l_j}^{(j)} \ket{l_1} \otimes \overset{j)}{\dotsm} \otimes  \ket{k}\otimes \dotsm \otimes \ket{l_{n^2}}.
\end{align*}
If $i \neq j$ we can write the second term on the right-hand side of \cref{eqHessianMultiGate} as
\begin{align*}
\bra{U}X_i X_j \ket{U} &=  \left(\sum_{k = 0}^1 \overline{\alpha}_{k, l_i}^{(i)} \bra{l_1} \otimes \overset{i)}{\dotsm} \otimes  \bra{k}\otimes \dotsm \otimes \bra{l_{n^2}}\right)
\\&\quad \times \left(\sum_{k = 0}^1 \alpha_{k, l_j}^{(j)} \ket{l_1} \otimes \overset{j)}{\dotsm} \otimes  \ket{k}\otimes \dotsm \otimes \ket{l_{n^2}}\right)
\\&= \left(\sum_{k = 0}^1 \overline{\alpha}_{k, l_i}^{(i)}\bra{k}\right)\ket{l_i}\bra{l_j} \left(\sum_{k = 0}^1 \alpha_{k, l_j}^{(j)} \ket{k}\right)
\\&= \overline{\alpha}^{(i)}_{l_i, l_i}\alpha^{(j)}_{l_j, l_j}.
\end{align*}
If $i = j$ we can write
\[
\bra{U} X_i^2 \ket{U} = \sum_{k = 0}^1 |\alpha^{(i)}_{k, l_i}|^2.
\]

On the other hand, if $i \neq j$ we can write the first term on the right-hand side of \cref{eqHessianMultiGate} as
\begin{align*}
\bra{U}X_i H X_j\ket{U}&= \bra{U} X_i H \left(\sum_{k = 0}^1 \alpha_{k, l_j}^{(j)} \ket{l_1} \otimes \overset{j)}{\dotsm} \otimes  \ket{k}\otimes \dotsm \otimes \ket{l_{n^2}}\right)
\\&= \bra{U} X_i \left(\sum_{k = 0}^1 \alpha_{k, l_j}^{(j)}E_{l_1, \overset{j)}{\dotsc}, k, \dotsc, l_{n^2}} \ket{l_1} \otimes \overset{j)}{\dotsm} \otimes  \ket{k}\otimes \dotsm \otimes \ket{l_{n^2}}\right)
\\&= \overline{\alpha}^{(i)}_{l_i, l_i}\alpha_{l_j, l_j}^{(j)}E_{l_1, \dotsc, l_j, \dotsc, l_{n^2}},
\end{align*}
where $E_{l_1, \dotsc, l_{n^2}}$ denotes the eigenvalue of $H$ associated with the eigenstate $\ket{l_1, \dotsc, l_{n^2}}$. 

If $i = j$ then 
\begin{align*}
\bra{U}X_i H X_i\ket{U}&= \bra{U} X_i H \left(\sum_{k = 0}^1 \alpha_{k, l_i}^{(i)} \ket{l_1} \otimes \overset{i)}{\dotsm} \otimes  \ket{k}\otimes \dotsm \otimes \ket{l_{n^2}}\right)
\\&= \bra{U} X_i \left(\sum_{k = 0}^1 \alpha_{k, l_i}^{(i)}E_{l_1, \overset{i)}{\dotsc}, k, \dotsc, l_{n^2}} \ket{l_1} \otimes \overset{i)}{\dotsm} \otimes  \ket{k}\otimes \dotsm \otimes \ket{l_{n^2}}\right)
\\&= \sum_{k = 0}^1 |\alpha^{(i)}_{k, l_i}|^2 E_{l_1, \overset{i)}{\dotsc}, k, \dotsc, l_{n^2}}.
\end{align*}

Substituting these expressions into \cref{eqHessianMultiGate} we obtain
\begin{align}
\label{expressionhessiano}
\begin{split}
\frac{1}{2}\left.\nabla^2 F\right|_{\overrightarrow{U}}(i\overrightarrow{XU},i\overrightarrow{XU}) &= \sum_{i, j = 1}^{n^2} \bra{U} X_i H X_j \ket{U} - E\sum_{i, j = 1}^{n^2} \bra{U} X_i X_j \ket{U}
\\&= \sum_{\substack{i, j = 1\\i \neq j}}^{n^2} \left(\overline{\alpha}^{(i)}_{l_i, l_i}\alpha_{l_j, l_j}^{(j)}E_{l_1, \dotsc, l_j, \dotsc, l_{n^2}} - E \overline{\alpha}^{(i)}_{l_i, l_i}\alpha^{(j)}_{l_j, l_j}\right) 
\\&\quad + \sum_{i = 1}^{n^2} \left(\sum_{k = 0}^1 |\alpha^{(i)}_{k, l_i}|^2 E_{l_1, \dotsc, k, \dotsc, l_{n^2}} - E \sum_{k = 0}^1 |\alpha^{(i)}_{k, l_i}|^2\right)
\\&= \sum_{i = 1}^{n^2} \sum_{k = 0}^1 \left(E_{l_1, \overset{i)}{\dotsc}, k, \dotsc, l_{n^2}} - E\right) |\alpha^{(i)}_{k, l_i}|^2
\\&= \sum_{i = 1}^{n^2} \left(E_{l_1, \overset{i)}{\dotsc}, l_i + 1 \textup{ mod } 2, \dotsc, l_{n^2}} - E\right) |\alpha^{(i)}_{l_i + 1 \textup{ mod } 2, l_i}|^2.
\end{split}
\end{align}
This way, if $E$ corresponds to the lowest (resp. highest) eigenvalue of $H$, the Hessian of $F$ is positive-semidefinite (resp. negative-semidefinite).

Observe that in \cref{expressionhessiano} $E_{l_1, \overset{i)}{\dotsc}, (l_i + 1 \mod 2), \dotsc, l_{n^2}}$ corresponds to the eigenvalue obtained when flipping a single spin---in the $i$-th position. Thus, if the eigenvector is not that with minimal eigenvalue of $H$ we can flip at least one spin from $\ket{1}$ to $\ket{0}$, in which case the associated eigenvalue decreases at least by $\min_{I \in V}\varepsilon_I$. 
\end{proof}

Lastly, we would like to see that the Hessian of $\tilde F$ at the global minimum is positive definite. To do this, we will see which directions $\overrightarrow{XU}$ lie in the kernel of $\nabla^2 F$.
\begin{proposition}
\label{lem:hessiandefpos}
The Hessian of $\tilde F$ at its global minimum is positive definite. 
\end{proposition}
\begin{proof}
We want to conclude that the kernel of $\nabla^2 F$ at its global minimum is included in the vertical tangent space of $\textup{SU}(2)^{\times n^2}$ with respect to the submersion $\pi: \textup{SU}(2)^{\times n^2} \to (\bb{S}^2)^{\times n^2}$. This way, using \cref{prop:equalityHessians}, we will conclude that $\tilde F$ has a positive-definite Hessian at the global minimum. 

Let us use the expression for the Hessian shown in \cref{expressionhessiano} and let us denote by $E$ the smallest (non-degenerate) eigenvalue of $H$, which corresponds to the eigenvector $\ket{0}^{\otimes n^2}$. Then  at any $U$ for which $U\ket{0}^{\otimes n^2} = e^{i\theta} \ket{0}^{\otimes n^2}$, it holds that
\[
\frac{1}{2}\left.\nabla^2 F\right|_{\overrightarrow{U}}(i\overrightarrow{XU},i\overrightarrow{XU}) = \sum_{i = 1}^{n^2} \left(E_{0, \overset{i)}{\dotsc}, 1, \dotsc, 0} - E\right) |\alpha^{(i)}_{1, 0}|^2,
\]
where $E_{0, \overset{i)}{\dotsc}, 1, \dotsc, 0}$ denotes the eigenvalue associated with $|0, \overset{i)}{\dotsc}, 1, \dotsc, 0\rangle$ and the terms $\alpha^{(i)}_{j,k}$ are the coordinates of the $i$-th component of the vector $\overrightarrow{X}$.

Since $E$ is the unique lowest eigenvalue of $H$, the only directions $i\overrightarrow{XU}$ that lie in the kernel of $\left.\nabla^2 F\right|_{\overrightarrow{U}}$ correspond to those such that $\alpha^{(i)}_{k, li} = 0$ for every $i$ and every $k \neq l_i$. Otherwise, we would have a non-zero term in the sum: $E_{l_1, \overset{i)}{\dotsc}, k, \dotsc, l_{n^2}} - E > 0$. 

For this reason, every $X_j$ is of the form
\[
X_j = \mathds{1} \otimes \dotsm \otimes \mathds{1} \otimes \left(\alpha_{0, 0}^{(j)} \ketbra{0} - \alpha_{0, 0}^{(j)} \ketbra{1}\right)\otimes \mathds{1} \otimes \dotsm \otimes \mathds{1},
\]
and so it can be written as 
\[
X_j = \mathds{1} \otimes \dotsm \otimes \mathds{1} \otimes \begin{pmatrix}
a & 0\\
0 & -a
\end{pmatrix}\otimes \mathds{1} \otimes \dotsm \otimes \mathds{1},
\]
for some $a$. These directions are vertical with respect to the action shown in \cref{eq:actiononSU2}, and so the Hessian of $\tilde F$ at its global minimum is positive definite.
\end{proof}

This result, along with the calculations presented in the proof of \cref{lem:escapingeigenvalueIsing} allow us to guarantee that at the global minimum $x^*$ of $\tilde F$, 
\[
\lambda_{\min} (\nabla^2 \tilde F(x^*)) \geq \lambda_*,
\]
with the same constant $\lambda_*$ as in \cref{lem:escapingeigenvalueIsing}.

\subsubsection{Rapid mixing and suboptimality for the Gibbs distributions}
Let us conclude this section by applying \cref{thm:MainFormal1,thm:MainFormal2,cor:rapidmixing} to our setting. In particular, we will study the function $F$ defined on the manifold $\textup{SU}(2)^{\times n^2}$ endowed with the metric $g$ which corresponds to the product of the bi-invariant metric of $\textup{SU}(2)$ and its projection $\tilde F$ on $(\bb{S}^2)^{\times n^2}$. 

\begin{proposition}
\label{prop:mainthmising1}
Let $F: \textup{SU}(2)^{\times n^2} \to \bb{R}$ be defined as in \cref{eq:isingfunction} and let $\tilde F: (\bb{S}^2)^{\times n^2} \to \bb{R}$ be the unique function such that $F = \tilde F \circ \pi$, where $\pi$ is the Riemannian submersion
\[
\pi:(\textup{SU}(2)^{\times n^2}, g) \to ((\bb{S}^2)^{\times n^2}, h),
\]
$g$ corresponds to the product of the bi-invariant metric of $\textup{SU}(2)$, and $h$ corresponds to the product of the metric $\frac{1}{2}g_{\mathit{round}}$, where $g_{\mathit{round}}$ denotes the usual round metric on the sphere. Let $A_2 \geq 1$ be the Lipschitz constant of the gradient of $F$ and $\tilde F$. For every $\varepsilon \in (0, 1]$ and $\delta \in (0, 1)$, if $\beta$ is such that 
\[
\beta \geq \frac{2}{\varepsilon}\bigg(\frac{1}{2} + 3^2n^2 \log (3n^2) + 3 n^4 \log(2^{5/2} \pi^2) + 3n^2 \log \frac{A_2\sqrt{2\pi}}{\delta\varepsilon}\bigg),
\]
the Gibbs distribution $\nu_F(x) = \frac{1}{Z_F} e^{-\beta F(x)}$ on $\textup{SU}(2)^{\times n^2}$ is such that
\[
\nu_F\Big(F - \min_{y \in \textup{SU}(2)^{\times n^2}} F(y) \geq \varepsilon\Big) \leq \delta.
\]
Now, if we define 
\[
\varepsilon_{\max} := \min \Big\{\frac{\pi^2 A_2}{16}, 1\Big\},
\]
then, for every $\varepsilon \in (0, \varepsilon_{\max}]$ and $\delta \in (0, 1)$, if $\beta$ satisfies 
\[
\beta \geq \frac{2}{\varepsilon}\bigg(\frac{1}{2} + 6n^2 \log (2n^2) + 2n^4 \log \frac{\pi^{3/2}}{\Gamma(3/2)} + 2n^2 \log \frac{A_2\sqrt{2\pi}}{\delta\varepsilon}\bigg),
\]
the Gibbs distribution $\nu_{\tilde F}(x) = \frac{1}{Z_{\tilde F}} e^{-\beta \tilde F(x)}$ on $(\bb{S}^{2})^{\times n^2}$ is such that
\[
\nu_{\tilde F}\Big(\tilde F - \min_{y \in (\bb{S}^{2})^{\times n^2}} \tilde F(y) \geq \varepsilon\Big) \leq \delta.
\]
\end{proposition}

Before we present the proof of \cref{prop:mainthmising1}, note that $((\bb{S}^2)^{\times n^2}, h)$ is assumed to be endowed with the product of the round metric scaled by $\frac{1}{2}$, i.e. 
\[
((\bb{S}^2)^{\times n^2}, h) \simeq (\bb{S}^2, \frac{1}{2}g_{\mathit{round}})^{\times n^2}.
\]
In contrast, in \cref{sec:studyingmanifolds} we studied the sphere endowed with the round metric. Nevertheless, the results derived in the aforementioned section can be easily adapted to hold for the scaled sphere. Indeed, given some $n$-dimensional Riemannian manifold $(M, g)$, if we consider the same manifold $M$ endowed with the scaled metric $\lambda g$, for some $\lambda > 0$, the Ricci curvature of the scaled manifold stays the same as the original \linebreak (\cite[Theorem 7.30]{lee2018introductionRiemannian}), while the distances get scaled by $\sqrt{\lambda}$ and the volume form gets scaled by $\lambda^{n/2}$. 

\begin{proof}[Proof of \cref{prop:mainthmising1}]
To prove the suboptimality result for the Gibbs distributions $\nu_F$ and $\nu_{\tilde F}$, we need to bound the injectivity radii and the volume of $(\textup{SU}(2)^{\times n^2}, g)$ and $((\bb{S}^{2})^{\times n^2}, h)$. First, as we saw in \cref{prop:isomSUS} that 
\[
(\textup{SU}(2), g_{\mathit{bi}}) \simeq (\bb{S}^3, 2g_{\mathit{round}}),
\]
where $g_{\mathit{bi}}$ denotes the bi-invariant metric of $\textup{SU}(2)$, we know that
\[
\textup{Vol}(\textup{SU}(2)) = 2^{3/2} \textup{Vol}(\bb{S}^3) = 2^{3/2} \frac{2\pi^2}{\Gamma(2)} = 2^{5/2} \pi^2, 
\]
where we used the expression for the volume of the sphere endowed with the round metric seen in \cref{dimensionandvolume}, and so 
\[
\textup{Vol}(\textup{SU}(2)^{\times n^2}) = (2^{5/2} \pi^2)^{n^2}. 
\]
Also 
\[
\textup{Vol}\Big(\Big(\bb{S}^2, \frac{1}{2}g_{\mathit{round}}\Big)\Big) = \frac{1}{2}\textup{Vol}(\bb{S}^2) = \frac{\pi^{3/2}}{\Gamma(3/2)},
\]
and so 
\[
\textup{Vol}(((\bb{S}^2)^{\times n^2}, h))) = \Big(\frac{\pi^{3/2}}{\Gamma(3/2)}\Big)^{n^2}.
\]
For the injectivity radii, we simply use the scaling property of the distance, along with the bounds obtained in \cref{cor:injectivityradii} to conclude that
\[
i(\textup{SU}(2)^{\times n^2}, g) = i(\textup{SU}(2), g_{\mathit{bi}}) = i(\bb{S}^3, 2g_{\mathit{round}}) = \sqrt{2}i(\bb{S}^3) \geq \sqrt{2}\pi. 
\]
\[
i(((\bb{S}^2)^{\times n^2}, h)) = i(\bb{S}^2, \frac{1}{2}g_{\mathit{round}}) = \frac{\sqrt{2}}{2}i(\bb{S}^2) \geq \frac{\sqrt{2}\pi}{2}. 
\]
Substituting these bounds into \cref{thm:MainFormal2}, the suboptimality for the Gibbs distributions follows. 
\end{proof}

Now, let us restate the assumption needed for \cref{thm:MainFormal1} to hold in this case.
\begin{assumption}
\label{ass:sec8.2assumption1} 
There exists a constant $0 < C_{\tilde F} \leq 1$ such that
\[
|\Grad{h} \tilde F(x)|_h \geq C_{\tilde F} d_h(x, \tilde{\mathcal{C}}),    
\]
for every $x \in (\bb{S}^2)^{\times n^2}$, where $\tilde{\mathcal{C}}$ denotes the set of critical points of $\tilde F$. 
\end{assumption}

\begin{proposition}
\label{prop:mainthmising2}
Let $F: \textup{SU}(2)^{\times n^2} \to \bb{R}$ be defined as in \cref{eq:isingfunction} and let $\tilde F: (\bb{S}^2)^{\times n^2} \to \bb{R}$ be the unique function such that $F = \tilde F \circ \pi$, where $\pi$ is the Riemannian submersion
\[
\pi:(\textup{SU}(2)^{\times n^2}, g) \to ((\bb{S}^2)^{\times n^2}, h),
\]
$g$ corresponds to the product of the bi-invariant metric of $\textup{SU}(2)$, and $h$ corresponds to the product of the metric $\frac{1}{2}g_{\mathit{round}}$, where $g_{\mathit{round}}$ denotes the usual round metric on the sphere. Assume that $\tilde F$ satisfies \cref{ass:sec8.2assumption1}, and let
\[
\lambda_* := \min\{1, \min_{I \in V} \varepsilon_I\}, 
\]
(cf. \cref{eq:definitionmagneticfield}). Furthermore, let $a, \beta > 0$ be such that
\begin{align*}
a^2 &\geq \max\Big\{\frac{48 A_2 n^2}{C^2_{\tilde F}}, \frac{288}{\lambda_{*}}\Big\},
\\\beta &\geq \max\Bigg\{\frac{384^2 (2n^2)^5 A^2_2 A^2_3 a^6}{\lambda_{*}^2}, \frac{18a^2}{\pi^2}\Bigg\},
\end{align*}
where $A_2 \geq 1$ is the Lipschitz constant of the gradient of $F$ and $\tilde F$, and $A_3 \geq 1$ is the Lipschitz constant of the Hessian of $\tilde F$. Let $\operatorname{L}_F$ and $\operatorname{L}_{\tilde F}$ be defined as 
\[
\operatorname{L}_F := -\textup{grad}_g\, F + \frac{1}{\beta}\Delta_g,\quad \textup{and}\quad \operatorname{L}_{\tilde F} := -\textup{grad}_h\, \tilde F + \frac{1}{\beta} \Delta_h.
\]
Then  the Markov triple $(\textup{SU}(2)^{\times n^2}, \nu_F, \Gamma_g)$ associated with $\operatorname{L}_F$ satisfies an $\textup{LSI}(\alpha)$, where
\[
\frac{1}{\alpha} = 16\beta A_2 \pi^2 n^2\max \Big\{\frac{120}{\lambda_*}, 2 n^2 \beta\Big\},
\]
and the Markov triple $((\bb{S}^2)^{\times n^2}, \nu_{\tilde F}, \Gamma_h)$ associated with $\operatorname{L}_{\tilde F}$ satisfies an $\textup{LSI}(\tilde \alpha)$, where
\[
\frac{1}{\tilde \alpha} = \frac{240\beta A_2  n^2 \pi^2}{\lambda_*}.
\]
\end{proposition}

\begin{proof}
In order to obtain the log-Sobolev inequalities, recall from \cref{thm:descriptioncriticalpointsIsing} and \cref{lem:distancecriticalpoints} that the critical points of $\tilde F$ are isolated and separated by a distance of at least $\frac{\sqrt{2}\pi}{2}$. Moreover, we proved in 
\cref{lem:escapingeigenvalueIsing,lem:hessiandefpos} that \cref{assumption3.5.1,assumption3.5.2} hold with $\lambda_*$ given by
\[
\min\{1, \min_{I \in V} \varepsilon_I\}.
\]

The Ricci curvatures of both $(\textup{SU}(2)^{\times n^2}, g)$ and $((\bb{S}^2)^{\times n^2}, h)$ are non-negative (cf. \cref{curvProduct,LemaCurvaturaSUn,sectionalcurvatureofsphere}). As in the proof of \cref{prop:mainthmtrace2}, we can use \cite[Proposition 31, Section 11]{OneilSemi} and \cref{ineqliebracket} to conclude that the coefficients of the Riemann curvature tensor of $((\mathbb{S}^2)^{\times n^2}, h)$ in normal coordinates when endowed with the product of scaled round metrics are bounded by $2$.

Also, it holds that the Riemannian submersion $\pi : (\textup{SU}(2)^{\times n^2}, g) \to ((\bb{S}^2)^{\times n^2}, h)$ has totally geodesic fibers, which means by \cref{CurvaturaSubmersion,vanishingT} that the sectional curvature of the fibers coincides with that of the total space. Since $(\textup{SU}(2)^{\times n^2}, g)$ has non-negative sectional curvature by \cref{LemaCurvaturaUn,curvatureproduct}, we conclude that the sectional curvature of the fibers is also non-negative and so their Ricci curvature is non-negative as well. 

Now, the result follows by substituting the manifold-dependent constants in the statement of \cref{thm:MainFormal1,} with the known bounds found in \cref{SectionExamples}. In particular, the dimensions of $\textup{SU}(2)^{\times n^2}$ and $(\bb{S}^2)^{\times n^2}$ are known to be $3n^2$ and $2n^2$, respectively. Also, using \cref{diameterProduct}, \cref{diamSUn} and the scaling property of the distances, we know that 
\begin{align*}
\textup{diam}(((\bb{S}^2)^{\times n^2}, h)) &= n\, \textup{diam}((\bb{S}^2, \frac{1}{2}g_{\mathit{round}})) = n\frac{\sqrt{2}\pi}{2},\\
\textup{diam}((\textup{SU}(2)^{\times n^2}, g)) &= n\, \textup{diam}((\textup{SU}(2), g_{\mathit{bi}})) \leq 2\pi n.
\end{align*}
Lastly, the diameter of the fibers with the induced metric can be bounded using \cref{diameterProduct} and \cref{prop:isomSUS}. Indeed, for every $x \in ((\bb{S}^2)^{\times n^2}, h)$
\[
\textup{diam}((\pi^{-1}(x), g|_{\pi^{-1}(x)})) \leq n\, \textup{diam}((\bb{S}^3, 2g_{\mathit{round}})) \leq \sqrt{2}\pi n,
\]
and by \cref{cor:injectivityradii} and the scaling of the metric, the convexity radius of $((\bb{S}^2)^{\times n^2}, h)$ is lower bounded by
\[
\mathit{conv}(((\bb{S}^2)^{\times n^2}, h)) = \mathit{conv}((\bb{S}^2, \frac{1}{2}g_{\mathit{round}})) = \frac{\sqrt{2}}{2}\mathit{conv}((\bb{S}^2, g_{\mathit{round}})) \geq \frac{\pi\sqrt{2}}{4},
\]
so the result follows.
\end{proof}

Lastly, let us restate \cref{cor:rapidmixing}.

\begin{corollary}
Let $X_t$ and $\tilde X_t$ be the Langevin diffusion processes generated by $\operatorname{L}_F$ and $\operatorname{L}_{\tilde F}$, with initial uniform distribution on $\textup{SU}(2)^{\times n^2}$ and $(\bb{S}^2)^{\times n^2}$, respectively. Under the assumptions of \cref{prop:mainthmising2}, they converge exponentially fast to their associated Gibbs distributions $\nu_F$, and $\nu_{\tilde F}$ i.e. 
\begin{align*}
\norm{\nu_F - \rho_t}^2_{\textup{TV}} &\leq  \beta\, e^{-2\alpha t} (\max_{y \in \textup{SU}(2)^{\times n^2}} F(y) - \min_{y \in \textup{SU}(2)^{\times n^2}} F(y)),\\
\norm{\nu_{\tilde F} - \tilde \rho_t}^2_{\textup{TV}} &\leq  \beta\, e^{-2\tilde\alpha t} (\max_{y \in (\bb{S}^2)^{\times n^2}} \tilde F(y) - \min_{y \in (\bb{S}^2)^{\times n^2}} \tilde F(y))
\\&= \beta\, e^{-2\tilde\alpha t} (\max_{y \in \textup{SU}(2)^{\times n^2}} F(y) - \min_{y \in \textup{SU}(2)^{\times n^2}} F(y)),
\end{align*}
for every $t \geq 0$, where $\rho_t$ and $\tilde \rho_t$ denote the distribution of $X_t$ and $\tilde X_t$, and $\alpha$ and $\tilde \alpha$ denote the log-Sobolev inequality constants of $(\textup{SU}(2)^{\times n^2}, \nu_F, \Gamma_g)$, and $((\bb{S}^2)^{\times n^2}, \nu_{\tilde F}, \Gamma_h)$. 
\end{corollary}

%% file: Chapters/Geometry.tex
\section{Background in Riemannian geometry}
\label{SectionExamples}

This section gathers most of the results from Riemannian geometry that are needed in this work. In particular, we will give a general introduction to Lie groups, and analyze the curvature of the unitary and special unitary groups. We will also introduce Riemannian submersions, and obtain results regarding the connection between the curvature of their domain and their image, focusing on the particular case of Stiefel and Grassmann manifolds. We will conclude the section with a study of fiber-invariant functions in the context of Riemannian submersions, and a set of auxiliary results regarding some properties about the sphere and the aforementioned manifolds, that are used throughout the text. 

\subsection{Lie groups}
\label{sec:SecLieGroups}
We begin by introducing some of the most elementary definitions and properties of Lie groups. Lie groups can be both studied from the point of view of algebra and differential/Riemannian geometry. We will focus on the latter approach---see \cite{fulton2013representation} for the algebraic viewpoint.

\begin{definition}[Lie group]
A \textup{Lie group} is a group $G$ that can be endowed with a smooth (real) manifold structure in such a way that both the group operation and the inverse map are smooth. 
\end{definition}

In this work, we will restrict our attention to a specific kind of Lie groups, namely \textit{linear} Lie groups, which are closed subsets of $\textup{GL}(n)$. In fact, the only Lie groups that we consider in this text are the unitary group $\textup{U}(n)$ and special unitary group $\textup{SU}(n)$. As we mentioned earlier, this section is only intended to be a short summary of the main definitions and results used in \cref{sec:proof,sec:traceratio}. For this reason, we have decided to omit many proofs. For a more general and detailed introduction to Lie groups, see for example \cite{gallier2020differential}.

To each Lie group $G$ we can associate a real vector space, known as its \textit{Lie algebra} $\mathfrak{g}$, which in the case of linear Lie groups is given by
\[
\mathfrak{g} := \{X \in \textup{M}_n(\bb{C}) : e^{tX} \in G \textup{ \textit{for all} } t \in \bb{R}\}.
\]

\begin{proposition}
\label{prop:liealgebracident}
For the Lie groups $\textup{U}(n)$ and $\textup{SU}(n)$, their associated Lie algebras $\mathfrak{u}(n)$ and $\mathfrak{su}(n)$ are given by the following sets: 
\[
\mathfrak{u}(n) = \{ X \in \textup{M}_n(\bb{C}) : X^\dagger = - X\},
\]
and
\[
\mathfrak{su}(n) = \{ X \in \textup{M}_n(\bb{C}) : X^\dagger = - X,\ \Tr(X) = 0\}. 
\]
\end{proposition}

The vector space $\mathfrak{g}$ is called a Lie algebra because it is endowed with a bilinear map, known as the Lie bracket. 
\begin{definition}[Lie bracket]
Given a Lie algebra $\mathfrak{g}$, a \textup{Lie bracket} is a bilinear map
\[
[\cdot, \cdot] : \mathfrak{g} \times \mathfrak{g} \to \mathfrak{g},
\]
satisfying
\begin{itemize}
    \item $[u, u] = 0$ for every $u \in \mathfrak{g}$.
    \item $[u, [v, w]] + [w, [u, v]] + [v, [w,u ]] = 0$, for every $u, v, w \in \mathfrak{g}$. This is known as the \textup{Jacobi identity}, and together with the previous properties, implies that $[u, v] = -[v, u]$ for every $u, v \in \mathfrak{g}$. 
\end{itemize}
\end{definition}

In the case of linear Lie groups, the Lie bracket associated with their Lie algebra is given by the matrix commutator operator, i.e. $[u, v] := uv - vu$, for every $u, v \in \mathfrak{g}$. 

Bear in mind that the notation $[\cdot, \cdot]$ is also used for the Lie bracket of vector fields on a general manifold $M$, not necessarily a Lie group. Both uses of the same notation are instances of the same structure; indeed, the space of vector fields on $M$ forms a Lie algebra---the algebra of the diffeomorphism group of $M$---endowed with the Lie bracket of vector fields. Nevertheless, the two can be distinguished by the context: while the Lie bracket of $\mathfrak{g}$ is a map defined on $\mathfrak{g} \times \mathfrak{g}$, the Lie bracket of vector fields is defined on $\mathfrak{X}(M) \times \mathfrak{X}(M)$. The Lie bracket of vector fields is also known as the Lie derivative. 

One essential property of Lie groups is that their Lie algebra $\mathfrak{g}$ can be identified with the tangent space at the identity, $\mathfrak{g} \equiv T_e G$ (cf. \cite[Theorem 4.8]{gallier2020differential}). Furthermore, for each element $p \in G$, 
\[
T_p G = p\mathfrak{g} = \mathfrak{g}p.
\]

Lie groups also give us a very convenient way to define vector fields on them, given a tangent vector at some fixed point. Indeed, given a Lie group $G$ and a vector $v \in \mathfrak{g}$, we can define the following vector fields $v^L, v^R \in \mathfrak{X}(G)$ in the following way: 
\[
v^L(a) = av,\quad v^R(a) = va,\quad \forall a \in G.
\]
Such vector fields are known as \textit{left} and \textit{right-invariant} vector fields, respectively. 

Invariant vector fields have various useful properties; one of which appears in the following proposition (cf. \cite[Section 19.2]{gallier2020differential}). 
\begin{proposition}
\label{propinvvectorfields1}
The Lie derivative of two left-invariant (resp. right-invariant) vector fields is again left-invariant (resp. right-invariant).
\end{proposition}

In fact, the Lie derivative of two invariant vector fields is related to the Lie bracket of the Lie algebra as follows:
\begin{proposition}[{\cite[Proposition 19.11]{gallier2020differential}}]
\label{propinvvectorfields2}
Let $G$ be a linear Lie group with unit $e$, and let $u, v \in \mathfrak{g}$. Then
\[
[u^L, v^L](e) = uv - vu
\]
and
\[
[u^R, v^R](e) = vu - uv.
\]
\end{proposition}

Since Lie groups are manifolds, they can be endowed with a Riemannian metric. In particular, some Lie groups can be endowed with what is known as a \textit{bi-invariant} metric.

\begin{definition}[Bi-invariant metric]
Given a linear Lie group $G$ endowed with a Riemannian metric $g$, we say that $g$ is \textup{bi-invariant} if, for every $a, b \in G$ and every $X, Y \in T_b G$
\[
g(X, Y) = g(aX, aY),
\]
and
\[
g(X, Y) = g(Xa, Ya).
\]
\end{definition}
Note that, in the above definition, if $X, Y \in T_b G$, then $aX, aY \in T_{ab} G$ and $Xa, Ya \in T_{ba} G$. 

\begin{remark}
\label{metrica}
We can define a bi-invariant metric on $G = \textup{U}(n)$ or $\textup{SU}(n)$ as follows; for every $X, Y \in \mathfrak{g}$, let
\[
g(X, Y) = \Tr(Y^\dagger X).
\]
Now, we can extend the metric to $TG$; let $pX, pY \in T_p G$, with $X, Y \in \mathfrak{g}$, we define
\[
g(pX, pY) := \Tr(Y^\dagger p^\dagger p X) = \Tr(Y^\dagger X) = g(X, Y),
\]
where the equality follows from $p$ being unitary. 
\end{remark}

We will refer to $g(X, Y) = \Tr (Y^\dagger X)$ as the bi-invariant metric of $\textup{SU}(n)$ and $\textup{U}(n)$. Working with this metric allows us to characterize all the geodesics of $\textup{SU}(n)$ and $\textup{U}(n)$ (cf. \cite[Proposition 21.20]{gallier2020differential}). 
\begin{proposition}
Consider $G = \textup{SU}(n)$ or $\textup{U}(n)$ endowed with the bi-invariant metric shown in \cref{metrica}. For every $p \in G$, the geodesic passing through $p$ in the direction of $Xp \in T_p G$, $X \in \mathfrak{g}$ is of the form 
\[
\gamma(t) = e^{tX}p. 
\]
\end{proposition}

One key property of Lie groups endowed with a bi-invariant metric is that they are \textit{symmetric spaces}. This fact is used several times in \cref{SectionPI,sec:LaplacianAndGradientBound} and allows us to obtain an explicit expression for the Taylor series of the metric in normal coordinates. 
\begin{definition}
\label{def:symmetricspace}
A connected Riemannian manifold $(M, g)$ is a \textup{symmetric space} if for each point $p \in M$ there exists a point reflexion, i.e. an isometry $\varphi : M \to M$ that fixes $p$ and such that $d\varphi|_p : T_p M \to T_p M$ is equal to $-\textup{Id}$.
\end{definition}

As noted previously, it is known that every  connected Lie group endowed with a bi-invariant metric is a symmetric space. Furthermore, the round sphere, and Grassmann manifolds endowed with their induced metric are symmetric spaces too (c.f. \cite[Example 3.24]{lee2018introductionRiemannian}). 

\subsubsection{Curvature of the unitary and special unitary groups}
\label{curvatureunitarygroup}

Let us now study the curvature of $\textup{U}(n)$ and $\textup{SU}(n)$ when endowed with the bi-invariant metric. Since $\textup{U}(n)$ and $\textup{SU}(n)$ are linear Lie groups, their curvature can be studied by inspecting what is known as the \textit{Killing form}. 

\begin{definition}[Killing form]
Given a Lie algebra $\mathfrak{g}$ of a linear Lie group $G$, the \textup{Killing form} of $G$ is a symmetric bilinear form $\mathfrak{B}: \mathfrak{g} \times \mathfrak{g} \to \bb{R}$ given by 
\[
\mathfrak{B}(u, v) := \Tr([u, [v, \cdot]]),
\]
where $[\cdot, \cdot]$ denotes the Lie bracket associated with the Lie algebra of $G$, which is given by the commutator. 
\end{definition}

The Killing form can in fact be defined for \textit{any} Lie algebra $\mathfrak{g}$, not just for Lie algebras associated with linear Lie groups. For further details, see \cite{gallier2020differential}. 

\begin{proposition}[{\cite[Proposition 21.19]{gallier2020differential}}]
\label{curvLieGroups}
For any Lie group $G$ equipped with a bi-invariant metric $g$, the following hold: 
\begin{enumerate}
\item For all pairs of orthonormal vectors $u, v \in \mathfrak{g}$,
\[
K_g(u,v) = \frac{1}{4}|[u,v]|^2_g.
\]
\item For all $u, v, w \in \mathfrak{g}$, 
\[
R(u, v)w = \frac{1}{4}[[u, v], w]. 
\]
\item For all $u, v \in \mathfrak{g}$, 
\[
\textup{Ric}_g(u,v) = -\frac{1}{4}\mathfrak{B}(u,v),
\]
where $\mathfrak{B}$ is the Killing form of $G$, and $K_g$, $R$ and $\textup{Ric}_g$ denote the sectional curvature, and the Riemann and Ricci curvature tensors of $(G, g)$, respectively.  
\end{enumerate}
\end{proposition}

This proposition is very useful in our context, as there exists an explicit formula for the Killing form of $\textup{U}(n)$ and $\textup{SU}(n)$.
\begin{remark}[{\cite[Remark of Page 648]{gallier2020differential}}]
\label{KillingLieGroups}
The Killing form of $\textup{U}(n)$ is 
\[
\mathfrak{B}(X,Y) = 2n\Tr(XY) - 2\Tr(X) \Tr(Y), \quad X, Y \in \mathfrak{u}(n).
\]
Furthermore, the Killing form of $\textup{SU}(n)$ is
\[
\mathfrak{B}(X, Y)= 2n\Tr(XY),\quad X, Y \in \mathfrak{su}(n).
\]
\end{remark}

The expressions shown in \cref{KillingLieGroups} allow us to bound the curvatures of $\textup{U}(n)$ and $\textup{SU}(n)$ when endowed with their bi-invariant metric.

\begin{proposition}
\label{LemaCurvaturaUn}
Let $\textup{U}(n)$ be endowed with the bi-invariant metric $g$ from \cref{metrica}. Then  
\[
\textup{Ric}_{g} \geq 0,\quad \textup{\textit{and}}\quad 
0 \leq K_{g} \leq \frac{1}{2}.
\]
\end{proposition}
\begin{proof}
Since left translations are isometries, it suffices to establish the curvature statements at the identity. By \cref{curvLieGroups,KillingLieGroups} it holds that
\[
\textup{Ric}_{g}(X, X) = -\frac{1}{2}\left(n\Tr(X^2) - \Tr(X)^2\right) \geq 0,
\]
for any $X \in \mathfrak{u}(n)$.

Now, given two orthonormal vectors $X,Y \in \mathfrak{u}(n)$, using \cref{curvLieGroups} we obtain
\begin{align*} 
0\leq K_{g}(X, Y) &= \frac{1}{4} \left|[X, Y]\right|_g^2.
\end{align*}
Furthermore, it is known \cite{FrobeniusNorm} that
\begin{equation}
\label{ineqliebracket}
\left|[X, Y]\right|^2_g \leq 2 |X|_g^2 |Y|_g^2,    
\end{equation}
which allows us to conclude that 
\[
0\leq K_{g}(X, Y)  \leq \frac{2}{4} \left|X\right|_g^2 \left|Y\right|_g^2 \leq \frac{1}{2}.
\]
\end{proof}

\begin{proposition}
\label{LemaCurvaturaSUn}
Consider $\textup{SU}(n)$ endowed with the bi-invariant metric $g$ from \cref{metrica}. Then  
\[
\textup{Ric}_{g} = \frac{n}{2}g,\quad \textup{\textit{and}}\quad 0 \leq K_{g} \leq \frac{1}{2}.
\]
\end{proposition}
\begin{proof}
As in the previous result, we may only consider the expressions at the identity. The bound of the sectional curvature can be proven following the same steps as in \cref{LemaCurvaturaUn}.

In order to obtain the equation for the Ricci curvature, note that by \cref{curvLieGroups,KillingLieGroups}
\[
\textup{Ric}_{g}(X, Y) = -\frac{2}{4}n\Tr(XY) = \frac{n}{2}g(X, Y),
\]
for every every pair $X, Y \in \mathfrak{su}(n)$. 
\end{proof}

\subsection{Curvature of product manifolds}
\label{curvatureproduct}

Let us now study the curvature of product manifolds and how it is related to that of each factor. The usual metric considered for the product of two Riemannian manifolds is known as the \textit{product metric}, which is given as the direct sum of the metrics each manifold is endowed with. This is the metric considered in the manifolds appearing in \cref{sec:traceratio}. 

Before we show how to bound the curvature of the product metric, let us present the following two auxiliary lemmas.

\begin{lemma}
\label{productmanifoldsmetric}
Let $(M, g)$ and $(N, h)$ be two Riemannian manifolds and let $M \times N$ be the product manifold endowed with the product metric $g \oplus h$. Let $X, Y \in \mathfrak{X}(M\times N)$ be two decomposable vector fields, i.e. $X = (X_1, X_2)$, and $Y = (Y_1, Y_2)$, with $X_1, Y_1 \in \mathfrak{X}(M)$, $X_2, Y_2 \in \mathfrak{X}(N)$. Then  $(g \oplus h)(X, Y)$ verifies 
\[
(g \oplus h)(X, Y) = g(X_1, Y_1) + h(X_2, Y_2),
\]
and
\[
[X, Y]_{M \times N} = [X_1, Y_1]_M + [X_2, Y_2]_N,
\]
where $[\cdot, \cdot]_{M \times N}$, $[\cdot, \cdot]_{M}$, and $[\cdot, \cdot]_{N}$ denote the Lie derivative in $M\times N$, $M$ and $N$, respectively, and the sum is written in the sense of $T_{(x_1, x_2)} (M\times N) = T_{x_1} M \oplus T_{x_2} N$. 

Furthermore, (c.f. \cite[Exercise 1.a, chapter 6]{doCarmoriemannian}), using the same sum notation for $T_{(x_1, x_2)} (M\times N)$, 
\[
\nabla_X Y = \nabla^M_{X_1} Y_1 + \nabla^N_{X_2} Y_2,
\]
where $\nabla$, $\nabla^M$ and $\nabla^N$ denote the Levi-Civita connection of $M \times N$, $M$ and $N$, respectively. 
\end{lemma}

Using \cref{productmanifoldsmetric} it is straightforward to see how the Riemann curvature tensor of a product is related to those of its factors. 
\begin{lemma}
\label{RiemannProd}
Let $(M, g)$ and $(N, h)$ be two Riemannian manifolds and let $M \times N$ be the product manifold endowed with the product metric. Then  for every point $(x, y) \in M \times N$ and every four tangent vectors $X, Y, Z, W \in T_{(x, y)} M \times N$,
\[
R_{g \oplus h}(X, Y, Z, W) = R_{g}(X_1, Y_1, Z_1, W_1) + R_{h}(X_2, Y_2, Z_2, W_2),
\]
where $R_{g \oplus h}$ denotes the Riemann curvature tensor of $M \times N$, and $R_{g}$, $R_{h}$ denote the Riemann curvature tensors of $M$ and $N$ respectively. We have used the same notation as in \cref{productmanifoldsmetric} for the decomposition of the tangent vectors.
\end{lemma}

With the two above lemmas in mind, it is now straightforward to prove the following statement. 

\begin{proposition}[Curvature of the product]
\label{curvProduct}
Let $(M, g)$, $(N, h)$ be two Riemannian manifolds of dimension $d_M$ and $d_N$, respectively, verifying
\[
K_g \leq D_1,\quad K_h \leq D_2,
\]
\[
\textup{Ric}_g \geq \delta_1 \,g,\quad \textup{Ric}_h \geq \delta_2\, h,
\]
for some constants $D_1, D_2, \delta_1, \delta_2 \in \bb{R}$, where $K_g$ (resp. $K_h$) denotes the sectional curvature of $M$ (resp. $N$) and $\textup{Ric}_g$ (resp. $\textup{Ric}_h$) denotes the Ricci curvature of $M$ (resp. $N$). Assume that for every $x \in M$ and every $y \in N$, the coefficients of the Riemann curvature tensors of $M$ and $N$, denoted as $R_g$ and $R_h$ respectively are such that for any unit vectors $X_1, Y_1, Z_1, W_1 \in T_x M$, $X_2, Y_2, Z_2, W_2 \in T_y N$,
\[
|R_g(X_1, Y_1, Z_1, W_1)| \leq \Lambda_1, \quad |R_h(X_2,Y_2, Z_2, W_2)| \leq \Lambda_2,
\]
for every $i,j,k,l \in \{1, \dotsc, d_M\}$ and every $m,n,p,q \in \{1, \dotsc, d_N\}$, for some constants $\Lambda_1, \Lambda_2 \in \bb{R}$. Then,  considering $M \times N$ endowed with the product metric, it holds that
\[
K_{g \oplus h} \leq \max\{0, D_1, D_2\},\quad \textup{Ric}_{g \oplus h} \geq \min\{\delta_1, \delta_2 \}\, (g \oplus h),
\]
and, for every $(x,y) \in M \times N$, and any unitary vectors $X, Y, Z, W \in T_{(x, y)} M \times N$, 
\[
|R_{g \oplus h}(X,Y,Z,W)| \leq \max\{\Lambda_1, \Lambda_2\},
\]
for every $i, j, k, l \in \{1, \dotsc, d_M + d_N\}$.
\end{proposition}
\begin{proof}
From the previous Lemma, we know that for any point $(x, y) \in M \times N$ and any four tangent vectors $X, Y, Z, W \in T_{(x, y)} M \times N$, it holds that
\[
R_{g \oplus h}(X, Y, Z, W) = R_g(X_1, Y_1, Z_1, W_1) + R_h(X_2, Y_2, Z_2, W_2),
\]
where we are using the same notation as in \cref{productmanifoldsmetric} for the decomposition of the tangent vectors. If either of the components is zero, the result follows automatically. Otherwise, we can write 
\begin{align*}
R_{g \oplus h}(X, Y, Z, W) &= R_g(X_1, Y_1, Z_1, W_1) + R_h(X_2, Y_2, Z_2, W_2)
\\&= a_1 b_1 c_1 d_1 R_g(\tilde X_1, \tilde Y_1, \tilde Z_1, \tilde W_1) + a_2 b_2c_2d_2 R_h(\tilde X_2, \tilde Y_2, \tilde Z_2, \tilde W_2),
\end{align*}
where $a_i, b_i, c_i, d_i$ correspond to the norms of $X_i, Y_i, Z_i, W_i$, respectively, with respect to $g$ and $h$, and $\tilde X_i, \tilde Y_i,\tilde Z_i,\tilde W_i$ correspond to the normalized vectors $X_i, Y_i, Z_i, W_i$, respectively. Thus
\begin{align*}
|R_{g \oplus h}(X, Y, Z, W)| &= a_1 b_1 c_1 d_1 |R_g(\tilde X_1, \tilde Y_1, \tilde Z_1, \tilde W_1)| + a_2 b_2c_2d_2 |R_h(\tilde X_2, \tilde Y_2, \tilde Z_2, \tilde W_2)|
\\&\leq \max\{\Lambda_1, \Lambda_2\} (a_1b_1c_1d_1 + a_2b_2c_2d_2)
\\&\leq \max\{\Lambda_1, \Lambda_2\} (a_1b_1 + a_2b_2)
\\&\leq \max\{\Lambda_1, \Lambda_2\} (a_1^2 + a_2^2)^{1/2}(b_1^2 + b^2_2)^{1/2}
\\&\leq \max\{\Lambda_1, \Lambda_2\},
\end{align*}
where we used that $0 < c_i, d_i < 1$ for $i \in \{1, 2\}$ and $a_1^2 + a_2^2 = b_1^2 + b_2^2 = 1$.

Now, let $\Pi$ be a plane in $T_{(x, y)} M \times N$. Let $\Pi = \mathit{span}\{X, Y\}$, for some orthonormal vectors $X, Y \in T_{(x,y)} M \times N$. We have that 
\begin{align*}
K_{g \otimes h}(\Pi) &= R_g(X_1,Y_1,X_1,Y_1) + R_h(X_2,Y_2,X_2,Y_2) 
\\&= K_g(X_1,Y_1) |X_1 \wedge Y_1|^2 + K_h(X_2,Y_2) |X_2 \wedge Y_2|^2,
\end{align*}
where $|X_1 \wedge Y_1|^2 := |X_1|_g^2|Y_1|_g^2 - g(X_1,Y_1)^2$ and $|X_2 \wedge Y_2|^2$ is defined analogously. Thus, it follows that
\[
K_{g \otimes h}(\Pi) \leq D_1 |X_1 \wedge Y_1|^2 + D_2 |X_2 \wedge Y_2|^2 \leq \max\{0, D_1, D_2\}. 
\]

For the last inequality, let $X \in T_{(x, y)} M \times N$. Then 
\begin{align*}
\textup{Ric}_{g \oplus h}(X, X) &= \sum_{i = 1}^{d_M+d_N} R_{g \oplus h}(X, e_i, X, e_i) \\&= \sum_{i = 1}^{d_M} R_g(X_1, e_i, X_1, e_i) + \sum_{i = d_M+1}^{d_M+d_N} R_h(X_2, e_i, X_2, e_i)
\\&= \textup{Ric}_g(X_1, X_1) + \textup{Ric}_h(X_2, X_2),
\end{align*}
where $\{e_i\}_{i = 1}^{d_M+d_N}$ is an orthonormal basis for $T_{(x, y)} M \times N$, which can be divided into an orthonormal basis for $T_x M$, and an orthonormal basis for $T_y N$, namely $\{e_i\}_{i = 1}^{d_M}$ and $\{e_i\}_{i = d_M+1}^{d_M+d_N}$, respectively. 
\end{proof}

There is one particular family of manifolds which have convenient curvature properties. These manifolds are known as Einstein manifolds.

\begin{definition}[Einstein manifold]
We say a manifold $M$ is \textup{Einstein} if
\[
\textup{Ric}_g = \lambda g,
\]
for some constant $\lambda \in \bb{R}$.
\end{definition}

\begin{proposition}[{\cite[Proposition 1.99]{besse2007einstein}}]
\label{productEinstein}
The product of two Riemannian manifolds which are Einstein with the same constant $\lambda$ is again Einstein with the same constant when endowed with the product metric.
\end{proposition}

\subsection{Injectivity and convexity radii}
\label{secinjectivity}

Before we study Riemannian submersions, let us first define the concept of injectivity and convexity radii of a Riemannian manifold. In fact, such quantities play a crucial role in the proof of the main results from \cref{SectionSuboptimality,SectionPI}. Having lower bounds for such quantities ultimately allows us to control the Poincaré and the log-Sobolev inequality constants studied in \cref{SectionPI,SectionPItoLSI}. 

In order to define the injectivity radius of a Riemannian manifold, we need to introduce the \textit{exponential map} and its inverse. 
\begin{definition}[Exponential map, logarithm]
\label{def:exponentialmap}
Let $(M, g)$ be a compact\footnote{We require $M$ to be compact for simplicity, in order to be able to define $\exp_x$ on $T_xM$, as in the general case the exponential map is only defined on a subset of $T_x M$, unless $M$ is complete.} Riemannian manifold. For every $p \in M$, the \textup{exponential map} at $p$, $\exp_p : T_p M \to M$, is defined as 
\[
\exp_p(v) := \gamma_{p, v}(1), 
\]
where $\gamma_{p, v}$ is the unique geodesic in $M$ such that $\gamma(0) = p$, $\gamma'(0) = v$. 

Whenever the exponential map is invertible, one can define its inverse, known as the \textit{logarithm} at $p$, denoted as $\log_p : U \to T_pM$, where $U \subset M$ is its domain.  
\end{definition}

\begin{remark}[Normal coordinates]
\label{def:normalcoordinates}
On any connected neighborhood of $0 \in T_xM$ on which $\exp_x$ is a diffeomorphism onto its image, the exponential map induces local coordinates in a natural way. These coordinates are known as \textup{normal coordinates centered at $x$} (see \cite[Chapter 5]{lee2018introductionRiemannian}). 
\end{remark}

\begin{definition}[Injectivity Radius]
\label{def:injectivityradius}
Let $(M, g)$ be a Riemannian manifold. The \textup{injectivity radius} of $M$, denoted by $i(M)$, is the greatest radius $r > 0$ for which the exponential map defined on the ball of radius $r$ centered at the origin of $T_x M$, denoted $\mathcal{B}_{T_xM}(r,0)$, is a diffeomorphism onto its image for all $x \in M$;
\[
i(M) := \inf_{\vphantom{\tilde{M}}x \in M} \sup \{r > 0 : \exp_x |_{\mathcal{B}_{T_xM}(r,0)} \text{ is a diffeomorphism onto its image}\}.
\]
\end{definition}

Bounding the injectivity radius is not an easy task in general. Nevertheless, for compact manifolds with upper bounded sectional curvature, it becomes easier.

\begin{proposition}[{\cite{Klingenberg}}]
\label{InjectivityRadiusBound}
Let $(M, g)$ be a compact Riemannian manifold, and let $K_g$ denote its sectional curvature. Assume that there exists a constant $D > 0$ such that $K_g \leq D$. Then
\[
i(M) \geq \min \left\{\frac{\pi}{\sqrt{D}}, \frac{1}{2}l(M)\right\},
\]
where $l(M)$ is the length of the shortest non-trivial periodic geodesic in $M$.
\end{proposition}

Before defining the convexity radius of a Riemannian manifold, let us define the \textit{cut locus} of a point in a Riemannian manifold, which is closely related to its injectivity radius. 

\begin{definition}[Cut locus]
Let $(M, g)$ be a complete and connected Riemannian manifold, and let $p \in M$, $v \in T_p M$. We define the \textup{cut time} associated with $p$ and $v$ as the maximum time $t$ for which the geodesic $\gamma : [0, t] \to M$, $\gamma_{(p, v)}(s) := \exp_p(sv)$ minimizes the distance between $p$ and $\exp_p(sv)$ for every $s \in [0, t]$,
\[
t_{\mathit{cut}}(p, v) := \sup \{t > 0 : \left.\gamma_{(p, v)}\right|_{[0, t]} \textup{ is minimizing}\}. 
\]
When $t_{\mathit{cut}}(p, v) < \infty$, we denote $\gamma_{(p, v)}(t_{\mathit{cut}}(p, v))$ as the \textup{cut point} of $p$ along $\gamma_{(p, v)}$.

This way, we define the \textup{cut locus} of $p \in M$ as the set of cut points of $p$ along all possible geodesics,
\[
C(p) := \{\gamma_{(p, v)}(t_{\mathit{cut}}(p, v)) : t_{\mathit{cut}}(p, v) < \infty, \ v \in T_p M, |v|_g = 1\}. 
\]
\end{definition}

From the definition of the cut locus, given $p \in M$ and $q \in C(p)$, it must be the case that $d_g(p, q) \geq i(M)$.

The next quantity of interest is the convexity radius, which is defined as the greatest radius for which any ball of such radius centered anywhere in the manifold is \textit{geodesically convex}. 

\begin{definition}
Given a manifold $M$ a subset $U \subset M$ is said to be \textup{geodesically convex}---or simply \textup{convex}---if, for any two points $x$ and $y$ in $U$, there exists a unique length-minimizing geodesic in $M$ connecting $x$ and $y$ which lies entirely in $U$. 
\end{definition}

\begin{definition}[Convexity radius]
The \textup{convexity radius} of a Riemannian manifold, denoted by $\mathit{conv}(M)$, is the greatest radius $r > 0$ for which the ball $\mathcal{B}(r,x)$ is geodesically convex for all $x \in M$; 
\[
\mathit{conv}(M) := \inf_{\vphantom{\tilde{M}} x \in M}\sup\{r > 0 : \mathcal{B}(s,x) \text{ is geodesically convex for every } 0 < s < r\}.  
\]
\end{definition}

Again, as for the injectivity radius, controlling the convexity radius is in general a complicated task, which becomes easier for compact manifolds with upper bounded sectional curvature.

\begin{proposition}[{\cite[Proposition 95]{berger2007panoramic}}]
\label{controlconvexityradius}
Let $(M, g)$ be a compact Riemannian manifold, and let $K_g$ denote its sectional curvature. Assume that there exists a constant $D > 0$ such that $K_g \leq D$, then 
\[
\mathit{conv}(M) \geq \frac{1}{2} \min\left\{\frac{\pi}{\sqrt{D}}, \frac{1}{2}l(M)\right\},
\]
where again $l(M)$ is the length of the shortest non-trivial periodic geodesic in $M$. 
\end{proposition}

\cref{InjectivityRadiusBound,controlconvexityradius} show that, in order to obtain lower bounds on the injectivity and convexity radii of a compact manifold, it is convenient to bound its sectional curvature and to control the length of its shortest non-trivial periodic geodesic.

\subsection{Riemannian submersions}
\label{chapterRiemannianSubmersions}
In this subsection we will gather some of the most basic definitions and results in the field of Riemannian submersions. For a more thorough analysis we refer the reader to classic references such as \cite{lee2018introductionRiemannian} or O'Neill's articles \cite{oneill1966fundamental, oneil1967submersions}.

\begin{definition}[Submersion]
Let $M$ and $B$ be two smooth manifolds, and let $F : M \rightarrow B$ be a smooth map. We say that $F$ is a \textup{submersion} if its differential $F_*|_p$ is surjective for all $p \in M$. $M$ and $B$ are often referred to as the \textup{total} and \textup{base} spaces, respectively. 
\end{definition}

Given a submersion, the tangent space of the total space at each point can be split into its \textit{vertical} and \textit{horizontal} components. 

\begin{definition}[Vertical and horizontal tangent spaces]
\label{verticalhorizontaltangent}
Given a submersion $\pi: M \rightarrow B$, we define the \textup{vertical tangent space} at $p \in M$ as $V_p M := \ker \pi_*|_p$. Note that $V_p M$ is a vector space of dimension $\dim(M) - \dim(B)$. This way, if $M$ is a Riemannian manifold, we can define the \textup{horizontal tangent space} as $H_p M = (V_p M)^{\bot}$. Thus,
\[
T_p M = V_p M \oplus H_p M,
\]
for every $p \in M$. 
\end{definition}

\begin{notation}
\label{notverticalhorizontal}
If $\pi: M \to B$ is a submersion and $X$ is a tangent vector at $x \in M$, $X \in T_x M$, we will denote the horizontal and vertical components of $X$ by $X^\mathcal{H}$ and $X^\mathcal{V}$, respectively. This way, 
\[
X = X^\mathcal{V} + X^\mathcal{H} \in V_x M \oplus H_x M.
\]
\end{notation}

Next, we will introduce \textit{Riemannian submersions}, which are submersions between manifolds that \textit{preserve} the metric between the total and the base space.

\begin{definition}[Riemannian Submersion]
Let $(M, g)$ and $(B, h)$ be two Riemannian manifolds, and let $\pi: M \rightarrow B$ be a submersion. We say $\pi$ is a \textup{Riemannian submersion} if for every $p \in M$, $\pi_*|_p$ maps $H_p M$ isometrically onto $T_{\pi(p)} B$. That is, for each $p \in M$ and $X, Y \in H_p M$, 
\[
\langle X, Y \rangle_g = \langle \pi_*|_p X,\pi_*|_p Y \rangle_h.
\]
\end{definition}

In particular, given a Riemannian submersion $\pi: (M, g) \to (B, h)$ the metric $g$ splits as 
\[
g = g^\mathcal{H} + g^\mathcal{V},
\]
where $g^\mathcal{H}$ and $g^\mathcal{V}$ correspond to the restrictions of $g$ on the horizontal and the vertical tangent spaces, respectively.

When given a Riemannian submersion between two Riemannian manifolds, one can relate the geodesics of the base and the total space using the following result. 

\begin{proposition}[{\cite[Corollary 2]{oneil1967submersions}}]
\label{LemmaHorizontalGeodesics}
Let $\pi: (M, g) \rightarrow (B, h)$ be a Riemannian submersion. If $\gamma$ is a geodesic on $M$ that is horizontal at some point $\gamma(t_0)$ ( i.e. $\gamma'(t_0) \in H_{\gamma(t_0)} M$), then $\gamma(t)$ is horizontal for every $t$ in the domain of $\gamma$, and so $\pi \circ \gamma$ is a geodesic of $B$. If $M$ is complete, every geodesic $ \tilde\gamma$ in $B$ is the projection of a horizontal geodesic $\gamma$ in $M$, i.e. $\tilde \gamma = \pi \circ \gamma$. 
\end{proposition}

One simple way to construct surjective Riemannian submersions given a Riemannian manifold is by considering quotient maps with respect to group actions. Therefore, we will first introduce some basic definitions about group actions on manifolds and then proceed to state the main results, which guarantee that certain group actions induce Riemannian submersions. 

\begin{definition}[Group action on a manifold]
Let $G$ be a group and $M$ be a manifold. A \textup{left action} of $G$ on $M$  is a map $G \times M \to M$, usually denoted as $(x, p) \mapsto x\cdot p$ verifying the following: 
\begin{itemize}
    \item $x_2 \cdot (x_1 \cdot p) = (x_2x_1)\cdot p$ for every $x_1, x_2 \in G$, $p \in M$.
    \item $e \cdot p = p$ for every $p \in M$, $e$ being the unit of $G$. 
\end{itemize}
Right actions are defined analogously. 
\end{definition}

Although group actions often give rise to Riemannian submersions, not every group action does. Nevertheless, if the group action fulfills certain properties, this is the case. 

\begin{definition}[Free action]
A group action of $G$ on a manifold $M$ is \textup{free} if for every $p \in M$, 
\[
\{x \in G : x p = p\} = \{e\},
\]
where $e$ denotes the unit of $G$. 
\end{definition}

\begin{definition}[Smooth action]
The action of a Lie group $G$ on a manifold $M$ is \textup{smooth} if the map $G \times M \rightarrow M$ given by $(x, p)\mapsto x \cdot p$ is smooth. 
\end{definition}

\begin{definition}[Proper action]
A group action of $G$ on a manifold $M$ is said to be \textup{proper} if the map 
\begin{align*}
G \times M &\rightarrow M \times M\\
(x, p) &\mapsto (x\cdot p, p)    
\end{align*} 
is proper, i.e., the preimage of every compact set is compact.
\end{definition}

In particular, every smooth action by a compact Lie group is proper (cf. \cite[Corollary C.16]{lee2018introductionRiemannian}). 

\begin{definition}[Isometric action]
A group action of $G$ on a Riemannian manifold $(M, g)$ is said to be \textup{isometric} if, for every $x \in G$, the map 
\begin{align*}
\alpha_x : M &\to M\\
p &\mapsto x \cdot p
\end{align*}
is an isometry on $(M, g)$ (i.e. a diffeomorphism such that for every point $p \in M$, its differential satisfies that $g(u, v) = g(d\alpha_x(u), d\alpha_x(v))$ for every $u, v \in T_p M$). 
\end{definition}

When a group action on a manifold satisfies all of the above conditions, the quotient map from the manifold to the orbit space is a Riemannian submersion. 

\begin{proposition}[{\cite[Corollary 2.29]{lee2018introductionRiemannian}}]
\label{existenceRiemannianSubmersionMetrics}
Suppose $(M, g)$ is a Riemannian manifold, and let $G$ be a Lie group acting smoothly, freely, properly and isometrically on $M$. Then  the orbit space $M/G$ has a unique smooth manifold structure and Riemannian metric such that $\pi : M \rightarrow M/G$ is a Riemannian submersion.
\end{proposition}

The Riemannian submersions that arise as the orbit space of a suitable group action are known as $G$-bundles. Let us briefly introduce them, as they are discussed in \cref{SectionLift}. For a more in-depth study of fiber bundles or $G$-bundles, we refer to \cite{cohen1998topology}. 

\begin{definition}
Let $M$ and $B$ be two connected manifolds. We say that $\pi: M \to B$ is a \textup{locally trivial fibration, or fiber bundle with fiber $F$} if it satisfies all of the following conditions:
\begin{itemize}
    \item $\pi^{-1}(x) \simeq F$ for every $x \in B$.
    \item $\pi$ is surjective. 
    \item For every $x \in B$ there exists an open neighborhood $U \subset B$ and a homeomorphism---also known as a local trivialization---$\Phi: \pi^{-1}(U) \to U \times F$ such that the following diagram commutes:
\[\begin{tikzcd}
	{\pi^{-1}(U)} && {U \times F} \\
	\\
	& U
	\arrow["\Phi", from=1-1, to=1-3]
	\arrow["\pi"', from=1-1, to=3-2]
	\arrow["{\pi_1}", from=1-3, to=3-2]
\end{tikzcd},\]
where $\pi_1$ denotes the projection onto the first component.
\end{itemize}
\end{definition}

\begin{definition}
Let $G$ be a Lie group. A \textup{principal $G$-bundle} is a fiber bundle $\pi: M \to B$ with fiber $G$ satisfying all of the following conditions: 
\begin{itemize}
    \item There exists a free action $\mu$ of $G$ on $M$ such that 
    \[\begin{tikzcd}
	{M \times G} &&& M \\
	\\
	{B \times\{e\}} &&& B
	\arrow["\mu", from=1-1, to=1-4]
	\arrow["{\pi \times \epsilon}"', from=1-1, to=3-1]
	\arrow["\pi", from=1-4, to=3-4]
	\arrow["{=}"{description}, draw=none, from=3-1, to=3-4]
\end{tikzcd}\]
commutes, where $\epsilon(g) = e$ for every $g \in G$, and $e$ denotes the unit of $G$. 
    \item The induced action on the fibers
    \[
    \mu|_{\pi^{-1}(x)}: \pi^{-1}(x) \times G \to \pi^{-1}(x),
    \]
    is free and \textit{transitive}---i.e. for any two points $x_1, x_2 \in \pi^{-1}(x)$ there exists some $g \in G$ such that $x_1 = \mu(x_2,g)$. 
    \item  There exist local trivializations $\Phi: \pi^{-1}(U) \to U \times G$ such that the following diagrams commute:
    \[\begin{tikzcd}
	{\pi^{-1}(U) \times G} &&& {U \times G \times G} \\
	\\
	{\pi^{-1}(U)} &&& {U \times G}
	\arrow["{\Phi \times \textup{Id}}", from=1-1, to=1-4]
	\arrow["\mu"', from=1-1, to=3-1]
	\arrow["{\textup{Id} \times \cdot}", from=1-4, to=3-4]
	\arrow["\Phi", from=3-1, to=3-4]
\end{tikzcd},\]
where $\cdot$ denotes the usual group multiplication of $G$. 
\end{itemize}
\end{definition}

As we mentioned earlier, whenever the assumptions described in \cref{existenceRiemannianSubmersionMetrics} hold, the Riemannian submersion is also a principal $G$-bundle. 
\begin{proposition}[{\cite[Theorem 1.3]{cohen1998topology}}]
Let $M$ be a manifold and let $G$ be a compact Lie group acting freely and smoothly on $M$. Then  
\[
\pi: M \to M/G
\]
is a principal $G$-bundle. 
\end{proposition}

The Riemannian submersions that we consider in this text satisfy one more property; all of the fibers of $\pi$ are \textit{totally geodesic}. 

\begin{definition}[Totally geodesic submanifold]
A Riemannian submanifold $(\tilde{M}, \tilde{g})$ of $(M, g)$ is said to be \textup{totally geodesic} if every $g$-geodesic that is tangent to $\tilde{M}$ at some point $t_0$ stays in $\tilde{M}$ for all $t$ in some interval $(t_0-\varepsilon, t_0+\varepsilon)$, $\varepsilon > 0$.
\end{definition}

The definition of a totally geodesic submanifold can also be understood in the terms of the following proposition. 

\begin{proposition}[{\cite[Proposition 8.12]{lee2018introductionRiemannian}}]
\label{geodesicstotallygeodesic}
Let $(\tilde{M}, \tilde{g})$ be an embedded Riemannian submanifold of a Riemannian manifold $(M, g)$. If $\tilde{M}$ is totally geodesic in $M$, then every $\tilde{g}$-geodesic in $\tilde{M}$ is also a $g$-geodesic in $M$.
\end{proposition}

\begin{remark}
\label{rmk:totallygeodesic}
In particular, given a Lie group $G$ endowed with a bi-invariant metric $g$ and some closed subgroup $H \subset G$, it holds that there exists a metric $h$ on $G/H$ such that $\pi: (G, g) \mapsto (G/H, h)$ is a Riemannian submersion with totally geodesic fibers \cite[Section 5]{oneill1966fundamental}. 
\end{remark}

\subsubsection{Curvature of Riemannian submersions}
\label{curvatureriemanniansub}
Let us now examine the relationship between the curvature of the total and the base space of a Riemannian submersion. For the purpose of this work, we will only define the so-called O'Neill tensor $T$. For an in-depth study of the O'Neill tensors, see \cite{oneill1966fundamental}.
\begin{definition}[O'Neill tensor]
\label{defOneilTensors}
Let $\pi : M \to B$ be a Riemannian submersion. For any two vector fields $E, F \in \mathfrak{X}(M)$, we define the \textup{O'Neill tensor} $T$ as
\[
T_E F := (\nabla_{E^\mathcal{V}} (F^\mathcal{V}))^\mathcal{H} + (\nabla_{E^\mathcal{V}} (F^\mathcal{H}))^\mathcal{V},
\]
where $\cdot ^\mathcal{V}$ and $\cdot^{H}$ denote the vertical and horizontal parts of a vector field, respectively.
\end{definition}

The tensor $T$ allows us to relate the sectional curvature of the base space and the fibers of a Riemannian submersion with that of the total space.
\begin{lemma}[{\cite[Corollary 1]{oneill1966fundamental}}]
\label{CurvaturaSubmersion}
Let $\pi : (M, g) \rightarrow (B, h)$ be a Riemannian submersion. Given two horizontal orthonormal vectors $X, Y \in H_xM$ and two vertical orthonormal vectors $V, W \in V_xM$ at some point $x \in M$, it holds that
\[
K_g(X, Y) = K^*_h(\pi_* X, \pi_* Y) - \frac{3}{4}\left|[X, Y]^\mathcal{V}\right|_g^2,
\]
and
\[
K_g(V, W) = \hat{K}(V, W) - \langle T_V V, T_W W\rangle_g + |T_V W|_g^2,
\]
where $K_g$ is the sectional curvature of $M$, $\hat{K}$ is the sectional curvature of the fibers of $\pi$, $K^*_h$ is the sectional curvature of $B$, and $T$ is the O'Neill tensor.
\end{lemma}

Notice that, while the sectional curvature takes a pair of tangent vectors as arguments, both the Lie derivative and the O'Neill tensor take two vector fields as arguments. Nevertheless, via a slight abuse of notation, the expressions are well-defined. 

\begin{remark}
\label{remarkcurvaturasubmersion}
The expressions shown in \cref{CurvaturaSubmersion} are well-defined, as for any two horizontal vector fields $X, Y \in \mathfrak{X}(M)$, and $p \in M$
\[
\left.\left([X, Y]^\mathcal{V}\right)\right|_p
\]
only depends on $X_p$ and $Y_p$. Similarly, for any two vertical vector fields $U, V \in \mathfrak{X}(M)$, and $p \in M$
\[
\left.\left(T_U V\right)\right|_p
\]
only depends on $U_p$ and $V_p$. 
\end{remark}
\begin{proof}
Choosing an adapted orthonormal frame on a neighborhood of $p$, we can write
\[
X = X^i(x) \left.\frac{\partial}{\partial h^i}\right|_x,\quad Y = Y^i(x) \left.\frac{\partial}{\partial h^i}\right|_x,
\]
where $\{\left.\frac{\partial}{\partial h^i}\right|_x\}$ form an orthonormal basis of $H_x M$, for every $x$ in a neighborhood of $p$. Using this frame, 
\[
\left.[X, Y]\right|_p = \sum_{i, j} X^i(p) Y^j(p)\Big[\frac{\partial}{\partial h^i}\bigg|_p,\, \frac{\partial}{\partial h^j}\bigg|_p\Big] + \sum_{j} X(Y^j)(p)\frac{\partial}{\partial h^j}\bigg|_p  - \sum_{i} Y(X^i)(p) \frac{\partial}{\partial h^i}\bigg|_p.
\]
Note that the two last summands in this expression are horizontal, thus
\[
\left(\left.[X, Y]\right|_p\right)^\mathcal{V} = \left(\sum_{i, j} X^i(p) Y^j(p)\Big[\frac{\partial}{\partial h^i}\bigg|_p, \frac{\partial}{\partial h^j}\bigg|_p\Big]\right)^\mathcal{V},
\]
which only depends on the point $p$. 

Working in the same frame, and using the definition of $T$
\[
(T_U V)|_p = \left.(\nabla_{U} V)^\mathcal{H}\right|_p = (U(V^k))(p) + U^i(p) V^j(p) \Gamma^k_{ij}(p)) \left.\frac{\partial}{\partial h^k}\right|_p.
\]
Note that, since $V$ is vertical, $V^k(x) \equiv 0$ for every $k$ (we are only summing over the horizontal components). Thus 
\[
(T_U V)|_p = (U^i(p) V^j(p) \Gamma^k_{ij}(p)) \left.\frac{\partial}{\partial h^k}\right|_p,
\]
which only depends on $p$.
\end{proof}

Let us conclude with a useful result, which can be seen in \cite{escobales1975riemannian}, and relates the O'Neill tensor $T$ to the fibers of a Riemannian submersion being totally geodesic.
\begin{lemma}
\label{vanishingT}
Let $\pi: (M, g) \to (B, h)$ be a Riemannian submersion. $\pi$ has totally geodesic fibers if and only if the O'Neill tensor $T$ from \cref{defOneilTensors} vanishes identically.
\end{lemma}
This result allows us to conclude that, whenever a Riemannian submersion has totally geodesic fibers, the sectional curvature of the fibers is the same as that of the total space. 

Moreover, using \cref{vanishingT}, we can guarantee that certain \textit{nested} submersions \textit{inherit} totally geodesic fibers. 
\begin{proposition}
\label{prop:inducedtotgeodfibers}
Let $(M_1, g_1)$, $(M_2, g_2)$ and $(M_3, g_3)$ be three Riemannian manifolds, and let $\pi_1, \pi_2$ and $\pi_3$ be three Riemannian submersions,
\begin{align*}
\pi_1 : (M_1, g_1) \to (M_2, g_2),\\
\pi_2: (M_2, g_2) \to (M_3, g_3),\\
\pi_3: (M_1, g_1) \to (M_3, g_3),
\end{align*}
such that $\pi_3 = \pi_2 \circ \pi_1$. If $\pi_3$ has totally geodesic fibers, then $\pi_2$ has totally geodesic fibers. 
\end{proposition}
\begin{proof}
We want to show that the O'Neill tensor $T$ associated with $\pi_2$ vanishes identically. To do so, we will show that for any two vertical vector fields  with respect to $\pi_2$, $U, V \in \mathfrak{X}(M_2)$, it holds that 
\begin{equation}
\label{eq:vanishingverticalcon}
(\nabla^{M_2}_U V)^{\mathcal{H}} = 0.    
\end{equation}
Should this hold, the result would follow: using the compatibility of the Levi-Civita connection $\nabla^{M_2}$, we know that for every two vertical vector fields with respect to $\pi_2$, $U, V \in \mathfrak{X}(M_2)$ and any horizontal vector field with respect to $\pi_2$, $X \in \mathfrak{X}(M_2)$, 
\[
0 = U g_2(V, X) = g_2(\nabla^{M_2}_U V, X) + g_2(V, \nabla^{M_2}_U X) = g_2(V, \nabla^{M_2}_U X),
\]
where the first equality follows from $V$ being vertical and $X$ horizontal, and the last equality follows from the assumption that \cref{eq:vanishingverticalcon} holds. This allows us to conclude that $\nabla^{M_2}_U X$ is always horizontal, and so $T$ vanishes (cf. \cref{defOneilTensors}). 

Let us then prove that \cref{eq:vanishingverticalcon} holds. Given $U, V \in \mathfrak{X}(M_2)$---which are vertical with respect to $\pi_2$ by assumption---we can \textit{lift} them to $M_1$, and obtain two basic lifts with respect to $\pi_1$, $\tilde U, \tilde V \in \mathfrak{X}(M_1)$ which are such that $d\pi_1(\tilde U) = U$ and $d\pi_1(\tilde V) = V$. Note that $\tilde U$ and $\tilde V$ are vertical with respect to $\pi_3$, 
\begin{align*}
d\pi_3(\tilde U) = d\pi_2(d\pi_1(\tilde U)) = d\pi_2(U) = 0,\\
d\pi_3(\tilde V) = d\pi_2(d\pi_1(\tilde V)) = d\pi_2(V) = 0,
\end{align*}
where the last equality follows from the assumption that $U$ and $V$ are vertical with respect to $\pi_2$. Now, from $\pi_3$ having totally geodesic fibers, we know that $\nabla^{M_1}_{\tilde U} \tilde V$ is vertical with respect to $\pi_3$. Moreover, it is known \cite[Lemma 1]{oneill1966fundamental} that $\nabla^{M_1}_{\tilde U} \tilde V$ is such that $d\pi_1(\nabla^{M_1}_{\tilde U} \tilde V) = \nabla^{M_2}_U V$, and so
\[
0 = d\pi_3(\nabla^{M_1}_{\tilde U} \tilde V) = d\pi_2(d\pi_1(\nabla^{M_1}_{\tilde U} \tilde V)) = d\pi_2(\nabla^{M_2}_U V),
\]
finishing the proof.
\end{proof}

\subsubsection{Curvature of complex Stiefel manifolds}
\label{curvaturestiefel}

\begin{definition}[Stiefel manifold]
Given $k, n \in \bb{N}$ with $1 \leq k \leq n$, the set of linear isometries from $\bb{C}^k$ to $\bb{C}^n$ is a smooth manifold known as the \textup{complex Stiefel manifold} (cf. \cite{autenried2014sub}),
\[
\textup{V}_k(\bb{C}^n) := \{X \in \bb{C}^{n \times k} : X^\dagger X = \mathds{1}\}.
\]
\end{definition}

We will restrict our attention to the case when $1 < k < n$, as $\textup{V}_n(\bb{C}^n) = \textup{U}(n)$, and $\textup{V}_1(\bb{C}^n)  = \bb{S}^{2n-1}$. These manifolds can be seen as the quotient space $\textup{U}(n)/\textup{U}(n-k)$, where $\textup{U}(n-k)$ acts on $\textup{U}(n)$ as 
\begin{align}
\label{eq:defactionstiefel}
\begin{split}
\textup{U}(n-k) \times \textup{U}(n) &\to \textup{U}(n)\\
(X, U) &\mapsto U\begin{pmatrix}
\mathds{1}_k & 0\\
0 & X
\end{pmatrix}.
\end{split}
\end{align}
It is easy to show that this action is free: let $U \in \textup{U}(n)$ and $X \in \textup{U}(n-k)$. If
\[
U\begin{pmatrix}
\mathds{1}_k & 0\\
0 & X
\end{pmatrix} = U,
\]
then 
\[
\begin{pmatrix}
\mathds{1}_k & 0\\
0 & X
\end{pmatrix} = \mathds{1}_n,
\]
and so it must be the case that $X = \mathds{1}_{n-k}$. Moreover, the action is smooth and isometric whenever $\textup{U}(n)$ is endowed with its bi-invariant metric $g$. Therefore, by \cref{existenceRiemannianSubmersionMetrics} we know that there exists a metric $h$ on $\textup{V}_k(\bb{C}^n)$ for which the quotient map 
\[
\pi: (\textup{U}(n), g) \to (\textup{V}_k(\bb{C}^n), h)
\]
is a Riemannian submersion. Furthermore, as we discussed in \cref{chapterRiemannianSubmersions}, the submersion has totally geodesic fibers. 

In order to control the curvature of the complex Stiefel manifolds, we will need an auxiliary result, which will allow us to \textit{lift} vector fields in a convenient way. 

\begin{proposition}
\label{horizontalleftinvariant}
Let $G$ be a Lie group endowed with a bi-invariant metric and let $H$ be a closed subgroup of $G$. Consider the action of $H$ on $G$ from the right:
\begin{align*}
H \times G \to G\\
(h, g) \mapsto gh.
\end{align*}
Let $X \in T_e G$ be a horizontal vector with respect to the fibration $G \to G/H$. Then its corresponding left-invariant vector field
\[
X^L(p) := pX
\]
is horizontal for every $p \in G$.
\end{proposition}
\begin{proof}
Note that 
\begin{align*}
L_g: G \to G\\
x \mapsto gx
\end{align*}
is an isometry for every $g \in G$. Furthermore, $L_g$ maps fibers to fibers: $L_g(xH) = (gx)H$. Therefore, if $X$ is horizontal at $e$---i.e. perpendicular to the fiber $H$---it follows that $(L_g)_* X$ is perpendicular to $(dL_g)_e (T_e H)$, where $(L_g)_* X = gX$ and $(dL_g)_e (T_e H) = T_g(gH)$. 
\end{proof}

\cref{horizontalleftinvariant} allows us to bound the curvature of the Stiefel and Grassmann manifolds. 

\begin{proposition}
\label{curvStiefel}
Let $\textup{U}(n)$ be endowed with its bi-invariant metric $g$ (cf. \cref{metrica}). Let $k < n$ and consider the induced metric $h$ on $\textup{V}_k(\bb{C}^{n})$ for which $\pi : (\textup{U}(n), g) \to (\textup{V}_k(\bb{C}^{n}), h)$ is a Riemannian submersion. Then
\[
0 \leq K_{h} \leq 2,\quad \textup{\textit{and}}\quad 
\textup{Ric}_{h} \geq 0.
\]
\end{proposition}

\begin{proof}
Recall that $\textup{V}_k(\bb{C}^{n})$ is the quotient space associated with the action shown in \cref{eq:defactionstiefel}. Now, let $p \in \textup{U}(n)$. We can use \cref{CurvaturaSubmersion,horizontalleftinvariant} to conclude that, given two horizontal orthonormal vectors $X, Y \in \mathfrak{u}(n)$, we can extend them as left-invariant vector fields $\tilde{X}, \tilde{Y} \in \mathfrak{X}(\textup{U}(n))$ which remain horizontal. Thus,
\begin{align*}
0 \leq K_{h}(\pi_* \tilde{X}(p), \pi_* \tilde{Y}(p)) &= K_{g}(\tilde{X}(p), \tilde{Y}(p)) + \frac{3}{4}|[\tilde{X}, \tilde{Y}]|_p^\mathcal{V}|^2_g
\\&\leq K_{g}(X, Y) + \frac{3}{4}|[\tilde{X}, \tilde{Y}]|_p|^2_g 
\\&= |[X, Y]|^2_g
\\&\leq 2.
\end{align*}
where the last equality follows from \cref{propinvvectorfields1,propinvvectorfields2,curvLieGroups}, and the last inequality follows from \cref{ineqliebracket}.

Furthermore, since the sectional curvature is non-negative, it holds that its Ricci curvature is non-negative as well. Indeed, let $V \in T_p \textup{V}_k(\bb{C}^n)$ be a unit vector, we can extend $V$ to an orthonormal basis of $T_p \textup{V}_k(\bb{C}^n)$, $\{V, e_2, \dotsc, e_d\}$, $d = 2nk - k^2$. This way, writing $\textup{Ric}_{h}(V, V)$ as 
\[
\textup{Ric}_{h}(V, V) = \sum_{i = 2}^d K_{h}(V, e_i) \geq 0, 
\]
the result follows. 
\end{proof}

\subsubsection{Curvature of complex Grassmann manifolds}
\label{curvaturegrassmann}

Lastly, we will bound the curvature of the manifold of linear subspaces of complex dimension $k$ in $\mathbb{C}^n$, which is known as the Grassmann manifold $\textup{Gr}_k(\bb{C}^n)$. For an in-depth study of Grassmann manifolds, see \cite{bendokat2024grassmann} and the references therein. 

\begin{definition}[Grassmann manifold]
We define the \textup{complex Grassmann manifold} $\textup{Gr}_k(\bb{C}^n)$ as
\[
\textup{Gr}_k(\bb{C}^n) = \{P \in \bb{C}^{n \times n} : P^\dagger = P,\ P^2 = P,\ \Tr(P) = k\}.
\]
\end{definition}

In particular, each point $P \in \textup{Gr}_k(\bb{C}^n)$ can be seen as a rank-$k$ projector in $\bb{C}^n$. Furthermore, $\textup{Gr}_k(\mathbb{C}^n)$ can also be seen as the quotient space
\[
\textup{U}(n) / (\textup{U}(n-k) \times \textup{U}(k)),
\]
where the quotient is considered with respect to the right action on $\textup{U}(n)$ defined as
\begin{align*}
\textup{U}(k) \times \textup{U}(n-k) \times \textup{U}(n) &\to \textup{U}(n)\\
((U, V),X) &\mapsto X \begin{pmatrix}
U & 0\\
0 & V
\end{pmatrix}. 
\end{align*}
This action is again free, smooth and isometric whenever $\textup{U}(n)$ is endowed with the bi-invariant metric $g$. Therefore, the quotient description of $\textup{Gr}_k(\mathbb{C}^n)$ allows us to conclude that it can be endowed with a Riemannian metric $h$ for which 
\[
\pi: (\textup{U}(n), g) \to (\textup{Gr}_k(\mathbb{C}^n), h)
\]
is a Riemannian submersion with totally geodesic fibers.

This way, following the same reasoning as in the proof of \cref{curvStiefel}, we obtain the same bounds on the sectional and Ricci curvatures of $\textup{Gr}_k(\bb{C}^n)$. 

\begin{proposition}
\label{prop:curvGrassmann}
Let $\textup{U}(n)$ be endowed with its bi-invariant metric $g$ (cf. \cref{metrica}). Consider the induced metric $h$ on $\textup{Gr}_k(\bb{C}^{n})$ for which $\pi : (\textup{U}(n), g) \to (\textup{Gr}_k(\bb{C}^{n}), h)$ is a Riemannian submersion. Then
\[
0 \leq K_{h} \leq 2,\quad \textup{\textit{and}}\quad 
\textup{Ric}_{h} \geq 0.
\]
\end{proposition}

Note that we can also define Grassmann manifolds via the action of $\textup{U}(k)$ on the Stiefel manifold $\textup{V}_k(\bb{C}^n)$  given as
\begin{align*}
\textup{U}(k) \times \textup{V}_k(\bb{C}^n) &\to \textup{V}_k(\bb{C}^n) \\
(U, X) &\mapsto XU.
\end{align*}
The action is smooth, free---$XU = X$ implies that $U = \mathds{1}$ by simply multiplying by $X^\dagger$ on both sides---and isometric whenever $\textup{V}_k(\bb{C}^n)$ is endowed with the induced metric $\tilde g$ making the projection map $\pi_1 : (\textup{U}(n), g) \to (\textup{V}_k(\bb{C}^n), \tilde g)$ a Riemannian submersion, where $g$ is the bi-invariant metric of $\textup{U}(n)$. This way, by \cref{existenceRiemannianSubmersionMetrics} we know that there exists a smooth manifold structure and a Riemannian metric $\tilde h$ on $\textup{Gr}_k(\bb{C}^n)$ for which the projection $\pi_2 : (\textup{V}_k(\bb{C}^n), \tilde g) \to (\textup{Gr}_k(\bb{C}^n), \tilde h)$ is a Riemannian submersion. Furthermore, since $\pi = \pi_2 \circ \pi_1$ is also a Riemannian submersion, by the uniqueness of the construction (cf. \cref{existenceRiemannianSubmersionMetrics}), we know that the manifold structures induced by $\pi$ and $\pi_2$ are the same, and that $h = \tilde h$. Lastly, by \cref{prop:inducedtotgeodfibers} we know that $\pi_2$ has totally geodesic fibers. 

\subsection{Riemannian submersions and fiber-invariant functions}
\label{sec:riemanniansubmersionsandopti}

In this subsection, we present several results concerning functions that are constant along the fibers of a surjective Riemannian submersion. The majority of the results presented here can be found in \cite[Section 9]{boumal2022intromanifolds}.

We will always be working with a surjective Riemannian submersion $\pi: (M, g) \to (B, h)$ and a sufficiently smooth function $f: M \to \bb{R}$ which is constant in the fibers of $\pi$, or a sufficiently smooth function $\tilde f : B \to \bb{R}$. The invariance of $f$ implies that there exists a unique function $\tilde f$ satisfying $f = \tilde f \circ \pi$ and vice versa. Furthermore, both functions have the same smoothness properties \cite[Theorem 9.21]{boumal2022intromanifolds}. 

We will conclude that having information about the gradient, the Hessian or the Laplacian of either $f$ or $\tilde f$ easily translates to having information about that of its \textit{projected} or \textit{lifted} version. This will be very convenient for us, as it will allow us to work in either the total or the base space of a Riemannian submersion, and then translate the properties obtained using the results that we will present in this section. 

Let us first define the horizontal lift of a tangent vector in a Riemannian submersion. 
\begin{definition}
Let $\pi: (M, g) \to (B, h)$ be a Riemannian submersion, let $x \in M$, and let $v \in T_{\pi(x)} B$. The \textup{horizontal lift} of $v$ at $x$, denoted as $\textup{lift}_x\, v$, is the unique horizontal vector $u \in H_x(M)$ such that $\pi_*|_x(u) = v$.
\end{definition}

The first result allows us to relate the gradient of a function in the base space $\tilde f$ with the gradient of the lifted function $f$ in the total space. In fact, as the following proposition claims, any function that is constant in the fibers of a Riemannian submersion has a horizontal gradient which is the lift of the gradient of $\tilde f$. 

\begin{definition}[Gradient]
Let $(M, g)$ be a Riemannian manifold. Given a differentiable function $f: M \rightarrow \bb{R}$, the \textup{gradient} of $f$, denoted by $\Grad{g} f$ is the unique vector field on $M$ such that
\[
\langle\Grad{g}f, X\rangle_g = df(X),
\]
for any vector field $X \in \mathfrak{X}(M)$, where $df$ denotes the differential of $f$, also written as $f_*$
\end{definition}

\begin{proposition}[{\cite[Proposition 9.39]{boumal2022intromanifolds}}]
\label{prop:gradientprojectedfunction}
Let $\pi: (M, g) \to (B, h)$ be a surjective Riemannian submersion. Let $\tilde f \in C^1(B)$, and let $f : M \to \bb{R}$ be the unique function such that $f = \tilde f \circ \pi$. Then  for every $x \in M$,
\[
\textup{grad}_g \, f(x) = \textup{lift}_x\, \textup{grad}_{h}\, \tilde f(\pi(x)).
\]
\end{proposition}

\cref{prop:gradientprojectedfunction} allows us to relate the critical points of a fiber-invariant function $f$ and its \textit{projection} $\tilde f$. Indeed, for any $p \in B$ which is a critical point of $\tilde f$, and any $x \in \pi^{-1}(p)$, it holds that $x$ is a critical point of $f$. Furthermore, as $\pi$ is a Riemannian submersion, we can also relate the norm of both gradients at any given point. 
\begin{remark}
Since $\pi$ preserves the norm of horizontal vectors, it holds that 
\[
|\textup{grad}_g\, f(x)|_g = |\textup{grad}_{h}\, \tilde f(\pi(x))|_{h},
\]
for every $x \in M$.
\end{remark}

The following propositions allow us to relate the Hessian of the fiber-invariant function in the total space with that of the function in the base space.

\begin{definition}[Hessian]
\label{definicionhessian}
Let $(M, g)$ be a Riemannian manifold. Let $f : M \rightarrow \mathbb{R}$ be a twice differentiable function on $M$. The \textup{Hessian} of $f$ is the tensor defined as
\[
\nabla^2 f := \nabla \nabla f = \nabla df,
\]
where $\nabla$ denotes the Levi-Civita connection of $(M, g)$, acting as
\[
\nabla^2 f(X, Y) = \langle \nabla_X \Grad{g}f, Y \rangle_g,
\]
for any vector fields $X, Y \in \mathfrak{X}(M)$. Moreover, $\nabla^2 f$ is symmetric.

We will also denote 
\[
\nabla^2 f(x)[X] := \nabla_X\, \Grad{g}f|_x
\]
for any $x \in M$ and any $X \in T_x M$.
\end{definition}

\begin{proposition}[{\cite[Proposition 9.45]{boumal2022intromanifolds}}]
\label{prop:projectedHessian}
Let $(M, g)$ and $(B, h)$ be two Riemannian manifolds. Let $\pi: (M, g) \to (B, h)$ be a surjective Riemannian submersion. Let $\tilde f \in C^2(B)$, and let $f = \tilde f \circ \pi$. Then  for every $x \in M$ and every horizontal vector $u \in H_x M$, 
\[
(\nabla^2_g\, f(x)[u])^{\mathcal{H}} = \textup{lift}_x \nabla^2_h\, \tilde f(\pi(x))[\pi_*u],
\]
where we used $\nabla^2_g$ and $\nabla^2_h$ to distinguish between the Hessian on $M$ and $B$.

In particular, if $x \in M$ is a critical point of $f$, then the eigenvalues of $\nabla^2_g\, f(x)$ and $\nabla^2_h\, \tilde f(\pi(x))$, are the same, together with a set of $\dim(M) - \dim(B)$ additional eigenvalues equal to zero. 
\end{proposition}

\begin{proposition}[{\cite[Lemma 9.41]{boumal2022intromanifolds}}]
\label{prop:kernelofhessian}
Let $\pi: (M, g) \to (B, h)$ be a surjective Riemannian submersion. If $x \in M$ is a critical point of a function $f \in C^2(M)$ which is constant in the fibers of $\pi$, then $\nabla^2_g\, f(x)[v] = 0$ for every $v \in V_x M$. 
\end{proposition}

The following proposition allows us to relate the Hessian of a fiber-invariant function $f$ on the total space of a Riemannian submersion to that of the \textit{projected} function $\tilde f$ in the base space.

\begin{proposition}
\label{prop:equalityHessians}
Let $\pi: (M, g) \to (B, h)$ be a surjective Riemannian submersion. Let $\tilde f \in C^2(B)$, and let $f: M \to \bb{R}$ be such that $f = \tilde f \circ \pi$. Let $x \in M$ be a critical point of $f$ in $M$. Then  for every $v \in T_xM$ it holds that 
\[
\nabla^2_g\,f(x)[v, v] = \nabla^2_g\,f(x)[v^{\mathcal{H}}, v^{\mathcal{H}}]= \nabla^2_h\, \tilde f(\pi(x))[\pi_* v, \pi_* v].
\]

In particular, if $p \in B$ is a critical point of $\tilde f$ in $B$, then for every $v \in T_p B$ and every $x \in \pi^{-1}(p)$ it holds that $\nabla^2_h\, \tilde f(p)[v, v] = \nabla^2_g\, f(x)[\textup{lift}_x v, \textup{lift}_x v]$. 
\end{proposition}
\begin{proof}
Assume that $x \in M$ is some critical point of $f$, and let $v \in T_x M$. Using \cref{prop:kernelofhessian} and the linearity of $\nabla$ we conclude that 
\[
\nabla_v \, \textup{grad}_g\, f|_x = \nabla_{v^{\mathcal{V}}} \,\textup{grad}_g\, f|_x + \nabla_{v^{\mathcal{H}}} \,\textup{grad}_g\, f|_x = \nabla_{v^{\mathcal{H}}} \,\textup{grad}_g\, f|_x. 
\]

Therefore, 
\begin{align*}
\nabla^2_g\, f(x)[v, v] = \langle \nabla_v \,\textup{grad}_g\, f, v\rangle_g &= \langle \nabla_{v^{\mathcal{H}}} \,\textup{grad}_g\, f, v\rangle_g 
\\&= \langle \nabla_{v^{\mathcal{H}}} \,\textup{grad}_g\, f, v^{\mathcal{H}}\rangle_g + \langle \nabla_{v^{\mathcal{H}}} \,\textup{grad}_g\, f, v^{\mathcal{V}}\rangle_g
\\&= \langle \nabla_{v^{\mathcal{H}}} \,\textup{grad}_g\, f, v^{\mathcal{H}}\rangle_g + \langle \nabla_{v^{\mathcal{V}}} \,\textup{grad}_g\, f, v^{\mathcal{H}}\rangle_g
\\&= \langle \nabla_{v^{\mathcal{H}}} \,\textup{grad}_g\, f, v^{\mathcal{H}}\rangle_g 
\\&= \nabla^2_g\, f(x)[v^{\mathcal{H}}, v^{\mathcal{H}}],
\end{align*}
where the third-last equality follows from the symmetry of the Hessian. Thus, 
\begin{align*}
\nabla^2_g\, f(x)[v, v] = \nabla^2_g \, f(x)[v^{\mathcal{H}}, v^{\mathcal{H}}] = \langle \nabla^2_g\, f(x)[v^H], v^{\mathcal{H}}\rangle_g &= \langle (\nabla^2_g\, f(x)[v^H])^{\mathcal{H}}, v^{\mathcal{H}}\rangle_g \\
&= \langle \nabla^2_h\, \tilde f(\pi(x))[\pi_* v^H], \pi_* v^{\mathcal{H}}\rangle_h 
\\&= \nabla^2_h\, \tilde f(\pi(x))[\pi_* v, \pi_* v],
\end{align*}
where the second-last equality follows from \cref{prop:projectedHessian} and from $\pi$ being a Riemannian submersion. This expression also allows us to conclude that, for every critical point $p \in B$ of $\tilde f$, every $x \in \pi^{-1}(p)$, and every $v \in T_p B$, 
\[
\nabla^2_h\, \tilde f(p)[v, v] = \nabla^2_g\, f(x)[\textup{lift}_x\, v, \textup{lift}_x\, v].
\]
\end{proof}

Let us now define the Laplace-Beltrami operator. 
\begin{definition}[Laplace-Beltrami operator]
Given a Riemannian manifold $(M, g)$, let $f : M \to \bb{R}$ be twice differentiable. We define the \textup{Laplacian} of $f$ as
\[
\Delta_g f = \textup{div}(\Grad{g}f),
\]
where 
\[
\textup{div} X:= \Tr(\nabla X),
\]
for any vector field $X \in \mathfrak{X}(M)$, and $\Tr(\nabla X)$ denotes the trace of $Y \mapsto \nabla_Y X$. The operator $\Delta_g$ is the \textup{Laplace-Beltrami operator}. 
\end{definition}

To conclude the section, we briefly discuss how the Lipschitz constants of a function $f$ that is constant in the fibers of a Riemannian submersion are valid for its \textit{projected} version $\tilde f$. Before we do so, let us first discuss \textit{parallel transport}. 
\begin{remark}[parallel transport]
\label{def:paralleltransport}
Given some curve $\gamma$ in $M$ and two points $\gamma(t_1), \gamma(t_2)$ along the curve, the \textup{parallel transport} along $\gamma$ is a vector space isomorphism between $T_{\gamma(t_1)} M$ and $T_{\gamma(t_2)} M$. For every $x \in M$ and every $v \in T_xM$, we denote the \textup{parallel transport} from $x$ to $\exp_x(v)$ along the curve $t \mapsto \exp_x (tv)$ as $\textup{P}_v$ . For an in-depth study of the parallel transport, see for example \cite[Section 4]{lee2018introductionRiemannian}.
\end{remark}

\begin{definition}[Lipschitz function]
\label{def:lipschitzfunction}
Let $(M, g)$ be a complete Riemannian manifold. We say that a function $f: M \rightarrow \bb{R}$ is $A_1$-Lipschitz if  
\[
|f(x) - f(y)| \leq A_1 d_g(x,y),\quad \forall x, y \in M.
\]
Similarly, we say that $\Grad{g}f$ is $A_2$-Lipschitz if 
\[
|\textup{P}^{-1}_{v}\,\Grad{g}f(\exp_x(v)) - \Grad{g}f(x)|_g \leq A_2 |v|_g,\quad \forall x \in M,\ \forall v \in T_xM,
\]
where $\textup{P}^{-1}_v$ denotes the inverse of $\textup{P}_v$. 
 
Lastly, we say that the Hessian $\nabla^2 f$ is $A_3$-Lipschitz if 
\[
\max_{\substack{s \in T_xM\\|s|_g=1}} |\textup{P}^{-1}_{v} \circ \nabla^2 f(\exp_x(v)) \circ \textup{P}_{v}[s] - \nabla^2 f(x)[s]|_g \leq A_3 |v|_g,    
\]
for every $x \in M$ and every $v \in T_xM$. 
\end{definition}

Alternatively, the Lipschitz constant of the gradient of a function $f$ corresponds to the operator norm of its Hessian. 
\begin{remark}[{\cite[Corollary 10.47]{boumal2022intromanifolds}}]
\label{rmk:LipschitzGrad}
Let $(M, g)$ be a complete Riemannian manifold and let $f: M \to \bb{R}$ be twice continuously differentiable. Then $\Grad{g} f$ is $L$-Lipschitz if and only if $\nabla^2 f$ has an operator norm bounded by $L$ for all $x \in M$, i.e. 
\[
\max_{\substack{s \in T_xM\\|s|_g = 1}}|\nabla^2 f(x)[s]|_g \leq L. 
\]
\end{remark}

To relate the Lipschitz constants of a fiber-invariant function $f$ to those of its \textit{projected} version, we can use the following result. 
\begin{proposition}
\label{prop:preservationofLipschitz}
Let $(M, g)$ be a compact Riemannian manifold, and let $\pi: (M, g) \to (B, h)$ be a surjective Riemannian submersion. Let $f \in C^2(M)$ be constant in the fibers of $\pi$ and assume that it is $A_1$-Lipschitz, with an $A_2$-Lipschitz gradient. Let $\tilde f : B \to \bb{R}$ be such that $f = \tilde f \circ \pi$. Then  it follows that $\tilde f$ is $A_1$-Lipschitz, with an $A_2$-Lipschitz gradient.
\end{proposition}
\begin{proof}
Let $p, q \in B$, and let $x \in \pi^{-1}(p)$, $y \in \pi^{-1}(q)$ be such that $d_h(p, q) = d_g(x, y)$. Then  
\[
|\tilde f(p) - \tilde f(q)| = |f(x) - f(y)| \leq A_1d_g(x, y) = A_1 d_h(p, q).
\]

Now, recall that the gradient of $f$ is $A_2$-Lipschitz if 
\[
\max_{\substack{s \in T_xM\\|s|_g = 1}}|\nabla^2 f(x)[s]|_g \leq A_2. 
\]
for every $x \in M$.

Now, let $x \in M$. Using \cref{prop:projectedHessian} we know that 
\begin{align*}
\max_{\substack{s \in T_xM\\|s|_g = 1}} |\nabla^2_g\, f(x)[s]|_g  
&\geq  \max_{\substack{s \in H_xM\\|s|_g = 1}}|(\nabla^2_g\, f(x)[s])^{\mathcal{H}}|_g
\\&= \max_{\substack{s \in H_xM\\|s|_g = 1}}|\textup{lift}_x \nabla^2_h\, \tilde{f}(\pi(x))[\pi_* s]|_g
\\&= \max_{\substack{s \in H_xM\\|s|_g = 1}}|\nabla^2_h\, \tilde{f}(\pi(x))[\pi_* s]|_h.
\end{align*}
From $\pi$ being surjective it follows that $\tilde f$ has an $A_2$-Lipschitz gradient. 
\end{proof}

\subsection{Periodic geodesics, diameter, dimension and volume}
\label{sec:studyingmanifolds}

In this subsection, we will first give upper bounds for the curvature of the round sphere. Then we will give lower bounds for the injectivity and convexity radii of the sphere, the unitary group, as well as the Stiefel and Grassmann manifolds via the study of their shortest non-trivial periodic geodesic and the upper bounds on their sectional curvature. Lastly, we will bound their diameter and volume and obtain an expression for their dimension.

First, the sectional and Ricci curvatures of the sphere endowed with the round metric are widely known. Let us therefore state the following result without proof. 
\begin{proposition}[{\cite[Theorem 8.34]{lee2018introductionRiemannian}}]
\label{sectionalcurvatureofsphere}
The sphere $\bb{S}^n$ endowed with the round metric has constant sectional curvature $1$. Furthermore, it is an Einstein manifold with constant $(n-1)$. 
\end{proposition}

This result, together with the bounds found for the sectional and Ricci curvatures of the unitary and special unitary groups---see \cref{LemaCurvaturaUn,LemaCurvaturaSUn}---and the Stiefel and Grassmann manifolds---see \cref{curvStiefel,prop:curvGrassmann}---are summarized in \cref{tab:TablaCurvaturas1}.

\begin{table}
    \centering
    \begin{tabular}{|c|c|c|c|}
    \hline
    Manifold & Metric & Sectional curvature & Ricci curvature\\
    \hline
    $\textup{U}(n)$ & Bi-invariant & $\leq 1/2$ & $\geq 0$\\
    \hline
    $\textup{SU}(n)$ & Bi-invariant & $\leq 1/2$ & $= n/2$\\
    \hline
    $\bb{S}^n$ & Round & $= 1$ & $= n-1$ \\
    \hline
    $\textup{V}_k(\bb{C}^n)$ & Induced & $\leq 2$ & $\geq 0$ \\
    \hline
    $\textup{Gr}_k(\bb{C}^n)$ & Induced & $\leq 2$ & $\geq 0$ \\
    \hline
    \end{tabular}
    \caption{Summary of the various bounds obtained for the sectional and Ricci curvatures of the manifolds studied.}
    \label{tab:TablaCurvaturas1}
\end{table}

\subsubsection{Bounding the injectivity and convexity radii}
\label{shortestgeod}

In order to obtain lower bounds on the injectivity and convexity radii in the manifolds of our interest, we saw in \cref{secinjectivity} that it is useful to have lower bounds on the length of their shortest non-trivial periodic geodesic. 

Let us first briefly discuss the length of the shortest non-trivial periodic geodesic in a product manifold. This is in fact a rather easy task: given two Riemannian manifolds, $(M, g)$ and $(N, h)$, every geodesic of $M \times N$ endowed with the product metric can be written as 
\[
\gamma(t) = (\gamma_1(t), \gamma_2(t)),
\]
where $\gamma_1(t)$ and $\gamma_2(t)$ are geodesics in $M$ and $N$, respectively. For this reason, the length of the shortest non-trivial periodic geodesic in $M \times N$, is 
\[
l(M\times N) = \min \{l(M), l(N)\}.
\] 
Certainly, the shortest non-trivial periodic geodesic of the product manifold corresponds to fixing a point for either $M$ or $N$, and considering the shortest non-trivial periodic geodesic of $N$ or $M$, respectively.

Let us now lower bound the length of the shortest non-trivial periodic geodesics in the manifolds of our interest. The lower bounds that will be derived are summarized in \cref{tab:TablaGeodDiametro} along with the upper bounds for the diameter of the manifolds, which will be discussed in \cref{sec:diamcontrol}.

\begin{table}
    \centering
    \begin{tabular}{|c|c|c|c|}
    \hline
    Manifold & Metric & $l(M)$ & Diameter\\
    \hline
    $\textup{U}(n)$ & Bi-invariant & $\geq 2\pi$ & $\leq \pi \sqrt{n} (\sqrt{2}+1)$ \\
    \hline
    $\bb{S}^n$ & Round & $= 2\pi$ & $\leq \pi$\\
    \hline
    $\textup{V}_k(\bb{C}^n)$ & Induced & $\geq 2\pi$ & $\leq \pi \sqrt{n} (\sqrt{2}+1)$ \\
    \hline
    $\textup{Gr}_k(\bb{C}^n)$ & Induced & $\geq \sqrt{2}\pi$ & $\leq \pi \sqrt{n} (\sqrt{2}+1)$ \\
    \hline
    \end{tabular}
    \caption{Summary of various bounds for the diameter and length of the shortest non-trivial periodic geodesic of the manifolds that are studied in \cref{sec:traceratio}.}
    \label{tab:TablaGeodDiametro}
\end{table}

The length of the shortest non-trivial periodic geodesic on the sphere when endowed with the round metric is known to be $2\pi$ \cite[Proposition 5.27]{lee2018introductionRiemannian}, as geodesics on the sphere correspond to great circles. Let us now compute the length of the shortest non-trivial periodic geodesic in $\textup{U}(n)$ endowed with the bi-invariant metric.
\begin{proposition}
\label{shortestperiodicgeodUn}
The length of the shortest non-trivial periodic geodesic of $\textup{U}(n)$ endowed with the bi-invariant metric is bounded from below by $2\pi$. 
\end{proposition}
\begin{proof}
Let us fix $X \in \mathfrak{u}(n)$ with norm one, i.e. $-\Tr(X^2) = 1$, and a periodic geodesic starting in $U \in \textup{U}(n)$ in the direction of $X$, i.e. $\gamma(t) = e^{tX}U$. We want to find
\[
\min\{t > 0 : e^{tX}U = U\} = \min\{t > 0 : e^{tX} = \mathds{1}\}.
\]

Since $X \in \mathfrak{u}(n)$ we know that $X^\dagger = -X$, therefore all its eigenvalues are imaginary. Furthermore, $X$ is normal (i.e. $X^\dagger X = X X^\dagger$). Therefore, without loss of generality, we can consider $X$ to be in its diagonal form
\[
\Sigma = \begin{bmatrix} 
    i\lambda_1 &  & & \\
     & i\lambda_2 &  &  \\
     &  & \ddots &  \\
     & &  & i\lambda_n 
    \end{bmatrix},
\]
where $\lambda_j \in \bb{R}$ for every $j \in \{1, \dotsc, n\}$. Otherwise, there would exist a unitary matrix $V$ such that $X = V \Sigma V^\dagger$ and $e^{tX} = V e^{t\Sigma} V^\dagger$, which yields
\[
e^{tX} = \mathds{1} \iff V e^{t\Sigma} V^\dagger = \mathds{1} \iff e^{t\Sigma} = \mathds{1}. 
\]
Now, 
\[
e^{tX} = \begin{bmatrix} 
    e^{it\lambda_1} & & \\
    & e^{it\lambda_2} & &  \\
    &  &  \ddots &\\
    & &  & e^{it\lambda_n} 
    \end{bmatrix},
\]
where all the off-diagonal elements are zero. In order for $e^{tX}$ to be the identity it has to hold that
\[
e^{it\lambda_j} = 1,\quad t > 0, \text{ for each } j \in \{1, \dotsc, n\},
\]
which implies that
\[
t\lambda_j = 2\pi k_j\quad \text{for some } k_j \in \bb{Z},\quad \text{for every }j \in \{1,\dotsc,n\}.
\]
Since for every eigenvalue $i\lambda_j$ it holds that $|\lambda_j|^2 \leq 1$, we can conclude that $t$ has to be at least $2\pi$ and so $l(\textup{U}(n)) \geq 2\pi$.
\end{proof}

Let us now finish by obtaining a lower bound on the length of the shortest non-trivial periodic geodesic on the Stiefel and Grassmann manifolds when endowed with their induced metrics.

Let $g$ be the bi-invariant metric of $\textup{U}(n)$. The length of the shortest non-trivial periodic geodesic on the Stiefel manifold $\textup{V}_k(\bb{C}^n)$ when endowed with the metric $h$ making $\pi: (\textup{U}(n), g) \to (\textup{V}_k(\bb{C}^n), h)$ a Riemannian submersion is known to be lower bounded by $2\pi$ \cite{rentmeesters2013algorithms,absil2025ultimate}.

Finally, let us obtain a lower bound for the shortest non-trivial periodic geodesic of Grassmann manifolds. 
\begin{proposition}
Consider $\textup{U}(n)$ endowed with the bi-invariant metric, and let $\textup{Gr}_k(\bb{C}^n)$ be endowed with the metric making the quotient map $\pi: \textup{U}(n) \to \textup{Gr}_k(\bb{C}^n)$ a Riemannian submersion. The length of its shortest non-trivial periodic geodesic is lower bounded by $\sqrt{2}\pi$. 
\end{proposition}
\begin{proof}
Let us fix some $P \in \textup{Gr}_k(\bb{C}^n)$, and some unit tangent vector $X \in T_P \textup{Gr}_k(\bb{C}^n)$. It is known (cf. \cite{bendokat2024grassmann}) that $X$ corresponds to a Hermitian matrix such that $X = PX + XP$ and that the geodesic starting at $P$ in the direction $X$ is given by 
\[
\gamma_{P, X}(t) := e^{t[X, P]}P e^{-t[X, P]}.
\]

Without loss of generality, we may assume that $P = \begin{pmatrix}
\mathds{1}_k & 0 \\
0 & 0_{n-k}
\end{pmatrix}$, 
and so the condition $X = XP + PX$ implies that $X$ is of the form 
\[
X = \begin{pmatrix}
0_k & A\\
A^\dagger & 0_{n-k}
\end{pmatrix},
\]
where $A$ is a complex $k \times (n-k)$ matrix. 

Standard calculations allow us to conclude that 
\begin{align*}
[X, P]^{2m} &= \begin{pmatrix}
(-1)^m (AA^\dagger)^m & 0\\
0 & (-1)^m (A^\dagger A)^m
\end{pmatrix},\\
[X, P]^{2m+1} &= \begin{pmatrix}
0 & (-1)^{m+1}(A A^\dagger)^m A\\
(-1)^{m}(A^\dagger A)^m A^\dagger & 0
\end{pmatrix},
\end{align*}
for every $m \in \bb{N}\cup\{0\}$. Thus, 
\begin{align*}
e^{t[X, P]} &= \begin{pmatrix}
\cos(t\sqrt{AA^\dagger}) & -\frac{\sin(t\sqrt{AA^\dagger})}{\sqrt{AA^\dagger}} A\\
\frac{\sin(t\sqrt{A^\dagger A})}{\sqrt{A^\dagger A}} A^\dagger & \cos(t\sqrt{A^\dagger A})
\end{pmatrix},
\end{align*}
where $\sin(t\cdot0)/0 := t$. 

This way, in order for $P = e^{t[X, P]} P e^{-t[X, P]}$ to hold, it must be the case that 
\begin{equation}
\label{eq:identitycosines}
\cos^2(t\sqrt{AA^\dagger}) = \mathds{1}_k.
\end{equation}

Let us assume without loss of generality that $A A^\dagger$ is diagonal with eigenvalues $\{\lambda_i\}_{i = 1}^k$. If \cref{eq:identitycosines} holds, every $\lambda_i$ must be non-negative and 
\begin{equation}
\label{eq:conditionclosedgeodGr}
t\sqrt{\lambda_i} = \pi k_i,\quad k_i \in \bb{Z},
\end{equation}
for every $i \in \{1, \dotsc, k\}$. 

Finally, note that, as $X$ has norm one,
\[
1 = \Tr(X^\dagger X) = \Tr(AA^\dagger) + \Tr(A^\dagger A) = 2\sum_{i = 1}^k \lambda_i,
\]
and so, as the eigenvalues are assumed to be non-negative, it holds that $\lambda_i \leq \frac{1}{2}$ for every $i \in \{1, \dotsc, k\}$, and $\sqrt{\lambda_i} \leq \frac{\sqrt{2}}{2}$. Thus, in order for \cref{eq:conditionclosedgeodGr} to hold, it must be the case that 
\[
t \geq \frac{2\pi}{\sqrt{2}} = \sqrt{2}\pi,
\]
finishing the proof.
\end{proof}

The lower bounds on the length of the shortest non-trivial periodic geodesic and the upper bounds on the sectional curvature of the manifolds considered allow us to obtain bounds on their injectivity and convexity radii. 
\begin{corollary}
\label{cor:injectivityradii}
Consider the manifolds $\textup{U}(n)$, $\bb{S}^n$, $\textup{V}_k(\bb{C}^n)$ and $\textup{Gr}_k(\bb{C}^n)$ endowed with the metrics shown in \cref{tab:TablaGeodDiametro}. Then their injectivity and convexity radii can be lower bounded as follows;
\begin{align*}
i(\textup{U}(n)) &\geq \pi,&\quad \mathit{conv}(\textup{U}(n)) &\geq \frac{\pi}{2},\\
i(\bb{S}^n) &\geq \pi, &\quad \mathit{conv}(\bb{S}^n) &\geq \frac{\pi}{2},\\
i(\textup{V}_k(\bb{C}^n)) &\geq \frac{\sqrt{2}\pi}{2}, &\quad \mathit{conv}(\textup{V}_k(\bb{C}^n)) &\geq  \frac{\sqrt{2}\pi}{4}, \\
i(\textup{Gr}_k(\bb{C}^n)) &\geq \frac{\sqrt{2}\pi}{2},&\quad \mathit{conv}(\textup{Gr}_k(\bb{C}^n)) &\geq \frac{\sqrt{2}\pi}{4}.\\
\end{align*}
\end{corollary}
\begin{proof}
To obtain the bounds, we can use \cref{InjectivityRadiusBound,controlconvexityradius}, along with the bounds on the sectional curvature and the length of their shortest non-trivial periodic geodesic shown in \cref{tab:TablaCurvaturas1,tab:TablaGeodDiametro}. 

For the unitary group endowed with the bi-invariant metric, we obtained that  
\[
K \leq \frac{1}{2},\quad \text{and}\quad l(\textup{U}(n)) \geq 2\pi.
\]
Thus
\[
i(\textup{U}(n)) \geq \min\{\sqrt{2}\pi, \pi\} = \pi,\quad \textup{and}\quad 
\mathit{conv}(\textup{U}(n)) \geq \frac{\pi}{2}. 
\]

For the sphere endowed with the round metric, we showed that 
\[
K \leq 1,\quad \textup{and}\quad l(\bb{S}^n) = 2\pi,
\]
so
\[
i(\bb{S}^n) \geq \min\{\pi, \pi\} = \pi,\quad \textup{and}\quad \mathit{conv}(\bb{S}^n) \geq \frac{\pi}{2}.
\]

For the Stiefel manifold $\textup{V}_k(\bb{C}^n)$ endowed with the metric $h$ making $\pi: (\textup{U}(n), g) \to (\textup{V}_k(\bb{C}^n), h)$ a Riemannian submersion, where $g$ is the bi-invariant metric of $\textup{U}(n)$, we proved that
\[
K \leq 2,\quad \textup{and}\quad l(\textup{V}_k(\bb{C}^n)) \geq 2\pi,
\]
which implies that
\[
i(\textup{V}_k(\bb{C}^n)) \geq \min\{\frac{\sqrt{2}\pi}{2}, \pi\} = \frac{\sqrt{2}\pi}{2},\quad \textup{and}\quad 
\mathit{conv}(\textup{V}_k(\bb{C}^n)) \geq \frac{\sqrt{2}\pi}{4}.    
\]

Lastly, for the Grassmann manifold also endowed with its induced metric, we showed that
\[
K \leq 2,\quad \textup{and}\quad l(\textup{Gr}_k(\bb{C}^n)) \geq \sqrt{2}\pi,
\]
yielding
\[
i(\textup{Gr}_k(\bb{C}^n)) \geq \min\{\frac{\sqrt{2}\pi}{2}, \frac{\sqrt{2}\pi}{2}\} = \frac{\sqrt{2}\pi}{2},\quad \textup{and}\quad 
\mathit{conv}(\textup{Gr}_k(\bb{C}^n)) \geq \frac{\sqrt{2}\pi}{4}.     
\]
\end{proof}

\subsubsection{Controlling the diameter}
\label{sec:diamcontrol}

We will now derive the diameter bounds shown in \cref{tab:TablaGeodDiametro}. In this subsection we will also show how the diameter of a product manifold is related to that of its components, and how the diameters of the total and the base space of a Riemannian submersion are related. 

\begin{proposition}
\label{diametersubmersion}
Let $M$ be a compact Riemannian manifold and let $\pi : (M, g) \to (B, h)$ be a Riemannian submersion. Denote by $d_g$ and $d_h$ the geodesic distances on $M$ and $B$, respectively. Then  for every $x, y \in M$
\[
d_g(x, y) \geq d_h(\pi(x), \pi(y)).
\]
In particular, if $\pi$ is surjective, $\textup{diam}(B) \leq \textup{diam}(M)$. 
\end{proposition}
\begin{proof}
Since $M$ is compact, the distance between $x$ and $y$ is attained by a geodesic. We denote such a geodesic by $\gamma$, and assume that $\gamma(0) = x$, $\gamma(1) = y$. Then 
\begin{align*}
d_g(x, y) = \int_0^1 |\gamma'(s)|_g\, ds \geq \int_0^1 |(\gamma')(s)^\mathcal{H}|_g\, ds = \int_0^1 |\pi_* (\gamma')(s)^\mathcal{H}|_h\, ds \geq d_h(\pi(x), \pi(y)),
\end{align*}
concluding the proof. 
\end{proof}

Lastly, the diameter of a product manifold endowed with the product metric can be easily studied, by simply obtaining an expression for the distance function in the product manifold. 
\begin{proposition}[Diameter of the product]
\label{diameterProduct}
Given two Riemannian manifolds $(M, g)$ and $(N, h)$, if we consider the product manifold $(M \times N, g\oplus h)$ endowed with the product metric, the distance induced by this metric is such that for every $(x_1,y_1), (x_2, y_2) \in M \times N$, 
\[
d_{g\oplus h}((x_1,y_1), (x_2, y_2)) = \sqrt{d^2_g(x_1, x_2) + d^2_h(y_1,y_2)},
\]
and so 
\[
\textup{diam}(M\times N) = \sqrt{\textup{diam}(M)^2 + \textup{diam}(N)^2}.
\]
\end{proposition}

With \cref{diameterProduct,diametersubmersion} in mind, we can now bound the diameter of the manifolds of our interest. We start by establishing upper bounds on the diameter of both the unitary group and the round sphere. To this end, we will use the following result.
\begin{proposition}[Bonnet-Myers]
Let $(M, g)$ be a compact, connected Riemannian manifold of dimension $n \geq 2$, and suppose that there is a positive constant $c = \frac{1}{R^2}$ such that the Ricci curvature of $M$ satisfies $\textup{Ric}_g(v,v) \geq (n-1)c$ for all unit vectors $v$. Then  the diameter of $M$ is less than or equal to $\pi R$.    
\end{proposition}

This result allows us to upper bound the diameter of the round sphere and $\textup{SU}(n)$.
\begin{corollary}
\label{diamSUn}
Let $\bb{S}^n$ be endowed with the round metric, and let $\textup{SU}(n)$ be endowed with the bi-invariant metric, where $n \geq 2$. Then  
\[
\textup{diam}(\bb{S}^n) \leq \pi,\quad \textup{\textit{and}}\quad \textup{diam}(\textup{SU}(n)) \leq \pi \sqrt{2n}.
\]
\end{corollary}
\begin{proof}
First, using \cref{sectionalcurvatureofsphere} we know that the round sphere $\bb{S}^n$ is an Einstein manifold with constant $n-1$, and so it satisfies the Bonnet-Myers theorem with constant $c = 1$. 

For the special unitary group, recall from \cref{LemaCurvaturaSUn} that $\textup{SU}(n)$ is an Einstein manifold with constant $\frac{n}{2}$. Moreover, $\textup{SU}(n)$ is connected, compact and has dimension $n^2-1$. Therefore $\textup{Ric}(v,v) = ((n^2-1) -1)c$ with
\[
c = \frac{n}{2(n^2 -2)}.
\]
This way, $R = \sqrt{\frac{2(n^2-2)}{n}}$ and so we can conclude that
\begin{align*}
\textup{diam}(\textup{SU}(n)) &\leq \pi \sqrt{\frac{2(n^2-2)}{n}} \leq \pi \sqrt{2n}.
\end{align*}
\end{proof}

Having obtained a bound on the diameter of $\textup{SU}(n)$, we can now upper bound the diameter of $\textup{U}(n)$ when endowed with its bi-invariant metric. 
\begin{corollary}
\label{diamUn}
Let $\textup{U}(n)$ be endowed with the bi-invariant metric. Then  
\[
\textup{diam}(\textup{U}(n)) \leq \pi\sqrt{n}(1+ \sqrt{2}).
\]
\end{corollary}
\begin{proof}
Let $A, B \in \textup{U}(n)$. We can rewrite $A$ and $B$ as $A = n_A h_A$, $B = n_B h_B$, where $n_A, n_B \in \textup{U}(1)$ and $h_A, h_B \in \textup{SU}(n)$. Using this decomposition we can write
\[
d_{\textup{U}(n)}(A, B) = d_{\textup{U}(n)}(n_A h_A, n_B h_B).
\]
Moreover, since the metric of $\textup{U}(n)$ is bi-invariant, it follows that 
\[
d_{\textup{U}(n)}(n_A h_A, n_B h_B) = d_{\textup{U}(n)}(h_A, \overline{n}_A n_B h_B).
\]
Now, applying the triangle inequality 
\[
d_{\textup{U}(n)}(h_A, \overline{n}_A n_B h_B) \leq d_{\textup{U}(n)}(h_A, h_B)  + d_{\textup{U}(n)}(h_B, \overline{n}_A n_B h_B).
\]

Using the fact that $d_{\textup{U}(n)}(h_A, h_B) \leq d_{\textup{SU}(n)}(h_A, h_B)$ and the bi-invariance of the metric again,
\begin{align*}
d_{\textup{U}(n)}(h_A, h_B)  + d_{\textup{U}(n)}(h_B, \overline{n}_A n_B h_B) &\leq d_{\textup{SU}(n)}(h_A, h_B) + d_{\textup{U}(n)} (h_B, \overline{n}_A n_B h_B)
\\&\leq \textup{diam}(\textup{SU}(n)) + d_{\textup{U}(n)}(\mathds{1}, \overline{n}_A n_B).
\end{align*}

Finally, since $n_A, n_B \in \textup{U}(1)$ and $d_{\textup{U}(n)}(\mathds{1}, \overline{n}_A n_B) = \sqrt{n}\, d_{\textup{U}(1)}(1, \overline{n}_A n_B)$, 
\begin{align*}
\textup{diam}(\textup{SU}(n)) + d_{\textup{U}(n)}(\mathds{1}, \overline{n}_A n_B) &\leq \textup{diam}(\textup{SU}(n)) + \sqrt{n}\, d_{\textup{U}(1)}(1, \overline{n}_A n_B)
\\&\leq \textup{diam}(\textup{SU}(n)) + \pi\sqrt{n},
\end{align*}
where the last inequality follows from $\textup{U}(1) \simeq \bb{S}^1$ having a diameter upper bounded by $\pi$. 
\end{proof}

Lastly, we can apply \cref{diametersubmersion} to obtain a bound on the diameter of $\textup{V}_k(\bb{C}^n)$ and $\textup{Gr}_k(\bb{C}^n)$ when endowed with their induced metrics. Indeed, it follows automatically that 
\begin{align}
\label{eq:diamaterboundstiefelgrassmann}
\begin{split}
\textup{diam}(\textup{V}_k(\bb{C}^n)) &\leq \textup{diam}(\textup{U}(n)) \leq \pi\sqrt{n}(1+ \sqrt{2}),\\
\textup{diam}(\textup{Gr}_k(\bb{C}^n)) &\leq \textup{diam}(\textup{U}(n)) \leq \pi\sqrt{n}(1+ \sqrt{2}).
\end{split}
\end{align}

\subsubsection{Dimension and volume}
\label{dimensionandvolume}

Lastly, let us briefly discuss the dimension and volume of the manifolds studied in this section, when endowed with the metrics shown in \cref{tab:TablaGeodDiametro}. 

\paragraph{Sphere.} The dimension of the $n$-sphere $\bb{S}^n$ is $n$. Its volume when endowed with the round metric is known to be
\[
\textup{Vol}_{g_{\textit{round}}}(\bb{S}^n) = \frac{2\pi^{(n+1)/2}}{\Gamma(\frac{n+1}{2})},
\]
(cf. \cite[Section 10]{lee2018introductionRiemannian}).

\paragraph{Unitary group.} The unitary group $\textup{U}(n)$ is a Lie group of dimension $n^2$. Furthermore, it is known (cf. \cite[Corollary 3.5.2]{lando2013graphs}) that when endowed with the bi-invariant metric, its volume is
\[
\textup{Vol}(\textup{U}(n)) = \frac{(2\pi)^{n(n+1)/2}}{\prod_{k = 1}^{n-1} k!}.
\]

Now, in order to study Stiefel and Grassmann manifolds, we will make use of two formulae which allow us to relate their dimension and volumes to those of the total space and the fibers. Given a Riemannian submersion $\pi: (M, g) \to (M/G, h)$, it holds that
\[
\dim(M/G) = \dim(M) - \dim(G).
\]
Moreover, if $\pi$ has totally geodesic fibres, 
\[
\textup{Vol}(M) = \textup{Vol(G)} \textup{Vol}(M/G).
\]
Note that the volume formula can be inferred using the Fubini formula from \cref{sec:fubini}. 

\paragraph{Stiefel manifold.} Using the formulae for the volume and the dimension of quotient spaces, it is easy to see that the dimension of the Stiefel manifold $\textup{V}_{k}(\bb{C}^{n})$ is 
\[
n^2 - (n-k)^2 = 2nk - k^2.
\]
Moreover, when endowed with the metric $h$ which makes $\pi: (\textup{U}(n), g) \to (\textup{V}_{k}(\bb{C}^{n}), h)$ a Riemannian submersion---where $g$ is the bi-invariant metric of $\textup{U}(n)$---its volume is given by
\begin{align*}
\textup{Vol}(\textup{V}_{k}(\bb{C}^{n})) &= \textup{Vol}(\textup{U}(n))/ \textup{Vol}(\textup{U}(n-k))
\\&= \frac{(2\pi)^{nk - \frac{k(k-1)}{2}}}{\prod_{a = n-k}^{n-1}a!}.
\end{align*}

\paragraph{Grassmann manifold.} Similarly, the dimension of the Grassmann manifold $\textup{Gr}_k (\bb{C}^n)$ is 
\[
n^2 - (n-k)^2 - k^2 = 2k(n - k),
\]
and its volume, when endowed with its corresponding induced metric is given by
\begin{align*}
\textup{Vol}(\textup{Gr}_{k}(\bb{C}^{n})) &= \textup{Vol}(\textup{U}(n))/ \textup{Vol}(\textup{U}(n-k) \times U(k))
\\&= \frac{(2\pi)^{nk - k^2} \prod_{a = 1}^{k-1} a!}{\prod_{a = n-k}^{n-1}a!}.
\end{align*}

%% file: Chapters/Sobolev_Spaces.tex
\section{Sobolev spaces on manifolds}
\label{sec:obolevspaces}

In this section, we will introduce some basic notions about Sobolev spaces on manifolds, which provide the natural setting for the Poincaré and log-Sobolev inequalities. We will follow the structure of Alexander Grigor'yan's notes \cite{Grigoryan_2024}. Let us first define \textit{weighted manifolds}.
\begin{definition}[Weighted manifold]
Let $(M, g)$ be a Riemannian manifold, and let $D(x)$ be a smooth and positive function on $M$, also known as a \textup{density function}. Let us define the measure $\mu$ on $M$ as
\[
d\mu(x) := D(x)\textup{dVol}_g(x).
\]
We say $(M, g, \mu)$ is a \textup{weighted manifold}.
\end{definition}

Since $(M, \mu)$ is a measure space, it is straightforward to define the space of square-integrable functions on $M$ with respect to $\mu$, denoted as $L^2(M, \mu)$. We also denote by $L^2_{\mathit{loc}}(M, \mu)$ the space of measurable functions on $M$ such that $f \in L^2(V, \mu)$ for every $V \subset M$ compact.

It is a well-known result that, whenever we consider a measure space endowed with a finite measure, the spaces $L^p$ are nested. 
\begin{lemma}
\label{lem:nestedLp}
Let $(E, \mu)$ be a measure space with finite measure $\mu$. Let $p, q \in [1, \infty]$, be such that $p > q$. Then $L^p(E, \mu) \subset L^q(E, \mu)$.    
\end{lemma}

Let us now introduce an analogue of the space $L^2(M, \mu)$ for vector fields on a weighted manifold $(M, g, \mu)$. 

\begin{definition}[Spaces $\overrightarrow{L}^2(M, \mu)$ and $\overrightarrow{L}^2_{\mathit{loc}}(M, \mu)$]
Given a weighted manifold $(M, g, \mu)$, we denote by $\overrightarrow{L}^2(M, \mu)$ the space of all measurable vector fields on $M$ such that $|X|_g \in L^2(M, \mu)$. 

We can define an inner product on $\overrightarrow{L}^2(M, \mu)$ given as
\[
(v, w)_{\overrightarrow{L}^2(M, \mu)} := \int_M \langle v, w\rangle_g\; d\mu,
\]
for every pair of vector fields $v, w \in \overrightarrow{L}^2(M, \mu)$. Its corresponding norm is given as
\[
\norm{v}^2_{\overrightarrow{L}^2(M, \mu)} = \int_M |v|^2_g\; d\mu.
\]

The space $\overrightarrow{L}^2_{\mathit{loc}}(M, \mu)$ is defined as the space of all measurable vector fields $X$ on $M$ such that $|X|_g \in L^2_{\mathit{loc}}(M, \mu)$.
\end{definition}

Before we define Sobolev spaces on a weighted manifold $(M, g, \mu)$, we need to define the notion of \textit{weak gradient}.
\begin{definition}[Weak gradient]
Let $(M, g)$ be a Riemannian manifold and let $u : M \to \bb{R}$ be such that $u \in L^2_{\mathit{loc}}(M, \mu)$. A \textup{weak gradient} of $u$ is a vector field $v \in \overrightarrow{L}^2_{\mathit{loc}}(M, \mu)$---also denoted $\Grad{g} u$---such that, for any smooth vector field $\psi$ on $M$ with compact support, it holds that
\begin{equation}
\label{defweakgradient}
\int_M u\, \textup{div}_{g, \mu} \psi\; d\mu = -\int_M \langle v, \psi\rangle_g\; d\mu,
\end{equation}
where 
\[
\textup{div}_{g, \mu} X := \frac{1}{D} \textup{div} (DX),\quad \forall X \in \mathfrak{X}(M).
\]
\end{definition}

The weak gradient of a given function $u$ is uniquely defined in $\overrightarrow{L}^2_{\mathit{loc}}(M)$. Moreover, if $u$ is assumed to be smooth, the usual gradient $\Grad{g}u$ coincides with the weak gradient. Indeed, the left hand side of \cref{defweakgradient} can be rewritten using the following proposition:
\begin{proposition}
\label{propdiv}
Given a smooth function $f: M \rightarrow \mathbb{R}$ and a smooth vector field $X \in \mathfrak{X}(M)$, the following identity holds
\[
\textup{div}(fX) = f\textup{div}(X) + X(f).
\]
\end{proposition}
Applying the divergence theorem---which we state below---\cref{defweakgradient} holds.
\begin{proposition}[Divergence theorem]
\label{divtheo}
Let $(M, g)$ be an oriented Riemannian manifold with boundary and let $X \in \mathfrak{X}(M)$ be a compactly supported smooth vector field, then
\[
\int_{M} \textup{div}(X)\, \textup{dVol}_g = \int_{\partial M} g(X,\hat{n})\, \textup{dVol}_{\tilde{g}},
\]
where $\hat{n}$ is the unit outward normal vector field to $\partial M$ and $\tilde{g}$ is the induced metric on the boundary.
\end{proposition}

Note that the weak gradient of a function $u$ is independent of the measure considered. Indeed, given a weighted manifold $(M, g, \mu)$, recall that for a given locally square-integrable function $u$ on $(M, g, \mu)$, its weak gradient is defined as the unique locally square-integrable vector field such that, for any smooth compactly supported vector field $\psi$ on $M$, it holds that
\[
\int_M u\, \frac{1}{D}\textup{div}(D \psi) D\; \textup{dVol}_g = -\int_M \langle v, \psi\rangle_g\, D\; \textup{dVol}_g,
\]
which can be rewritten as
\[
\int_M u\, \textup{div}(D \psi) \textup{dVol}_g = -\int_M \langle v, D\psi\rangle_g\; \textup{dVol}_g.
\]
Since this identity is satisfied by every smooth and compactly supported vector field on $M$, it suffices to consider $\tilde{\psi} := D \psi$ to conclude that the definition is independent of the measure $\mu$. 

Having defined the weak gradient of a locally square-integrable function, we can define the Sobolev space $H^1$. 
\begin{definition}[Sobolev space $H^1$]
\label{def:Sobolevspace}
We denote by $H^1(M, \mu)$ the \textup{Sobolev space}
\[
H^1(M, \mu) := \left\{u \in L^2(M, \mu) : \Grad{g}u \in \overrightarrow{L}^2(M, \mu) \right\},
\]
which can be endowed with the following inner product
\[
(u, v)_{H^1(M, \mu)} := (u, v)_{L^2(M, \mu)} + (\Grad{g}u, \Grad{g}v)_{\overrightarrow{L}^2(M, \mu)},\quad \forall u, v \in H^1(M, \mu).
\]
The associated norm is given by
\[
\norm{u}^2_{H^1(M, \mu)} = \norm{u}^2_{L^2(M, \mu)} + \norm{\Grad{g}u}^2_{\overrightarrow{L}^2(M, \mu)}.
\]
\end{definition}

The Sobolev space $H^1(M, \mu)$ endowed with the inner product described above is in fact a Hilbert space \cite[Lemma 2.5]{Grigoryan_2024}. Moreover, whenever $\mu$ is a finite measure on $M$, it follows from \cref{lem:nestedLp} that $H^1(M, \mu) \subset L^1(M, \mu)$.

Note that both the Sobolev and $L^2$ spaces have been defined for a general weighted manifold $(M, g, \mu)$. In particular, given any open set $U \subset M$ and some measure $\tilde \mu$ on $U$ given by a density function, we can define $L^2(U, \tilde\mu)$ and $H^1(U, \tilde \mu)$.

\begin{notation}
When the measure $\mu$ considered is the one given by the volume form associated with $g$---i.e. $D(x) \equiv 1$---we will simply write $L^2(M)$, $L^2_{\mathit{loc}}(M)$, $\overrightarrow{L}^2(M)$, $\overrightarrow{L}^2_{\mathit{loc}}(M)$ and $H^1(M)$ for simplicity.  
\end{notation}

Whenever $M$ is compact, the following holds.
\begin{remark}
\label{equivLpspaces}
Let $(M, g, \mu)$ be a compact weighted manifold, and let $U$ be an open subset of $M$ endowed with the restricted measure $\mu$ on $U$. Then 
\begin{enumerate}
    \item $L^2(U, \mu|_U) = L^2(U),$
    \item $L^2_{\mathit{loc}}(U, \mu|_U) = L^2_{\mathit{loc}}(U),$
    \item $\overrightarrow{L}^2(U, \mu|_U) = \overrightarrow{L}^2(U),$
    \item $\overrightarrow{L}^2_{\mathit{loc}}(U, \mu|_U) = \overrightarrow{L}^2_{\mathit{loc}}(U),$
\end{enumerate}
where equality is meant in the sense of the spaces being equal and having equivalent norms. 
\end{remark}
\begin{proof}
Let us only prove $1$, as it implies the rest of the equalities. Recall that $d\mu(x) = D(x) \textup{dVol}_g(x)$, where $D(x)$ is assumed to be smooth and positive. Thus, since $M$ is compact, there exist two constants $k_1$ and $k_2$ such that
\[
0 < k_1 \leq D(x) \leq k_2 < \infty,
\]
and so
\[
k_1 \int_U |f|^2 \textup{dVol}_g \leq \int_U |f|^2 d\mu|_U \leq k_2\int_U |f|^2 \textup{dVol}_g,
\]
which finishes the proof. 
\end{proof}

In fact, the same holds for Sobolev spaces. 
\begin{remark}
\label{equivSobolev}
Let $(M, g, \mu)$ be a compact weighted manifold, and let $U$ be an open subset of $M$ endowed with the restricted measure $\mu$ on $U$. Then 
\[
H^1(U) = H^1(U, \mu|_U),
\]
where equality is meant in the sense of the spaces being equal and having equivalent norms. 
\end{remark}

%% file: Chapters/BakrEmery.tex
\section{Bakry-Émery theory in domains with a convex boundary}
\label{BakryEmeryandLyapunov}

The goal of this section is to introduce the \textit{curvature-dimension condition}, formulated by Bakry and Émery \cite{bakry2013analysis}, and to explore its connection to the Poincaré inequality in a \textit{domain} with \textit{convex boundary}. In particular, given a Riemannian manifold $(M, g)$ and a smooth function $F$ on $M$, the curvature-dimension condition relates the Hessian of $F$, the Ricci curvature of $M$ and the metric $g$. Moreover, when the curvature-dimension condition holds on a domain of $M$ with a convex boundary, we can obtain a lower bound on the first non-zero eigenvalue of the Langevin dynamics generator
\[
\operatorname{L} =  -\Grad{g} F + \frac{1}{\beta} \Delta_g,
\]
which in turn yields a Poincaré inequality. 

\begin{definition}[Domain]
Let $M$ be a manifold, and let $U \subset M$ be a subset. We say that $U$ is a \textup{domain} if the closure of $U$ coincides with that of its interior. 
\end{definition}

Throughout the section we will only consider domains which are connected and have a smooth boundary. 

Let us now give two definitions of convexity. 
\begin{definition}[Weakly geodesically convex set]
Let $(M, g)$ be a Riemannian manifold. A subset $U \subset M$ is said to be \textup{weakly geodesically convex} if, for each $p, q \in U$ there exists a distance-minimizing geodesic in $M$ joining them and staying inside $U$. 
\end{definition}

For the second notion of convexity, we first need to define the \textit{second fundamental form} of a manifold with boundary. 
\begin{definition}[Second fundamental form of a manifold with boundary]
Let $(M, g)$ be a Riemannian manifold with boundary, and let $N$ be the inward pointing unit normal vector field of $\partial M$. We define the \textup{second fundamental form} of $\partial M$ as 
\[
\bb{I}(X, Y) := - \langle \nabla_X N, Y\rangle_g,
\]
for every $X, Y \in \mathfrak{X}(\partial M)$. 
\end{definition}

\begin{definition}[Convex boundary]
Let $(M, g)$ be a Riemannian manifold and let $U \subset M$ be a domain. We say that $\partial U$ is \textup{convex} if $\bb{I}(X,X) \geq 0$ for every $X \in \mathfrak{X}(\partial U)$.
\end{definition}

Although the two notions of convexity may seem quite different, it is known that if a closed domain is weakly geodesically convex, then its boundary is convex \cite[Proposition 2.1.5]{wang2006functional}.

In order to derive a Poincaré inequality on a convex domain, let us consider a slightly more general setting than in the main text. This way, let $(M, g)$ be a compact Riemannian manifold. Let us define the operator
\begin{equation}
\label{eq:generaloperatorL}
\mathcal{L} := \alpha(\Grad{g} V + \Delta_g), 
\end{equation}
where $\alpha > 0$ is some constant and $V \in C^\infty(M)$. In particular, choosing $\alpha = \frac{1}{\beta}$, $V = -\beta F$, for some smooth function $F$ on $M$, we recover the Langevin dynamics generator $\operatorname{L}$. 

As we mentioned in \cref{SectionPI}, $\mathcal{L}$ is naturally defined on $C^2(M)$, and it can be extended via its closure in $L^2(M, \mu)$ considering boundary conditions. Let us also denote this extension by $\mathcal{L}$. To each operator $\mathcal{L}$, we can associate a Markov triple $(M, \mu, \mathbf{\Gamma})$.

\begin{definition}[Carré du champ operator]
\label{def:carreduchampGeneric}
Let $(M, g)$ be a compact manifold, given some operator $\mathcal{L}$ as in \cref{eq:generaloperatorL}, we define the \textup{carré du champ operator} $\mathbf{\Gamma}$ on $C^2(M) \times C^2(M)$ as
\[
\mathbf{\Gamma}(f_1, f_2) := \frac{1}{2}(\mathcal{L}(f_1f_2) - f_1\mathcal{L}f_2 - f_2\mathcal{L}f_1),
\]
and denote $\mathbf{\Gamma}(f,f)$ by $\mathbf{\Gamma}(f)$. 
\end{definition}

Following an analogous reasoning to that followed in \cref{rem:explicitexprLangevin}, we can obtain an explicit expression for the carré du champ operator associated with $\mathcal{L}$. 

\begin{remark}
The carré du champ operator associated with $\mathcal{L}$ as defined in \cref{eq:generaloperatorL} can be written as
\[
\mathbf{\Gamma}(f_1, f_2) = \alpha\langle \Grad{g} f_1, \Grad{g} f_2\rangle_g,    
\]
for any $f_1, f_2 \in C^2(M)$, and so 
\[
\mathbf{\Gamma}(f) = \alpha |\Grad{g} f|^2_g,
\]
for any $f \in C^2(M)$.
\end{remark}

\begin{definition}[Markov triple]
\label{def:markovtripleGeneric}
Let $(M, g)$ be a compact manifold, and let $\mathcal{L}$ be defined as 
\[
\mathcal{L} = \alpha(\Grad{g} V + \Delta_g),
\]
where $\alpha > 0$ is some constant and $V \in C^\infty(M)$. The \textup{Markov triple} $(M, \mu, \mathbf{\Gamma})$ associated with $\mathcal{L}$ consists of
\[
d\mu := e^{V} \textup{dVol}_g,
\]
which can be assumed to be a probability measure by simply shifting $V$ by a suitable constant ($\mathcal{L}$ is invariant under this shift), and the associated carré du champ operator given by 
\[
\mathbf{\Gamma}(f_1, f_2) = \alpha \langle \Grad{g} f_1, \Grad{g} f_2\rangle_g,
\]
for any $f_1, f_2 \in C^2(M)$.
\end{definition}

To define the curvature-dimension condition, we first need to define the second order carré du champ operator.
\begin{definition}[Second order carré du champ operator]
\label{secondordercarreduchamp}
Given a compact manifold $(M, g)$ and an operator $\mathcal{L}$ defined as above, we define the \textup{second order carré du champ operator} as
\[
\mathbf{\Gamma}_2(f, h) := \frac{1}{2}(\mathcal{L}\mathbf{\Gamma}(f, h) - \mathbf{\Gamma}(f,\mathcal{L}h) - \mathbf{\Gamma}(\mathcal{L}f, h)),\quad \forall f, h \in C^{\infty}(M),
\]
where $\mathbf{\Gamma}$ is the carré du champ operator associated with $\mathcal{L}$. 
\end{definition}

\begin{definition}[Curvature-dimension condition]
Let $(M, g)$ be a compact Riemannian manifold, let $\mathcal{L}$ be as above and let $(M, \mu, \mathbf{\Gamma})$ be the associated Markov triple. Given some constant $\kappa \in \bb{R}$, we say that $(M, \mu, \mathbf{\Gamma})$ satisfies a (tight) \textup{curvature-dimension condition} $\textup{CD}(\kappa)$\footnote{Typically, the curvature-dimension condition as stated here is denoted as $\textup{CD}(\kappa, \infty)$. Nevertheless, we will use $\textup{CD}(\kappa)$ for readability.} if, for all $f \in C^{\infty}(M)$,
\begin{equation}
\label{eq:CDOriginal}
\mathbf{\Gamma}_2(f, f) \geq \kappa \mathbf{\Gamma}(f, f).
\end{equation}
\end{definition}

While this definition of the curvature-dimension condition is given with respect to a general Markov triple, we are interested in curvature-dimension condition associated with the Markov triple induced by the Langevin diffusion generator
\[
\operatorname{L} = -\Grad{g} F + \frac{1}{\beta} \Delta_g,
\]
where $F$ is some smooth function on $M$ and $\beta$ is some positive constant. To obtain such an expression, let us first present the second order carré du champ operator associated with $\operatorname{L}$. 
\begin{lemma}[{\cite[Lemma G.1]{LiErd2022Supp}}]
\label{lem:explicitsecondorder}
Let $(M, g)$ be a compact Riemannian manifold. Let $\beta >0$, let $F$ be some smooth function on $M$ and consider the Langevin diffusion generator $\operatorname{L}$. Then its associated second order carré du champ operator can be written as
\[
\Gamma_2(f, f) = \frac{1}{\beta^2}|\nabla^2 f|^2_g + \frac{1}{\beta^2} \textup{Ric}_g(\Grad{g} f, \Grad{g} f) + \frac{1}{\beta} \nabla^2F(\Grad{g} f, \Grad{g} f),
\]
for any $f \in C^{\infty}(M)$.
\end{lemma}

With this lemma and the explicit expression for the carré du champ operator $\Gamma$ associated with $\operatorname{L}$ in mind, we can rewrite the curvature-dimension condition associated with $\operatorname{L}$ in terms of the Ricci curvature of $M$, the metric $g$ and the Hessian of $F$. 
\begin{remark}
Let $(M, g)$ be a compact Riemannian manifold, let $\beta > 0$, let $F$ be some smooth function on $M$ and let $\operatorname{L}$ be the Langevin diffusion generator. Then  the curvature-dimension condition $\textup{CD}(\kappa)$ from \cref{eq:CDOriginal} is equivalent to 
\begin{equation}
\label{eq:CDineqaux}
\nabla^2 F(X,X) + \frac{1}{\beta} \textup{Ric}_g(X,X) \geq \kappa g(X,X),
\end{equation}
for every $X \in \mathfrak{X}(M)$.
\end{remark}
\begin{proof}
Assume first that \cref{eq:CDineqaux} holds. Using the explicit expression for the carré du champ operator $\Gamma$ associated with $\operatorname{L}$, which is given by 
\[
\Gamma(f, h)  = \frac{1}{\beta} \langle \Grad{g}f, \Grad{g}h\rangle_g,
\]
and using \cref{lem:explicitsecondorder} we know that
\begin{align*}
\Gamma_2(f, f) &= \frac{1}{\beta^2}|\nabla^2 f|^2_g + \frac{1}{\beta^2}\textup{Ric}_g(\Grad{g}f,\Grad{g}f) + \frac{1}{\beta}\nabla^2 F(\Grad{g}f, \Grad{g}f)     
\\&\geq \frac{1}{\beta^2}\textup{Ric}_g(\Grad{g}f,\Grad{g}f) + \frac{1}{\beta}\nabla^2 F(\Grad{g}f, \Grad{g}f) 
\\&\geq \frac{\kappa}{\beta}\langle\Grad{g}f, \Grad{g}f \rangle_g
\\&= \kappa \Gamma(f, f).
\end{align*}

To prove the converse result, let $p \in M$ and $X \in T_pM$. Let $f$ be a smooth function such that $\Grad{g}f(p) = X$ and $\nabla^2 f(p) = 0$. Then, at $p$,   
\begin{align*}
\Gamma_2(f, f) &= \frac{1}{\beta^2}|\nabla^2 f|^2_g + \frac{1}{\beta^2} \textup{Ric}_g(\Grad{g}f, \Grad{g}f) + \frac{1}{\beta} \nabla^2F(\Grad{g}f, \Grad{g}f)
\\&= \frac{1}{\beta^2} \textup{Ric}_g(X, X) + \frac{1}{\beta} \nabla^2F(X, X),
\end{align*}
and 
\[
\Gamma(f, f) = \frac{1}{\beta} |\Grad{g}f|^2_g = \frac{1}{\beta} |X|^2_g.
\]
Lastly, since
\[
\Gamma_2(f, f) \geq \kappa \Gamma(f, f),
\]
holds by assumption, we obtain that 
\[
\frac{1}{\beta^2} \textup{Ric}_g(X, X) + \frac{1}{\beta} \nabla^2F(X, X) \geq \frac{\kappa}{\beta} |X|^2_g,
\]
finishing the proof. 
\end{proof}

Let us finish the section by showing the connection between the curvature-dimension condition and the Poincaré inequality on a domain of a compact Riemannian manifold. As we mentioned before, we will be working on $U \subset M$, a \textit{connected domain} with smooth boundary $\partial U$ of a compact Riemannian manifold $(M, g)$. We will consider the operator 
\begin{equation}
\label{eq:operatorsinalpha}
\hat{\mathcal{L}} := \Grad{g}V + \Delta_g,
\end{equation}
where $V \in C^\infty(M)$, and the \textit{Neumann eigenvalue problem}
\begin{equation*}
\begin{cases}
\hat{\mathcal{L}}f = -\lambda f,\\
Nf|_{\partial U} = 0,
\end{cases}
\end{equation*}
where $N$ denotes the inward pointing unit normal vector field of $\partial U$. 

More specifically, we will study the minimum value $\lambda > 0$ for which there exists some $f \in C^2_N(U)$---the space of $C^2$ functions on $U$ with Neumann boundary condition---such that $\hat{\mathcal{L}} f = - \lambda f$. This quantity is known as the first Neumann eigenvalue of $\hat{\mathcal{L}}$, and is denoted by $\lambda_{1, \hat{\mathcal{L}}}(U)$. 

The first Neumann eigenvalue of $\hat{\mathcal{L}}$ corresponds to the Poincaré inequality associated with the Markov triple $(U, \mu|_U, |\Grad{g}\cdot|_g^2)$, where $\mu|_U = \mathds{1}_U e^{V(x)}\textup{dVol}_g(x)$ is assumed to be a probability measure. Before we prove this result in \cref{rmk:gapimpliesPI}, let us present a simple characterization of $\lambda_{1, \hat{\mathcal{L}}}$ (cf. \cite[Section 2.2]{wang2006functional}). In the following Lemma, given an operator $\mathscr{L}$, we will use the notation $(\mathscr{L}, D(\mathscr{L}))$ to highlight that its domain is $D(\mathscr{L})$. 
\begin{lemma}
Let $(M, g)$ be a compact Riemannian manifold and let $U$ be a connected bounded open domain of $M$ with smooth boundary. Let $\hat{\mathcal{L}}$ be defined as in \cref{eq:operatorsinalpha} and let $(\hat{\mathcal{L}}, D(\mathcal{\hat{L}}))$ be the unique self-adjoint extension (cf. \cite[Theorem 0.3.3]{wang2006functional}) of $(\hat{\mathcal{L}}, C^2_N(U))$ on $L^2(U, \mu)$. Then $\lambda_{1, \hat{\mathcal{L}}}(U)$ is the spectral gap of $(\hat{\mathcal{L}}, D(\hat{\mathcal{L}}))$ and so 
\begin{equation}
\label{firstNeumanneig}
\lambda_{1, \hat{\mathcal{L}}}(U) = \inf \{\mu|_U(|\Grad{g}f|_g^2) : f \in C^2(U),\ \left.\mu\right|_U(f) = 0,\ \left.\mu\right|_U(f^2) = 1\},
\end{equation}
where $d\mu|_U := \mathds{1}_U e^{V(x)}\,\textup{dVol}_g(x)$.
\end{lemma}

Let us now show that the first Neumann eigenvalue $\lambda_{1, \hat{\mathcal{L}}}(U)$ corresponds to the Poincaré inequality constant of $(U, \mu|_U, |\Grad{g}\cdot|_g^2)$. 
\begin{remark}
\label{rmk:gapimpliesPI}
Let $(M, g)$ be a compact Riemannian manifold and let $U$ be a connected bounded open domain of $M$ with smooth boundary. Let $\hat{\mathcal{L}}$ be defined as in \cref{eq:operatorsinalpha}, and consider its associated Markov triple $(U, \mu|_U, |\Grad{g}\cdot|_g^2)$, where $\mu|_U$ is assumed to be a probability measure. Let $\lambda_{1, \hat{\mathcal{L}}}(U)$ be the first Neumann eigenvalue of $\hat{\mathcal{L}}$ on $U$. Then $(U, \mu|_U, |\Grad{g}\cdot|_g^2)$ satisfies a Poincaré inequality with constant $\lambda_{1, \hat{\mathcal{L}}}(U)$, i.e. 
\[
\int_U f^2\, d\mu|_U - \Big(\int_U f\, d\mu|_U\Big)^2 \leq \frac{1}{\lambda_{1, \hat{\mathcal{L}}}(U)} \int_U |\Grad{g}f|_g^2 \, d\mu|_U,
\]
for every $f \in C^2(U) \cap H^1(U, \mu|_U)$. 
\end{remark}
\begin{proof}
Let $f \in C^2(U) \cap H^1(U,\mu|_U)$. If $f \equiv \mu|_U(f)$, the result follows automatically. Otherwise, consider 
\[
\tilde f := \frac{f - \left.\mu\right|_U(f)}{\left.\mu\right|_U((f - \left.\mu\right|_U(f))^2)^{\frac{1}{2}}}.
\]
Since $f \in L^2(U, \mu|_U)$ and $\mu|_U$ is a probability measure on $U$, \cref{lem:nestedLp} ensures that $f \in L^1(U, \mu|_U)$, and so $\tilde f$ is well-defined and belongs to $C^2(U) \cap H^1(U, \mu|_U)$. 

It is clear that $\left.\mu\right|_U(\tilde f) = 0$ and $\left.\mu\right|_U(\tilde{f}^2) = 1$. Thus by \cref{firstNeumanneig} it holds that 
\[
\lambda_{1, \hat{\mathcal{L}}}(U) \leq \left.\mu\right|_U(|\Grad{g}\tilde{f}|_g^2) = \frac{\left.\mu\right|_U(|\Grad{g}f|_g^2)}{\left.\mu\right|_U((f - \left.\mu\right|_U(f))^2)},
\]
i.e. 
\begin{equation}
\label{eq:preliminaryineqpoincare}
\left.\mu\right|_U((f - \left.\mu\right|_U(f))^2) \leq \frac{1}{\lambda_{1, \hat{\mathcal{L}}}(U)}\left.\mu\right|_U(|\Grad{g}f|_g^2).   
\end{equation}

Moreover, since $\mu|_U$ is a probability measure on $U$, it holds that 
\[
\left.\mu\right|_U((f - \left.\mu\right|_U(f))^2) = \left.\mu\right|_U(f^2) - 2\left.\mu\right|_U(f)^2 + \left.\mu\right|_U(f)^2 = \left.\mu\right|_U(f^2) - \left.\mu\right|_U(f)^2,
\]
and so we can rewrite \cref{eq:preliminaryineqpoincare} as
\[
\left.\mu\right|_U(f^2) - \left.\mu\right|_U(f)^2 \leq \frac{1}{\lambda_{1, \hat{\mathcal{L}}}(U)}\left.\mu\right|_U(|\Grad{g}f|_g^2)
\]
for every $f \in C^2(U) \cap H^1(U, \mu|_U)$. 
\end{proof}

The following proposition allows us to relate the curvature-dimension condition with the existence of a lower bound on the first eigenvalue of $\hat{\mathcal{L}}$. To do this, it is necessary to ask for $\partial U$ to be convex, as otherwise, the curvature-dimension condition does not necessarily imply a Poincaré inequality. See \cite{wang2009log} for more details. 

\begin{proposition}[{\cite[Theorem 2.2.5]{wang2006functional}}]
\label{thm:boundneumanneigval}
Let $(M, g)$ be a compact Riemannian manifold and let $U$ be a bounded open domain of $M$ with smooth boundary $\partial U$. Consider the operator
\[
\hat{\mathcal{L}} = \Grad{g}V + \Delta_g,
\]
where $V \in C^\infty(M)$. Assume that there exists some constant $\upkappa \in \bb{R}$ such that
\begin{equation}
\label{modifiedCD}
\textup{Ric}_g(X,X) - \langle \nabla_X\, \Grad{g}V, X\rangle_g \geq -\upkappa |X|^2_g,
\end{equation}
for every $X \in \mathfrak{X}(U)$. If $\partial U$ is convex, then the first Neumann eigenvalue of $\hat{\mathcal{L}}$ satisfies
\[
\lambda_{1, \hat{\mathcal{L}}}(U) \geq \frac{8}{D^2} - \frac{1}{2} \upkappa,
\]
where $D$ is the diameter of $(U, g|_U)$.
\end{proposition}

As we mentioned above, this result allows us to relate the curvature-dimension condition with the Poincaré inequality of the Markov triple $(U, \nu|_U, \Gamma)$ associated with the operator
\[
\operatorname{L} f = \langle -\Grad{g}F, \Grad{g}f\rangle_g + \frac{1}{\beta} \Delta_g\, f,
\]
where $F$ is some smooth function on $M$ and $\beta$ is some positive constant. 
\begin{proposition}
\label{rem:CDimpliesPI}
Let $(M, g)$ be a compact Riemannian manifold and let $U$ be a bounded open domain with smooth convex boundary $\partial U$. Let $F$ be some smooth function on $M$ and assume that there exists some constant $\kappa > 0$ for which 
\begin{equation}
\label{eq:CDconditionremark}
\nabla^2 F(X, X) + \frac{1}{\beta} \textup{Ric}_g(X, X) \geq \kappa g(X, X),
\end{equation}
for every $X \in \mathfrak{X}(U)$. Then  it holds that 
\[
\int_U f^2\, d\nu|_U  - \Big(\int_U f\, d\nu|_U\Big)^2 \leq \frac{2}{\kappa\beta} \int_U |\Grad{g}f|_g^2 \, d\nu|_U,
\]
for every $f \in C^2(U) \cap H^1(U, \nu|_U)$. 
\end{proposition}
\begin{proof}
First, let $V$ be defined as 
\begin{equation}
\label{eq:defofV}
V(x) := -\beta F(x) - \log(Z),
\end{equation}
where $Z = \int_U e^{-\beta F(x)} \textup{dVol}_g(x)$, and define 
\[
\textup{\textbf{L}} := \beta \langle - \Grad{g}F, \cdot\rangle_g + \Delta_g.
\]

Let us obtain a lower bound on the first Neumann eigenvalue $\lambda_{1, \textup{\textbf{L}}}(U)$. In order to apply \cref{thm:boundneumanneigval}, let us begin by rewriting the curvature-dimension condition. 

The curvature-dimension condition from \cref{eq:CDconditionremark} can be rewritten as 
\begin{equation}
\label{eq:compare1}
\beta \nabla^2 F +\textup{Ric}_g \geq \beta \kappa g.    
\end{equation}

If we now rewrite the left-hand side of \cref{modifiedCD} for $V(x)$ as defined in \cref{eq:defofV}, using the fact that
\[
\nabla^2 F(X, Y) = \langle \nabla_X \Grad{g}F, Y\rangle_g,
\]
we arrive at
\begin{equation}
\label{eq:compare2}
\textup{Ric}_g(X,X) +\beta \nabla^2 F(X,X) \geq -\upkappa |X|^2_g.    
\end{equation}

Comparing \cref{eq:compare1,eq:compare2}, we conclude that \cref{eq:compare2} holds in our case for $-\upkappa := \beta\kappa$. Thus, we can apply \cref{thm:boundneumanneigval} to find that 
\[
\lambda_{1, \textup{\textbf{L}}}(U) \geq \frac{8}{D^2} + \frac{1}{2}\beta \kappa \geq \frac{\beta \kappa}{2},
\]
where $D$ is the diameter of $U$. And so, using \cref{rmk:gapimpliesPI}, since $\nu(x)|_U := \mathds{1}_U\frac{1}{Z} e^{-\beta F}$ is a probability measure on $U$, it holds that
\[
\int_U f^2\, d\nu|_U - \Big(\int_U f \, d\nu|_U \Big)^2 \leq \frac{2}{\kappa \beta} \int_U |\Grad{g}f|_g^2 \, d\nu|_U,
\]
for every $f \in C^2(U) \cap H^1(U, \nu|_U)$. 
\end{proof}

%% file: Chapters/Laplacian.tex
\section{Studying the pseudo-radial process}
\label{sec:LaplacianAndGradientBound}

The goal of this section is to analyze in detail the function $\tilde r_{p, v}$ defined in \cref{sec2.2} and used in \cref{prop9.6} to construct quasi-Lyapunov functions around the saddle points of the objective function. Given a manifold $(B, h)$, some fixed point $p \in B$ and a fixed tangent vector $v \in T_p B$ with norm one, $\tilde r_{p, v}$ is defined in normal coordinates outside the cut locus of $p$ as
\begin{equation}
\label{eq:deftilder}
\tilde r_{p, v}(x) = \langle v, \log_p x\rangle,
\end{equation}
where $\log_p x$ denotes the inverse of the exponential mapping, and $\langle \cdot , \cdot \rangle$ is the usual Euclidean inner product. 

In the proof of \cref{prop9.6} both the norm of the gradient of $\tilde r_{p, v}$ and its Laplacian need to be bounded. A natural strategy would be to try to show that its Laplacian is non-negative and its gradient has constant norm one. However, in this section we will first show that the Laplacian of $\tilde r_{p, v}$ cannot be assumed to be non-negative, nor the norm of its gradient can be assumed to be constant and equal to one. Then we will show how, under the assumption of $(B, h)$ being a symmetric space (cf. \cref{def:symmetricspace}), the absolute value of its Laplacian and the norm of its gradient stay close to zero and one, respectively, when restricting $\tilde r_{p, v}$ to a geodesic ball centered at $p$ of sufficiently small radius. Lastly, we will state and prove the auxiliary lemmas \ref{lem:boundtermgradlaplacian} and \ref{AuxLemma3} used in \cref{sec:secondlyapunovfunction}. 

\subsection{No-go results for the Laplacian and the gradient of \texorpdfstring{$\tilde r_{p,v}$}{rpv}}
\label{sec:notrivialboundlaplacian}

Consider a very simple Riemannian manifold; the two-dimensional round sphere.

\begin{remark}
\label{rem:laplaciannotpositive}
Let $(\bb{S}^2, g_{\mathit{round}})$ be the two-sphere endowed with the usual round metric, let $p \in \bb{S}^2$ be a fixed point and let $v \in T_p \bb{S}^2$ with $|v|_{g_{\mathit{round}}} = 1$. There is no radius $R$ such that $\Delta_{g_{\mathit{round}}} \tilde r_{p, v}$ is non-negative when restricted to $\mathcal{B}(R, p)$. 
\end{remark}

\begin{proof}
Without loss of generality, we can assume that $p \in \bb{S}^2$ has coordinates $(0, 0, 1)$, and therefore $v = (v_1, v_2, 0)$ with $v_1^2 + v_2^2 = 1$. In these coordinates, 
\[
\log_p(x) = (\theta \cos\varphi, \theta\sin\varphi, 0)
\]
where $\theta \in (0, \pi)$ and $\varphi \in [0, 2\pi)$. Therefore,
\[
\tilde r_{p, v}(x) = \theta(v_1\cos\varphi + v_2 \sin\varphi).
\]

It is a standard fact that the round metric of $\bb{S}^2$ in geodesic polar coordinates is of the form
\[
g_{\mathit{round}}(\theta, \varphi) = d\theta^2 + \sin^2 \theta d\varphi^2,
\]
which implies that
\[
\Gamma^\theta_{\varphi \varphi} = - \sin\theta\cos\theta,\quad \Gamma^\varphi_{\varphi \theta} = \Gamma^\varphi_{\theta\varphi} = \cot\theta,
\]
while the rest of the Christoffel symbols vanish identically. 

Since the Laplacian can be written in local coordinates as
\begin{equation}
\label{eq:Laplacianincoords}
\Delta_g = g^{ij}(x) \frac{\partial^2}{\partial x^i \partial x^j} - g^{kl}(x) \Gamma^n_{kl}(x) \frac{\partial}{\partial x^n},
\end{equation}
we can use the expressions for the Christoffel symbols, the metric and its inverse to conclude that
\begin{align*}
\Delta_{g_\mathit{round}} &= g^{\theta \theta}_{\mathit{round}}(x) \frac{\partial^2}{\partial \theta^2} + g^{\varphi \varphi}_{\mathit{round}}(x) \frac{\partial^2}{\partial \varphi^2} - g^{\varphi\varphi}_{\mathit{round}}(x) \Gamma^\theta_{\varphi\varphi}(x) \frac{\partial}{\partial \theta} 
\\&= \frac{\partial^2}{\partial \theta^2} + \frac{1}{\sin^2 \theta} \frac{\partial^2}{\partial \varphi^2} + \cot \theta \frac{\partial}{\partial \theta},
\end{align*}
and so 
\begin{align*}
\Delta_{g_\mathit{round}} \tilde r_{p, v}(x) &= \Big(\frac{\partial^2}{\partial \theta^2} + \frac{1}{\sin^2 \theta} \frac{\partial^2}{\partial \varphi^2} + \cot \theta \frac{\partial}{\partial \theta}\Big)\theta(v_1\cos\varphi + v_2 \sin\varphi)
\\&= -\frac{\theta}{\sin^2 \theta}(v_1\cos\varphi+v_2\sin\varphi) + \cot\theta (v_1\cos\varphi + v_2 \sin\varphi)
\\&=  (v_1\cos\varphi + v_2 \sin\varphi)\Big(\cot\theta -\frac{\theta}{\sin^2 \theta}\Big).
\end{align*}

In order to obtain the result in Cartesian normal coordinates, we can simply substitute 
\[
y_1 := \theta \cos \varphi,\ y_2 := \theta \sin\varphi,\ \rho := \sqrt{y_1^2 + y_2^2} = \theta,
\]
to obtain that 
\[
\Delta_{g_\mathit{round}} \tilde r_{p, v}(x) = (v_1 y_1 + v_2 y_2 )\Big(\frac{\cot \rho}{\rho} - \frac{1}{\sin^2\rho}\Big). 
\]
Observing this identity we can conclude that, even if the Laplacian of $\tilde r_{p, v}$ tends to zero as $\rho$ becomes small, this does not imply that $\Delta_{g_\mathit{round}}\, \tilde r_{p, v}(x) \geq 0$ near $p$. 
\end{proof}

Now, working again with the two-sphere, let us show that the norm of the gradient of $\tilde r_{p, v}$ cannot be assumed to be constant. 

\begin{remark}
There is no radius $R > 0$ such that $|\Grad{g_{\mathit{round}}} \tilde r_{p, v}|_{g_{\mathit{round}}} \equiv 1$ when restricted to the ball $\mathcal{B}(R, p)$. 
\end{remark}
\begin{proof}
Using the expressions for $\tilde r_{p, v}$ and the round metric in normal coordinates obtained in the previous remark, along with the expression of $\Grad{g_{\mathit{round}}} \tilde r_{p, v}$ in local coordinates, namely
\[
\Grad{g_{\mathit{round}}} \tilde r_{p, v}(x) = g_{\mathit{round}}^{ij}(x) \frac{\partial \tilde r_{p, v}}{\partial x^i} \partial_j,
\]
we conclude that, in geodesic polar coordinates, 
\[
\Grad{g_{\mathit{round}}} \tilde r_{p, v}(x) = (v_1 \cos \varphi + v_2 \sin \varphi) \partial_{\theta} + \frac{1}{\sin^2\theta} \theta (v_2 \cos\varphi - v_1 \sin \varphi)\partial_{\varphi}.
\]
Therefore
\[
|\Grad{g_{\mathit{round}}} \tilde r_{p, v}(x)|^2_{g_{\mathit{round}}} = (v_1 \cos \varphi + v_2 \sin \varphi)^2 + \frac{1}{\sin^2\theta} \theta^2 (v_2 \cos\varphi - v_1 \sin \varphi)^2.
\]
Again, writing this expression in Cartesian normal coordinates we obtain
\begin{align*}
|\Grad{g_{\mathit{round}}} \tilde r_{p, v}(x)|^2_{g_{\mathit{round}}} &= \frac{(v_1y_1 + v_2 y_2)^2}{\rho^2} + \frac{(v_2 y_1 - v_1 y_2)^2}{\sin^2\rho}
\\&= 1 + \Big(\frac{1}{\sin^2\rho} - \frac{1}{\rho^2}\Big) (v_2 y_1 - v_1 y_2)^2.
\end{align*}
where we used the fact that $(v_1y_1 + v_2 y_2)^2 + (v_2 y_1 - v_1 y_2)^2 = (v^2_1 + v_2^2)y_1^2 + (v_1^2 + v_2^2)y_2^2 = \rho^2$, as $v_1^2 + v_2^2 = 1$. 

Although it is easy to see that the norm of the gradient tends to $1$ as $\rho$ approaches zero, it is clear that the expression is not constant in $\mathcal{B}(R, p)$ for any value of $R$. 
\end{proof}
\subsection{Bounding the Laplacian in symmetric spaces}
\label{sec:boundingLaplaciansymmetricspaces}

Although $\Delta_h \tilde r_{p, v}$ cannot be assumed to be non-negative, here we will prove that if $\tilde r_{p, v}$ is defined on a symmetric space, its absolute value can be bounded when restricting $\tilde r_{p, v}$ to a sufficiently small geodesic ball. Throughout this section, $(B, h)$ will be assumed to be a complete and connected symmetric space of dimension $d$. 

Let $p \in B$ be fixed. Using normal coordinates centered at $p$ we can rewrite \cref{eq:deftilder} as 
\[
\tilde r_{p, v}(x) = \delta_{ij} v^i x^i.
\]
Therefore, using the expression of the Laplacian in local coordinates (cf.  \cref{eq:Laplacianincoords}), and assuming without loss of generality that $v = \frac{\partial }{\partial x^1}$, we can rewrite the Laplacian of $\tilde r_{p, v}$ as
\begin{equation}
\label{eq:exprlaplacianr}
\Delta_h \tilde r_{p, v}(x) = -h^{kl}(x) \Gamma^{1}_{kl}(x).
\end{equation}
Thus, if we are able to bound the coefficients of the inverse metric and the Christoffel symbols near $p$, we will be able to bound $\Delta_h \tilde r_{p, v}(x)$, as desired. To this end, we will first obtain an explicit expression for the Taylor series of the coefficients of the metric $h$ in normal coordinates. This expression could be of independent interest; although there are many references in which the Taylor expansion of the metric is computed up to a certain order---to the best of our knowledge, always lower than 7th---the full Taylor series of $h$ is not so easy to find in the literature. This will be done below. Then, in \cref{sec:laplacianbound} we will obtain bounds for the absolute value of $h^{ij}(x)$ and $\partial_a h^{ij}(x)$, for every $a, i, j \in \{1, \dotsc, d\}$ using the Taylor expansion of $h$. These bounds will allow us to relate the Laplacian of $\tilde r_{p, v}$ to the coefficients of the Riemann curvature tensor at $p$ and the radius of the ball to which $\tilde r_{p, v}$ is restricted. 

\subsubsection{Taylor expansion of the metric coefficients}
\label{sec:taylorexpansionmetric}

As we mentioned earlier, we will give an explicit expression for the Taylor series of the coefficients of $h$ in normal coordinates, when $(B, h)$ is a complete and connected symmetric space. Indeed, whenever this is the case, the metric $h$ is analytic (cf. \cite[Proposition 5.5, Chapter IV]{helgason1979differential}). This subsection is based on the proof of \cite[Theorem 2.65]{andrews2010ricci}.

\begin{proposition}
\label{thm:taylorexpansionmetric}
Let $(B, h)$ be a complete and connected $d$-dimensional Riemannian manifold which is furthermore a symmetric space, and let $p \in B$ be some fixed point in $B$. For every $x \in B$ outside the cut locus of $p$, the coefficients $h_{ij}(x)$ can be written in normal coordinates centered in $p$ as
\[
h_{ij}(x) = \delta_{ij} + \sum_{n = 1}^{\infty} \frac{1}{(2n)!}h^{(2n)}_{ij}(0),
\]
whenever the right-hand side is a convergent series, with
\[
h^{(k)}_{ij}(t) := (\partial_t)^k h_{ij}(\gamma(t)),
\]
and $\gamma$ is defined as
\begin{align*}
\gamma: \mathbb{R} &\to B\\
t &\mapsto \exp_p(t\log_p x).
\end{align*}
\end{proposition}

Before we proceed to the proof of this result, let us prove some auxiliary lemmas, which will be fundamental when obtaining an explicit expression for the Taylor expansion of $h$. 

\begin{definition}
Let $(B, h)$ be a complete and connected Riemannian manifold. For every fixed $p \in B$ and any $x \in B$ outside the cut locus of $p$, we define $\varphi$ as
\begin{align*}
\varphi: \mathbb{R}^2 &\to B\\
(s, t) &\mapsto \exp_p(tV(s)),
\end{align*}
where $V: \bb{R} \to T_p B$ is some smooth function satisfying $V(0) = \log_p x$. 
\end{definition}

\begin{notation}
Let us denote the pullback connection by $\varphi$ as ${}^{\varphi}\nabla$, which is defined as 
\[
{}^{\varphi}\nabla_X Y|_{(s, t)} := \nabla_{\varphi_* X} Y|_{\varphi(s, t)},
\]
for every $X \in \mathfrak{X}(\bb{R}^2)$ and any $Y \in \mathfrak{X}(\varphi)$. 
\end{notation}

\begin{lemma}
\label{lem:pushforwardandcovariant}
Let $(B, h)$ be a complete and connected Riemannian manifold. Let $p \in B$ be fixed, let $x \in B$ be outside the cut locus of $p$ and let $\varphi$ be as defined previously. Consider $(\partial_s, \partial_t)$ to be the (global) frame of $\bb{R}^2$ associated with the coordinates $(s, t)$. Let us denote $u := V(0), w:= V'(0)$. Then  it holds that
\begin{align}
\label{eq:firstcovder}
\begin{split}
\varphi_* \partial_s|_{(0,0)} &=  0,\\
\varphi_* \partial_t|_{(0, 0)} &= u,\\
{}^{\varphi}\nabla_{\partial_t}\, \varphi_* \partial_t|_{(0,0)} & = 0,\\
{}^{\varphi}\nabla_{\partial_t}\,\varphi_* \partial_s|_{(0,0)} &= w.
\end{split}
\end{align}
Moreover, it follows that ${}^{\varphi}\nabla_{\partial_t}\, \varphi_* \partial_t \equiv 0$. 
\end{lemma}
\begin{proof}
Writing $\varphi$ in normal coordinates centered at $p$ one gets
\[
\varphi(s, t) = (tV^1(s), \dotsc, tV^d(s)),
\]
so it is easy to see that the pushforward of the vector fields $\partial_s$ and $\partial_t$ by $\varphi$ are given by
\[
\varphi_* \partial_s|_{(s,t)} = d\varphi(\partial_s)|_{(s, t)} = \partial_s \varphi(s, t) = t\,\partial_s V^i(s) \partial_i|_{\varphi(s, t)},
\]
and
\[
\varphi_* \partial_t|_{(s, t)} = V^{i}(s) \partial_i|_{\varphi(s, t)}.
\]

Now, since in normal coordinates the Christoffel symbols all vanish at $p = \varphi(0, 0)$, we can conclude that
\[
\nabla_X Y|_p = X(Y^k)\frac{\partial }{\partial x^k}, \quad \forall X, Y \in \mathfrak{X}(B). 
\]
Therefore, 
\begin{align*}
{}^{\varphi}\nabla_{\partial_t}\, \varphi_* \partial_t|_{(0,0)} &= \nabla_{\varphi_* \partial_t} \,\varphi_* \partial_t|_{\varphi(0,0)} = (\partial_t V^i(s)) \partial_i|_{\varphi(0,0)} = 0,\\
{}^{\varphi}\nabla_{\partial_t}\,\varphi_* \partial_s|_{(0,0)} &= \nabla_{\varphi_* \partial_t} \,\varphi_* \partial_s|_{\varphi(0,0)} = \partial_t (t \partial_s V^i(s)) \partial_i|_{\varphi(0,0)} = \partial_s V^i(0) \partial_i|_{\varphi(0,0)} = w,
\end{align*}
where we have used the fact that $(\varphi_* \partial_t) f = \partial_t (f \circ \varphi)$.     

Moreover, since $\gamma_{V(s)}(\cdot) := \varphi(s, \cdot)$ is a geodesic for every $s$, it holds that ${}^{\varphi}\nabla_{\partial_t}\, \varphi_* \partial_t \equiv 0$, finishing the proof. 
\end{proof}

\begin{lemma}
\label{lem:covariantderivativeszero}
Under the same hypothesis of \cref{eq:firstcovder}, for every $k \in \mathbb{N}$, it holds that
\[
{}^{\varphi}\nabla^{(2k)}_{\partial_t} \varphi_* \partial_s|_{(0,0)}  = 0.
\]
\end{lemma}
\begin{proof}
First, note that for each $k \geq 2$ (cf. \cite[Claim 2.67]{andrews2010ricci})
\begin{equation}
\label{eq:originalcovderkth}
{}^{\varphi}\nabla^{(k)}_{\partial_t} \varphi_* \partial_s = \sum_{l = 0}^{k-2} {\binom{k-2}{l}} (\nabla^{(k-2-l)} R)_{\varphi} (\underbrace{\varphi_* \partial_t, \dotsc, \varphi_*\partial_t}_{k-2-l \text{ times}}, {}^{\varphi}\nabla^{(l)}_{\partial t} (\varphi_* \partial s),  \varphi_* \partial_t)\varphi_* \partial_t, 
\end{equation}
where $R$ denotes the Riemann curvature tensor of $(B, h)$. 

Since $(B, h)$ is a complete and connected symmetric space by assumption, its Riemann curvature tensor is parallel, i.e. $\nabla R \equiv 0$ (cf. \cite[Corollary 12.7]{lee2018introductionRiemannian}). Therefore, \cref{eq:originalcovderkth} simplifies to
\begin{equation}
\label{eq:kthcovdertivative}
{}^{\varphi}\nabla^{(k)}_{\partial_t} \varphi_* \partial_s = R_{\varphi} ({}^{\varphi}\nabla^{(k-2)}_{\partial t} (\varphi_* \partial s),  \varphi_* \partial_t)\varphi_* \partial_t. 
\end{equation}
Moreover, recall that $\varphi_* \partial_s|_{\varphi(0,0)} = 0$
and so 
\[
{}^{\varphi}\nabla^{(2)}_{\partial_t} \varphi_* \partial_s|_{(0,0)}  = R_{\varphi(0,0)} (\varphi_* \partial s,  \varphi_* \partial_t)\varphi_* \partial_t = R_{\varphi(0,0)}(0, u, u) = 0. 
\]
The result now follows by induction on $k$. 
\end{proof}

With the two previous lemmas in mind, let us now prove \cref{thm:taylorexpansionmetric}.

\begin{proof}[{Proof of \cref{thm:taylorexpansionmetric}}]

Let $x \in B$ be outside the cut locus of $p$, and consider the curve $\gamma(t) = \varphi(0, t) = \exp_p(t\log_p x)$. Whenever 
\[
\sum_{n = 1}^{\infty} \frac{1}{n!}h^{(n)}_{ij}(0) t^n
\]
converges, we can write
\[
h_{ij}(\gamma(t)) = \delta_{ij} + \sum_{n = 1}^{\infty} \frac{1}{n!}h^{(n)}_{ij}(0) t^n,
\]
where $h^{(n)}_{ij}$ is as defined in \cref{thm:taylorexpansionmetric}. Choosing $t = 1$ we obtain that 
\[
h_{ij}(x) = \delta_{ij} + \sum_{n = 1}^{\infty} \frac{1}{n!}h^{(n)}_{ij}(0).
\]

In order to obtain an explicit formula for $h^{(n)}_{ij}(0)$, we start by showing that, using the expression for $\varphi_* \partial_s$ shown in \cref{lem:pushforwardandcovariant}, we have 
\begin{equation}
\label{eq:pullbackmetric}
h_{\varphi(s, t)}(\varphi_* \partial_s, \varphi_* \partial_s) = t^2 \partial_s V^i(s) \partial_s V^j(s) h_{ij}(\varphi(s, t)). 
\end{equation}
Thus, using the general Leibniz rule on \cref{eq:pullbackmetric} we find that for every $k \geq 2$
\begin{align*}
(\partial_t)^{k} h_{\varphi(s, t)}(\varphi_* \partial_s, \varphi_* \partial_s) &= \big\{k(k-1) (\partial_t)^{k-2} h_{ij}(\varphi(s, t)) + 2k t (\partial_t)^{k-1} h_{ij}(\varphi(s, t))
\\&\quad + t^2 (\partial_t)^{k} h_{ij}(\varphi(s, t))\big\} \partial_s V^i(s) \partial_s V^j(s).
\end{align*}
In particular, for $(s, t) = (0, 0)$,
\begin{equation}
\label{eq:derivativemetric2}
(\partial_t)^k h_{\varphi(0,0)}(\varphi_* \partial_s, \varphi_* \partial_s) = k(k-1)  h^{(k-2)}_{ij}(0) \partial_s V^i(0) \partial_s V^j(0) = k (k-1) h^{(k-2)}_{ij}(0) w^i w^j,
\end{equation}
where $w = V'(0)$. 

On the other hand, since ${}^{\varphi}\nabla$ is compatible with $h_{\varphi}$ (cf. \cite[Proposition 2.58]{andrews2010ricci}), it follows that
\[
\partial_t h_{\varphi(s, t)}(\varphi_* \partial_s, \varphi_* \partial_s) = h_{\varphi(s, t)}({}^{\varphi}\nabla_{\partial_t} \varphi_* \partial_s, \varphi_* \partial_s) + h_{\varphi(s, t)}(\varphi_* \partial_s, {}^{\varphi}\nabla_{\partial_t} \varphi_* \partial_s),
\]
and so, for every $k \geq 0$ it holds that
\begin{equation}
\label{eq:derivativemetric1}
(\partial_t)^k h_{\varphi(s, t)}(\varphi_* \partial_s, \varphi_* \partial_s) = \sum_{l = 0}^k \binom{k}{l} h_{\varphi(s, t)}({}^{\varphi}\nabla^{(k-l)}_{\partial_t} \varphi_* \partial_s, {}^{\varphi}\nabla^{(l)}_{\partial_t} \varphi_* \partial_s).
\end{equation}

Putting together \cref{eq:derivativemetric1,eq:derivativemetric2} we conclude that for every $k \geq 0$
\begin{equation}
\label{eq:expressionderivativemetric}
(k+2)(k+1)h^{(k)}_{ij}(0) w^i w^j = \sum_{l = 0}^{k+2} \binom{k+2}{l} h_{\gamma(0)}({}^{\varphi}\nabla^{(k+2-l)}_{\partial t} \varphi_* \partial s, {}^{\varphi}\nabla^{(l)}_{\partial t} \varphi_* \partial s ).
\end{equation}
In particular, this identity implies that if $k$ is odd, 
\begin{equation}
\label{eq:oddderivativesmetriczero}
(k+2)(k+1)h^{(k)}_{ij}(0) w^i w^j = 0,
\end{equation}
since for every value of $l$ either $k+2-l$ or $l$ is even, and so by \cref{lem:covariantderivativeszero}  the right-hand side of \cref{eq:expressionderivativemetric} is zero. 

If $k$ is even, following an analogous reasoning we conclude that
\begin{equation}
\label{eq:identitywithderivativesofg}
(k+2)(k+1)h^{(k)}_{ij}(0) w^i w^j = \sum_{\substack{l = 1\\l \text{ odd}}}^{k+1} \binom{k+2}{l} h_{\gamma(0)}({}^{\varphi}\nabla^{(k+2-l)}_{\partial t} \varphi_* \partial s, {}^{\varphi}\nabla^{(l)}_{\partial t} \varphi_* \partial s ).    
\end{equation}
Going back to the Taylor expansion of $h_{ij}(x)$, we get that 
\[
h_{ij}(x) = \delta_{ij} + \sum_{n = 1}^{\infty} \frac{1}{(2n)!}h^{(2n)}_{ij}(0),
\]
finishing the proof. 
\end{proof}

\begin{remark}
Although we did not give an explicit formula for the terms in the Taylor series of $h_{ij}(x)$ in the statement of \cref{thm:taylorexpansionmetric}, we did find such an explicit formula in the proof. In particular, we proved that for every $n \in \bb{N}$, the coefficients $h_{ij}^{(2n)}(0)$ satisfy \cref{eq:identitywithderivativesofg}.
\end{remark}

Before we obtain a bound for the Laplacian of the function $\tilde r_{p, v}$, let us first give a simpler and more explicit expression for the even derivatives of the metric coefficients of $(B, h)$. 
\begin{proposition}
\label{prop:explicitexpresisonmetricderivatives}
Let $(B, h)$ be a complete and connected symmetric space of dimension $d$. Let us denote by $R$ its Riemann curvature tensor. Let $p \in B$ be fixed and let $x \in B$ be outside the cut locus of $p$. Let $\gamma(t) := \exp_p(t \log_p x)$. Then, in normal coordinates centered at $p$,
\[
h^{(2n)}_{ij}(0) = \frac{2^{2n+1} }{(2n+2)(2n+1)} \delta_{\alpha j} T^\alpha_{i, i_1, \dotsc, i_{2n}} x^{i_1}\dotsm x^{i_{2n}}, \quad \forall n \geq 0, \forall i, j \in \{1, \dotsc, d\},
\]
where
\[
h^{(k)}_{ij}(t) = (\partial_t)^k h_{ij}(\gamma(t)),
\]
and
\begin{align*}
T^\alpha_\beta &:= \delta^\alpha_\beta,\\
T^\alpha_{\beta i_1i_2} &:= R^\alpha_{\beta i_1 i_2},\\
T^\alpha_{\beta i_1, \dotsc, i_{2n}} &:= R^\alpha_{\alpha_1 i_1 i_2} \dotsm R^{\alpha_{n-2}}_{\alpha_{n-1} i_{2n-3} i_{2n-2}} R^{\alpha_{n-1}}_{\beta i_{2n-1} i_{2n}},\quad n \geq 2,
\end{align*}
where the coefficients of $R$ are considered in normal coordinates centered at $p$ and evaluated at $p$.  
\end{proposition}
\begin{proof}
Let us define $\varphi$ as in \cref{lem:pushforwardandcovariant}. We will obtain expressions for ${}^{\varphi}\nabla^{(n)}_{\partial_t} \varphi_* \partial_s|_{(0, 0)}$ in normal coordinates, for odd values of $n$---the even case was studied in \cref{lem:covariantderivativeszero}. Using \cref{eq:firstcovder,eq:kthcovdertivative} we see that, in normal coordinates centered at $p$, 
\begin{align*}
{}^{\varphi}\nabla_{\partial_t}\, \varphi_* \partial_s|_{(0, 0)} &= w^b \partial_b,\\
{}^{\varphi}\nabla^{(3)}_{\partial_t} \varphi_* \partial_s|_{(0, 0)} &= R_{\varphi(0,0)} ({}^{\varphi}\nabla^{}_{\partial t} (\varphi_* \partial s),  u)u = R_{\varphi(0,0)} (w,  u)u = R_{ai_1 i_2}^b w^a x^{i_1} x^{i_2} \partial_b,\\
{}^{\varphi}\nabla^{(2n+1)}_{\partial_t} \varphi_* \partial_s|_{(0, 0)} &= R^b_{\alpha_1 i_1 i_2} \dotsm R^{\alpha_{n-2}}_{\alpha_{n-1} i_{2n-3} i_{2n-2}} R^{\alpha_{n-1}}_{a i_{2n-1} i_{2n}} w^a x^{i_1} \dotsm x^{i_{2n}} \partial_b, \quad n \geq 2,
\end{align*}
where again $w = V'(0)$ and $u = \log_p x$. To simplify the notation, we can write these expressions using the notation from the statement of the proposition,
\[
{}^{\varphi}\nabla^{(2n+1)}_{\partial_t} \varphi_* \partial_s|_{(0, 0)} = T^\alpha_{\beta i_1, \dotsc, i_{2n}} w^\beta x^{i_1}\dotsm x^{i_{2n}}\partial_\alpha,
\]
for every $n \geq 0$. 

This way, since in normal coordinates $h_{\varphi(0,0)}$ is the Euclidean metric, for every $n \in \bb{N}$ and every odd value $2l+1$, $0 \leq l \leq n$ we can write
\begin{align*}
h_{\gamma(0)}({}^{\varphi}\nabla^{(2n+2-(2l+1))}_{\partial t} \varphi_* \partial s, {}^{\varphi}\nabla^{(2l+1)}_{\partial t} \varphi_* \partial s ) &= \delta_{\alpha \beta}T^\alpha_{i, i_1, \dotsc, i_{2(n-l)}} T^\beta_{j, j_1, \dotsc, j_{2l}}
\\&\quad \times w^i w^j x^{i_1}\dotsm x^{i_{2(n-l)}}  x^{j_1}\dotsm x^{j_{2l}}\\
&= \delta_{\alpha \beta}T^\alpha_{i, i_1, \dotsc, i_{2(n-l)}} T^\beta_{j, i_{2(n-l)+1}, \dotsc, i_{2n}} w^i w^j x^{i_1}\dotsm x^{i_{2n}}
\\&= \delta_{\alpha j} T^\alpha_{i, i_1, \dotsc, i_{2n}} w^i w^j x^{i_1}\dotsm x^{i_{2n}},
\end{align*}
where the last equality follows simply by re-ordering the indices and the symmetry of $R$. Inserting this identity into \cref{eq:identitywithderivativesofg}, since $w$ was arbitrary, we obtain that for every $i, j \in \{1, \dotsc, d\}$
\begin{align}
\label{eq:expansiontaylorterms}
\begin{split}
(2n+2)(2n+1)h^{(2n)}_{ij}(0) &= \sum_{\substack{l = 1\\l \text{ odd}}}^{2n+1} \binom{2n+2}{l}\delta_{\alpha j} T^\alpha_{i, i_1, \dotsc, i_{2n}} x^{i_1}\dotsm x^{i_{2n}}
\\&= 2^{2n+1} \delta_{\alpha j} T^\alpha_{i, i_1, \dotsc, i_{2n}} x^{i_1}\dotsm x^{i_{2n}}.
\end{split}
\end{align}

\end{proof}

\subsubsection{Bounding the Laplacian}
\label{sec:laplacianbound}

Having obtained an explicit expression for the Taylor series of the coefficients of the metric $h$ in normal coordinates, we now proceed to upper bound the absolute value of $h_{ij}(x)$ and $\partial_a h_{ij}(x)$ for every $a, i, j \in \{1, \dotsc, d\}$. By bounding the former, we will be able to upper bound the absolute value of the coefficients of the inverse metric, which will allow us to upper bound the absolute value of the Christoffel symbols, and ultimately $|\Delta_h \tilde r_{p, v}(x)|$. 

\begin{proposition}
\label{prop:boundmetricterms}
Let $(B, h)$ be a complete and connected $d$-dimensional Riemannian manifold which is also a symmetric space. Let $p \in B$ be a fixed point and let us consider normal coordinates centered at $p$. Let $0 < r < i(B)$ and let $x \in B$ be such that $d_h(x, p) \leq r$. Assume that, at $p$, there exists some constant $\mathbf{K} \geq 1$ such that the coefficients of the Riemann curvature tensor in these coordinates satisfy $|R_{ijk}^l(p)| \leq \mathbf{K}$ for every $i, j, k, l \in \{1, \dotsc, d\}$. Then  
\[
\left|\sum_{n = 1}^{\infty} \frac{1}{(2n)!}h^{(2n)}_{ij}(0)\right| \leq\frac{2}{d} \sum_{n = 1}^{\infty}\frac{d^{3n}  \mathbf{K}^n (2r)^{2n} }{(2n+2)!},
\]
for every $i, j \in \{1, \dotsc, d\}$.
\end{proposition}
\begin{proof}
In \cref{prop:explicitexpresisonmetricderivatives} we showed that for every $n \in \bb{N}$
\[
h^{(2n)}_{ij}(0) = \frac{2^{2n+1} }{(2n+2)(2n+1)} \delta_{\alpha j} T^\alpha_{i, i_1, \dotsc, i_{2n}} x^{i_1}\dotsm x^{i_{2n}}, \quad \forall i, j \in \{1, \dotsc, d\}. 
\]
It is now straightforward to obtain a bound for this expression. As in normal coordinates centered at $p$ it holds that $|R_{ijk}^l(p)| \leq \mathbf{K}$ for every $i, j, k, l \in \{1, \dotsc, d\}$, it follows that for each $n \in \bb{N}$
\begin{equation}
\label{eq:tensorTbound}
|T^\alpha_{\beta i_1, \dotsc, i_{2n}} | \leq d^{n-1} \mathbf{K}^{n}, \quad \forall \alpha, \beta, i_1, \dotsc, i_{2n} \in \{1, \dotsc, d\}. 
\end{equation}
Since we are also assuming that $d_h(p, x) = |\log_p x|_h \leq r$, it follows that for every $n \in \bb{N}$
\begin{align*}
(2n+2)(2n+1)|h^{(2n)}_{ij}(0)| &\leq 2^{2n+1} d^{2n} d^{n-1} \mathbf{K}^{n} r^{2n}
\\&= 2^{2n+1} d^{3n-1} \mathbf{K}^n r^{2n},
\end{align*}
for every $i, j \in \{1, \dotsc, d\}$. 

This way, we proved that for every $i, j \in \{1, \dotsc, d\}$
\[
\left|\sum_{n = 1}^{\infty} \frac{1}{(2n)!}h^{(2n)}_{ij}(0)\right| \leq \frac{2}{d} \sum_{n = 1}^{\infty}\frac{d^{3n}  \mathbf{K}^n (2r)^{2n} }{(2n+2)!},
\]
finishing the proof. 
\end{proof}

\cref{prop:boundmetricterms} allows us to obtain bounds which can be as small as desired, simply by requiring $r$ to be small enough.

\begin{remark}
\label{rem:boundtaylor}
Under the conditions of \cref{prop:boundmetricterms}, let
\[
r \leq \frac{1}{d^{3/2} \mathbf{K}^{1/2}}. 
\]
Then  
\[
\left|\sum_{n = 1}^{\infty} \frac{1}{(2n)!}h^{(2n)}_{ij}(0)\right| < \frac{1}{d}. 
\]
\end{remark}
\begin{proof}
By simply inserting the bound on $r$ into the expression found in \cref{prop:boundmetricterms} we see that
\begin{align*}
\left|\sum_{n = 1}^{\infty} \frac{1}{(2n)!}h^{(2n)}_{ij}(0)\right| &\leq \frac{2}{d} \sum_{n = 1}^{\infty}\frac{d^{3n}  \mathbf{K}^n (2r)^{2n} }{(2n+2)!}
\\&\leq \frac{2}{d} \sum_{n = 1}^{\infty}\frac{d^{3n}  \mathbf{K}^n (2 \frac{1}{d^{3/2} \mathbf{K}^{1/2}})^{2n} }{(2n+2)!}
\\&= \frac{2}{d} \sum_{n = 1}^{\infty}\frac{4^n}{(2n+2)!}
\\&\leq \frac{1}{d}.
\end{align*}
\end{proof}

As we mentioned earlier, we need to find an upper bound on the absolute value of the coefficients of $h^{-1}$ at $x$ in normal coordinates, as they appear both in the expression of the Christoffel symbols and in the Laplacian of $\tilde r_{p, v}$. In order to do so, we will use two auxiliary lemmas.

\begin{lemma}
\label{lem:frobeniusnorm}
For every real matrix $A$ of size $n \times m$ it holds that 
\[
\max_{ij} |a_{ij}| \leq \norm{A}_F \leq \sqrt{nm} \max_{ij} |a_{ij}|,
\]
where $\norm{A}_F$ denotes the Frobenius norm of $A$, namely
\[
\norm{A}_F := \Tr(A^T A)^{1/2} = \Big(\sum_{ij} |a_{ij}|^2\Big)^{1/2}.
\]
\end{lemma}

The next lemma will allow us to obtain a Taylor series expansion for the coefficients of $h^{-1}$ in normal coordinates using what is known as a Neumann series approximation. We refer the reader to \cite[Theorem 4.20]{stewart1998matrix} for more details. 

\begin{lemma}
\label{prop:neumannseries}
Let $X$ be some matrix such that $\norm{X} < 1$ for some (submultiplicative) matrix norm $\norm{\cdot}$. Then $(\mathds{1}-X)$ is non-singular and 
\[
(\mathds{1} - X)^{-1} = \sum_{k = 0}^{\infty} X^k. 
\]    
\end{lemma}

As we mentioned previously, these lemmas allow us to prove the following result, which gives us a bound on the absolute value of the coefficients of $h^{-1}$ in normal coordinates. 

\begin{proposition}
\label{prop:boundinversemetric}
Let $(B, h)$ be a complete and connected $d$-dimensional Riemannian manifold which is also a symmetric space. Let $p \in B$ be a fixed point and let us consider normal coordinates centered at $p$. Let $0 < r < i(B)$ and let $x \in B$ be such that $d_h(x, p) \leq r$. Assume that, at $p$, there exists some constant $\mathbf{K} \geq 1$ such that the coefficients of the Riemann curvature tensor in these coordinates satisfy $|R_{ijk}^l(p)| \leq \mathbf{K}$ for every $i, j, k, l \in \{1, \dotsc, d\}$. Further assume that
\[
r \leq \frac{1}{d^{3/2} \mathbf{K}^{1/2}}.
\]
Then  for every $i, j \in \{1, \dotsc, d\}$,
\[
|h^{ij}(x) | \leq \frac{1}{1-L},
\]
where 
\begin{equation}
\label{eq:defconstanteL}
L := 2\sum_{n = 1}^{\infty}\frac{d^{3n}  \mathbf{K}^n (2r)^{2n} }{(2n+2)!}.
\end{equation}
\end{proposition}
\begin{proof}
First, we can write $h$ as 
\[
h(x) = \mathds{1} + X(x),
\]
where
\[
X_{ij}(x) = h_{ij}(x) - \delta_{ij} = \sum_{n = 1}^\infty \frac{h^{(2n)}_{ij}(0)}{(2n)!},
\]
for every $i, j \in \{1, \dotsc, d\}$. 

From \cref{prop:boundmetricterms} we can conclude that for every $i, j \in \{1, \dotsc, d\}$, 
\[
|X_{ij}(x)| = |h_{ij}(x) - \delta_{ij}| \leq \frac{L}{d},
\]
where $L$ is as defined in \cref{eq:defconstanteL}. Therefore, by \cref{rem:boundtaylor,lem:frobeniusnorm} we can conclude that 
\begin{equation}
\label{eq:normboundX}
\norm{X(x)}_F \leq d \max_{ij}|X_{ij}(x)| \leq L < 1.
\end{equation}

Next, we define the norm,
\[
\norm{X(x)}_2 := \sigma_{\max}(X(x)),
\]
where $\sigma_{\max}(X(x))$ denotes the greatest singular value of $X(x)$ . This norm is upper bounded by the Frobenius norm from \cref{lem:frobeniusnorm}, which can be written as 
\[
\norm{X(x)}_F = \sqrt{\sigma_1(X(x))^2 + \dotsc + \sigma_d(X(x))^2 } ,
\]
and so $\norm{X(x)}_2 \leq \norm{X(x)}_F \leq L < 1$. Thus, by \cref{prop:neumannseries} the metric $h$ is non-singular and 
\[
h^{-1}(x) = (\mathds{1} + X(x))^{-1} = \sum_{k = 0}^{\infty} (-1)^k X(x)^k,
\]
which implies that 
\begin{equation}
\label{eq:normboundinverseg}
\norm{h^{-1}(x)}_2 \leq \sum_{k= 0}^{\infty} \norm{X(x)}^k_2 \leq \frac{1}{1-L}.    
\end{equation}
Lastly, we can use the fact that for every $i, j \in \{1, \dotsc, d\}$, $|h^{ij}(x)| \leq \norm{h^{-1}(x)}_2$ to guarantee that 
\[
|h^{ij}(x)| \leq \norm{h^{-1}}_2 \leq \frac{1}{1-L}, 
\]
for every $i, j \in \{1, \dotsc, d\}$
\end{proof}

To finish, we will follow the same steps as in \cref{prop:boundmetricterms} to give an upper bound on $|\partial_a h_{ij}(x)|$ in normal coordinates centered at $p$, depending on $d_h(x, p)$ and the bound on the Riemann curvature tensor.  
\begin{proposition}
\label{prop:bounddermetric}
Let $(B, h)$ be a complete and connected $d$-dimensional Riemannian manifold which is also a symmetric space. Let $p \in B$ be a fixed point and let us consider normal coordinates centered at $p$. Let $0 < r < i(B)$ and let $x \in B$ be such that $d_h(x, p) \leq r$. Assume that, at $p$, there exists some constant $\mathbf{K} \geq 1$ such that the coefficients of the Riemann curvature tensor in these coordinates satisfy $|R_{ijk}^l(p)| \leq \mathbf{K}$ for every $i, j, k, l \in \{1, \dotsc, d\}$. Then  
\[
|\partial_a h_{ij}(x)| \leq \sum_{n = 1}^\infty \frac{n 2^{2n+2} d^{3n-2} \mathbf{K}^{n} r^{2n-1}}{(2n+2)!},
\]
for every $a, i, j \in \{1, \dotsc, d\}$. 
\end{proposition}
\begin{proof}
For every $a, i \in \{1, \dotsc, d\}$, we can write
\[
\partial_a x^i = \delta^i_a. 
\]
Therefore, if we take the derivative with respect to the $a$-th coordinate in \cref{eq:expansiontaylorterms} we obtain that for every $n \in \bb{N}$ and every $i, j \in \{1, \dotsc, d\}$
\[
(2n+2)(2n+1)\partial_a h^{(2n)}_{ij}(0) = 2^{2n+1} \delta_{\alpha j} T^\alpha_{i, i_1, \dotsc, i_{2n}}  x^{i_1, \dotsc, i_{2n}}_a
\]
where 
\[
x^{i_1,\dotsc,i_{2n}}_a := \delta^{i_1}_a x^{i_2} \dotsm x^{i_{2n}} + x^{i_1}\delta^{i_2}_a x^{i_3} \dotsm x^{i_{2n}} + \dotsm + x^{i_1} \dotsm x^{i_{2n-1}} \delta^{i_{2n}}_a.
\]
The norm of these elements can be bounded as 
\begin{align*}
|x^{i_1,\dotsc,i_{2n}}_a| &\leq |\delta^{i_1}_a x^{i_2} \dotsm x^{i_{2n}}| + |x^{i_1}\delta^{i_2}_a x^{i_3} \dotsm x^{i_{2n}}| + \dotsm + |x^{i_1} \dotsm x^{i_{2n-1}} \delta^{i_{2n}}_a|\\
&\leq 2n r^{2n-1}.
\end{align*}
Therefore, using this inequality together with \cref{eq:tensorTbound} we conclude that for every $n \geq 1$ and every $a, i, j \in \{1, \dotsc, d\}$,
\begin{align*}
(2n+2)(2n+1)|\partial_a h^{(2n)}_{ij}(0)| &= |2^{2n+1} \delta_{\alpha j} T^\alpha_{i, i_1, \dotsc, i_{2n}}  x^{i_1, \dotsc, i_{2n}}_a|
\\&\leq 2^{2n+1} |T^j_{i, i_1, \dotsc, i_{2n}}|  |x^{i_1, \dotsc, i_{2n}}_a|
\\&\leq 2^{2n+1} d^{2n-1} d^{n-1} \mathbf{K}^n 2nr^{2n-1}
\\&= 2^{2n+2} n d^{3n-2} \mathbf{K}^n r^{2n-1},
\end{align*}
and so 
\begin{align*}
|\partial_a h_{ij}(x)| &\leq \sum_{n = 1}^\infty \frac{1}{(2n)!} |\partial_ah^{(2n)}_{ij}(0)|
\\&\leq \sum_{n = 1}^\infty \frac{n 2^{2n+2} d^{3n-2} \mathbf{K}^{n} r^{2n-1}}{(2n+2)!},
\end{align*}
for every $a, i, j \in \{1, \dotsc, d\}$. 
\end{proof}

Having obtained a bound on $|\partial_a h_{ij}(x)|$ and $|h^{ij}(x)|$ for every $a, i, j \in \{1, \dotsc, d\}$, the bound on the Laplacian of $\tilde r_{p, v}$ follows almost immediately, as we will show in the next theorem. Note that the bound for the absolute value of the Laplacian can be made arbitrarily small by shrinking $r$. Nevertheless, the bound we obtain below is sufficiently good for the purpose of this work. 

\begin{proposition}[Upper bound on the absolute value of the Laplacian of $\tilde r_{p, v}$]
\label{thm:lowerboundlaplacianr}
Let $(B, h)$ be a complete and connected $d$-dimensional Riemannian manifold, which is also a symmetric space. Let $p \in B$ be a fixed point and let us consider normal coordinates centered at $p$. Let $v \in T_p B$, $|v|_h = 1$, $0 < r < i(B)$ and $x \in B$ be such that $d_h(x, p) \leq r$. Assume that, at $p$, there exists some constant $\mathbf{K} \geq 1$ such that the coefficients of the Riemann curvature tensor in these coordinates satisfy $|R_{ijk}^l(p)| \leq \mathbf{K}$ for every $i, j, k, l \in \{1, \dotsc, d\}$. Then  if 
\[
r \leq \frac{1}{d^{3/2} \mathbf{K}},
\]
it holds that
\[
|\Delta_h \tilde r_{p, v}(x)| \leq 24 d^{5/2}. 
\]  
\end{proposition}
\begin{proof}
In normal coordinates, the Christoffel symbols can be written as 
\begin{equation}
\label{eq:expressionChristoffel}
\Gamma^1_{kl}(x) = \frac{1}{2} h^{1m}(x) \Big(\partial_l h_{mk}(x) + \partial_k h_{ml}(x) - \partial_m h_{kl}(x)\Big),
\end{equation}
for every $k, l \in \{1, \dotsc, d\}$. Thus, we can rewrite \cref{eq:exprlaplacianr} as
\begin{equation}
\label{eq:explicitLaplacianr}
\Delta_h \tilde r_{p, v}(x) = -h^{ij}(x) \Gamma^1_{ij}(x) = -\frac{1}{2}h^{ij}(x)h^{1k}(x)(\partial_j h_{ki}(x) + \partial_i h_{kj}(x) - \partial_k h_{ij}(x)). 
\end{equation}

By \cref{prop:bounddermetric}, we know that for every $a, i, j \in \{1, \dotsc, d\}$, 
\[
|\partial_a h_{ij}(x)| \leq \sum_{n = 1}^\infty \frac{n 2^{2n+2} d^{3n-2} \mathbf{K}^{n} r^{2n-1}}{(2n+2)!}.
\]
Therefore, since
\[
r \leq \frac{1}{d^{3/2} \mathbf{K}},
\]
and since we are assuming without loss of generality that $\mathbf{K} \geq 1$, we see that
\begin{align}
\begin{split}
\label{eq:ineqvaluederivativemetric}
|\partial_a h_{ij}(x)| &\leq \sum_{n = 1}^\infty \frac{n 2^{2n+2} d^{3n-2} \mathbf{K}^{n} (\frac{1}{d^{3/2} \mathbf{K}})^{2n-1}}{(2n+2)!}
\\&= \sum_{n = 1}^\infty \frac{n 2^{2n+2} d^{-1/2} \mathbf{K}^{-n+1}}{(2n+2)!}
\\&\leq \frac{4}{\sqrt{d}} \sum_{n = 1}^\infty \frac{n 4^{n}}{(2n+2)!}
\\&\leq \frac{4}{\sqrt{d}}. 
\end{split}
\end{align}
Now, we can use \cref{prop:boundinversemetric} to conclude that for every $i, j \in \{1, \dotsc, d\}$
\[
|h^{ij}(x)| \leq \frac{1}{1-L},
\]
where $L$ is as defined in \cref{eq:defconstanteL}.

Again, using that 
\[
r \leq \frac{1}{d^{3/2} \mathbf{K}} \leq \frac{1}{d^{3/2} \mathbf{K}^{1/2}},
\]
we can conclude that 
\[
L \leq 2 \sum_{n = 1}^{\infty}\frac{4^{n}}{(2n+2)!} < \frac{1}{2}. 
\]
and so 
\begin{equation}
\label{eq:ineqvalueinversemetric}
|h^{ij}(x)| \leq \frac{1}{1-\frac{1}{2}} = 2.     
\end{equation}

Inserting \cref{eq:ineqvaluederivativemetric,eq:ineqvalueinversemetric} into \cref{eq:explicitLaplacianr} it follows that
\[
|\Delta_h \tilde r_{p, v}(x)| \leq 24 d^{5/2},
\]
concluding the proof.
\end{proof}

\subsection{Bounding the norm of the gradient in symmetric spaces}
\label{sec:boundinggradientsymmetric}

To finish, let us focus on the gradient of $\tilde r_{p, v}(x)$, for which we will obtain a lower and an upper bound on its norm. To this end, we will use the bounds on the metric derived in the previous section. 

\begin{lemma}
\label{AuxLemma2}
Let $(B, h)$ be a complete and connected Riemannian manifold of dimension $d$. Let $p \in B$ be a fixed point and let us consider normal coordinates centered at $p$. Let $0 < r < i(B)$ and let $x \in B$ be such that $d_h(x, p) \leq r$. In these coordinates, the gradient of $\tilde r_{p, v}$ can be written as
\[
\Grad{h} \tilde r_{p, v}(x) = h^{ij}(x) \delta_{ik} v^k \partial_j,
\]
where $v = v^i \partial_i  \in T_p B$ is the vector with respect to which $\tilde r_{p, v}$ is defined. 
\end{lemma}
\begin{proof}
In normal coordinates centered at $p$, $\tilde r_{p, v}(x) = \langle v, x\rangle = \delta_{ij} v^i x^j$. Recall that, in these coordinates, 
\[
\Grad{h} \tilde r_{p, v}(x) = h^{ij}(x) \frac{\partial \tilde r_{p, v}(x)}{\partial x^i} \partial_j. 
\]
Therefore, using that $\partial_a x^i = \delta_a^i$ for every $a, i \in \{1, \dotsc, d\}$ it holds that
\[
\partial_a \tilde r_{p, v}(x) = \delta_{ai} v^i,
\]
and so 
\[
\Grad{h} \tilde r_{p, v}(x) = h^{ij}(x) \delta_{ik} v^k \partial_j.
\]
\end{proof}

\begin{proposition}
\label{prop:lowerboundgradtilder}
Under the conditions of \cref{AuxLemma2}, assume that $(B, h)$ is symmetric, $|v|_h = 1$ and that there exists a constant $\mathbf{K} \geq 1$ such that for every $x \in B$, in normal coordinates centered at $x$, the coefficients of the Riemann curvature tensor are such that $|R^l_{ijk}(x)| \leq \mathbf{K}$ for every $i, j, k, l \in \{1, \dotsc, d\}$. If
\[
r \leq \frac{1}{d^{3/2} \mathbf{K}^{1/2}},
\]
then  
\[
\frac{1}{1 + L} \leq |\Grad{h} \tilde r_{p, v}(x)|^2_h  \leq \frac{1}{1 - L},
\]
where
\[
L = 2\sum_{n = 1}^{\infty}\frac{d^{3n}  \mathbf{K}^n (2r)^{2n} }{(2n+2)!}.
\]
\end{proposition}
\begin{proof}
In \cref{AuxLemma2} we showed that
\[
\Grad{h} \tilde r_{p, v}(x) = h^{ij}(x) \delta_{ik} v^k \partial_j,
\]
and so 
\[
|\Grad{h}\tilde r_{p, v}(x)|^2_h = h_{ij}(x) h^{ai}(x) \delta_{ab}v^b h^{cj}(x) \delta_{cd}v^d =  \delta_{jb} \delta_{cd} h^{cj}(x) v^b v^d = \sum_{ij} h^{ij}(x) v^i v^j. 
\]
This expression can be easily lower bounded,
\begin{align*}
|\Grad{h}\tilde r_{p, v}(x)|^2_h = \sum_{ij} h^{ij}(x) v^i v^j \geq \lambda_{\min}(h^{-1}(x)) |v|_h^2 = \frac{1}{\lambda_{\max}(h(x))},
\end{align*}
as $|v|_h^2 = 1$. 

Thus, using \cref{eq:normboundX} we can conclude that
\begin{equation}
\label{eq:boundmaxeigg}
\lambda_{\max}(h(x)) = \norm{h}_2 = \norm{\mathds{1} + X}_2 \leq \norm{1}_2 + \norm{X}_2 \leq 1 + \norm{X}_F \leq 1 + L,
\end{equation}
and so 
\[
|\Grad{h}\tilde r_{p, v}(x)|^2_h  \geq \frac{1}{1 + L}.
\]

Similarly, 
\[
|\Grad{h}\tilde r_{p, v}(x)|^2_h \leq \lambda_{\max}(h^{-1}(x)) = \norm{h^{-1}(x)}_2 \leq \frac{1}{1-L},
\]
where the last inequality follows from \cref{eq:normboundinverseg}.
\end{proof}

\subsection{Two auxiliary lemmas}
\label{sec:secpruebalemma}

We now state and prove the auxiliary lemmas \ref{lem:boundtermgradlaplacian} and \ref{AuxLemma3} used in \cref{sec:secondlyapunovfunction}.

\subsubsection{First auxiliary lemma}

\begin{lemma}
\label{lem:boundtermgradlaplacian}
Let $(B, h)$ be a compact and connected symmetric space. Assume that there exists some constant $\mathbf{K} \geq 1$ such that, for every $x \in B$, the coefficients of the Riemann curvature tensor in normal coordinates centered at $x$ satisfy $|R_{ijk}^l(x)| \leq \mathbf{K}$ for every $i, j, k, l \in \{1, \dotsc, \dim(B)\}$. Let 
\[
r \leq \min\Big\{\frac{i(B)}{2}, \frac{1}{72 \dim(B)^{5/2} \mathbf{K}}\Big\}.
\]
Let $p \in B$ be some fixed point of $B$, and let $v \in T_p B$ be a unit tangent vector at $p$. Let us consider 
\[
\tilde r_{p, v}(q) = \langle v, \log_p q\rangle,
\]
for every $q \in \mathcal{B}(r, p)$ it holds that
\[
\frac{10}{11} \leq |\textup{grad}_h \Tilde{r}_{p, v}(q)|^2_h \leq \frac{10}{9},
\]
and
\[
|\textup{grad}_h \Tilde{r}_{p, v}(q)|^2_h + \Tilde{r}_{p, v}(q)\Delta_h \Tilde{r}_{p, v}(q) \geq \frac{1}{2}.
\]
\end{lemma}

We will use the results from \cref{sec:boundingLaplaciansymmetricspaces,sec:boundinggradientsymmetric} to obtain bounds for the norm of the gradient and the Laplacian of $\Tilde{r}_{p, v}$, using the fact that it is defined on a symmetric space.

\begin{proof}[Proof of \cref{lem:boundtermgradlaplacian}]
As we are assuming that 
\[
r \leq \frac{1}{72 \dim(B)^{5/2} \mathbf{K}}, 
\]
we can apply \cref{thm:lowerboundlaplacianr,prop:lowerboundgradtilder} to $\Tilde{r}_{p, v}$ and conclude that for every $q \in \mathcal{B}(r, p)$ it holds that 
\begin{align*}
|\Delta_h\, \Tilde{r}_{p, v}(q)| &\leq 24 \dim(B)^{5/2},\\
\frac{1}{1+L} \leq |\textup{grad}_h \, \Tilde{r}_{p, v}(q)|^2_h &\leq \frac{1}{1-L}.
\end{align*}

Since $r \leq \frac{1}{2\text{dim}(B)^{3/2} \mathbf{K}}$, and $\mathbf{K} \geq 1$, we can bound $L$ as
\begin{equation}
\label{eq:boundonL}
L \leq 2\sum_{n = 1}^\infty \frac{\mathbf{K}^{-n}}{(2n+2)!} \leq 2\sum_{n = 1}^\infty \frac{1}{(2n+2)!} \leq \frac{1}{10},
\end{equation}
and so 
\[
\frac{10}{11} \leq \frac{1}{1+L} \leq |\textup{grad}_h \,\tilde r_{p, v}(q)|^2_h \leq \frac{1}{1-L} \leq \frac{10}{9},
\]
for every $q \in \mathcal{B}(r, p)$.

Lastly, note that $|\tilde r_{p, v}(q)| \leq d_h(p,q)$ for every $q \in \mathcal{B}(r, p)$. Thus, since $r \leq \frac{1}{72\dim(B)^{5/2}}$ by assumption, it holds that
\[
|\textup{grad}_h \,\tilde r_{p, v}(q)|^2_h + \tilde r_{p, v}(q) \Delta_h \tilde r_{p, v}(q) \geq \frac{10}{11} - \frac{1}{3} \geq \frac{1}{2},
\]
for every $q \in \mathcal{B}(r, p)$.
\end{proof}

\subsubsection{Second auxiliary Lemma}
\label{sec:proofsecondlemma}

In order to simplify the notation, let us first define the following auxiliary function. 
\begin{definition}
\label{def:auxfunctionH}
Let $(B, h)$ be a connected Riemannian manifold, and let $\tilde F$ be a smooth function on $B$. Let $y \in B$ be some fixed critical point of $\tilde F$. Using normal coordinates centered at $y$, let us define $H_y(x)$ for every $x$ outside the cut locus of $y$ as 
\[
H_y(x) := \delta^{ij} \partial_{jk} \tilde F(y)x^k \partial_i|_y = \nabla^2 \tilde F(y)[\log_y x] \in T_y B,
\]
where
\[
\nabla^2 \tilde F(y)[\log_y x] = \nabla_{\log_y x} \, \Grad{h}\tilde F|_y.
\]
\end{definition}

\begin{lemma}
\label{AuxLemma3}
Let $(B, h)$ be a compact and connected symmetric space. Assume that there exists some constant $\mathbf{K} \geq 1$ such that, for every $x \in B$, the coefficients of the Riemann curvature tensor in normal coordinates centered at $x$ satisfy $|R_{ijk}^l(x)| \leq \mathbf{K}$ for every $i, j, k, l \in \{1, \dotsc, \dim(B)\}$. Let $\tilde F$ be a smooth function on $B$ satisfying \cref{modassumption3.4,modassumption3.5.1}. Let 
\[
r \leq \min\Big\{\frac{i(B)}{2}, \frac{1}{72\dim(B)^{5/2}\mathbf{K}}\Big\}.
\]

Let $p \in B$ be a fixed saddle point of $\tilde F$ and let $v$ be a unit eigenvector of $\nabla^2 \tilde F(p)$ that corresponds to its minimum eigenvalue, $-\lambda_v$. Defining $\tilde{r}_{p, v}(q)$ with respect to $v$, i.e. 
\[
\tilde{r}_{p, v}(q) = \langle v, \log_p q\rangle,
\]
for every $q$ outside the cut locus of $p$, it holds that
\[
|\langle - \textup{grad}_h \tilde F(q) , \textup{grad}_h \tilde r_{p, v}(q) \rangle_h - \lambda_v \tilde r_{p, v}(q)| \leq  (24 \dim(B)^{5/2} A_2  + A_3) r^2,
\]
for every $q \in \mathcal{B}\big(r, p\big)$, where $A_2$ and $A_3$ are the Lipschitz constants of the gradient and the Hessian of $\tilde F$, respectively. 
\end{lemma}

Before we present the proof, let us state some auxiliary results. We start with a version of the mean value theorem in Riemannian manifolds. 
\begin{lemma}[{\cite[Proposition 10.55]{boumal2022intromanifolds}}]
\label{proposition10.55}
Let $(B, h)$ be a Riemannian manifold and let $f \in C^2(B)$. Let $x \in B$, $v \in T_xB$, and let $\gamma(t) = \exp_x(tv)$ be defined on $[0, 1]$, where $\exp_x: T_xB \to B$ denotes the exponential map at $x$. Also assume that there exists some $\kappa \geq 0$ such that, for all $t \in [0, 1]$, 
\[
\norm{(\textup{P}^\gamma_{x \to \gamma(t)})^{-1} \circ \nabla^2 f(\gamma(t)) \circ \textup{P}^\gamma_{x \to \gamma(t)} - \nabla^2 f(x)} \leq \kappa|tv|_h.
\]
Then  
\[
|\textup{P}^{-1}_v \Grad{h}f(\exp_x(v)) - \Grad{h}f(x) - \nabla^2f(x)[v]|_h \leq \frac{\kappa}{2}|v|^2_h,
\]
where $\textup{P}^\gamma_{x \to \gamma(t)}$ denotes the parallel transport from $x$ to $\gamma(t)$ along $\gamma$. 
\end{lemma}

Applying this result to our setting we can obtain the following bound.  
\begin{lemma}
\label{lem:normdifferencegradFH}
Let $y \in \tilde{\mathcal{S}}$ be some fixed saddle point of $\tilde F$. Then  for every $x$ outside the cut locus of $y$ it holds that 
\[
|\textup{P}^{-1}_{\log_y x}\,  \Grad{h}\tilde F(x) - H_y(x)|_h \leq \frac{A_3}{2} d_h(x, y)^2,
\]
where $A_3$ is the Lipschitz constant of the Hessian of $\tilde F$. 
\end{lemma}
\begin{proof}
The proof follows automatically by applying \cref{proposition10.55}, using the fact that $\Grad{h}\tilde F(y) = 0$ and that $\tilde F$ has a $A_3$-Lipschitz Hessian.
\end{proof}

Lastly, the following lemma allows us to bound the error incurred when transporting a vector field via parallel transport. 
\begin{lemma}[{\cite[Proposition 10.46]{boumal2022intromanifolds}}]
\label{lem:bounderrorparallel}
If $V$ is a continuously differentiable vector field on a Riemannian manifold $(B, h)$, for which there exists a constant $K$ such that for every $x \in B$ and every $s \in T_xB$ with $|s|_h = 1$
\[
|\nabla_s V(x)|_h \leq K. 
\]
Then, for any $x \in B$ and any $y$ outside the cut locus of $x$, it holds that 
\[
|\textup{P}^{-1}_{\log_x y} V(y) - V(x)|_h \leq K d_h(x,y).
\] 
\end{lemma}

With the above lemmas in mind, we can finally give a proof for \cref{AuxLemma3}. 
\begin{proof}[Proof of \cref{AuxLemma3}]
First, since parallel transport is an isometry, we can bound
\begin{align*}
|\langle - \textup{grad}_h \tilde F(q) , \textup{grad}_h \tilde r_{p, v}(q) \rangle_h - \lambda_v \tilde r_{p, v}(q)| &\leq
|\langle - \textup{P}_{\log_p q}\, H_p(q), \textup{grad}_h \tilde r_{p, v}(q)\rangle_h - \lambda_v \tilde r_{p, v}(q)| 
\\&\quad + |\langle \textup{grad}_h \tilde F(q)- \textup{P}_{\log_p q}\, H_p(q), \textup{grad}_h \tilde r_{p, v}(q)\rangle_h| 
\\&\leq |\langle - H_p(q), \textup{grad}_h \tilde r_{p, v}(p)\rangle_h - \lambda_v \tilde r_{p, v}(q)|  
\\&\quad + |\langle - H_p(q), \textup{P}_{\log_p q}^{-1}\, \textup{grad}_h \tilde r_{p, v}(q) - \textup{grad}_h \tilde r_{p, v}(p)\rangle_h | 
\\&\quad + |\langle \textup{grad}_h \tilde F(q)- \textup{P}_{\log_p q}\, H_p(q), \textup{grad}_h \tilde r_{p, v}(q)\rangle_h| 
\\&\leq |\langle - H_p(q), \textup{grad}_h \tilde r_{p, v}(p)\rangle_h - \lambda_v \tilde r_{p, v}(q)|  
\\&\quad + |H_p(q)|_h |\textup{P}_{\log_p q}^{-1}\, \textup{grad}_h \tilde r_{p, v}(q) - \textup{grad}_h \tilde r_{p, v}(p)|_h 
\\&\quad + | \textup{grad}_h \tilde F(q)- \textup{P}_{\log_p q}\, H_p(q)|_h |\textup{grad}_h \tilde r_{p, v}(q)|_h.
\end{align*}
Let us now study each of the summands on the right-hand side separately. First, using the expressions for the gradient of $\tilde r_{p, v}$ and $H_p$ in normal coordinates centered at $p$ (cf. \cref{def:auxfunctionH,AuxLemma2}), we can rewrite, denoting $q$ in coordinates as $(x^1, \dotsc, x^{\dim(B)})$, 
\begin{align}
\label{eq:expressionatz}
\begin{split}
\langle -H_p(q), \textup{grad}_h\,\tilde{r}_{p, v}(p)\rangle_h &= -(H_p(q))^i\left(\textup{grad}_h\, \tilde{r}_{p, v}(p)\right)^j \delta_{ij}
\\&= -\delta^{ik}\partial_{kl} \tilde{F}(0) x^l v^j \delta_{ij}
\\&= -\delta^{k}_j \partial_{kl} \tilde{F}(0) x^l v^j
\\&= - \partial_{kl} \tilde{F}(0) x^l v^k.
\end{split}
\end{align}
Recall that
\[
\nabla^2_h\, \tilde F(p)[v] = -\lambda_v v.
\]
Thus, using \cref{eq:expressionatz} we can conclude that
\begin{equation}
\label{eq:boundfirstsummand}
\langle -H_p(q), \textup{grad}_h\,\tilde{r}_{p, v}(p)\rangle_h - \lambda_v \tilde{r}_{p,v}(q) = \lambda_v\delta_{ij} x^i v^j - \lambda_v \delta_{ab} x^a v^b = 0. 
\end{equation}

For the second summand, we can bound the norm of $H_p(q)$ using \cref{rmk:LipschitzGrad} as
\begin{equation}
\label{eq:firsttermsecondsummand}
|H_p(q)|_h = |\nabla^2_h \tilde{F}(p)[\log_p q]|_h \leq A_2 r,    
\end{equation}
where $A_2$ is the Lipschitz constant of the gradient of $\tilde F$. Let us now bound the norm of 
\[
\textup{P}_{\log_p q}^{-1}\, \textup{grad}_h \tilde r_{p, v}(q) - \textup{grad}_h \tilde r_{p, v}(p).
\]
To this end, we will use \cref{lem:bounderrorparallel}. First note that for every $x \in \mathcal{B}(r, p)$ and every two unit vectors $a, b \in T_x B$ we can bound, using normal coordinates centered at $p$, 
\begin{align}
\big|\nabla^2_h \tilde r_{p, v}(x)[a, b]\big| \leq \norm{\nabla^2_h \tilde r_{p, v}(x)}_2 |a|_2 |b|_2,
\end{align}
where $|a|_2$ and $|b|_2$ denote the Euclidean norm of $a$ and $b$, which can be upper bounded as
\[
|a|^2_h \geq m |a|^2_2, 
\]
where $m$ denotes the minimum eigenvalue of $h(x)$ in normal coordinates centered at $p$. In particular, $m$ is equal to $\frac{1}{\norm{h^{-1}(x)}_2}$. Thus
\begin{align}
\big|\nabla^2_h \tilde r_{p, v}(x)[a, b]\big| \leq \norm{\nabla^2_h \tilde r_{p, v}(x)}_2 |a|_2 |b|_2  \leq \norm{\nabla^2_h \tilde r_{p, v}(x)}_2 \norm{h^{-1}(x)}_2.
\end{align}
Since 
\[
r \leq \frac{1}{72\dim(B)^{5/2}K}
\]
we can upper bound the norm of $h^{-1}$ using \cref{eq:normboundinverseg,eq:boundonL}:  
\begin{align*}
\norm{h^{-1}(x)}_2 \leq  \frac{1}{1-L} \leq \frac{10}{9}.
\end{align*}
The norm of the Hessian of $\tilde r_{p, v}$ can also be bounded, using \cref{lem:frobeniusnorm} by 
\begin{align*}
\norm{\nabla^2_h \tilde r_{p, v}(x)}_2 \leq \norm{\nabla^2_h \tilde r_{p, v}(x)}_F \leq \dim(B) \max_{ij} |(\nabla^2_h \tilde r_{p, v}(x))_{ij}|.
\end{align*}
where $(\nabla^2_h \tilde r_{p, v}(x))_{ij}$ denote the coefficients of the Hessian of $\tilde r_{p,v}$ evaluated at $x$ in normal coordinates centered at $p$. We can obtain an explicit expression for these (cf. \cite[Example 4.22]{lee2018introductionRiemannian}). Indeed, for every $i, j \in \{1, \dotsc, \dim(B)\}$, 
\[
(\nabla^2_h \tilde r_{p, v}(x))_{ij} = \partial_{ij} \tilde r_{p, v}(x) - \Gamma^k_{ij}(x) \partial_k \tilde{r}_{p, v}(x) = - \Gamma^k_{ij}(x) v_k, 
\]
where $v_k := \delta_{ki} v^i$. This way, since $v$ has norm one, and using the expression for the Christoffel symbols found in \cref{eq:expressionChristoffel} along with \cref{eq:ineqvaluederivativemetric,eq:ineqvalueinversemetric} we can bound
\[
|(\nabla^2_h \tilde r_{p, v}(x))_{ij}| \leq \dim(B) \max_k |\Gamma^k_{ij}(x)| \leq \frac{3}{2} \dim(B)^2 \max_{kl}|h^{kl}(x)| \max_{kln} |\partial_k h_{ln}(x)| \leq 12 \dim(B)^{3/2},
\]
which allows us to conclude that 
\[
\big|\nabla^2_h \tilde r_{p, v}(x)[a, b]\big| \leq 24 \dim(B)^{5/2}.
\]
Using this inequality, we apply \cref{lem:bounderrorparallel} and reach
\begin{align}
\label{eq:bounderrorparallel}
\begin{split}
|\textup{P}_{\log_p q}^{-1}\, \textup{grad}_h \tilde r_{p, v}(q) - \textup{grad}_h \tilde r_{p, v}(p)|_h &\leq \sup_{\substack{x \in \mathcal{B}(r, p)\\s \in T_x B\\|s|_h = 1}}|\nabla^2_h \tilde r_{p, v}(x)[s]|_h\, r
\\&= \sup_{\substack{x \in \mathcal{B}(r, p)\\a,b \in T_x B\\|a|_h = 1\\|b|_h=1}}\big|\nabla^2_h \tilde r_{p, v}(x)[a, b]\big|\, r
\\&\leq 24\dim(B)^{5/2} r.
\end{split}
\end{align}
Using this bound, along with \cref{eq:firsttermsecondsummand} we see that 
\begin{equation}
\label{eq:boundsecondsummand}
|H_p(q)|_h |\textup{P}_{\log_p q}^{-1}\, \textup{grad}_h \tilde r_{p, v}(q) - \textup{grad}_h \tilde r_{p, v}(p)|_h  \leq 24 \dim(B)^{5/2} A_2 r^2.
\end{equation}

For the third summand, we can bound it using \cref{lem:normdifferencegradFH,lem:boundtermgradlaplacian}; 
\begin{align}
\label{eq:boundthirdsummand}
\begin{split}
| \textup{grad}_h \tilde F(q)- \textup{P}_{\log_p q}\, H_p(q)|_h|\textup{grad}_h \tilde r_{p, v}(q)|_h &= |\textup{P}_{\log_p q}^{-1} \, \textup{grad}_h \tilde F(q)- H_p(q)|_h |\textup{grad}_h \tilde r_{p, v}(q)|_h  
\\&\leq A_3 r^2.
\end{split}
\end{align}
Putting \cref{eq:boundfirstsummand,eq:boundsecondsummand,eq:boundthirdsummand} together, we obtain
\[
|\langle - \textup{grad}_h \tilde F(q) , \textup{grad}_h \tilde r_{p, v}(q) \rangle_h - \lambda_v \tilde r_{p, v}(q)| \leq  (24 \dim(B)^{5/2} A_2  + A_3) r^2,
\]
concluding the proof.
\end{proof}